\documentclass[10pt,twoside]{report}
\usepackage{amsthm,amsmath,amssymb,color,manfnt,amscd}
\usepackage{subfigure}
\usepackage{graphicx}
\usepackage[alwaysadjust]{paralist}
\input{Qcircuit}
\usepackage[colorlinks=true]{hyperref} 
\usepackage{supertabular}
\usepackage{lscape}
\usepackage{algorithmic}
\usepackage{algorithm2e}
\usepackage[bf,bf,outermarks]{titlesec}
\renewcommand{\thechapter}{\arabic{chapter}}
\titleformat{\chapter}[display]
{\bfseries\Large}
{\filleft\MakeUppercase{\chaptertitlename} \Huge\thechapter}
{4ex}
{\titlerule \titlerule
\vspace{2ex}%
\Huge\filright}
[\vspace{2ex}%
\titlerule]

\newpagestyle{front}[]{
\headrule
\sethead[\usepage][][Chapter \thechapter: \chaptertitle]
{Chapter \thechapter: \chaptertitle}{}{\usepage}}

\newpagestyle{pref}[]{
\headrule
\sethead[\usepage][][\chaptertitle]
{\chaptertitle}{}{\usepage}}

\newpagestyle{main}[]{
\headrule
\sethead[\usepage][][Chapter \thechapter: \chaptertitle]
{\thesection. \sectiontitle}{}{\usepage}}

\newtheorem{theorem}{Theorem}
\newtheorem{lemma}[theorem]{Lemma}
\newtheorem{remark}[theorem]{Remark}

\newtheorem{corollary}[theorem]{Corollary}
\newtheorem{exampleX}[theorem]{Example}

\newtheorem{example}[theorem]{Example}
\newtheorem{fact}[theorem]{Fact}
\theoremstyle{definition}
\newtheorem{definition}[theorem]{Definition}
\newtheorem{defn}[theorem]{Definition}

\theoremstyle{remark}
\newtheorem{proposition}[theorem]{Proposition}

\newcommand{\nix}[1]{}

 \DeclareMathOperator{\ord}{ord}

 \DeclareMathOperator{\tr}{tr}
\DeclareMathOperator{\Tr}{Tr}

\DeclareMathOperator{\swt}{swt}
\DeclareMathOperator{\wt}{wt}
\DeclareMathOperator{\spann}{span}

\newcommand{\qr}{q\equiv\square\bmod{n}}
\newcommand{\qrx}{q\equiv\square\bmod}

\newcommand{\problem}{\noindent\textit{\textbf{Research Problems.}$\quad$}}

\def\H{\widetilde{H}}
\newcommand{\mbf}{\mathbf}

\DeclareMathOperator{\coeff}{coeff} 
\DeclareMathOperator{\ex}{ex} \DeclareMathOperator{\rk}{rk}
\newcommand{\ceil}[1]{\left\lceil #1\right\rceil}
\newcommand{\floor}[1]{\left\lfloor #1\right\rfloor}
\newcommand{\B}{\mathcal{BCH}}
\newcommand{\zero}{\mathbf{0}}
\newcommand{\G}{\mathcal{G}}
\newcommand{\C}{\mathbb{C}}
\newcommand{\F}{\mathbb{F}}
\newcommand{\N}{\mathcal{N}}
\newcommand{\Z}{\mathbb{Z}}

\newcommand{\h}{\mathcal{H}}
\newcommand{\sdual}{{\perp_s}}
\newcommand{\adual}{{\perp_a}}
\newcommand{\hdual}{{\perp_h}}
\newcommand{\acal}[2]{\langle #1\mid #2\rangle_a}

\newcommand{\ds}{\displaystyle}

\newcommand{\qccs}{quantum convolutional codes}
\newcommand{\one}{\mathbf{1}}
\newcommand{\RM}{{\mathcal{R}}}

\DeclareMathOperator{\dirlimit}{{\displaystyle\lim_{\longrightarrow}}}
\newcommand{\NN}{\mathbf{N}}
\newcommand{\scal}[2]{\langle #1\,|\,#2\rangle}

\makeatletter
\renewcommand\section{\@startsection {section}{1}{\z@}%
                                   {-3.5ex \@plus -1ex \@minus -.2ex}%
                                   {2.3ex \@plus.2ex}%
                                  {\normalfont\Large\bfseries}}
\renewcommand\subsection{\@startsection {subsection}{1}{\z@}%
                                   {-3.5ex \@plus -1ex \@minus -.2ex}%
                                   {2.3ex \@plus.2ex}%
                                  {\normalfont\large\bfseries}}
 \makeatother
\begin{document}
 \pagenumbering{roman}

\color{black}
\title{\bf On Quantum and Classical Error Control Codes:\\ Constructions and Applications}
\author{ \bigskip \\ {\LARGE \bf  Salah A. Aly } \\ \bigskip
 \vspace{5.0cm} \\
 \\  {\bf Major Subject:  Computer Science }}
\date{\bf \copyright~~All Rights Reserved \break \\  {\Large  2008 }  }
\maketitle

\newpage

{\bf On Quantum and Classical Error Control Codes: Constructions and Applications}

\bigskip

\noindent Parts of this work were submitted to the Department of Computer Science at  Texas A\&M University for the degree of doctoral of Philosophy on Fall 2007. Some other parts were added later without peer reviews. The document is reformatted.  Please report all typos or errors to the author with all due haste.

\vspace{5cm}

\bigskip
\bigskip
\bigskip
\bigskip
\bigskip

\begin{center}{\large \bf
To my family and teachers \nix{\bigskip \\ To every child, who
was born\\  of ignorant or poor parents.}}\end{center}
 \maketitle

\pagestyle{empty}
\headrule
 
\chapter*{Abstract}

\noindent It is conjectured that quantum computers are able to solve
certain problems more quickly than any deterministic or
probabilistic computer.  For instance, Shor's algorithm is able to
factor large integers in polynomial time on a quantum computer. A
quantum computer exploits the rules of quantum mechanics to speed up
computations. However, it is a formidable task to build a quantum
computer, since the quantum mechanical systems storing the
information unavoidably interact with their environment. Therefore,
one has to mitigate the resulting noise and decoherence effects to
avoid computational errors.

In this work, I study various aspects of quantum error
control codes~-- the key component of fault-tolerant quantum
information processing. I present the fundamental theory and
necessary background of quantum codes and construct many families of
quantum block and convolutional codes over finite fields, in
addition to families of subsystem codes. This work is
organized into these parts:

\begin{description}
\item[\emph{Quantum Block Codes.}]  After introducing the theory of quantum block codes, I establish conditions when BCH
codes are self-orthogonal (or dual-containing) with respect to
Euclidean and Hermitian inner products. In particular, I derive two
families of nonbinary quantum BCH codes using the stabilizer
formalism. I study duadic codes and establish the existence of
families of degenerate quantum codes, as well as  families of
quantum codes derived from projective geometries.

\item[\emph{Subsystem Codes.}]  Subsystem codes form a new class of
quantum codes in which the underlying classical codes do not need to
be self-or\-thogonal. I give an introduction to subsystem codes and
present several methods for subsystem code constructions. I derive
families of subsystem codes from classical BCH and RS codes and
establish a family of optimal MDS subsystem codes. I establish
propagation rules of subsystem codes and construct tables of upper
and lower bounds on subsystem code parameters.

\item[\emph{Quantum Convolutional Codes.}]  Quantum convolutional
codes are particularly well-suited for communication applications. I
develop the theory of quantum convolutional codes and give families
of quantum convolutional codes based on RS codes. Furthermore, I
establish a bound on the code parameters of quantum convolutional
codes -- the generalized Singleton bound. I develop a general
framework for deriving convolutional codes from block codes and use
it to derive families of non-catastrophic quantum convolutional
codes from BCH codes.
\item[\emph{Quantum and Classical LDPC Codes.}]  LDPC codes are a class of modern error control codes that  can be decoded
using iterative decoding algorithms. In this part, I derive classes of quantum LDPC codes based on finite geometries, Latin squares and combinatorial objects. In addition, I construct families of LDPC codes derived from classical BCH codes and elements of cyclotomic cosets.
\item[\emph{Asymmetric Quantum Codes.}] Recently, the theory of quantum error control codes has been extended to
include quantum codes over asymmetric quantum channels --- qubit-flip and
phase-shift errors may occur with different probabilities.  I derive families of  asymmetric quantum codes derived from
classical BCH and RS codes over finite fields. In addition, I derive a generic method to derive asymmetric quantum cyclic codes.
\end{description}

 \pagenumbering{arabic}
\pagestyle{pref}
\tableofcontents \listoffigures \listoftables

\chapter*{Acknowledgement}

This work
would not be a reality without the kind people whom I met during my
graduate studies.

\smallskip

I  thank my  advisor Dr. Andreas Klappenecker for his
support, guidance, and patience. He kindly introduced me to this
pioneering research.   Andreas taught me how to write high quality
research papers. Throughout countless emails, I cannot remember how
many times I thought my code constructions and paper drafts were
good enough, and he kindly challenged me to make them correct and
outstanding.

\smallskip

I  thank all my committee members: Dr. M. Suhail Zubairy, Dr.
Mahmoud El-Halwagi, Dr. Rabi Mahapatra, and Dr. Andrew Jiang. They
were all supportive and kind. A special gratefulness goes to my
mentor  Dr. El-Halwagi for his encouragement. He was always an
inspiration for me, whenever I faced tough times.

\smallskip

I  thank Zhenning Kong,  Pradeep K. Sarvepalli,  and Ahmad El-Guindy. I
 thank Martin Roetteler and Marcus Grassl for
their collaboration. I would like to thank Emina Soljanin and the
Mathematical Science Research Group at Bell Labs \& Alcatel-Lucent.
\smallskip

In a weighty remarkable document like this where the precision of
every word counts with caution;  remaining silent is too difficult.
During the last five years of my life, I was undoubtedly isolated
from people and life. Words can not describe how I felt. I would
like to thank my parents and extended family members for their
patience while I was away from them for many unseen years.  Absolutely, this work is
dedicated to them and I
also wish this work will ignite a light for my nephews and
all youth in my home city to encourage them to learn. Finally, from infancy until now,
I have always been blessed by the prayers of my relatives and
elders; I can now be sure that my work is not based on my cleverness
or intelligence.  I owe all praise, gratitude, and everything to Him.

\smallskip

\begin{quote}
Salah A. Aly\\ December 1, 2007.
\end{quote}


\setcounter{page}{1}
\pagestyle{front}

\chapter{Introduction}

Quantum computing is a relatively new interdisciplinary field that
has recently attracted many researchers from physics, mathematics,
and computer science.  The main idea of quantum computing is to
utilize the laws of quantum physics to perform fast computations.
Quantum information processing can be beneficial in numerous
applications, such as secure key exchange or quick search. Arguably,
one of the most attractive features is that quantum algorithms are
conjectured to solve certain computational problems exponentially
faster than any classical algorithm.  For instance, Shor's quantum
algorithm can factor integers faster than any known classical
algorithm.

Quantum information is represented by the states of quantum mechanical
systems. Since the information-carrying quantum systems will
inevitably interact with their environment, one has to deal with
decoherence effects that tend to destroy the stored information.
Hence, it is infeasible to perform quantum computations without
introducing techniques to remedy this dilemma. One method is to apply
fault-tolerant operations that make the computations permissible under
a certain threshold value. These fault-tolerant techniques employ
quantum error control codes to protect quantum information.

The main contribution of this work  is the development of novel
techniques for quantum error control, including the construction of numerous
quantum error control codes to guard quantum information.\\

\section{Background} The state space of a discrete quantum mechanical system
is given by a finite-dimensional Hilbert space, namely by a finite-dimensional
complex vector space that is equipped with the standard Hermitian inner
product.  The states of the quantum system are assumed to be vectors of unit
length in the induced norm. Any quantum mechanical operation other than a
measurement is given by a unitary linear operation.

For quantum information processing, one chooses a fixed orthonormal
basis of the state space of the quantum mechanical system, called the
computational basis. The basis vectors represent classical information
that is processed by the quantum computer. To fix ideas, consider a
quantum system with two-dimensional state space $\C^2$. The basis
vectors
$$ v_0=\left( \begin{array}{c} 1 \\ 0 \end{array} \right), \texttt{ }
v_1=\left( \begin{array}{c} 0\\ 1 \end{array}\right)$$ can be used to
represent the classical bits 0 and 1. As the indices of the basis
vectors can be difficult to read, it is customary in quantum
information processing to use Dirac's ket notation for the basis
vectors; namely, the vector $v_0$ is denoted by $\ket{0}$ and the
vector $v_1$ is denoted by $\ket{1}$.  Therefore, any possible state
of such a two-dimensional quantum system is given by a linear
combination of the form
$$a \ket{0}+b\ket{1}=\left(\begin{array}{c} a\\ b \end{array}\right),
\quad \mbox{ where } a, b \in \C \mbox{ and } |a|^2+|b|^2=1,$$ as any
vector of unit length is a possible state.  One refers to the state
vector of a two-dimensional quantum system as a quantum bit or qubit.

The superposition or linear combination of the basis vectors $\ket{0}$ and
$\ket{1}$ of a quantum bit is one marked difference between classical and
quantum information processing. One can measure a quantum bit in the
computational basis. Such a measurement of a quantum bit in the state
$a\ket{0}+b\ket{1}$ leaves the quantum bit with a probability of $|a|^2$ in
state $\ket{0}$ and with probability $|b|^2$ in state $\ket{1}$. Furthermore,
the outcome of this probabilistic operation is recorded as a measurement
result.

In quantum information processing, the operations manipulating quantum bits
follow the rules of quantum mechanics, that is, an operation that is not a
measurement must be realized by a unitary operator.  For example, a quantum bit
can be flipped by a quantum NOT gate $X$ that transfers the qubits $\ket{0}$
and $\ket{1}$ to $\ket{1}$ and $\ket{0}$, respectively. Thus, this operation
acts on a general quantum state as follows.
$$X(a\ket{0}+b\ket{1})=a \ket{1}+ b \ket{0}.$$ With respect to the
computational basis, the quantum NOT gate $X$ is represented by the
matrix $\left( \begin{array}{cc} 0 &1\\ 1&0\\
\end{array}\right)$. Other popular operations include the phase flip $Z$, the combined bit and
phase-flip $Y$, and the Hadamard gate $H$, which are represented with respect
to the computational basis by the matrices
$$Z=\left( \begin{array}{cc} 1 &0\\ 0&-1\\
\end{array}\right), Y=\left(
\begin{array}{cc} 0 &-i\\ i&0\\ \end{array}\right), H=\frac{1}{\sqrt{2}}\left(
\begin{array}{cc} 1 &1\\ 1&-1\\ \end{array}\right).$$

The state space of a joint quantum system is described by the tensor
product of the state spaces of its parts. Consequently, a quantum
register of length $n$, which is by definition a combination of $n$
qubits, can be represented by the normalized complex linear
combination of the $2^n$ mutually orthogonal basis states in
$\C^{2^n}$, namely as a linear combination of the vectors
$$\ket{\psi}=\ket{\psi_1} \otimes \ket{\psi_2} \otimes ...
\otimes \ket{\psi_n}=\ket{\psi_1\psi_2...\psi_n}
\mbox{  where } \ket{\psi_i} \in \{\ket{0},\ket{1}\}.$$

Operations acting on two (or more) quantum bits include the controlled
not operation CNOT, which realizes the map
$$ \ket{00}\mapsto \ket{00}, \ket{01}\mapsto \ket{01}, \ket{10}\mapsto
\ket{11}, \ket{11}\mapsto \ket{10}.$$
In the computational basis, the CNOT operation is described by the matrix
$$
CNOT= \left(\begin{array}{cccc}1&0&0&0\\0&1&0&0\\0&0&0&1\\0&0&1&0
        \end{array}\right).
$$

\section{Quantum Codes}
Quantum error control codes like their classical counterparts are
means to protect quantum information against noise and decoherence.
Quantum codes can be classified into additive or nonadditive codes.
If the code is defined based on an abelian subgroup (stabilizer),
then it is called an additive (stabilizer) code. The structure and
construction of additive codes are well-known. Additive codes are
also defined over a vector space, therefore addition (or
subtraction) of two codewords is also a valid codeword in the
codespace~\cite{calderbank98}.

Shor's demonstrated the first quantum error correcting
code~\cite{shor95}. The code encodes one qubit into nine qubits, and
is able to correct for one error and detect two errors. Shortly
Gottesman~\cite{gottesman97}, Steane~\cite{steane96}, and
Calderbank, Rains, Shor, Sloane~\cite{calderbank98} developed the
stabilizer codes and the problem transferred to finding classical
additive codes over the finite fields~$\F_q$ and $\F_{q^2}$ that are
self-orthogonal or dual-containing with respect to the Euclidean or
Hermitian inner products, respectively. Since then, many families of
quantum error-correcting codes have been constructed, also, bounds
on the minimum distance and code parameters of quantum codes have
been driven. In~\cite{calderbank98}, a table of upper bounds on the
minimum distance of binary quantum codes has been given. Moreover,
propagation rules to drive new quantum codes from existing quantum
codes have been shown.

Nonbinary quantum codes, inspired by their classical counterparts, might be
useful for some applications.  For example, in quantum concatenated codes, the
underline finite field would be $\F_{2^m}$, which is useful for decoding
operations~\cite{bennett96}. In this work  I derive both binary and
nonbinary quantum block and convolutional codes in addition to subsystem codes.
The foundation materials that will be used in the next chapters are presented
in Chapters I, II, and III.

In contrast, the nonadditive codes do not have  uniform structure and are not
equivalent to any nontrivial additive codes. Knill showed in~\cite{knill96}
that nonadditive codes can give better performance. As far as I know, the
literature lacks a comparative analytical study among these two classifications
of codes. Roychowdhury and Vatan~\cite{roychowdhury98} established sufficient
conditions on the existence of nonadditive codes, introduced strongly
nonadditive codes, and proved Gilbert-Varshimov bounds for these codes.
Furthermore, they also showed that the nonadditive codes that correct $t$
errors satisfy asymptotically rate $R \geq 1- 2H_2(2t/n)$. Arvind el al.
developed the theory of non-stabilizer quantum codes from Abelian subgroup of
the error group~\cite{arvind02}.

There is also a different approach, to design quantum codes, that is
known as entangled-assisted quantum codes.  Designing quantum codes
by entanglement property assumes a shared entangled qubits between
two parties (sender and receiver). Some progress in this theory and
constructing quantum codes using entanglement are shown in
~\cite{hsieh07,brun06}.

\section{Problem Statement}
In this section, I will state some of the open research problems
that I have been investigating. My goal is to construct good
families of quantum codes to protect quantum information against
noise and decoherence. I will construct quantum block and
convolutional codes in addition to subsystem codes.

\medskip

\noindent \textbf{Quantum Block Codes.} A well-known method of
constructing quantum error-correcting codes is by using the
stabilizer formalism. Let $S$ be a stabilizer abelian subgroup of an
error group $G$, and $C(S)$ be a  subgroup in $G$ that contains all
elements which commute with every element in $S$, ((i.e. $S
\subseteq C(S)$, An expanded explanation is provided in
Chapter~\ref{ch_QBC_basics}). If we also assume  that $S$ and $C(S)$
can be mapped to a classical code $C$ and its dual $C^\perp$,
respectively. Then  a quantum code $Q$ exists, stabilized by the
subgroup $S$ as shown by the independent work of Calderbank and Shor
\cite{calderbank96} and Steane \cite{steane96b}. The quantum code
$Q$ is a $q^k$ dimensional subspace of the Hilbert space $C^{q^n}$,
and it has parameters $[[n,k,d]]_q$ with $k$ information logic
qubits and $n$ encoded qubits. The code $Q$ is able to correct all
errors up to $\lfloor (d-1)/2 \rfloor$, see
Chapter~\ref{ch_QBC_basics} for more details. A quantum code is
called impure if there is a vector in $C$ with weight less than any
vector in $(C^\perp \backslash C)$; otherwise it is called pure.
Pure quantum codes have been constructed based on good classical
codes (i.e. codes with high minimum distance). However, the
construction of  impure quantum codes from classical codes with poor
distances has not been widely
 investigated. Surprisingly, one can construct good impure quantum codes based on bad
classical codes (i.e. codes with low minimum distance).

 \problem The goals of my research in quantum block codes are to:
\begin{compactenum}[a)]

\item
  Construct families of quantum block codes over finite fields based on
self-orthogonal (or dual-containing) classical codes. Determine
whether there are families of impure quantum codes such that the
stabilizer has many vectors with small weights and these families
are not extended codes.
\item
Study the probability of undetected errors for some families of stabilizer codes and
search for codes with undetected error probability that approaches zero.
\item
Determine whether  stabilizer codes be constructed from polynomial
and Euclidean geometry codes since these codes have the feature of
majority list decoding, and what are the conditions  that will
determine whether these codes will be self-orthogonal (or
dual-containing)?
\item
Analyze the method by which a family of stabilizer codes uses fault-tolerant
quantum computing. What is its threshold value?  Can it be improved? And if so,
what assumptions must be made to improve it?
\item
Determine whether quantum stabilizer codes, in which errors have
some nice structure, can correct beyond the minimum distance,  since
we know that  fire  and burst-error classical codes can correct
errors beyond half of their minimum
distance.\\
\end{compactenum}
\noindent \textbf{Subsystem Codes.} Subsystem codes are a relatively
new construction of quantum codes based on isolating the active
errors into two subsystems. Hence, a quantum code $Q$ is a tensor
product of two subsystems $A$ and $B$, i.e. $Q=A \otimes B$. The
dimension of the subsystem A is $q^k$ while the dimension of the
subsystem $B$ is $q^r$; the code $Q$ has parameters $[[n,k,r,d]]_q$.
A special feature of subsystem codes is that any classical additive
code $C$ can be used to construct a subsystem code. One should
contrast this with stabilizer codes, where the classical codes are
required to satisfy self-orthogonality (or dual-containing)
conditions. Many interesting problems have not yet been addressed on
subsystem codes such as bounds, weight enumerators, encoding
circuits and families of subsystem codes. Also, there are no tables
of upper bounds, lower bounds, or best known subsystem codes.

\problem The goals of my research in subsystem codes are to:
\begin{compactenum}[a)]
\item

Investigate properties of subsystem codes and find good subsystem codes with high
rates and large minimum distances. How do stabilizer codes compare with subsystem
 codes with $r \geq 1$? How are families of subsystem codes constructed based on classical
codes?

\item
Analyze the conditions under which classical codes will give us
subsystem codes with large gauge qubits $r \geq 1$. Assuming we have
RS or BCH codes with length $n$ and designed distance $\delta$ that
can be used to construct subsystem codes. How much does the minimum
distance for subsystem RS or BCH codes increase, if $k$ and $r$ are
exchanged?

\item

Implement the linear programming and Gilbert-Varshimov bounds, using
Magma computer algebra,  to derive tables of upper bounds, lower
bounds, and best known codes of subsystem codes over finite fields.

\item  Determine what  the efficient encoding and decoding circuits look like for subsystem
codes, and whether we can   draw an
encoding circuit for a subsystem code from a given encoding circuit of a stabilizer code.\\
\end{compactenum}
\noindent \textbf{Quantum Convolutional Codes.} Quantum
convolutional codes (QCC's) seem to be useful for quantum
communication because they have online encoder and decoder
algorithms (circuits). One main property of quantum convolutional
codes is the delay operator where the encoder has some memory set.
However, quantum convolutional codes still have  not been  studied
extensively. Furthermore, many interesting and open questions remain
regarding the properties and the usefulness of quantum convolutional
codes. At this time, it is not known whether quantum convolutional
codes offer a decisive advantage over quantum block codes, since we
do not yet have a well-defined formalism of quantum convolutional
codes. For example, the CSS construction, projectors, and
non-catastrophic encoders are not clearly defined for quantum
convolutional codes. In other words, except for the work by
Ollivier~\cite{ollivier04}, there are only some examples of quantum
convolutional codes with $1/3$, $1/4$, and $1/n$ code rates.

\problem The goals of my research in quantum convolutional codes are
to:

\begin{compactenum}[a)]
\item
Formulate a stabilizer formalism for convolutional codes that is similar to the
well-defined stabilizer formalism of quantum block codes, and to construct
families of quantum convolutional codes based on classical convolutional codes.

\item
Determine whether it is possible to construct quantum convolutional
codes, given RS and BCH codes with length $n$ and designed
distance $\delta$, and to determine under which conditions these
codes can be mapped to self-orthogonal convolutional codes, what the
restrictions are on $\delta$, and whether parameters of quantum
convolutional codes can be bounded using a generalized Singleton
 bound.

\item
Design online efficient encoding and decoding circuits for quantum
convolutional codes.

\item
Establish whether a scenario for quantum convolutional codes, where the errors
can be isolated into subsystems, exists that is similar to error
avoiding codes (subsystem codes) that can be constructed from block codes.\\
\end{compactenum}

\noindent \textbf{Quantum and Classical LDPC Codes.} Low-density parity check (LDPC) codes are a significant class of
classical codes with many applications. Several good LDPC codes have
been constructed using random, algebraic, and finite geometries
approaches, with containing cycles of length at least six in their
Tanner graphs. However, it is impossible to design a self-orthogonal
parity check matrix of an LDPC code without introducing cycles of
length four.

\problem The goals of my research in subsystem codes are to:
\begin{compactenum}[a)]
\item
Construct many families of quantum LDPC codes, and study their prosperities.   Will the performance of classical LDPC codes be the same as performance of quantum LDPC codes over asymmetric or symmetric quantum channels?

\item What are the conditions for  classical LDPC codes to have less cycles of length four and still give us good quantum LDPC codes. \item Study the decoding aspects of quantum LDPC codes.\\
    \end{compactenum}

    \noindent \textbf{Asymmetric Quantum Codes.}
    Recently, the theory of quantum error control codes has been extended to
include quantum codes over asymmetric quantum channels --- qubit-flip and
phase-shift errors may occur with different probabilities.  I derive families of  asymmetric quantum codes derived from
classical BCH and RS codes over finite fields. In addition, I derive a generic method to derive asymmetric quantum cyclic codes.

\section{Work Outline}
Some of the research problems stated in the previous subsection are
completely solved up on this work, some are left as an extension
work, and obviously some will remain open. In this work I construct
many families of quantum error control codes and study their
properties. The work is structured into these parts and the
main results are stated  as follows.
\begin{enumerate}[I)]
\item
In part I, Chapters~\ref{ch_QBC_basics}, \ref{ch_QBC_BCH},
\ref{ch_QBC_Qduadic}, \ref{ch_QBC_QPRM}, I study families of quantum
block codes constructed using the CSS construction. I establish
conditions when nonbinary primitive BCH codes are dual-containing
with respect to Euclidean and Hermitian products; consequently I
derived families of quantum BCH codes. Also, I compute the dimension
and bound the minimum distance of BCH codes under some restricted
conditions. I derive impure quantum codes with remarkable minimum
distance based on duadic codes. Also, I construct one family of
quantum codes from project geometry codes.

\item In part II, Chapters~\ref{ch_subsys_basic}, \ref{ch_subsys_construction}, \ref{ch_subsys_families}, \ref{ch_subsys_rules_tables}, I study families of subsystem codes. I give various methods
for  subsystem code constructions, and, in addition, I derive
families of subsystem codes based on  BCH and RS codes. I generate
tables of  upper and lower bounds of subsystem code parameters.
Finally, I trade the dimensions of subsystem code parameters and
present a fair comparison between stabilizer and subsystem codes.

\item In part III, Chapters~\ref{ch_QCC_bounds}, \ref{ch_QCC_RS}, \ref{ch_QCC_BCH}, I study  quantum convolutional codes. I
establish the stabilizer formalism of quantum convolutional codes
using the direct limit, and  I derive the generalized Singleton
bound for quantum convolutional codes. Finally, I demonstrate two
families of quantum convolutional codes derived from RS and BCH
codes.

\item In part IV,  I derive classes of quantum LDPC codes based on finite geometries, Latin squares and combinatorial objects. In addition, I construct families of LDPC codes derived from classical BCH codes and elements of cyclotomic cosets.
\item In part V,  Recently, the theory of quantum error control codes has been extended to
include quantum codes over asymmetric quantum channels --- qubit-flip and
phase-shift errors may occur with different probabilities.  I derive families of  asymmetric quantum codes derived from
classical BCH and RS codes over finite fields. In addition, I derive a generic method to derive asymmetric quantum cyclic codes.
\end{enumerate}


\pagestyle{main}

\chapter{Background}\label{ch_background}

In this chapter I will present background material and terminologies
of classical coding theory and quantum error control codes that are
necessary to assist the reader in understanding  the families of
quantum codes presented in the following chapters. I will also cite
previous work on quantum error control codes that is relevant to this work.

\smallskip

The power of quantum computers comes from their ability to use
quantum mechanical principles such as entanglement, interference,
superposition, and measurement. These fascinating natural types of
computers can solve certain problems exponentially faster than any
known classical computers. Some well known examples of problems that
can be solved are factorization of large primes and
searching~\cite{Chang00}. It was recently demonstrated that  quantum
key distribution schemes can be used to exchange private keys over
public communication channels.

\smallskip

Finding problems that can be solved by quantum computers is an
interesting research subject, yet a difficult task. With the
exception of a few problems, it is not well-known what types of
problems that quantum computers can solve exponentially fast.
However, there is no doubt about the usefulness and powerfulness of
quantum computers. The most difficult problem associated with
building quantum computers is isolating the~\emph{noise}. The
term~\emph{noise} can be defined as quantum errors that are caused
by decoherence from an environment.

\section{Classical Coding Theory}

Let $q$ be a power of a prime $p$. Let $\F_q$ denote a finite field with $q$
elements. If $q=p^m$ then
\begin{eqnarray}\F_q^n[x]= \{f(x) \in \F_q[x] \mid deg f(x) < m
\},\end{eqnarray}
where $f(x)$ is a polynomial of max degree $m$, and $\F_q[x]$ is a
polynomial ring. If $q=p$, then the field has the integer elements
$\{0,1,...,p-1 \}$ with the normal addition and multiplication
operations module $p$. The addition and multiplication of elements
in $\F_q$, where $q=p^m$, are done  by adding and multiplying in
$\F_p[x]$ module a known irreducible polynomial  $P_m(x)$ in
$\F_p[x]$ of degree $m$. A detailed survey on finite fields is
reported in~\cite{huffman03}. Let $\beta$ be an element in $\F_q$.
The  smallest positive integer $\ell$ such that $\beta^\ell =1$ is
called the order
 of $\beta$. The order of a finite field is the number of elements on it, i.e.,
 the cardinality of the field. If $\alpha \in \F_q$ and the order of $\alpha$ is $q-1$,
 then $\alpha$ is called a primitive element in $\F_q$. In this case, all
 nonzero elements in $\F_q$ can be represented in $q-1$ consecutive powers of a
 primitive element
 $ \{1, \alpha, \alpha^2,...,\alpha^{q-1}, \alpha^{q}=\alpha, \alpha^\infty=0 \}.$

\medskip

\noindent  \textbf{Linear Codes.} Let $\F_q^n$ be a vector space
with dimension $n$ and size $q^n$. A code $C$ is a subspace of the
vector space $\F_q^n$ over $\F_q$. Every linear code is  generated
by a generator matrix $G$ of size $k \times n$. Let $u$ be a vector
in $\F_q^k$, then \begin{eqnarray}C=\{u G \mid ~~~\forall ~~~u \in
\F_q^k \}, \end{eqnarray} where $G$ is a generator matrix of size $k
\times n$ over $\F_q$. The $k$ basis vectors of $G$ are the basis
for the code $C$. The code $C$ has $q^k$ codewords, the size of $C$.
We can also generate a dual matrix $H$ of size $(n-k) \times n$ from
the matrix $G$ such that
\begin{eqnarray}G H^{T}=0.
\end{eqnarray}
 The $n-k$ rows
of $H$ are also linearly independent.   $H$ is called the parity
check matrix of $C$. We say that $v$ is a valid codeword in $C$, if
and only if, $H v^T=0.$ The parity check matrix $H$ can also be used
to define the $C$ as \begin{eqnarray}C= \{ v \in \F_q^n \mid
Hv^T=0\}.\end{eqnarray}
The dual of a code $C$ is denoted by $C^\perp$ and is defined by

\begin{eqnarray} C^\perp =\{w \mid w \in \F_q^n, ~~w.v=0 ~~ \forall~~~ v \in C
\},\end{eqnarray} where $w.v$ is the Euclidean inner product between
two vectors in $\F_q$. If we assume that $w=(w_1,w_2,\ldots,w_n)$
and $v=(v_1,v_2,\ldots,v_n)$ then $w.v= \sum_{i=1}^n w_iv_i.$
We can say that $w$ is orthogonal to $v$ if their inner product vanishes, i.e.,
$w.v=0$.
If $C^\perp \subseteq C$, then the code is called dual-containing.
It means that all codewords in $C^\perp$ lie in $C$ as well. Also,
if all codewords in $C$ lie in $C^\perp$, then the code $C$ is
called self-orthogonal, i.e., $C \subseteq C^\perp$. Self-orthogonal
or dual-containing codes are of particular interest to our work
because they are used to derive quantum codes. If $C=C^\perp$, then
the code is called self-dual. If $[n,k,d]_q$ are parameters of a
code $C$, then $[n,n-k,d]_q$ are parameters of the dual code
$C^\perp$.

\medskip

\noindent \textbf{Minimum Distance and Hamming Weight.} Some
important criteria's of a code are the weight and minimum distance
among its codewords. The weight of a codeword $v$ in a code $C$ is
the number of nonzero positions (coordinates) in $v$. Let $w$ and
$v$ be two codewords in a code $C \subseteq \F_q^n$. The Hamming
distance between $w$ and $v$ is given by the number of positions in
which $w$ and $v$ differ. It is  weight of the difference codeword.
\begin{eqnarray}
d(w,v)=\mid \{i \mid 1 \leq i \leq n, w_i \neq v_i \} \mid =
\wt(w-v).
\end{eqnarray}
The minimum distance of a code is the smallest distance between two different
codewords in $C$. If $C \subseteq \F_q^n$, then the minimum distance $d$ is the
minimum weight of a nonzero codeword.

The code performance can be measured by its rate, decoding and
encoding complexity, and minimum distance. If the minimum distance
is large, the code has a better ability to correct errors. Given a
minimum distance $d$ of  a code $C$, the maximum number of errors
$t$ that can be corrected by $C$ is $t = \lfloor  (d-1)/2\rfloor,$
where the errors are distributed in random positions. The rate of a
a code $C$ is given by the ratio of its dimension to its length,
i.e., $k/n$. The linear code parameters are given by $[n,k,d]_q$ or
$(n,q^k,d)_q$.

Let $A_i$ and $B_i$ be the number of codewords in $C$ and $C^\perp$ of weight
$i$, respectively. The list of codewords $A_i$ and $B_i$ are called the weight
distributions of $C$ and $C^\perp$, respectively. If $C$ is a code with
parameters $[n,k,d]$ over $\F_q$, then it is a well-known fact that
$A_0+A_1+\ldots+A_n=q^k$. Furthermore, $A_0=1$ and $ A_1=A_2=\ldots=A_{d-1}
=0$.

\medskip

\noindent \textbf{Error Corrections.} Now assume a codeword $v \in
C$ is sent over a noise communication channel. Let $r=v+e$ be the
received vector where $e$ is the added noise. Then one can use the
matrix $H$ to perform error correction and detection capabilities of
the code $C$.
\begin{eqnarray}s=rH^T=(v+e)H^T=eH^T.\end{eqnarray}
Based on the value of the syndrome $s$, one might be able to correct
the received codeword $r$ to the original codeword $v$,
see~\cite{huffman03,macwilliams77} for further details.

\smallskip

\subsection{Bounds on the Code Parameters}
The relationship between the code parameters $n,k, d$ and $q$ has
been well studied in order to compare the performance of codes. The
minimum distance $d$ is used to measure the ability of a code to
correct errors. Good error correcting codes are designed with a
large  minimum distance $d$ and as large a number of codewords $q^k$
as possible, for a given length $n$ and alphabet size $q$. So, it is
crucial to establish upper and lower bounds on the code parameters.
There have been many upper bounds on the code parameters such as
Singleton, Hamming and sphere packing, and linear programming
bounds. Also, there have been some lower bounds such as
Gilbert-Varshamov bound.
\medskip

\noindent \textbf{Singleton Bound and MDS Codes.} Given a code $C$
with parameters $[n,k,d]_q$ for $d \leq n$, the classical Singleton
bound can be stated as
\begin{eqnarray}
q^k \leq q^{n-d+1}.
\end{eqnarray}
If $C$ is a linear code, then $k \leq n-d+1$. Codes that attain the
Singleton bound with equality are called Maximum Distance Separable
(MDS) codes. MDS codes are also optimal codes. This class of codes
is of particular interest because it has the maximum distance that
can be achieved among all other codes with the same length,
dimension, and alphabet size. No other codes of length $n$ and size
$q^k$ have larger minimum distances than MDS codes, with the same
parameters. Also, it is known
 that the dual of a classical MDS code is also an MDS code.

\medskip
\noindent \textbf{Hamming Bound and Perfect Codes.} Given a code $C$
with parameters $[n,k,d]_q$ for $d \leq n$, the classical Hamming
bound can be stated as
\begin{eqnarray}
\sum_{i=0}^t \binom{n}{i} (q-1)^i \leq q^{n-k}, \end{eqnarray} where
$t= \lfloor (d-1)/2 \rfloor$.
Codes that attain Hamming bound with equality are classified as
perfect codes. Let every codeword be represented by a sphere of
radius $t$. The interpretation of Hamming bound, or sometimes called
sphere packing bound, is that all codewords or the $q^k$ spheres are
pairwise disjoint in the space $\F_q^n$. For further details on
bound on the classical code parameters, see for
example~\cite{huffman03,macwilliams77,lin04}.

\subsection{Families of Codes}
There have been numerous families of classical codes. The most
notable are the Bose-Chaudhuri-Hocquenghem (BCH), Reed-Solomon (RS),
Reed-Muller (RM), algebraic and projective geometry, and LDPC codes,
see~\cite{huffman03,macwilliams77,lin04}. In this work I
will describe some of these families. I will establish the
conditions required for these codes to be self-orthogonal (or
dual-containing) over finite fields, and, consequently, they can be
used to derive quantum error control codes.

\section{Quantum Error Control Codes}
There has been a tremendous amount of research work in quantum error
correcting codes during the last ten years. As such, the theory of
stabilizer codes is well developed over binary and nonbinary fields.
Many  families of stabilizer codes are constructed based on BCH, RS,
RM, finite geometry classical codes, where these families of codes
are shown to be self-orthogonal (or dual-containing). Recently, the
theory of stabilizer codes over finite fields has been extended to
subsystem codes, where  families of classical codes do not need to
be self-orthogonal (or dual-containing). Also,  new families and
code constructions of subsystem codes have been investigated. I will
summarize previous work related to my research in the following
subsections.

\subsection{Quantum Block Codes} The first quantum code was introduced by Shor as an impure
quantum code with parameters $[[9,1,3]]_2$ in a landmark paper in
1995~\cite{shor95}. The idea was to protect one qubit against bit
flip and phase errors into nine qubits.  Gottesman developed the
theory and introduced quantum encoding circuits and fault-tolerant
quantum computing~\cite{gottesman05,gottesman96,gottesman97}.
Calderbank and Shor extended the theory to codes over $\F_4$ and
introduced the CSS construction independently with
Steane~\cite{calderbank98,calderbank96,steane96}.
The quantum code $Q$ can be defined as follows.
\begin{defn}
A $q$-ary quantum code $Q$, denoted by $[[n,k,d]]_q$, is a $q^k$ dimensional
subspace of the Hilbert space $\mathbb{C}^{q^n}$ and can correct all errors up
to $\lfloor \frac{d-1}{2}\rfloor$.
\end{defn}
\noindent The code $Q$ is able to encode $k$ logical qubits into $n$
physical qubits with a minimum distance of at least $d$ between any
two codewords. The $Q$ can be constructed based on two  classical
codes $C_1$ and $C_2$ such that $C_2^\perp\le C_1$ as follows.
\begin{fact}[CSS Code Construction]\label{th:css}
Let $C_1$ and $C_2$ denote two classical linear codes with parameters
$[n,k_1,d_1]_q$ and $[n,k_2,d_2]_q$ such that $C_2^\perp\le C_1$. Then there exists
a $[[n,k_1+k_2-n,d]]_q$ stabilizer code with minimum distance $d=\min\{ \wt(c) \mid
c\in (C_1\setminus C_2^\perp)\cup (C_2\setminus C_1^\perp)\}\ge \min\{ d_1,d_2\}$.
\end{fact}
\noindent Constructing a quantum code $Q$ reduces to constructing a
self-orthogonal (or dual-containing) classical code $C$ defined over
$\F_q$ or $\F_{q^2}$ as follows.
\begin{fact}\label{th:css2}
If there exists an $\F_{q}$-linear $[n,k,d]_{q}$ classical code $C$
containing its dual, $C^\perp \subseteq C$, then there exists an
$[[n,2k-n,\geq d]]_q$ quantum stabilizer code that is pure to $d$.
\end{fact}

\begin{fact}If there exists an $\F_{q^2}$-linear $[n,k,d]_{q^2}$ classical code $C$
such that $C^{\hdual}\subseteq C$, then there exists an
$[[n,2k-n,\ge d]]_q$ quantum stabilizer code that is pure to $d$.
\end{fact}

There have been many families of quantum codes based on binary
classical codes,
see~\cite{grassl99b,grassl00,grassl99,kim02,steane99}. These classes
of codes are derived from BCH, RS, algebraic geometry codes in
addition to codes over graphs. The theory has been generalized to
finite fields,
see~\cite{ashikhmin01,feng02,feng02b,gottesman99,kim04,
rains99,roetteler04,schlingemann02}.  Recently, new bounds, encoding
circuits, and new families have been investigated,
see~\cite{aly06a,aly06b,feng04,grassl03,feng02,li04,roetteler04}.

We will describe foundations of  quantum block codes, as well as
bounds and families of such codes in
Chapters~\ref{ch_QBC_basics},\ref{ch_QBC_BCH},\ref{ch_QBC_Qduadic},
\ref{ch_QBC_QPRM}.

\subsection{Subsystem Codes}
Subsystem codes are a generalization of the theory of quantum error
correction and decoherence free subspaces. Such codes are an
extension of quantum codes that are constructed based on
self-orthogonal(or dual-containing) classical codes.   The
assumption is that a quantum code $Q$ can be decomposed as a tensor
product of two subsystems $A$ and $B$, i.e. $Q=A \otimes B$. The
source qubits are stored in the subsystem $A$ and  gauge qubits are
stored in subsystem $B$. Therefore, subsystem codes are quantum
error control codes where errors can be avoided as well as
corrected. One can correct only errors on the subsystem  $A$ and
completely neglect the errors affecting the subsystem
$B$~\cite{bacon06,kribs05}; for a group representation of operator
quantum codes, see~\cite{klappenecker0608,knill06,poulin05}.

It has been shown in~\cite{aly06c,aly08a} that subsystem codes over
$\F_q$ can be derived from classical additive codes over $\F_q$ and
$\F_{q^2}$ without the needed for self-orthogonal or dual-containing
conditions. An approach for code construction and bounds on the code
parameters is shown in~\cite{aly06c}. It has been claimed that
subsystem codes seem to offer some attractive features for
protection of quantum information and fault-tolerant quantum
computing. They  can be self-correcting codes~\cite{bacon06}. Let
$\mathcal{H}=C^{q^n}$ be the Hilbert space  such that
$\mathcal{H}=Q\oplus Q^\perp$, where $Q^\perp$ is the orthogonal
complement of $Q$. An $[[n,k,r,d]]_q$ subsystem code $Q$ can be
described as
\begin{defn}
An $[[n,k,r,d]]_q$ subsystem code is a decomposition of the subspace
$Q$ into a tensor product of two vector spaces $A$ and $B$ such that
$Q=A\otimes B$. If $\dim A=k$ and $\dim B= r$, then the code $Q$ is
able to detect all errors  of weight less than $d$ on subsystem $A$.
\end{defn}

Subsystem codes can be constructed  from  classical codes  over
$\F_q$ and $\F_{q^2}$.

\begin{fact}[Euclidean Construction]\label{lem:css-Euclidean-subsys}
If $C$ is a $k'$-dimensional $\F_q$-linear code of length $n$ that
has a $k''$-dimensional subcode $D=C\cap C^\perp$ and $k'+k''<n$,
then there exists an
$$[[n,n-(k'+k''),k'-k'',\wt(D^\perp\setminus C)]]_q$$
subsystem code.
\end{fact}

\begin{fact}[Hermitian Construction]\label{lem:css-Hermitina-subsys}
Let $C \subseteq \F_{q^2}^n$ be an $\F_{q^2}$-linear $[n,k,d]_{q^2}$
code such that $D=C\cap C^\hdual$ is of dimension
$k'=\dim_{\F_{q^2}} D$. Then there exists an
$$[[n,n-k-k',k-k',\wt(D^\hdual \setminus C)]]_q$$ subsystem code.
\end{fact}

We will describe foundations of  subsystem codes; in addition to
bounds and families of such codes in
Chapters~\ref{ch_subsys_basic},\ref{ch_subsys_construction},\ref{ch_subsys_families},
\ref{ch_subsys_rules_tables}.

\subsection{Quantum Convolutional Codes}
Quantum convolutional codes (QCC's) seem to be useful for quantum
communication because they have online encoders and decoders. One
main property of quantum convolutional codes is the delay operator
where the encoder has some memory set. However, quantum
convolutional codes still have  not been  studied extensively. As
pointed out earlier by several authors~\cite{grassl05}, many
interesting and unsolved questions remain regarding the properties
and the usefulness of quantum convolutional codes. At this time, it
is not known if quantum convolutional codes offer a decisive
advantage over quantum block codes. We do not yet have a
well-defined formalism of quantum convolutional codes. For example,
the CSS construction, projector of a quantum convolutional code, and
non-catastrophic encoders are not clearly defined for quantum
convolutional codes. In other words, except for the work by Ollivier
\cite{ollivier04}, there are only some examples of quantum
convolutional codes with $1/3$, $1/4$, and $1/n$ code rates.
There have been examples of quantum convolutional codes in the
literature; the most notable being are the $((5,1,3))$ code of
Ollivier and Tillich, the $((4,1,3))$ code of Almeida and Palazzo
and the rate $1/3$ codes of Forney and Guha. We present the most
notable results as follows
\begin{itemize}
\item Ollivier and Tillich developed the stabilizer framework for quantum convolutional codes.
They also addressed the encoding and decoding aspects of quantum
convolutional codes
(cf.~\cite{ollivier05,ollivier03,ollivier04,ollivier05}).
Furthermore, they provided a maximum likelihood error estimation
algorithm. They showed, as an example, a quantum convolutional code
of rate $k/n=1/5$ that can correct only one error.
\item
Forney and Guha constructed quantum convolutional codes with rate $1/3$
\cite{forney05a}. Also, together with Grassl, they derived rate $(n-2)/n$
quantum convolutional codes \cite{forney05b}. They gave tables of optimal rate
$1/3$ quantum convolutional codes and they also constructed good quantum block
codes obtained by tail-biting convolutional codes.
\item
Grassl and R{\"{o}}tteler constructed quantum convolutional codes
from product codes. They showed that starting with an arbitrary
convolutional code and a self-orthogonal block code, a quantum
convolutional code can be constructed. (cf. \cite{grassl05}).
Recently, Grassl and R{\"{o}}tteler~\cite{grassl06b} stated a
general algorithm to construct quantum circuits for non-catastrophic
encoders and encoder inverses for channels with memories.
Unfortunately, the encoder they derived is for a subcode of the
original code.
\end{itemize}

Recall that one can construct convolutional stabilizer codes from
self-orthogonal (or dual-containing) classical convolutional codes
over $\F_q$ (cf. \cite[Corollary~6]{aly07b}) and $\F_{q^2}$ (see
\cite[Theorem~5]{aly07b}) as stated in the following theorem.

\begin{fact}\label{CSS:F_q}
An $[(n,k,nm;\nu,d_f)]_q$ convolutional stabilizer code exists if
and only if there exists an $(n,(n-k)/2,m;\nu)_q$ convolutional code
such that $C \leq C^\perp$ where the dimension of $C^\perp$ is given
by $(n+k)/2$ and $d_f=\wt(C^\perp \backslash C).$
\end{fact}

We will describe foundations of  quantum convolutional codes, as
well as bounds and families of such codes in
Chapters~\ref{ch_QCC_bounds},\ref{ch_QCC_RS},\ref{ch_QCC_BCH}.

\section{Fault Tolerant Quantum Computing}

Fault tolerant quantum
computing is needed to speed up building quantum computers, if it has to happen in reality.
The main purpose of fault tolerant quantum computing is to limit the
number of errors that may occur in practical quantum computers.
These errors may happen in the quantum error correcting operations
or in the quantum circuits (i.e. gate operations). First, Shor
presented the idea of applying fault tolerant quantum operations
into quantum gates~\cite{shor96}. He applied it on controlled-not
and phase gates, and showed how to perform fault tolerant operations
even if an error happened in one single qubit. Some  progress
 in fault tolerant quantum computing is  included ~\cite{preskill98,gottesman99,steane04,knill04}. Fault tolerant quantum
computing seems to speed up the process of building quantum
computers under a certain threshold value, known as threshold
theorem~\cite{knill04,steane04,aliferis06}.


\part{Quantum Block Codes}
\chapter{Fundamentals of Quantum Block Codes}\label{ch_QBC_basics}

In this chapter I aim to provide an accessible introduction to the
theory of  quantum error-correcting codes over finite fields. Many
definitions that are stated in this chapter will be also used
through out the following parts.  I will recall certain
definitions concerning the error group and bounds of quantum code
parameters from this chapter in the later chapters. Whenever, there
is a definition or result that has not been mentioned in this
chapter and will be used in the later chapters, I will
state it accordingly if needed. I tried to keep the prerequisites to
a minimum, though I assume that the reader has a minimal background
in coding theory and quantum computing as introduced in the first
two chapters or as shown in any introductory textbook such
as~\cite{Chang00}. Also, I recommend the introductory
textbooks~\cite{huffman03} and~\cite{macwilliams77} as sources for
the classical coding theory. I will cite most of the known previous
work in quantum error control codes. Finally, part of this chapter
has been done in a joint work with A. klappenecker and P. Sarvepalli
and has been presented in~\cite{klappenecker065}.

\smallskip

This chapter focuses only on quantum block codes and it is organized
as follows. Section~\ref{sec:stabcodes} gives a brief overview of
the main ideas of stabilizer codes while Section~\ref{sec:additive}
reviews the relation between quantum stabilizer codes and classical
codes. This connection makes it possible to reduce the study of
quantum stabilizer codes to the study of self-orthogonal (or
dual-containing) classical codes, though the definition of
self-orthogonality is a little broader than the classical one.
Further, it allows us to use all the tools of classical codes to
derive bounds on the parameters of good quantum codes.
Section~\ref{sec:bounds} gives an overview of the important bounds
for quantum codes. I will state quantum Singleton and Hamming bounds
on quantum code parameters. I will prove quantum Hamming bound for
impure quantum codes that can correct one or two errors. After that
I will introduce many families of quantum error-correcting codes
derived from self-orthogonal (or dual-containing) classical codes in
the following chapters.

\smallskip

\textit{Notations.}  The finite field with $q$ elements is denoted
by $\F_q$, where $q=p^m$ and $p$ is assumed to be a prime and $m$ is
an integer number. The trace function from $\F_{q^r}$ to $\F_q$ is
defined as $\tr_{q^r/q}(x)=\sum_{i=0}^{r-1} x^{q^k}$, and we may
omit the subscripts if $\F_q$ is the prime field.  The center of a
group $G$ is denoted by $Z(G)$ and the centralizer of a subgroup $S$
in $G$ by $C_G(S)$. We denote by $H\le G$ the fact that $H$ is a
subgroup of~$G$. The trace~$\Tr(M)$ of a square matrix~$M=[m_{ij}]$
of size $n\times n$ is the sum of the diagonal elements of~$M$,
i.e., $\sum_{i=1}^n m_{ii}=\Tr(M)$.

\section{Stabilizer Codes}\label{sec:stabcodes}
In this chapter, we use $q$-ary quantum digits, shortly called
qudits, as the basic unit of quantum information. The state of a
qudit is a nonzero vector in the complex vector space $\C^q$. This
vector space is equipped with an orthonormal basis whose elements
are denoted by $\ket{x}$, where $x$ is an element of the finite
field $\F_q$.  The state of a system of $n$ qudits is then a nonzero
vector in $\C^{q^n}$.  In general, quantum codes are just nonzero
subspaces of $\C^{q^n}$. A quantum code that encodes $k$ logical
qudits of information into $n$ physical qudits is denoted by
$[[n,k,d]]_q$, where the subscript $q$ indicates that the code is
$q$-ary and $d$ is the minimum distance of this code. More
generally, an $((n,K,d))_q$ quantum code is a $K$-dimensional
subspace encoding $\log_qK$ qudits into $n$ qudits and it can
correct up to $t=\lfloor (d-1)/2 \rfloor$ errors.

The first quantum error-correcting code was introduced  by Shor in
1995 as an impure quantum code with parameters $[[9,1,3]]_2$~
\cite{shor95}. The idea was to protect one qubit against bit flip
and phase flip errors by encoding this qubit into nine qubits.
Calderbank and Shor extended the theory and formalized the CSS
construction independently with
Steane~\cite{calderbank98,calderbank96,steane96}. Shortly, Gottesman
introduced stabilizer codes, quantum concatenated codes and quantum
encoding circuits~\cite{gottesman96,gottesman97,gottesman02}.

As the quantum codes are subspaces, it seems natural to describe
them by giving a basis for the subspace. However, in case of quantum
codes this turns out to be an inconvenient description. For
instance, consider a $[[7,1,3]]_2$ Steane code that encodes one
logical qubit into seven physical qubits with a minimum distance
three among its codewords. We can describe a basis for this code as
follows
$$
\begin{array}{ll}
\ket{0_L}&=\ket{0000000}+\ket{1010101} +\ket{0110011}+\ket{1100110} \\
& + \ket{0001111}+\ket{0111100}+\ket{1011010}+\ket{1101001},\\
\ket{1_L}& =\ket{0000000}+\ket{1010101} +\ket{0110011}+\ket{1100110} \\
& + \ket{0001111}+\ket{0111100}+\ket{1011010}+\ket{1101001}.
\end{array}
$$ An alternative description of the quantum error-correcting codes
that will be discussed in this chapter relies on error operators
that act on $\C^{q^n}$.  Let $E$ be an error operator. If we make
the assumption that the errors are independent on each qudit, then
each error operator $E$ can be decomposed as $E=E_1\otimes \cdots
\otimes E_n$. Furthermore, linearity of quantum mechanics allows us
to consider only a discrete set of errors. The quantum
error-correcting codes that we consider here can be described as the
joint eigenspace of an abelian subgroup of error operators. The
subgroup of error operators is called the stabilizer of the code
(because it leaves each state in the code unaffected) and the code
is called a stabilizer code. In the next four subsections, we will
describe the error group and stabilizer codes in details.

\subsection{Error Bases}
Let $P$ be a set of Pauli matrices given by $\{I,X,Z,Y\}$. In
general, we can regard any error as being composed of an amplitude
error (qubit flip) and a phase error (qubit shift). Let $a$ and $b$
be elements in $\F_q$. We can define unitary operators $X(a)$ and
$Z(b)$ on~$\C^q$ that generalize the Pauli $X$ and $Z$ operators to
the $q$-ary case; they are defined as
\begin{eqnarray} X(a)\ket{x}=\ket{x+a},\qquad
Z(b)\ket{x}=\omega^{\tr(bx)}\ket{x},\end{eqnarray} where $\tr$ denotes the
trace operation from $\F_q$ to $\F_p$, and $\omega=\exp(2\pi i/p)$ is a
primitive $p$th root of unity.

Let $\mathcal{E}=\{X(a)Z(b)\,|\, a,b\in \F_q\}$ be the set of error operators.
The error operators in $\mathcal{E}$ form a basis of the set of complex
$q\times q$ matrices as the trace $\Tr(A^\dagger B)=0$ for distinct elements
$A,B$ of $\mathcal{E}$.  Further, we observe that
\begin{equation}\label{eq:multrule}
X(a)Z(b)\,X(a')Z(b')=\omega^{\tr(ba')}X(a+a')Z(b+b').
\end{equation}

The error basis for $n$ $q$-ary quantum systems can be obtained by tensoring
the error basis for each system. Let $\mathbf{a}=(a_1,\dots, a_n)\in \F_q^n$.
Let us denote by $ X(\mathbf{a}) = X(a_1)\otimes\, \cdots \,\otimes X(a_n)$ and
$Z(\mathbf{a}) = Z(a_1)\otimes\, \cdots \,\otimes Z(a_n)$ for the tensor
products of $n$ error operators.  Then we have the following result whose proof
follows from the definitions of $X(\mathbf{a})$ and $Z(\mathbf{b})$.
\begin{lemma}\label{th:nice}
The set $\mathcal{E}_n= \{ X(\mathbf{a})Z(\mathbf{b})\,|\, \mathbf{a},
\mathbf{b}\in \F_q^n\}$ is an error basis on the complex vector
space~$\C^{q^n}$.
\end{lemma}

\subsection{Stabilizer Codes}
We will describe the quantum codes using a set of error bases.
 Consider the error group $G_n$ defined as
\begin{eqnarray} G_n = \{ \omega^{c}X(\mathbf{a})Z(\mathbf{b})\,|\,
\mathbf{a, b} \in \F_q^n, c\in \F_p\}.\end{eqnarray} $G_n$ is simply
a finite group of order $pq^{2n}$ generated by the matrices in the
error basis $\mathcal{E}_n$. Two elements $E_1$ and $E_2$ in $G_n$
are abelian if $E_1E_2=E_2E_1$.

Let $S$ be the largest abelian subgroup of the error group $G_n$
fixes every element in a quantum code $Q$. Then a \textsl{stabilizer
code} $Q$ is a non-zero subspace of $\C^{q^n}$ defined as
\begin{equation}\label{eq:stab}
Q = \bigcap_{E \in S} \{ \ket{\psi} \in \C^{q^n} \mid E\ket{\psi}=\ket{\psi}\}.
\end{equation}
Alternatively, $Q$ is the joint +1 eigenspace of the stabilizer
subgroup $S$. The notation of eigenspace and eigen value are
described for example in~\cite{cohn05}. A stabilizer code contains
\textsl{all} joint eigenvectors of~$S$ with eigenvalue 1, as
equation~(\ref{eq:stab}) indicates.  If the code is smaller and does
not contain all the joint eigenvectors of $S$ with eigenvalue 1,
then it is not a stabilizer code for $S$. In other words, every
error operator $E$ in $S$ fixes every codeword $\ket{\psi}$ in $Q$.

\subsection{Stabilizer and Error Correction}
Now, we define the quantum code via its stabilizer $S$, then we can
 be able to describe the performance of the code, that is, we
should be able to tell how many errors it can error and how the
error-correction is done, in addition to how many errors it can
detect.

The central idea of error detection is that a detectable error acting on $Q$
should either act as a scalar multiplication on the code space (in which case
the error did not affect the encoded information) or it should map the encoded
state to the orthogonal complement of $Q$ (so that one can set up a measurement
to detect the error).  Specifically, we say that $Q$ is able to detect an error
$E$ in the unitary group $U(q^n)$ if and only if the condition $\langle c_1 |
E| c_2\rangle=\lambda_E\langle c_1 |c_2\rangle$ holds for all $c_1, c_2\in Q$,
see~\cite{knill97}.

We can show that a stabilizer code~$Q$ with stabilizer $S$ can
detect all errors in $G_n$ that are scalar multiples of elements in
$S$ or that do not commute with some element of $S$, see
Lemma~\ref{th:detectable}. In particular, an undetectable error in
$G_n$ has to commute with all elements of the stabilizer. Let $S\le
G_n$ and $C_{G_n}(S)$ denote the centralizer of $S$ in $G_n$,
\begin{eqnarray}
C_{G_n}(S)=\{ E\in G_n\,|\, EE'=E'E \text{ for all } E'\in
S\}.\end{eqnarray} Let $SZ(G_n)$ denote the group generated by $S$
and the center $Z(G_n)$. We need the following characterization of
detectable errors.
\begin{lemma}\label{th:detectable}
Suppose that $S \le G_n$ is the stabilizer group of a stabilizer code~$Q$ of
dimension $\dim Q>1$.  An error $E$ in $G_n$ is detectable by the quantum code
$Q$ if and only if either $E$ is an element of $SZ(G_n)$ or $E$ does not belong
to the centralizer $C_{G_n}(S)$.
\end{lemma}
\begin{proof}
See \cite{ketkar06,ashikhmin01}; the interested reader can find a
more general approach in~\cite{knill96b,klappenecker034}.
\end{proof}

Since detectability of errors is closely associated to commutativity of error
operators, we will derive the following condition on commuting elements in
$G_n$:
\begin{lemma}\label{th:commute}
Two elements $E=\omega^cX(\mathbf{a})Z(\mathbf{b})$ and
$E'=\omega^{c'}X(\mathbf{a'})Z(\mathbf{b'})$ of the error group
$G_n$ satisfy the relation $EE' = \omega^{\tr(\mathbf{b\cdot
a'-b'\cdot a})} E'E.$ In particular, the elements $E$ and $E'$
commute if and only if the trace symplectic form $\tr(\mathbf{b\cdot
a'-b'\cdot a})$ vanishes.
\end{lemma}
\begin{proof}{}
We can easily verify that $EE'=\omega^{\tr(\mathbf{b\cdot
a'})}X(\mathbf{a+a'})Z(\mathbf{b+b'})$ and $E'E=\omega^{\tr(\mathbf{b'\cdot
a})} X(\mathbf{a+a'})Z(\mathbf{b+b'})$ using equation~(\ref{eq:multrule}).
Therefore, $\omega^{\tr(\mathbf{b\cdot a'-b'\cdot a})}E'E$ yields $EE'$, as
claimed.
\end{proof}

\noindent \textbf{Minimum Distance.} We shall also define the
minimum distance of a quantum code $Q$. In order to do so, we need
to define the symplectic weight of a vector $(a|b)$ in $\F_q^{2n}$.
The \textsl{symplectic weight} $\swt$ of a vector $(\mbf a|\mbf b)$
in $\F_q^{2n}$ is defined as
\begin{eqnarray}\swt((\mbf a|\mbf b)) = | \{\, k\, |\, (a_k,b_k)\neq
(0,0)\}|.\end{eqnarray} The weight $\wt(E)$ of an element
$E=\omega^c E_1\otimes\cdots\otimes E_n=\omega^c
X(\mbf{a})Z(\mbf{b})$ in the error group~$G_n$ is defined to be the
number of nonidentity tensor components i.e., $\wt(E)=|\{E_i\neq I
\}|=\swt((\mbf a|\mbf b))$.

A quantum code $Q$ is said to have \textsl{minimum distance} $d$\/
if and only if it can detect all errors in $G_n$ of weight less
than~$d$, but cannot detect some error of weight~$d$.  We say that
$Q$ is an $((n,K,d))_q$ code if and only if $Q$ is a $K$-dimensional
subspace of $\C^{q^n}$ that has minimum distance~$d$.  An
$((n,q^k,d))_q$ code is also called an $[[n,k,d]]_q$ code. One of
these two notations will be used when needed.

Due to the linearity of quantum mechanics, a quantum error-correcting code that
can detect a set $\mathcal{D}$ of errors,  can also detect all errors in the
linear span of $\mathcal{D}$.  A code of minimum distance~$d$ can correct all
errors of weight $t=\lfloor (d-1)/2\rfloor$ or less.

\noindent \textbf{Pure and Impure Codes.} We say that a quantum code
$Q$ is \textsl{pure to} $t$\/ if and only if its stabilizer group
$S$ does not contain non-scalar error operators of weight less than
$t$. An $[[n,k,d]]_q$ quantum code is called pure if and only if it
is pure to its minimum distance $d$.  We will follow the same
convention as in~\cite{calderbank98}, that an $[[n,0,d]]_q$ code is
pure. Impure codes are also referred to as degenerate codes.
Degenerate codes are of interest because they have the potential for
passive error-correction and they are difficult to construct as we
will explain later.

\subsection{Encoding Quantum Codes}
The Stabilizer $S$ of a quantum code $Q$  provides also a means for
encoding quantum codes. The essential idea is to
 encode the information into the code space through a projector. For an
$((n,K,d))_q$ quantum code with stabilizer $S$, the projector $P$ is
defined as \begin{eqnarray} P = \frac{1}{|S|}\sum_{E\in S}
E.\end{eqnarray} It can be checked that $P$ is an orthogonal
projector onto a vector space $Q$. Further, we have
\begin{eqnarray}K=\dim Q = \Tr P = q^n/|S|.\end{eqnarray}
The stabilizer allows us to derive encoded operators, so that we can
operate directly on the encoded data instead of decoding and then
operating on them. These operators are in $C_{G_n}(S)$.
See~\cite{gottesman97} and~\cite{grassl03} for more details.

%
\begin{figure}[t]
  \includegraphics[scale=0.8]{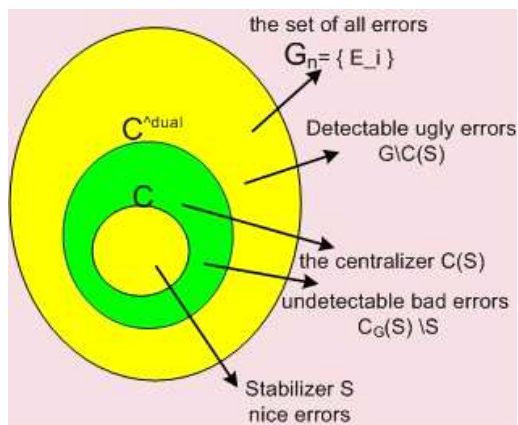}
 \centering
  \caption{The relationship between a quantum stabilizer code $Q$  and a classical code $C$, where $C \subseteq C^\perp$.}\label{fig:qecc1}
\end{figure}
\section{Deriving Quantum Codes from Self-orthogonal Classical Codes}\label{sec:additive}

In this section we show how stabilizer codes are related to
classical codes (additive codes over $\F_q$ or over $\F_{q^2}$). The
central idea behind this relation is the fact insofar as the
detectability of an error is concerned the phase information is
irrelevant. This means we can factor out the phase defining a map
from $G_n$ onto $\F_q^{2n}$ and study the images of $S$ and
$C_{G_n}(S)$. We will denote a classical code $C \le \F_q^n$ with
$K$ codewords and distance $d$ by $(n,K,d)_q$. If it is linear then
we will also denote it by $[n,k,d]_q$ where $k=\log_q K$. We define
the Euclidean inner product of $x,y\in \F_q^n$ as $x\cdot y=
\sum_{i=1}^n x_i y_i$. The dual code $C^\perp$ is the set of vectors
in $\F_q^n$ orthogonal to $C$ i.e., $C^\perp=\{x\in \F_q^n\mid
x\cdot c =0 \mbox{ for all } c\in C \}$. For more details on
classical codes see \cite{huffman03} or \cite{macwilliams77}.

Constructing a quantum code  $Q$ reduces to constructing a
self-orthogonal classical code $C$ over $\F_q$ and $\F_q^2$,
see~\cite{cleve97,cleve97b,calderbank98,gottesman97,
grassl01,steane99,steane96,shor95}. This relationship is shown in
Fig.~\ref{fig:qecc1}.

\begin{fact}[CSS Code
Construction]\label{th:css} Let $C_1$ and $C_2$ denote two classical linear
codes with parameters $[n,k_1,d_1]_q$ and $[n,k_2,d_2]_q$ such that
$C_2^\perp\le C_1$. Then there exists a $[[n,k_1+k_2-n,d]]_q$ stabilizer code
with minimum distance $d=\min\{ \wt(c) \mid c\in (C_1\setminus C_2^\perp)\cup
(C_2\setminus C_1^\perp)\}\ge \min\{ d_1,d_2\}$.
\end{fact}
Also, we can construct quantum codes from classical codes that contain their
duals or are self-orthogonal as follows:
\begin{fact}\label{fact:css2}
If $C$ is a classical linear $[n,k,d]_q$ code containing its dual, $C^\perp\le
C$, then there exists a $[[n,2k-n,d]]_q$ stabilizer code.
\end{fact}

Fact~\ref{fact:css2} is particularly interesting because it helps us to
construct a quantum code from a classical code and its dual. There have been
many families of quantum codes based on binary classical codes, see
\cite{grassl99b,grassl00,grassl99,kim02}. The theory has been generalized to
finite fields, see~\cite{ashikhmin01, feng02,feng02b,gottesman99,kim04,
rains99,roetteler04,schlingemann02}. Recently, new bounds, encoding circuits,
and new families have been investigated,
see~\cite{aly06a,aly06b,feng04,grassl03,feng02,li04,roetteler04}.

\subsection{Codes over $\F_q$.}
If we associate with an element $\omega^cX(\mathbf{a})Z(\mathbf{b})$
of $G_n$ an element $(\mathbf{a}|\mathbf{b})$ of $\F_q^{2n}$, then
the group $SZ(G_n)$ is mapped to the additive code
\begin{eqnarray}C=\{ (\mathbf{a}|\mathbf{b})\,|\,
\omega^cX(\mathbf{a})Z(\mathbf{b})\in SZ(G_n)\}=
SZ(G_n)/Z(G_n).\end{eqnarray} To relate the images of the stabilizer
and its centralizer, we need the notion of a trace-symplectic form
of two vectors $(\mathbf{a}|\mathbf{b})$ and
$(\mathbf{a'}|\mathbf{b'})$ in $\F_{q}^{2n}$,
\begin{eqnarray} < (\mathbf{a}|\mathbf{b})\, |\, (\mathbf{a'}|\mathbf{b'})>_s =
\tr_{q/p}(\mathbf{b}\cdot \mathbf{a}' - \mathbf{b}'\cdot
\mathbf{a}).\end{eqnarray} Let $C^\sdual$ be the trace-symplectic
dual of $C$ defined as
\begin{eqnarray}C^\sdual = \{x\in \F_q^{2n}\mid < x |\, c >_s  =0 \mbox{ for all } c\in C \}.\end{eqnarray}
The centralizer $C_{G_n}(S)$ contains all elements of $G_n$ that
commute with each element of $S$; thus, by Lemma~\ref{th:commute},
$C_{G_n}(S)$ is mapped onto the trace-symplectic dual code
$C^\sdual$ of the code $C$, \begin{eqnarray} C^\sdual =\{
(\mathbf{a}|\mathbf{b})\,|\, \omega^cX(\mathbf{a})Z(\mathbf{b})\in
C_{G_n}(S)\}.
\end{eqnarray}

The next theorem illustrates this connection between classical codes
and stabilizer codes and generalizes the well-known connection to
symplectic codes~\cite{calderbank98,gottesman96} of the binary case.
\begin{theorem}\label{th:stabilizer}
An $((n,K,d))_q$ stabilizer code exists if and only if there exists
an additive code $C \le \F_q^{2n}$ of size $|C|=q^n/K$ such that
$C\le C^\sdual$ and $\swt(C^{\sdual} \setminus C)=d$ if $K>1$( and
$\swt(C^\sdual)=d$ if $K=1$).
\end{theorem}
\begin{proof}{}
See~\cite{ashikhmin01,ketkar06} for the proof.
\end{proof}
In 1996, Calderbank and Shor~\cite{calderbank96} and
Steane~\cite{steane96} introduced the following construction of
quantum codes. It is perhaps the simplest method to build quantum
codes via classical codes over $\F_q$.
\begin{lemma}[CSS Code Construction]\label{th:css}
Let $C_1$ and $C_2$ denote two classical linear codes with parameters
$[n,k_1,d_1]_q$ and $[n,k_2,d_2]_q$ such that $C_2^\perp\le C_1$. Then there
exists a $[[n,k_1+k_2-n,d]]_q$ stabilizer code with minimum distance $d=\min\{
\wt(c) \mid c\in (C_1\setminus C_2^\perp)\cup (C_2\setminus C_1^\perp)\}$ that
is pure to $\min\{ d_1,d_2\}$.
\end{lemma}
\begin{proof}{}
Let $C=C_1^\perp\times C_2^\perp\le \F_q^{2n}$. Clearly $C\le
C_2\times C_1$. If $(c_1\mid c_2) \in C$ and $(c_1'\mid c_2') \in
C_2\times C_1$, then we observe that $ \tr( c_2\cdot c_1' - c_2'
\cdot c_1) = \tr(0-0)=0.$ Therefore, $C\le C_2\times C_1 \le
C^\sdual$. Since $|C|=q^{2n-k_1-k_2}$,
$|C^\sdual|=q^{2n}/|C|=q^{k_1+k_2}=|C_2\times C_1|$. Therefore,
$C^\sdual = C_2\times C_1$. By Theorem~\ref{th:stabilizer} there
exists an $((n,K,d))_q$ quantum code with $K=q^n/|C|=q^{k_1+k_2-n}$.
The claim about the minimum distance and purity of the code is
obvious from the construction.
\end{proof}
\begin{corollary}\label{th:css2}
If $C$ is a classical linear $[n,k,d]_q$ code containing its dual, $C^\perp\le
C$, then there exists an $[[n,2k-n,\geq d]]_q$ stabilizer code that is pure
to~$d$.
\end{corollary}

We will use Lemma~\ref{th:css} and Corollary~\ref{th:css2} to derive
many families of quantum error-correcting codes based on BCH, RS,
duadic, and projective geometry codes as shown in the following
sections.

\subsection{Codes over $\F_{q^2}$.}
We can also  extend the connection of the quantum codes and
classical codes that are defined over $\F_{q^2}$, especially as it
allows us the use of codes over quadratic extension fields. The
binary case was done in~\cite{calderbank98} and partial
generalizations were done in~\cite{matsumoto00,kim04} and
\cite{rains99}. We provide a slightly alternative generalization
using a trace-alternating form. Let $(\beta,\beta^q)$ denote a
normal basis of $\F_{q^2}$ over $\F_q$. We define a
trace-alternating form of two vectors $v$ and $w$ in $\F_{q^2}^n$ by
\begin{equation}\label{eq:alternating}
 (v|w)a =\tr_{q/p}\left(\frac{v\cdot w^q - v^q\cdot w }{\beta^{2q}-\beta^2}\right).
\end{equation}
The argument of the trace is an element of $\F_q$ as it is invariant under the
Galois automorphism $x\mapsto x^q$.

Let  $\phi:\F_q^{2n}\rightarrow \F_{q^2}^n$ take $(\mbf a|\mbf b)\mapsto \beta
\mbf a + \beta^q \mbf b.$
The map $\phi$ is isometric in the sense that the symplectic weight of $(\mbf
a|\mbf b)$ is equal to the Hamming weight of $\phi((\mbf a|\mbf b))$.  This map
allows us to transform the trace-symplectic duality into trace-alternating
duality. In particular it can be easily verified that if $c,d\in \F_{q}^{2n}$,
then $< c,|\,d>s=(\phi(c),|, \phi(d))a$.  If $D\le\F_{q^2}^n$, then we denote
its trace-alternating dual by $D^\adual=\{v\in \F_{q^2}^n \mid (v|w)a=0 \mbox{
for all }w\in D \}$. Now Theorem~\ref{th:stabilizer} can be reformulated as:

\begin{theorem}\label{th:alternating}
An\/ $((n,K,d))_q$ stabilizer code exists if and only if there exists an
additive subcode $D$ of\/ $\F_{q^2}^{n}$ of cardinality $|D|=q^n/K$ such that
$D\le D^\adual$ and $\wt(D^{\adual} \setminus D)=d$ if $K>1$ (and
$\wt(D^\adual)=d$ if $K=1$).
\end{theorem}
\begin{proof}{}
From Theorem~\ref{th:stabilizer} we know that an $((n,K,d))_q$ stabilizer code
exists if and only if there exists a code $C\le \F_q^{2n}$ such that
$|C|=q^n/K$, $C\le C^\sdual$, and $\swt(C^\sdual\setminus C)=d$ if $K>1$ (and
$\swt(C^\sdual)=d$ if $K=1$). The theorem follows simply by applying the
isometry $\phi$.
\end{proof}
If we restrict our attention to linear codes over $\F_{q^2}$, then the
hermitian form is more useful. The hermitian inner product of two vectors
$\mathbf{x}$ and $\mathbf{y}$ in $\F_{q^2}^n$ is given by $\mathbf{x}^q\cdot
\mathbf{y}$. From the definition of the trace-alternating form it is clear that
if two vectors are orthogonal with respect to the hermitian form they are also
orthogonal with respect to the trace-alternating form. Consequently, if $D\le
\F_{q^2}^n$, then $D^\hdual \le D^\adual$, where $D^\hdual=\{v\in \F_{q^2}^n
\mid v^q\cdot w=0 \mbox{ for all }w\in D \}$.

Therefore, any self-orthogonal code with respect to the hermitian inner product
is self-orthogonal with respect to the trace-alternating form. In general, the
two dual spaces $D^\hdual$ and $D^\adual$ are not the same. However, if $D$
happens to be $\F_{q^2}$-linear, then the two dual spaces coincide.

\begin{corollary}\label{co:classical}
If there exists an $\F_{q^2}$-linear $[n,k,d]_{q^2}$ code $D$ such that
$D^\hdual\le D$, then there exists an $[[n,2k-n,\geq d]]_{q}$ quantum code that
is pure to~$d$.
\end{corollary}
\begin{proof}{}
Let $q=p^m$, $p$ prime. If $D$ is a $k$-dimensional subspace of $\F_{q^2}^n$,
then $D^\hdual$ is a $(n-k)$-dimensional subspace of $\F_{q^2}^n$. We can also
view~$D$ as a $2mk$-dimensional subspace of $\F_p^{2mn}$, and $D^\adual$ as a
$2m(n-k)$-dimensional subspace of $\F_p^{2mn}$.  Since $D^\hdual \subseteq
D^\adual$ and the cardinalities of $D^\adual$ and $D^\hdual$ are the same, we
can conclude that $D^\adual =D^\hdual$. The claim follows from
Theorem~\ref{th:alternating}.
\end{proof}
So it is sufficient to consider the hermitian form in case of
$\F_{q^2}$-linear codes. For additive codes (that are not linear)
over $\F_{q^2}$ we have to use the rather inconvenient
trace-alternating form. Finally, using the hermitian construction,
we will derive many families of quantum error-correcting codes in
the following sections.

\section{Bounds on Quantum Codes}\label{sec:bounds}

We need some bounds on the achievable minimum distance of a quantum stabilizer
code. Perhaps the simplest one is the Knill-LaFlamme bound, also called the
quantum Singleton bound. The binary version of the quantum Singleton bound was
first proved by Knill and Laflamme in~\cite{knill97}, see
also~\cite{ashikhmin99,ashikhmin00b}, and later generalized by Rains using
weight enumerators in~\cite{rains99}.
\begin{theorem}[Quantum Singleton Bound]\label{th:singleton}
An  $((n,K,d))_q$ stabilizer code with $K>1$ satisfies
\begin{eqnarray} K\le q^{n-2d+2}.\end{eqnarray}
\end{theorem}
All binary and nonbinary quantum codes obeys the quantum Singleton
bound as shown in Theorem~\ref{th:singleton}. In addition all pure
and impure quantum codes satisfies this bound as well.  Codes which
meet the quantum Singleton bound are called quantum MDS codes.
In~\cite{ketkar06}, it was showed that these codes cannot be
indefinitely long and the maximal length of a $q$-ary quantum MDS
codes is upper bounded by $2q^2-2$. This could probably be tightened
to $q^2+2$. It would be interesting to find quantum MDS codes of
length greater than $q^2+2$ since it would disprove the MDS
Conjecture for classical codes \cite{huffman03}. A related open
question is regarding the construction of codes with lengths between
$q$ and $q^2-1$. At the moment there are no analytical methods for
constructing a quantum MDS code of arbitrary length in this range
(see \cite{grassl04} for some numerical results).

Another important bound for quantum codes is the quantum Hamming bound. The
quantum Hamming bound states (see~\cite{gottesman96,feng04}) that:
\begin{theorem}[Quantum Hamming Bound]\label{th:hamming}
Any pure $((n,K,d))_q$ stabilizer code satisfies \begin{eqnarray}
\sum_{i=0}^{\lfloor (d-1)/2\rfloor} \binom{n}{i}(q^2-1)^i \le
q^n/K.\end{eqnarray}
\end{theorem}
While the quantum Singleton bound holds for all quantum codes, it is
not known if the quantum Hamming bound is of equal applicability. So
far no degenerate quantum code has been found that  beats this
bound. Gottesman showed that impure  binary quantum codes cannot
beat the quantum Hamming bound~\cite{gottesman97}.

 In~\cite{ashikhmin99} Ashikhmin and Litsyn derived many bounds for quantum
codes by extending a novel method originally introduced by
Delsarte~\cite{delsarte72} for classical codes. Using this method
they proved the binary versions of Theorem~\ref{th:hamming} and
Theorem~\ref{th:singleton}. We use this method to show that the
Hamming bound holds for all double error-correcting quantum codes.
See~\cite{ketkar06} for a similar result for single error-correcting
codes. But first we need  Theorem~\ref{th:lp2} and the Krawtchouk
polynomial of degree $j$ in the variable $x$, \begin{eqnarray}
K_j(x) = \sum_{s=0}^j (-1)^s(q^2-1)^{j-s} {x \choose s}{n-x \choose
j-s}.\end{eqnarray}

\begin{theorem}\label{th:lp2}
Let $Q$ be an $((n,K,d))_q$ stabilizer code of dimension $K>1$.
Suppose that $S$ is a nonempty subset of $\{0,\dots,d-1\}$ and
$N=\{0,\dots,n\}$.  Let \begin{eqnarray} f(x)= \sum_{i=0}^n f_i
K_i(x)\end{eqnarray} be a polynomial satisfying the conditions
\begin{enumerate}
\item[i)] $f_x> 0$ for all $x$ in $S$, and $f_x\ge 0$ otherwise;
\item[ii)] $f(x)\le 0$ for all $x$ in $N\setminus S$.
\end{enumerate}
Then \begin{eqnarray} K \le \frac{1}{q^n}\max_{x\in S}
\frac{f(x)}{f_x}.\end{eqnarray}
\end{theorem}
\begin{proof}{}
See~\cite{ketkar06}.
\end{proof}
We demonstrate usefulness of the previous theorem by showing that
the quantum Hamming bound holds for impure nonbinary codes  when
$d=5$.
\begin{lemma}[Quantum Hamming Bound]\label{th:hamming2}
An $((n,K,5))_q$ stabilizer code with $K>1$ satisfies
\begin{eqnarray} K\le
q^{n}\big/(n(n-1)(q^2-1)^2/2+n(q^2-1)+1).\end{eqnarray}
\end{lemma}
\begin{proof}{}
 Let $f(x)=\sum_{j=0}^n f_j K_j(x)$, where $f_x=(\sum_{j=0}^{e}K_j(x))^2$,  $S = \{0,1,\ldots,4\}$ and N=\{0,1,\ldots,n\}. Calculating $f(x)$ and $f_x$
gives us
\begin{eqnarray*}
  f_0 &=& (1+n(q^2-1)+n(n-1)(q^2-1)^2/2)^2\\
  f_1 &=& \frac{1}{4} (n-1)^2 (n-2)^2 (q^2-1)^4  \\
  f_2 &=& (\frac{1}{2} (n-3) (n-2) (q^2-1)^2 -(n-2) (q^2-1))^2\\
  f_3 &=& (1-2 (n-3) (q^2-1)+\frac{1}{2} (n-4) (n-3) (q^2-1)^2)^2 \\
  f_4 &=& (3-3 (n-4) (q^2-1)+\frac{1}{2} (n-5) (n-4) (q^2-1)^2)^2\\
  \mbox{ and, }\\
    f(0) &=& q^{2 n} (1+n (q^2-1)+\frac{1}{2} (n-1) n (q^2-1)^2)\\
    f(1) &=& q^{2 n} (q^2+2 (n-1) (q^2-1)+(n-1) (q^2-2) (q^2-1)) \\
    f(2) &=& q^{2 n} (4+4 (q^2-2)+(q^2-2)^2+2 (n-2) (q^2-1)) \\
    f(3) &=& q^{2 n} (6+6 (q^2-2)) \\
    f(4) &=& 6 q^{2 n}.
\end{eqnarray*}
Clearly $f_x>0$ for all $x \in S$ . Also, $f(x) \leq 0$ for all $x
\in N \backslash S$ since the binomial coefficients for negative
values are zero. The Hamming bound is given by \begin{eqnarray} K
\leq q^{-n} \max_{s \in S} \frac{f(x)}{f_x} \end{eqnarray} So, there
are four different comparisons where $f(0)/f_0 \geq f(x)/f_x$, for
$x=1,2,3,4$. We find a lower bound for $n$ that holds for all values
of $q$.  For $n \geq 7$ it follows
 that \begin{eqnarray} \max
\{f(0)/f_0,f(1)/f_1,f(2)/f_2,f(3)/f_3,f(4)/f_4\} = f(0)/f_0
\end{eqnarray}
\end{proof}
The detailed prove of  Lemma~\ref{th:hamming2} can be found
in~\cite{aly07hamming}. While the above method is a general method
to prove Hamming bound for impure quantum codes,  the number of
terms increases with a large minimum distance. It becomes difficult
to find the true bound using this method. However, one can derive
more consequences from Theorem~\ref{th:lp2}; see, for
instance,~\cite{ashikhmin99,ashikhmin00b,levenshtein95,mceliece77}.

\section{Perfect Quantum Codes} A quantum code that meets the quantum Hamming
bound with equality is known as a perfect quantum code. In fact the famous
$[[5,1,3]]_2$ code \cite{laflamme96} is one such. We will show that there do
not exist any pure perfect quantum codes other than the ones mentioned in the
following theorem. It is actually a very easy result and follows from known
results on classical perfect codes, but we had not seen this result earlier in
the literature.
\begin{theorem}
There do not exist any pure perfect quantum codes with distance greater than 3.
\end{theorem}
\begin{proof}{}
Assume that $Q$ is a pure perfect quantum code with the parameters
$((n,K,d))_q$. Since it meets the quantum Hamming bound we have
\begin{eqnarray}
{\sum_{j=0}^{\lfloor(d-1)/2\rfloor}\binom{n}{j}(q^2-1)^j}={q^n}/K.
\end{eqnarray}
By Theorem~\ref{th:alternating} the associated classical code $C$ is such that
$C^\adual \le C\le \F_{q^2}^{n}$ and has parameters $(n,q^nK,d)_{q^2}$. Its
distance is $d$ because the quantum code is pure. Now $C$ obeys the classical
Hamming bound (see \cite[Theorem~1.12.1]{huffman03} or any textbook on
classical codes). Hence
\begin{eqnarray}
|C|=q^nK\leq
\frac{q^{2n}}{\sum_{j=0}^{\lfloor(d-1)/2\rfloor}\binom{n}{j}(q^2-1)^j}.
\end{eqnarray}
Substituting the value of $K$ we see that this implies that $C$ is a perfect
classical code. But the only perfect classical codes with distance greater than
$3$ are the Golay codes and the repetition codes~\cite{huffman03}. The perfect
Golay codes are over $\F_2$ and $\F_3$ not over a quadratic extension field as
$C$ is required to be. The repetition codes are of dimension $1$ and cannot
contain  their duals as $C$ is required to contain. Hence $C$ cannot be anyone
of them. Therefore, there are no pure quantum codes of distance greater than 3
that meet the quantum Hamming bound.
\end{proof}
Since it is not known if the quantum Hamming bound holds for
nonbinary degenerate quantum codes with distance $d >5$, it would be
interesting to find degenerate quantum codes that either meet or
beat the quantum Hamming bound~\cite{aly07hamming}. This is
obviously a challenging open research problem.


\chapter{Quantum BCH Codes}\label{ch_QBC_BCH}

 An attractive feature of BCH codes is that one can infer valuable information
from their design parameters (length, size of the finite field, and designed
distance), such as bounds on the minimum distance and dimension of the code. In
this chapter, we show that  one can also deduce from the design parameters
whether or not a primitive, narrow-sense BCH contains its Euclidean or
Hermitian dual code. This information is invaluable in the construction of
quantum BCH codes. A new proof is provided for the dimension of BCH codes with
small designed distance, and simple bounds on the minimum distance of such
codes and their duals are derived as a consequence. These results allow us to
derive the parameters of two families of primitive quantum BCH codes as a
function of their design parameters. This chapter is based on a joint work with
P.K. Sarvepalli and A. Klappenecker and it was presented
in~\cite{aly07a,aly06a}.


\section{BCH Codes}
The Bose-Chaudhuri-Hocquenghem (BCH)
codes~\cite{Bose60a,Bose60b,Gorenstein61,Hocquenghem59} are a well-studied
class of cyclic codes that have found numerous applications in classical and
more recently in quantum information processing. Recall that a cyclic code of
length $n$ over a finite field $\F_q$ with $q$ elements, and $\gcd(n,q)=1$, is
called a \textsl{BCH code with designed distance $\delta$} if its generator
polynomial is of the form
$$g(x)=\prod_{z\in Z} (x-\alpha^z), \qquad Z=C_b\cup \cdots \cup
C_{b+\delta-2},$$ where $C_x=\{ xq^k\bmod n \,|\, k\in \Z, k\ge 0\,\}$ denotes
the $q$-ary cyclotomic coset of $x$ modulo~$n$, $\alpha$ is a primitive element
of $\F_{q^m}$, and $m=\ord_n(q)$ is the multiplicative order of $q$ modulo $n$.
Such a code is called primitive if $n=q^m-1$, and narrow-sense if $b=1$.

An attractive feature of a (narrow-sense) BCH code is that one can derive many
structural properties of the code from the knowledge of the parameters $n$,
$q$, and $\delta$ alone.  Perhaps the most well-known facts are that such a
code has minimum distance $d\ge \delta$ and dimension $k\ge
n-(\delta-1)\ord_n(q)$.  In this chapter, we will show that a necessary
condition for a narrow-sense BCH code which contains its Euclidean dual code is
that its designed distance $\delta=O(qn^{1/2})$. We also derive a sufficient
condition for dual containing BCH codes. Moreover, if the codes are primitive,
these conditions are same. These results allow us to derive families of quantum
stabilizer codes. Along the way, we find new results concerning the minimum
distance and dimension of classical BCH codes.

To put our results into context, we give a brief overview of related work in
quantum BCH codes. This chapter was motivated by problems concerning quantum
BCH codes; specifically, our goal was to derive the parameters of the quantum
codes as a function of the design parameters.  Examples of certain binary
quantum BCH codes have been given by many authors, see, for
example,~\cite{calderbank98,grassl99b,grassl00,steane96}. Steane
\cite{steane99} gave a simple criterion to decide when a binary narrow-sense
primitive BCH code contains its dual, given the design distance and the length
of the code.  We generalize Steane's result in various ways, in particular, to
narrow-sense (not necessarily primitive) BCH codes over arbitrary finite fields
with respect to Euclidean and Hermitian duality. These results allow one to
derive quantum BCH codes; however, it remains to determine the dimension,
purity, and minimum distance of such quantum codes.

The dimension of a classical BCH code can be bounded by many different standard
methods, see~\cite{berlekamp68,huffman03,macwilliams77} and the references
therein. An upper bound on the dimension was given by
Shparlinski~\cite{shparlinski88}, see also~\cite[Chapter~17]{konyagin99}. More
recently, the dimension of primitive narrow-sense BCH codes of designed
distance $\delta< q^{\lceil m/2\rceil}+1$ was apparently determined by Yue and
Hu~\cite{yue96}, according to reference~\cite{yue00}.  We generalize their
result and determine the dimension of narrow-sense BCH codes  for a certain
range of designed distances. As desired, this result allows us to explicitly
obtain the dimension of the quantum codes without computation of cyclotomic
cosets.

The purity and minimum distance of a quantum BCH code depend on the minimum
distance and dual distance of the associated classical code. In general, it is
a difficult problem to determine the true minimum distance of BCH codes,
see~\cite{charpin98}. A lower bound on the dual distance can be given by the
Carlitz-Uchiyama-type bounds when the number of field elements is prime, see,
for example,~\cite[page~280]{macwilliams77} and \cite{stichtenoth94}. Many
authors have determined the true minimum distance of BCH codes in special
cases, see, for instance,~\cite{peterson72},\cite{yue00}.

 We refer to such a code as a
$\B(n,q;\delta)$ code, and call $Z$ the defining set of the code. The basic
properties of these classical codes are discussed, for example, in the books
\cite{huffman03,kabatiansky05,macwilliams77}.

Given a classical BCH code, we can use one of the following well-known constructions
to derive a quantum stabilizer code:
\begin{enumerate}
\item If there exists a classical linear $[n,k,d]_q$ code $C$ such
that $C^\perp\subseteq C$, then there exists an $[[n,2k-n,\ge d]]_q$ stabilizer code
that is pure to $d$. If the minimum distance of $C^\perp$ exceeds $d$, then the
quantum code is pure and has minimum distance $d$.
\item If there exists a classical linear $[n,k,d]_{q^2}$ code $D$ such
that $D^\hdual\subseteq D$, then there exists an $[[n,2k-n,\ge d]]_{q}$ stabilizer
code that is pure to $d$. If the minimum distance of $D^\hdual$ exceeds $d$, then
the quantum code is pure and has minimum distance $d$.
\end{enumerate}
The orthogonality relations are defined in the \textit{Notations}\/ at the end of
this section.  Examples of certain binary quantum BCH codes have been given in
\cite{calderbank98,grassl99b,grassl04,steane96}.

Our goal is to derive the parameters of the quantum stabilizer code as a function of
their design parameters $n$, $q$, and $\delta$ of the associated primitive,
narrow-sense BCH code $C$. This entails the following tasks:
\begin{enumerate}
\item[a)] Determine the design parameters for which $C^\perp\subseteq C$;
\item[b)] determine the dimension of $C$;
\item[c)] bound the minimum distance of $C$ and~$C^\perp$.
\end{enumerate}
In case $q$ is a perfect square, we would also like to answer the Hermitian versions
of questions a) and c):
\begin{enumerate}
\item[a')] Determine the design parameters for which $C^\hdual\subseteq C$;
\item[c')] bound the minimum distance of $C$ and~$C^\hdual$.
\end{enumerate}
\smallskip

To put our work into perspective, we sketch our results and give a brief overview of
related work.

Let $C$ be a primitive, narrow-sense BCH code $C$ of length $n=q^m-1$, $m\ge 2$,
over $\F_q$ with designed distance $\delta$.

To answer question a), we prove in Theorem~\ref{th:dual} that $C^\perp\subseteq C$
holds if and only if $\delta\le q^{\ceil{m/2}}-1-(q-2)[m \text{ odd}]$. The
significance of this result is that allows one to identify all BCH codes that can be
used in the quantum code construction~1). Fortunately, this question can be answered
now without computations. Steane proved in~\cite{steane99} the special case $q=2$,
which is easier to show, since in this case there is no difference between even and
odd $m$.

In Theorem~\ref{th:hdual}, we answer question a') and show that $C^\hdual\subseteq
C$ if and only if $\delta\le q^{(m+[\text{$m$ even}])/2}-1-(q-2)[\text{$m$ even}]$,
where we assume that $q$ is a perfect square.  This result allows us to determine
all primitive, narrow-sense BCH codes that can be used in construction 2). We are
not aware of any prior work concerning the Hermitian case.

In the binary case, an answer to question b) was given by MacWilliams and
Sloane~\cite[Chapter 9, Corollary 8]{macwilliams77}.  Apparently, Yue and Hu
answered question b) in the case of small designed distances~\cite{yue96}. We give a
new proof of this result in Theorem~\ref{th:bchdimension} and show that the
dimension $k=n-m\ceil{(\delta-1)(1-1/q)}$ for $\delta$ in the range $2\le
\delta<q^{\ceil{m/2}}+1$. As a consequence of our answer to b), we obtain the
dimensions of the quantum codes in constructions 1) and 2).

Finding the true minimum distance of BCH codes is an open problem for which a
complete answer seems out of reach, see~\cite{charpin98}.  As a simple consequence
of our answer to b), we obtain better bounds on the minimum distance for some BCH
codes, and we derive simple bounds on the (Hermitian) dual distance of BCH codes
with small designed distance, which partly answers c) and c').

In Section~\ref{sec:qcodes}, all these results are used to derive two families of
quantum BCH codes. Impatient readers should now browse this section to get the
bigger picture. Theorem~\ref{sh:euclid} yields the result that one obtains using
construction 1). Cohen, Encheva, and Litsyn derived in \cite{cohen99} the special
case $q=2$ of our theorem by combining the results of Steane, and MacWilliams and
Sloane that we have mentioned already. The result of construction 2) is given in
Theorem~\ref{sh:hermite}.

\medskip
\textit{Notations.} We denote the ring of integers by $\mathbf{Z}$ and a finite
field with $q$ elements by $\mathbf{F}_q$. We follow Knuth and attribute to $[P(k)]$
the value 1 if the property $P(k)$ of the integer $k$ is true, and 0 otherwise. For
instance, we have $[k \text{ even}]= k-1\bmod 2$, but the left hand side seems more
readable.  If $x$ and $y$ are vectors in $\F_q^n$, then we write $x\perp y$ if and
only if $x\cdot y=0$. Similarly, if $x$ and $y$ are vectors in $\F_{q^2}^n$, then we
write $x\,\hdual\, y$ if and only if $x^q\cdot y=0$.

\section{Dimension and Minimum Distance}
In this section we determine the dimension of primitive, narrow-sense BCH codes of
length $n$ with small designed distance. Furthermore, we derive bounds on the
minimum distance of such codes and their duals.

\subsection{Dimension}
First, we make some simple observations about cyclotomic cosets that are essential
in our proof.

\begin{lemma} \label{th:bchcosetsize}
If\/ $q$ be a power of a prime, $m$ a positive integer and $n=q^m-1$, then all
$q$-ary cyclotomic cosets $C_x=\{xq^\ell\bmod n\,|\,\ell \in \Z\}$ with $x$ in the
range $1\le x< q^{\lceil m/2\rceil}+1$ have cardinality $|C_x|=m$.
\end{lemma}
\begin{proof}
Seeking a contradiction, we assume that $|C_x|< m$. If $m=1$, then $C_x$ would have
to be the empty set, which is impossible. If $m>1$, then $|C_x|<m$ implies that
there must exist an integer $j$\/ in the range $1\le j<m$ such that $j$ divides $m$
and $xq^j\equiv x \mod n$. In other words, $q^m-1$ divides $x(q^j-1)$; hence, $x\geq
(q^m-1)/(q^j-1)$.

If $m$ is even, then $j\leq m/2$; thus, $x\geq q^{ m/2}+1$. If $m$ is odd, then
$j\leq m/3$ and it follows that $x\geq (q^m-1)/(q^{m/3}-1)$, and it is easy to see
that the latter term is larger than $q^{\lceil m/2\rceil}+1$. In both cases this
contradicts our assumption that $1\le x\le q^{\lceil m/2\rceil}$; hence $|C_x|=m$.
\end{proof}

\begin{lemma}\label{th:disjointcosets}
Let $q$ be a power of a prime, $m$ a positive integer, and $n=q^m-1$. Let $x$ and
$y$ be integers in the range $1\le x,y< q^{\lceil m/2\rceil}+1$ such that $x,y\not
\equiv 0 \bmod q$. If $x \neq y$, then the $q$-ary cosets of\/ $x$ and $y$
modulo~$n$ are disjoint, i.e., $C_x\neq C_y$.
\end{lemma}

\begin{proof}
Seeking a contradiction, we assume that $C_x=C_y$. This assumption implies that
$y\equiv xq^\ell \bmod n$ for some integer $\ell$ in the range $1\le \ell <m$.

If $xq^\ell<n$, then $xq^\ell \equiv 0\bmod q$; this contradicts our assumption
$y\not\equiv 0 \bmod q$, so we must have $xq^\ell \ge n$. It follows from the range
of $x$ that $\ell$ must be at least ${\lfloor m/2\rfloor}$.

If $\ell={\lfloor m/2\rfloor}$, then we cannot find an admissible $x$ within the
given range such that $y\equiv xq^{\lfloor m/2\rfloor}\bmod n$. Indeed, it follows
from the inequality $xq^{\lfloor m/2\rfloor}\ge n$ that $x\ge q^{\lceil m/2\rceil}$,
so $x$ must equal $q^{\lceil m/2\rceil}$, but that contradicts $x\not\equiv 0\bmod
q$. Therefore, $\ell$ must exceed $\lfloor m/2\rfloor$.

Let us write $x$ as a $q$-ary number $x=x_0+x_1q+\cdots+x_{m-1}q^{m-1}$, with $0\le
x_i<q$. Note that $x_0\neq 0$ because $x\not\equiv 0\bmod q$. If $\lfloor
m/2\rfloor<\ell<m$, then $xq^\ell$ is congruent to $ y_0=x_{m-\ell}+\cdots +
x_{m-1}q^{\ell-1}+x_0q^{\ell}+\cdots + x_{m-\ell-1}q^{m-1}$ modulo $n$. We observe
that $y_0\ge x_0q^\ell\ge q^{\lceil m/2\rceil}$. Since $y\not\equiv 0 \bmod q$, it
follows that $y=y_0\geq q^{\lceil m/2\rceil}+1$, contradicting the assumed range of
$y$.
\end{proof}

The previous two observations about cyclotomic cosets allow us to derive a closed
form for the dimension of a primitive BCH code. This result generalizes binary case
\cite[Corollary~9.8, page~263]{macwilliams77}. See also \cite{stichtenoth90} which
gives estimates on the dimension of BCH codes among other things.
\begin{theorem}\label{th:bchdimension}
A primitive, narrow-sense BCH code of length $q^m-1$ over $\F_q$ with designed
distance $\delta$ in the range $2 \leq \delta \le q^{\lceil m/2 \rceil}+1$ has
dimension
\begin{equation}\label{eq:dimension}
k=q^m-1-m\lceil (\delta-1)(1-1/q)\rceil.
\end{equation}
\end{theorem}

\begin{proof}
The defining set of the code is of the form $Z=C_1\cup C_2\cdots \cup C_{\delta-1}$,
a union of at most $\delta -1$ consecutive cyclotomic cosets. However, when $1\leq
x\leq \delta-1$ is a multiple of $q$, then $C_{x/q}=C_x$. Therefore, the number of
cosets is reduced by $\lfloor(\delta-1)/q \rfloor$. By
Lemma~\ref{th:disjointcosets}, if $x, y\not\equiv 0 \bmod q$ and $x\neq y$, then the
cosets $C_x$ and $C_y$ are disjoint. Thus, $Z$ is the union of $(\delta-1)-\lfloor
(\delta-1)/q\rfloor= \lceil (\delta-1)(1-1/q)\rceil$ distinct cyclotomic cosets. By
Lemma~\ref{th:bchcosetsize} all these cosets have cardinality~$m$.  Therefore, the
degree of the generator polynomial is $m\lceil (\delta-1)(1-1/q)\rceil$, which
proves our claim about the dimension of the code.
\end{proof}

If we exceed the range of the designed distance in the hypothesis of the previous
theorem, then our dimension formula (\ref{eq:dimension}) is no longer valid, as our
next example illustrates.

\begin{exampleX}
Consider a primitive, narrow-sense BCH code of length $n=4^2-1=15$ over $\F_4$.  If
we choose the designed distance $\delta =6 > 4^1+1$, then the resulting code has
dimension $k=8$, because the defining set $Z$ is given by
$$ Z= C_1\cup C_2\cup \cdots \cup C_5 = \{1,4\}\cup \{2,8\}\cup
\{3,12\}\cup \{5\}.$$ The dimension formula (\ref{eq:dimension}) yields
$4^2-1-2\lceil (6-1)(1-1/4)\rceil=7$, so the formula does not extend beyond the
range of designed distances given in Theorem \ref{th:bchdimension}.
\end{exampleX}

\subsection{Distance Bounds}

The true minimum distance $d_{min}$ of a primitive BCH code over $\F_q$ with
designed distance $\delta$ is bounded by $\delta\le d_{min}\le q\delta-1$, see
\cite[p.~261]{macwilliams77}.  If we apply the Farr bound (essentially the sphere
packing bound) using the dimension given in Theorem~\ref{th:bchdimension}, then we
obtain:

\begin{corollary}\label{th:mindist}
If $C$ is primitive, narrow-sense BCH code of length $q^m-1$ over $\F_q$ with
designed distance $\delta$ in the range $2\le \delta\le q^{\lceil m/2\rceil}+1$ such
that
\begin{eqnarray}\label{eqa3}
\sum_{i=0}^{\lfloor (\delta+1)/2\rfloor} \binom{q^m-1}{i} (q-1)^i >q^{m\lceil
(\delta-1)(1-1/q)\rceil},
\end{eqnarray}
then $C$ has minimum distance $d= \delta$ or $\delta+1$; if, furthermore,
$\delta\equiv 0\bmod q$, then $d=\delta+1$.
\end{corollary}
\begin{proof}
Seeking a contradiction, we assume that the minimum distance~$d$ of the code
satisfies $d \geq \delta+2$. We know from Theorem~\ref{th:bchdimension} that the
dimension of the code is $k=q^m-1-m\lceil (\delta-1)(1-1/q)\rceil.$ If we substitute
this value of $k$ into the sphere-packing bound
$$ q^{k}   \sum_{i=0}^{\lfloor (d-1)/2\rfloor}
\binom{q^m-1}{i} (q-1)^i \leq q^n,
$$
then we obtain
\begin{eqnarray*}
\begin{split}
\sum_{i=0}^{\lfloor (\delta+1)/2\rfloor}\binom{q^m-1}{i}(q-1)^i &\le
\sum_{i=0}^{\lfloor (d-1)/2\rfloor}\binom{q^m-1}{i}(q-1)^i\\&\le q^{m\lceil
(\delta-1)(1-1/q)\rceil},
\end{split}
\end{eqnarray*}
but this contradicts condition~(\ref{eqa3}); hence, $\delta\le d\le \delta+1$.

If $\delta\equiv 0\bmod q$, then the cyclotomic coset $C_\delta$ is contained in the
defining set $Z$ of the code because $C_\delta=C_{\delta/q}$. Thus, the BCH bound
implies that the minimum distance must be at least $\delta+1$.
\end{proof}

%

\begin{corollary}
A primitive, narrow sense BCH code of length $n=q^m-1$ over $\F_q$ with designed
distance $\delta$ in the range $2 \leq \delta \le q^{\lceil m/2 \rceil}+1$ that
satisfies
\begin{eqnarray}\label{eq:griesmer}
n< \sum_{i=0}^{k-1} \left\lceil \frac{\delta+1}{q^i}\right\rceil, \quad\text{ with }
\quad k=n-m\lceil (\delta-1)(1-1/q)\rceil,
\end{eqnarray}
has minimum distance $\delta$.
\end{corollary}
\begin{proof}
This follows from Theorem~\ref{th:bchdimension} and the Griesmer bound.
\end{proof}
\textit{Remark.} The two competing requirements on the designed distance in the
hypothesis of this corollary limit its applicability. We can use the same proof
technique for codes with larger minimum distance if we replace $k$ in equation
(\ref{eq:griesmer}) by a suitable bound. Generalizing our observations about
cyclotomic cosets in the previous section could improve the trivial bound $k\ge
q^m-1-m(\delta-1)$.

\begin{exampleX} Consider a primitive, narrow-sense BCH code of length
$n=3^{2}-1$ over $F_3$. Let $\delta = 4$,  it can be seen that
 $ \sum_{i=0}^{2} 2^i \left(%
\begin{array}{c}
  8 \\
  i \\
\end{array}%
\right)  > 3^{4}$.  This means that condition (\ref{eqa3}) holds, then by
Corollary~\ref{th:mindist}, the code of length 8 and designed distance $ \delta =4 $
has a minimum distance $ d_{min} = 4$.  To verify that, let us construct  a
primitive narrow-sense BCH code with length $n=8$ and designed distance $\delta=4$.
We have $k= q^m-1 -m \lceil2t(1-1/q) \rceil =4$ and the generator polynomial is
$g(x) = 2+x+x^3+x^4$ and the parity check polynomial is $h(x) = 1+x+x^2+2x^3+x^4$.

So, $h_R(x) = 1+2x+x^2+x^3+x^4$ and the parity check matrix is
 \[ H = \left(%
\begin{array}{cccccccc}
  1 & 1 & 1 & 2 & 1 & 0 & 0 &0 \\
  0 & 1 & 1 & 1 & 2 & 1 & 0 & 0\\
  0 & 0 & 1 & 1 & 1 & 2 & 1 & 0\\
  0 & 0 & 0 & 1 & 1 & 1 & 2 & 1\\
\end{array}%
\right)\] by subtracting columns 4 and 5 then add the result to columns 1 and 2,  we
found that the min distance for this matrix H is 4 that verifies our claim in
Corollary~\ref{th:mindist} where $2t+1 \equiv 0 \bmod 3$.
\end{exampleX}

\begin{lemma}\label{th:dualdist}
Suppose that $C$ is a primitive, narrow-sense BCH code of length $n=q^m-1$ over
$\F_q$ with designed distance $2\leq \delta\le \delta_{\max}=q^{\lceil
m/2\rceil}-1-(q-2)[m \textup{ odd}])$, then the dual distance $d^\perp \geq
\delta_{\max} + 1$.
\end{lemma}
\begin{proof}
Let $N=\{0,1,\ldots,n-1 \}$ and $Z_{\delta}$ be the defining set of $C$. We know
that $Z_{\delta_{\max}}\supseteq Z_{\delta}\supset \{1,\ldots,\delta-1 \}$.
Therefore $N\setminus Z_{\delta_{\max}} \subseteq N\setminus Z_{\delta}$.  Further,
we know that $Z\cap Z^{-1}=\emptyset$ if $2\leq \delta\leq \delta_{\max }$ from
Lemma~\ref{th:selforthogonal} and Theorem~\ref{th:dual}. Therefore,
$Z^{-1}_{\delta_{\max}}\subseteq N\setminus Z_{\delta_{\max}}\subseteq N\setminus
Z_{\delta}$.

Let $T_{\delta}$ be the defining set of the dual code. Then $T_{\delta}=(N\setminus
Z_{\delta})^{-1} \supseteq Z_{\delta_{\max}}$. Moreover $\{0\}\in N\setminus
Z_{\delta}$ and therefore $T_{\delta}$. Thus there are at l east $\delta_{\max}$
consecutive roots in $T_{\delta}$. Thus the dual distance $d^\perp \geq
\delta_{\max}+1$.
\end{proof}

\begin{lemma}\label{th:hdualdist}
Suppose that $C$ is a primitive, narrow-sense BCH code of length $n=q^{2m}-1$ over
$\F_{q^2}$ with designed distance $2\leq \delta\le \delta_{\max}=q^{m+[\text{$m$
even}]}-1-(q^2-2)[m \textup{ even}])$, then the dual distance $d^\perp \geq
\delta_{\max} + 1$.
\end{lemma}
\begin{proof}
The proof is analogous to the one of Lemma~\ref{th:dualdist}; just keep in mind that
the defining set $Z_\delta$ is invariant under multiplication by $q^2$ modulo $n$.
\end{proof}

\section{Euclidean Dual Codes}
Recall that the Euclidean dual code $C^\perp$ of a code $C\subseteq \F_q^n$ is given
by $C^\perp = \{ y\in \F_q^n \,|\, x\cdot y =0 \mbox{ for all } x \in C \}.$ Steane
showed in~\cite{steane99} that a primitive binary BCH code of length $2^m-1$
contains its dual if and only if its designed distance $\delta$ satisfies $\delta
\leq 2^{\lceil m/2\rceil}-1$. In this section we derive a similar condition for
nonbinary BCH codes.

\begin{lemma}\label{th:selforthogonal}
Suppose that $\gcd(n,q)=1$. A cyclic code of length $n$ over $\F_q$ with defining
set $Z$ contains its Euclidean dual code if and only if $Z\cap Z^{-1} = \emptyset$,
where $Z^{-1}$ denotes the set $Z^{-1}=\{-z\bmod n\mid z \in Z \}$.
\end{lemma}
\begin{proof}
See, for instance,~\cite[Theorem 4.4.11]{huffman03}.
\end{proof}

\begin{theorem} \label{th:dual}
A primitive, narrow-sense BCH code of length $q^m-1$, with $m\ge 2$, over the finite
field\/ $\F_q$ contains its dual code if and only if its designed distance $\delta$
satisfies
$$\delta \leq
\delta_{\max}=q^{\lceil m/2\rceil}-1-(q-2)[m \textup{ odd}].$$
\end{theorem}
\begin{proof}
Let $n=q^m-1$.  The defining set $Z$ of a primitive, narrow-sense BCH code $C$ of
designed distance $\delta$ is given by $Z=C_1\cup C_2\cdots\cup C_{\delta-1}$, where
$C_x=\{xq^j\bmod n \mid j\in \mathbf{Z} \}$.
\begin{enumerate}
\item We will show that the code
$C$ cannot contain its dual code if the designed distance $\delta>\delta_{\max}$.
Seeking a contradiction, we assume that the defining set $Z$ contains the set
$\{1,\dots,s\}$, where $s=\delta_{\max}$. By Lemma~\ref{th:selforthogonal}, it
suffices to show that $Z\cap Z^{-1}$ is not empty. If $m$ is even, then $s =
q^{m/2}-1$, and $Z^{-1}$ contains the element $-s q^{m/2} \equiv q^{m/2}-1 \equiv s
\bmod n$, which means that $Z\cap Z^{-1} \neq \emptyset$; contradiction. If $m$ is
odd, then $s=q^{(m+1)/2}-q+1$, and the element given by $-sq^{(m-1)/2}\equiv
q^{(m+1)/2}-q^{(m-1)/2} -1 \bmod n$ is contained in $Z^{-1}$. Since this element is
less than $s$ for $m\geq 3$, it is contained in $Z$, so $Z\cap Z^{-1}\neq
\emptyset$; contradiction.  Combining these two cases, we can conclude that
$\delta\leq q^{\lceil m/2\rceil}-1-(q-2)[m\mbox{ is odd}]$ for $m\geq 2$.

\item For the converse, we prove that if $\delta\leq \delta_{\max}$,
then $Z\cap Z^{-1}=\emptyset$, which implies $C^\perp \subseteq C$ by
Lemma~\ref{th:selforthogonal}. It suffices to show that $\min C_{-x} \geq
\delta_{\max}$ for any coset $C_x$ in~$Z$. Since $1\leq x<\delta_{\max}\leq
q^{\lceil m/2\rceil}-1$, we can write $x$ as a $q$-ary integer of the form
$x=x_0+x_1q+\cdots+x_{m-1}q^{m-1}$ with $0\le x_i<q$, and $x_i=0$ for $i\ge \lceil
m/2\rceil.$ If $\bar{y}=n-x$, then $\bar{y}= \bar{y}_0+\bar{y}_1q + \cdots +
\bar{y}_{m-1}q^{m-1} =\sum_{i=0}^{m-1}(q-1-x_i)q^i. $ Set $y=\min C_{-x}$. We note
that $y$ is a conjugate of $\bar{y}$.  Thus, the digits of $y$ are obtained by
cyclically shifting the digits of $\bar{y}$.

\item[3a)] First we consider the case when $m$ is even. Then the $q$-ary
expansion of $x$ has at least $m/2$ zero digits. Therefore, at least $m/2$ of the
$\bar{y}_i$ are equal to $q-1$.  Thus, $y\geq \sum_{i=0}^{m/2-1}(q-1)
q^i=q^{m/2}-1=\delta_{\max}$.

\item[3b)] If $m$ is odd, then as $1\leq x<q^{(m+1)/2}-q+1$, we have
$m>1$ and $\bar{y}=\bar{y}_0+\bar{y}_1q + \cdots + (\bar{y}_{(m-1)/2})q^{(m-1)/2} +
(q-1)q^{(m+1)/2}+\cdots + (q-1)q^{m-1}$. For $0\leq j\leq (m-1)/2$, we observe that
$xq^j<n$, and since $\bar{y}q^j \equiv -xq^j \bmod n, \bar{y}q^j  = n-xq^j \geq
q^m-1-(q^{(m+1)/2}-q)q^{(m-1)/2} = q^{(m+1)/2}-1 \geq \delta_{\max}$. For $(m+1)/2
\leq j\leq m-1$, we find that
\begin{eqnarray*}
\begin{split}
\bar{y}q^j \bmod n &=\bar{y}_{m-j}+\cdots+\bar{y}_{(m-1)/2}q^{j-(m+1)/2}
\\& + (q-1)q^{j-(m-1)/2} + \cdots + (q-1)q^{j-1}\\&+\bar{y}_0q^{j} + \cdots +
\bar{y}_{m-j-1}q^{m-1},\\
&\geq  (q^{(m-1)/2}-1)q^{j-(m-1)/2} +\bar{y}_0+\cdots\\&+\bar{y}_{(m-1)/2},\\
& \geq q^{(m+1)/2}-q+ 1 = \delta_{\max},
\end{split}
\end{eqnarray*}
where $\bar{y}_0+\cdots+\bar{y}_{(m-1)/2}\geq 1$ because $x<q^{(m+1)/2}-q+1$. Hence
$y =\min\{\bar{y}q^j\mid j \in \mathbf{Z}\} \geq \delta_{\max}$ when $m$ is odd.
\end{enumerate}
Therefore a primitive BCH code contains its dual if and only if $\delta \leq
\delta_{\max}$, for $m\geq 2$.
\end{proof}

\section{Hermitian Dual Codes}
If the cardinality of the field is a perfect square, then we can define another type
of orthogonality relation for codes.  Recall that if the code $C$ is a subspace of
the vector space $\F_{q^2}^n$, then its Hermitian dual code $C^{\perp_h}$ is given
by $C^{\perp_h}=\{ y\in \F_{q^2}^n\,|\, y^q\cdot x = 0 \mbox{ for all } x \in C\}$,
where $y^q=(y_1^q,\dots,y_n^q)$ denotes the conjugate of the vector
$y=(y_1,\dots,y_n)$. The goal of this section is to establish when a primitive,
narrow-sense BCH code contains its Hermitian dual code.

\begin{lemma}\label{th:hermitian}
Assume that $\gcd(n,q)=1$. A cyclic code of length $n$ over $\F_{q^2}$ with defining
set $Z$ contains its Hermitian dual code if and only if $Z\cap Z^{-q} = \emptyset$,
where $Z^{-q}=\{-qz \bmod n \mid z \in Z \}$.
\end{lemma}
\begin{proof}
Let $N=\{0,1,\dots,n-1\}$. If $g(z)=\prod_{x\in Z} (z-\alpha^x)$ is the generator
polynomial of a cyclic code $C$, then $h^\dagger(z)=\prod_{x\in N\setminus Z}
(z-\alpha^{-qx})$ is the generator polynomial of $C^{\perp_h}$.  Thus,
$C^{\perp_h}\subseteq C$ if and only if $g(z)$ divides $h^\dagger(z)$. The latter
condition is equivalent to $Z\subseteq \{ -qx\,|\, x\in N\setminus Z\}$, which can
also be expressed as $Z\cap Z^{-q}=\emptyset$.
\end{proof}

\begin{theorem}\label{th:hdual}
A primitive, narrow-sense BCH code of length $q^{2m}-1$ over $\F_{q^2}$, where
$m\neq 2$, contains its Hermitian dual code if and only if its designed distance
$\delta$ satisfies
\begin{equation*}
\delta\leq \delta_{\max}=q^{m+[m \text{ even}]}-1-(q^2-2)[m \text{ even}].
\end{equation*}
\end{theorem}
\begin{proof}
Let $n=q^{2m}-1$. Recall that the defining set $Z$ of a primitive, narrow-sense BCH
code $C$ over the finite field $\F_{q^2}$ with designed distance $\delta$ is given
by $Z=C_1\cup \cdots \cup C_{\delta-1}$ with $C_x=\{ xq^{2j}\bmod n\,|\, j\in \Z
\}$.
\begin{enumerate}
\item We will show that the code $C$ cannot contain its Hermitian dual
code if the designed distance $\delta> \delta_{\max}$. Seeking a contradiction, we
assume that the defining set $Z$ contains $\{1,\dots,s\}$, where $s=\delta_{\max}$.
By Lemma~\ref{th:hermitian}, it suffices to show that $Z\cap Z^{-q}$ is not empty.
If $m$ is odd, then $s=q^{m}-1$. Notice that $n-qs q^{2(m-1)/2}=q^{m}-1=s$, which
means that $s\in Z\cap Z^{-q}$, and this contradicts our assumption that this set is
empty. If $m$ is even, then $s=q^{m+1}-q^2+1$. We note that $n-qsq^{m-2} =
q^{m+1}-q^{m-1}-1 < s=q^{m+1}-q^2+1$, for $m>2$. It follows that
$q^{m+1}-q^{m-1}-1\in Z\cap Z^{-q}$, contradicting our assumption that this set is
empty.
Combining the two cases, we can conclude that $s$ must be smaller than the value
$q^{m+[m \text{ even}]}-1-(q^2-2)[m \text{ even}]$.

\item For the converse, we show that if $\delta<\delta_{\max}$, then
$Z\cap Z^{-q}=\emptyset$, which implies $C^{\perp_h}\subseteq C$ thanks to
Lemma~\ref{th:hermitian}. It suffices to show that $\min \{n-qC_x\} \geq
\delta_{\max}$ or, equivalently, that $\max qC_x \leq n-\delta_{\max}$ holds for
$1\le x\le \delta-1$.

\item If $m$ is odd, then the $q$-ary expansion of $x$ is of the form
$x=x_0+x_1q+\cdots+x_{m-1}q^{m-1}$, with $x_i=0$, for $m \leq i\leq 2m-1$ as
$x<q^{m}-1$. So at least $m$ of the $x_i$ are equal to zero, which implies $\max
qC_x <q^{2m}-1-(q^m-1)=n-\delta_{\max}$.

\item Let $m$ be even and $qxq^{2j}$ be the $q^2$-ary conjugates of
$qx$. Since $x<q^{m+1}-q^2+1$, $x=x_0+x_1q+\cdots+x_{m}q^m$ and at least one of the
$x_i \leq q-2$. If $0\leq 2j\leq m-2$, then $qxq^{2j}\leq q(q^{m+1}-q^2)q^{m-2} =
q^{2m}-q^{m+1} = n-q^{m+1}+1< n-\delta_{\max}$. If $2j=m$, then $qxq^{m} =
x_{m-1}+x_{m}q+0.q^2+\cdots+ 0.q^{m}+x_0q^{m+1}\cdots+x_{m-2}q^{2m-1}$. We note that
there occurs a consecutive string of $m-1$ zeros and because one of the $x_i\leq
q-2$, we have $qxq^{2j}<n- q^{2}(q^{m-1}-1)-1\leq n-\delta_{\max}$. For $m+2 \leq
2j\leq 2m-2$, we see that $qxq^{2j}< n-q^4(q^{m-1}-1)<n-\delta_{\max}$.
\end{enumerate}
Thus we can conclude that the primitive BCH codes contain their Hermitian duals when
$\delta\leq q^{m+[m \text{ even}]}-1-(q^2-2)[m \text{ even}]$.
\end{proof}

\section{Families of Quantum BCH Codes}\label{sec:qcodes}
We use the results of the previous sections to prove the existence of quantum
stabilizer codes. We use the CSS construction as shown in the previous Chapter.

\begin{theorem}\label{sh:euclid}
If $q$ is a power of a prime, and $m$ and $\delta$ are integers such that $m\ge 2$
and $2\le \delta\le \delta_{\max}=q^{\lceil m/2\rceil}-1-(q-2)[m \text{ odd}]$, then
there exists a quantum stabilizer code $Q$ with parameters
$$[[q^m-1,q^m-1-2m\lceil(\delta-1)(1-1/q)\rceil, d_Q\ge \delta]]_q$$
that is pure up to $\delta$. If $\B(n,q;\delta)$ has true minimum distance $d$, and
$d\le \delta_{\max}$, then $Q$ is a pure quantum code with minimum distance $d_Q=d$.
\end{theorem}
\begin{proof}
Theorem~\ref{th:bchdimension} and \ref{th:dual} imply that there exists a classical
BCH code with parameters $[q^m-1,q^m-1-m\lceil(\delta-1)(1-1/q)\rceil,\ge \delta]_q$
which contains its dual code. An $[n,k,d]_q$ code that contains its dual code
implies the existence of the quantum code with parameters $[[n,2k-n,\ge d]]_q$ by
the CSS construction, see~\cite{grassl04}, \cite{grassl99b}.  By
Lemma~\ref{th:dualdist}, the dual distance exceeds $\delta_{\max}$; the statement
about the purity and minimum distance is an immediate consequence.
\end{proof}

\begin{theorem}\label{sh:hermite}
If $q$ is a power of a prime, $m$ is a positive integer, and $\delta$ is an integer
in the range $2\le \delta \le \delta_{\max}=q^{m+[m \textup{ even}]}-1 -(q^2-2)[m
\textup{ even}]$, then there exists a quantum code $Q$ with parameters
$$ [[q^{2m}-1, q^{2m}-1-2m\lceil(\delta-1)(1-1/q^2)\rceil , d_Q\ge
\delta]]_q$$ that is pure up to $\delta$. If $\B(n,q^2;\delta)$ has true minimum
distance $d$, with $d< \delta_{\max}$, then $Q$ is a pure quantum code of minimum
distance $d_Q=d$.
\end{theorem}
\begin{proof}
It follows from Theorems~\ref{th:bchdimension} and~\ref{th:hdual} that there exists
a primitive, narrow-sense $[q^{2m}-1,q^{2m}-1-m\lceil
(\delta-1)(1-1/q^2)\rceil,\ge\delta]_{q^2}$ BCH code that contains its Hermitian
dual code.  Recall that if a classical $[n,k,d]_{q^2}$ code $C$ exists that contains
its Hermitian dual code, then there exists an $[[n,2k-n,\ge d]]_q$ quantum code that
is pure up to $d$, see~\cite{ashikhmin01}; this proves our claim. By
Lemma~\ref{th:hdualdist}, the Hermitian dual distance exceeds $\delta_{\max}$, which
implies the last statement of the claim.
\end{proof}


\section{Quantum BCH   from Self-orthogonal Product Codes }
It has been shown that product codes have a special interest because they have
simple decoding algorithms and high bit rates. Furthermore, the Quantum BCH
 codes have much higher rates than the corresponding classical product codes. We
apply an important result by Grassl~\cite[Theorem 5-8 ]{grassl05} in quantum
block codes.

Let $C_i=[n_i,k_i,d_i]_q$ be a linear code over finite field $\F_q$ with
generator matrix $G_i$ for $i \in \{1,2 \} $. Then the linear code $C=[
n_1n_2,k_1k_2,d_1d_2]_q$ is the product code of $C_1 \otimes C_2$ with
generator matrix $G=G_1\otimes G_2$, see~\cite{forney05b,grassl05,ollivier04}.

\begin{lemma}\label{QCC-productcodes}
Let $C_E \subseteq C_E^{\perp}$ and  $C_H \subseteq C_H^\perp$ denote two codes
which are self-orthogonal with respect to the Euclidean and Hermitian inner
products, respectively. Also, Let C and D denote arbitrary linear codes over
$\F_q$ and $\F_{q^2}$, respectively. Then $C \otimes C_E$ and $D \otimes C_H$
are Euclidean and Hermitian self-orthogonal codes, respectively. Furthermore,
the minimum distance of the dual of the product code $C \otimes C_E$  ($D
\otimes C_H$) cannot exceed the minimum distance of the dual distance of $C
(D)$ and the dual distance of $C_E (C_H)$.
\end{lemma}

\begin{proof}
See~\cite[Theorem 7, Corollary 6 ]{grassl05}.
\end{proof}

We can explicitly determine dimension of the new self-orthogonal product code if we
know dimension of the original two self-orthogonal codes. 
Therefore, we apply our previous result in dimension of BCH codes  as shown in
section 2 into Lemmas \ref{BCH-twoproductcodes} and \ref{BCH-RS-productcodes}.

\begin{lemma}\label{BCH-twoproductcodes}
Let $C_i$ be a primitive narrow-sense BCH code with length $n_i=q^{m_i}-1$ and
designed distance $2 \leq \delta_i \leq q^{\lceil m_i/2\rceil}-1-(q-2)[m_i \textup{
odd}]$ over finite field $\F_q$ for $i \in \{1,2\}$. Then the product code
$$C_1 \otimes C_2^\perp= [n_1n_2,k_1 (n_2-k_2),\geq \delta_1 \wt(C_2^\perp)]_q$$
is self-orthogonal and its Euclidean dual code is
$$(C_1 \otimes C_2^\perp)^\perp= [n_1n_2,n_1n_2-k_1(n_2-k_2),\geq
\min(\wt(C_1^\perp),\delta_2)]_q$$ where $k_i=q^{m_i}-1-m_i\lceil
(\delta_i-1)(1-1/q)\rceil$ and $\wt(C_i^\perp) \geq \delta_i$.
\end{lemma}
\begin{proof}
We know that if $2\leq \delta_2 \leq q^{m/2}-1$, then $C_2$ contains its Euclidean
dual as shown in Theorem \ref{th:dual}. From \cite[Theorem 5]{grassl05} and Lemma
\ref{QCC-productcodes}, we conclude that the product code $C_1 \otimes C_2^\perp$ is
Euclidean self-orthogonal.
\end{proof}

\begin{lemma}\label{BCH-RS-productcodes}
Let $C_1=[n,k,d]$ be a primitive narrow-sense BCH code with length $n=q^{m}-1$ and
designed distance $2 \leq \delta \leq q^{m/2}-1$ over
 $\F_q$ . Furthermore, let $C_2=[q-1,q-\delta_2,\delta_2]$ be a
self-orthogonal Reed-Solomon code. Then the product code
$$C_1 \otimes C_2= [(q-1)n,k (q-\delta_2),\geq \delta_1\delta_2]_q$$
is self-orthogonal with parameters

\begin{eqnarray*}
\begin{split}(C_1 \otimes C_2)^\perp &= [(q-1)n,(q-1)n-k
(q-\delta_2), \\&\geq \min(\wt(C_1^\perp),q-\delta_2)]_q
\end{split}
\end{eqnarray*}

 where $k=q^m-1-m\lceil (\delta_1-1)(1-1/q)\rceil$ and
$\wt(C_1^\perp) \geq \delta_1$.
\end{lemma}
\begin{proof}
Since $C_2$ is a self-orthogonal code, then the dual code $C_2^\perp$ has minimum
distance $ q-\delta_2$ and dimension $\delta_2-1$. From \cite[Theorem 5]{grassl05}
and Lemma \ref{QCC-productcodes}, we conclude that $C_1 \otimes C_2$ is
self-orthogonal. The dual distance of $(C_1 \otimes C_2)^\perp$ comes from lemma
\ref{QCC-productcodes} such that the dual distance of $C_2^\perp$ is
$\wt(C_2^\perp)=q- \delta_2$.
\end{proof}

Now, we generalize the previous two lemmas to any arbitrary primitive BCH codes.
\begin{lemma}\label{BCH-productcodes-general}
Let $C_i$ be a primitive BCH code with length $n_i=q^{m_i}-1$ and designed distance
$2 \leq \delta_i \leq q^{\lceil m_i/2\rceil}-1-(q-2)[m_i \textup{ odd}]$ over $\F_q$
for $i \in \{1,2\}$. Then the product code
$$C_1 \otimes C_2= [n_1n_2,k_1 k_2,\geq \delta_1\delta_2]_q$$
is self-orthogonal with parameters
$$C_1^\perp \otimes C_2^\perp= [n_1n_2,n_1n_2-k_1k_2,\geq min(\delta_1^\perp,\delta_2^\perp)]_q$$
where $k_i=q^m_i-1-m_i\lceil (\delta_i-1)(1-1/q)\rceil$ and $\delta_i^\perp \geq
\delta_i$.
\end{lemma}
\begin{proof}
Direct conclusion and similar proof as Lemma~\ref{BCH-twoproductcodes}.
\end{proof}

Note: Lemmas~\ref{BCH-RS-productcodes} and~\ref{BCH-twoproductcodes} can be
extended to Hermitian self-orthogonal codes. Finally, we can construct families
of quantum error-correcting codes using Lemmas~\ref{BCH-twoproductcodes}
and~\ref{BCH-RS-productcodes}.
\begin{lemma}\label{Qubit-BCH-twocodes}
Let $C_i$ be a primitive narrow-sense BCH code with length $n_i=q^{m_i}-1$ and
designed distance $2 \leq \delta_i \leq q^{\lceil m_i/2\rceil}-1-(q-2)[m_i \textup{
odd}]$ over $\F_q$ for $i \in \{1,2\}$. Furthermore, the product code
$$C_1 \otimes C_2^\perp= [n_1n_2,k_1 (n_2-k_2),\geq \delta_1 \wt(C_2^\perp)]_q$$
is self-orthogonal where $k_i=q^{m_i}-1-m_i\lceil (\delta_i-1)(1-1/q)\rceil$ and
$\wt(C_i^\perp) \geq \delta_i$. Then there exists a quantum error-correcting codes
with parameters $$ [[n_1n_2, n_1n_2- 2k_1 (n_2-k_2), d_{min}  ]]_q.$$
\end{lemma}

\begin{proof}
The proof is a direct consequence.
\end{proof}


\section{Conclusions and Discussion}
We have investigated primitive, narrow-sense BCH codes in this chapter. A
careful analysis of the cyclotomic cosets in the defining set of the code
allowed us to derive a formula for the dimension of the code when the designed
distance is small. We were able to characterize when primitive, narrow-sense
BCH codes contain their Euclidean and Hermitian dual codes, and this allowed us
to derive two series of quantum stabilizer codes.

BCH are an interesting class of codes because one in advance can choose their
design parameters. In the following chapters, we will show that BCH can be used
to derived families of unit memory quantum convolutional codes as well as
families of subsystem codes.

It remains open problem to establish conditions when nonprimitive non-narrow
sense BCH codes contain their Euclidean and Hermitian duals. In general, we do
not know the exact minimum distance of a BCH code with given parameters.

BCH codes can be used to derive LDPC codes.  One can represent elements of the
finite field as zero vectors of the code length except at positions of power of
those elements. In~\cite{aly08c} we derive LDPC codes derived from nonprimitive
BCH codes. This construction can be used to derive families of quantum LDPC
codes.

\chapter{Quantum Duadic Codes}\label{ch_QBC_Qduadic}

 Good quantum codes, such as quantum MDS codes, are typically
nondegenerate (pure), meaning that errors of small weight require
active error-correction, which is---paradoxically---itself prone to
errors. Decoherence free subspaces, on the other hand, do not
require active error correction, but perform poorly in terms of
minimum distance. In this chapter, examples of degenerate (impure)
quantum codes are constructed that have better minimum distance than
decoherence free subspaces and allow some errors of small weight
that do not require active error correction.  In particular, two new
families of $[[n,1,\geq \sqrt{n}]]_q$ degenerate quantum codes are
derived from classical duadic codes. This chapter is based on a
joint work with A. Klappenecker and P.K. Sarvepalli,
see~\cite{aly07c,aly06b}. I aim to provide enough details in
classical duadic codes and degenerate quantum codes, so my results
on quantum duadic codes will be readable.

\section{Introduction}\label{sec:intro}
Suppose that $q$ is a power of a prime $p$. Recall that an
$[[n,k,d]]_q$ quantum stabilizer code $Q$ is a $q^k$-dimensional
subspace of ${\C^{q^n}}$ such that $\langle u|E|u\rangle=\langle
v|E|v\rangle$ holds for any error operator $E$ of weight $\wt(E)<d$
and all $\ket{u}, \ket{v} \in Q$, see~\cite{ashikhmin01,ketkar06}
for details.  The stabilizer code $Q$ is called nondegenerate (or
pure) if and only if $\langle v|E|v\rangle=q^{-n}\tr E$ holds for
all errors $E$ of weight $\wt(E)<d$ where $\tr$ is the trace of $E$;
otherwise, $Q$ is called degenerate. Recall that purity and
nondegeneracy are equivalent notions in the case of stabilizer
codes, see~\cite{calderbank98,gottesman97}.

In spite of the negative connotations of the term ``degenerate'', we
will argue that degeneracy is an interesting and in some sense
useful quality of a quantum code. Let us call an error nice if and
only if it acts by scalar multiplication on the stabilizer code.
Nice errors do not require any correction, which is a nice feature
considering the fact that operational imprecisions of a quantum
computer can introduce errors in a correction step (which is the
main reason why elaborate fault-tolerant implementations are
needed).

If we assume a depolarizing channel, then errors of small weight are more likely to
occur than errors of large weight. If the stabilizer code $Q$ is nondegenerate, then
all nice errors have weight $d$ or larger, so the most probable errors \textit{all}
require (potentially hazardous) active error correction. On the other hand, if the
stabilizer code is degenerate, then there exist nice errors of weight less than the
minimum distance. Given these observations, it would be particularly interesting to
find degenerate stabilizer codes with many nice errors of small weight.

Although the first quantum error-correcting code by Shor was a degenerate
$[[9,1,3]]_2$ stabilizer code, it turns out that most known quantum stabilizer code
families provide pure codes. If one insists on a large minimum distance, then
nondegeneracy seems more or less unavoidable (for example, quantum MDS codes are
necessarily nondegenerate, see~\cite{rains99}). However, the fact that most known
stabilizer codes do not have nice errors of small weight is the result of more
pragmatic considerations.

Let us illustrate this last remark with the CSS construction;
similar points can be made for other stabilizer code constructions.
Suppose we start with a classical self-orthogonal $[n,k,d]_q$ code
$C$, i.e., $C \subseteq C^{\perp}$, then one can obtain with the CSS
construction an $[[n,n-2k, \delta]]_q$ stabilizer code, where
$\delta=\wt(C^\perp\setminus C)$. Since we often do not know the
weight distribution of the code $C$, the easiest way to obtain a
stabilizer code with minimum distance at least $\delta_0$ is to
choose $C$ such that its dual distance $d^\perp\ge \delta_0$, as
this ensures $\delta\ge d^\perp\ge \delta_0$. However, since
$C\subseteq C^\perp$, the side effect is that all nonscalar nice
errors have a weight of at least $d\ge d^\perp\ge \delta_0$.

Our considerations above suggest a different approach. Since we
would like to have nice errors of small weight, we start with a
classical self-orthogonal code $C$ that has a small minimum
distance, but is chosen such that the vector of smallest Hamming
weight in the difference set $C^\perp\setminus C$ is large.  In
general, it is of course difficult  to find a good lower bound for
the weights in this difference set.

We illustrate this approach for degenerate quantum stabilizer codes that are derived
from classical duadic codes. Recall that the duadic codes generalize the quadratic
residue codes, see~\cite{leon84}, \cite{smid86},\cite{smid87}. We show that one can
still obtain a surprisingly large minimum distance, considering the fact we start
with classical codes that are really bad.

The chapter is organized as follows. In Section~\ref{sec:duadic}, we
recall basic properties of duadic codes. In
Section~\ref{sec:euclid}, we construct degenerate quantum stabilizer
codes using the CSS construction. Finally, in
Section~\ref{sec:hermitian}, we obtain further quantum stabilizer
codes using the Hermitian code construction.

\paragraph*{Notation}
Throughout this chapter, $n$ denotes a positive odd integer.  If $a$ is an
integer coprime to $n$, then we denote by $\ord_n(a)$ the multiplicative order
of $a$ modulo $n$. We briefly write $\qr$ to express the fact that $q$ is a
quadratic residue modulo~$n$. We write $p^\alpha\|n$ if and only if the integer
$n$ is divisible by $p^\alpha$ but not by $p^{\alpha+1}$.  If $\gcd(a,n)=1$,
then the map $\mu_a:i\mapsto a i \bmod n$ denotes a permutation on the set $\{
0,1,\ldots,n-1\}$. An element $c=(c_1,\ldots,c_n )\in \F_q^n$ is said to be
even-like if $\sum_i{c_i}=0$, and odd-like otherwise. A code $C\subseteq
\F_q^n$ is said to be even-like if every codeword in $C$ is even-like, and
odd-like otherwise.

\section{Classical Duadic Codes} \label{sec:duadic}
In this section, we recall the definition and basic properties of duadic codes of
length $n$ over a finite field $\F_q$ such that $\gcd(n,q)=1$. For each choice, we
will obtain a quartet of codes: two even-like cyclic codes and two odd-like cyclic
codes.

Let $S_0$, $S_1$ be the defining sets of two cyclic codes of length~$n$ over~$\F_q$
such that
\begin{compactenum}
\item  $S_0\cap S_1=\emptyset$,
\item $S_0\cup S_1=S =
\{1,2,\ldots,n-1 \}$, and
\item $aS_i \bmod n=S_{(i+1 \bmod 2)} $ for some $a$
coprime to $n$.
\end{compactenum}
In particular, each $S_i$ is a union of $q$-ary cyclotomic cosets modulo $n$. Since
condition 3) implies $|S_0|=|S_1|$, we have $|S_i|=(n-1)/2$, whence $n$ must be odd.
The tuple $\{ S_0,S_1,a \}$ is called a \textit{splitting} of $n$ given by the
permutation~$\mu_a$.

Let $\alpha$ be a primitive $n$-th root of unity over $\F_q$.  For
$i\in \{0,1\}$, the odd-like duadic code $D_i$ is a cyclic code of
length~$n$ over~$\F_q$  with defining set $S_i$ and generator
polynomial \begin{eqnarray}g_i(x) = \prod_{j\in S_i}
(x-\alpha^j).\end{eqnarray} The even-like duadic code $C_i$ is
defined as the even-like subcode of $D_i$; thus, it is a cyclic code
with defining set $S_i\cup \{ 0\}$ and generator polynomial
$(x-1)g_i(x)$. The dimension of a cyclic code $D_i$ of length $n$
and generator polynomial $g_i(x)$ is given by
\begin{eqnarray}
k_i=n-deg(g_i(x)). \end{eqnarray}
 The dimension of $D_i$ is $(n+1)/2$ and that of $C_i$ is $(n-1)/2$
respectively. Obviously $C_i\subset D_i$. We have the following
results on the classical duadic codes.

\begin{theorem}\label{th:duadicexist}
Duadic codes of length $n$ over $\F_q$ exist if and only if $q$ is a
quadratic residue modulo $n$, i.e., $\qr$.
\end{theorem}
\begin{proof}
This is well-known, see for example,~\cite[Theorem~1]{smid87}
or~\cite[Theorem~6.3.2, pages~220-221]{huffman03}.
\end{proof}

 It is natural to ask when duadic codes are self-orthogonal, so that the CSS
construction~\cite{calderbank98} can be used.

\begin{lemma}\label{th:duadicduals}
Let $C_i$ and $D_i$ be the even-like and odd-like duadic codes of length $n$ over
$\F_q$, where $i\in \{0,1\}$. Then
\begin{compactenum}[i)]
\item $C_i^\perp=D_i$ if and only if $-S_i \equiv
S_{(i+1\bmod 2)}\bmod n$.
\item $C_i^\perp=D_{(i+1 \bmod 2)}$ if and only if $-S_i \equiv
S_i\bmod n$.
\end{compactenum}
\end{lemma}
\begin{proof}
See \cite[Theorems 6.4.2-3]{huffman03}
\end{proof}
In other words, if the splitting is given by $\mu_{-1}$, then the
even-like duadic codes $C_i$ are self-orthogonal. If $\mu_{-1}$
fixes the set~$S_i$, then $C_1\subset C_0^\perp=D_1$ and $C_0\subset
C_1^\perp=D_0$. This naturally raises the question when $\mu_{-1}$
gives a splitting of $n$ and when it only fixes the codes. For some
special cases of $n$ this is known.
When all prime factors of $n=\prod p_i^{m_i}$ are such that
$p_i\equiv -1\mod 4$, then we have the following result.

\medskip

\begin{lemma}\label{fact:splittings}
Let $n=\prod p_i^{m_i}$ be the prime factorization of an odd integer
$n$, where each $m_i>0$ and $q$ is a quadratic residue modulo $n$.
If every $p_i\equiv -1 \mod 4$, then all the splitters of $n$ are
given by $\mu_{-1}$. On the other hand if at least one $p_i\equiv
1\mod 4$, then there exists a splitting given by $\mu_a$ where
$a\neq -1$.
\end{lemma}
\begin{proof}
 See~\cite[Theorem~8]{smid87}.
\end{proof}

\medskip

Although the weight distribution of a duadic code is not known in
general, the following well-known fact gives partial information
about the weights of odd-like codewords.

\begin{lemma}[Square Root Bound]\label{th:duadicdist}
Let $D_0$ and $D_1$ be a pair of odd-like duadic codes of length $n$ over $\F_q$.
Then their minimum odd-like weights in both codes are same, say $d_o$.  We have
\begin{compactenum}
\item $d_o^2\geq n$,
\item $d_o^2-d_o+1\geq n$ if the splitting is given by $\mu_{-1}$.
\end{compactenum}
\end{lemma}
\begin{proof}
See \cite[Theorem~6.5.2]{huffman03}.
\end{proof}

\section{Quantum Duadic Codes -- Euclidean Case} \label{sec:euclid}

In this section, we derive quantum stabilizer codes from classical duadic code using
the well-known CSS construction.  Recall that in the CSS construction, the existence
of an $[n,k_1]_q$ code~$C$ and an $[n,k_2]_q$ code~$D$ such that $C\subset D$
guarantees the existence of an $[[n,k_2-k_1,d]]_q$ quantum stabilizer code with
minimum distance $d=\min \mbox{wt}\{(D\setminus C)\cup (C^\perp\setminus
D^\perp)\}$.

\subsection{Basic Code Constructions}
Recall that two $\F_q$-linear codes $C_1$ and $C_2$ are said to be equivalent if and
only if there exists a monomial matrix $M$ and automorphism $\gamma$ of $\F_q$ such
that $C_2=C_1M\gamma$, see \cite[page~25]{huffman03}. We denote equivalence of codes
by $C_1\sim C_2$.  For us it is relevant that equivalent codes have the same weight
distribution, see~\cite[page~25]{huffman03}.

The permutation map $\mu_a:i\mapsto ai\bmod n$ also defines an action on polynomials
in $\F_q[x]$ by $f(x)\mu_a=f(x^a)$. This induces an action on a cyclic code $C$ over
$\F_q$ by
$$C\mu_a = \{c(x)\mu_a \mid c(x)\in C\} =\{ c(x^a)\mid c(x)\in C\}.$$
\begin{lemma}\label{th:equivcode}
Let $C$ be a cyclic code of length $n$ over $\F_q$ with defining set $T$. If
$\gcd(a,n)=1$, then the cyclic code $C\mu_a$ has the defining set $a^{-1}T$.
Furthermore, we have $C\mu_a\sim C$.
\end{lemma}
\begin{proof}
This follows from the definitions, see also \cite[Corollary~4.4.5]{huffman03} and
\cite[page~141]{huffman03}.
\end{proof}

\begin{theorem}\label{th:quantumduadic1}
Let $n$ be a positive odd integer, and let $\qr$.  There exist quantum duadic codes
with the parameters $[[n,1,d]]_q$, where $d^2\geq n$.  If $\ord_n(q)$ is odd, then
there also exist quantum duadic codes with minimum distance $d^2-d+1\geq n$.
\end{theorem}
\begin{proof}
Let $N=\{0,1,\dots,n-1\}$.  If $\qr$, then there exist duadic codes $C_i\subset
D_i$, for $i\in \{0,1\}$. Suppose that the defining set of $D_i$ is given by $S_i$;
thus, the defining set of the even-like subcode $C_i$ is given by $S_i\cup \{0\}$.
It follows that $C_i^\perp$ has defining set $-(N\setminus (\{0 \}\cup S_{i})) =
-S_{(i+1\bmod 2)}.$ Using Lemma~\ref{th:equivcode}, we obtain
$C_i^\perp=D_{(i+1\bmod 2)}\mu_{-1} \sim D_{(i+1\bmod 2)}$ and
$D_i^\perp=C_{(i+1\bmod 2)}\mu_{-1}\sim C_{(i+1\bmod 2)}$. By the CSS construction,
there exists an $[[n,(n+1)/2-(n-1)/2,d]]_q$ quantum stabilizer code with minimum
distance $d=\min\{\wt((D_i\setminus C_i)\cup (C_i^\perp\setminus D_i^\perp)) \} $.
Since $C_i^\perp\sim D_{(i+1\bmod 2)}$ and $D_i^\perp\sim C_{(i+1\bmod 2)}$, the
minimum distance $d=\min \{\wt((D_i\setminus C_i)\cup (D_{(i+1\bmod 2)}\setminus
C_{(i+1\bmod 2)}) \}$, which is nothing but the minimum odd-like weight of the
duadic codes; hence $d^2\geq n$.  If $\ord_n(q)$ is odd, then $\mu_{-1}$ gives a
splitting of $n$\cite[Lemma~5]{rushanan86}. In this case,  Lemma~\ref{th:duadicdist}
implies that the odd-like weight $d$ satisfies $ d^2-d+1 \geq n$.~\end{proof}

In the binary case, it is possible to derive degenerate codes with similar
parameters using topological constructions~\cite{bravyi98,freedman01,kitaev02}, but
the codes do not appear to be equivalent to the construction given here.

\subsection{Degenerate Codes} \label{sec:impure1}
The next result proves the existence of degenerate duadic quantum stabilizer codes.
This results shows that the classical duadic codes, such as $C_i\subseteq D_i$,
contain codewords of very small weight but their set difference $D_i\setminus C_i$
(and $ C_i^\perp\setminus D_i^\perp$) does not. First we need the following lemma,
which shows the existence of duadic codes of low distance.

It is always possible to construct a degenerate code of distance $d$
and pure to 1  by the method discussed
in~\cite[Theorem~6]{calderbank98}; see
also~\cite[Lemma~69]{ketkar06}. An alternative method to construct
impure codes is to use concatenation
\cite{calderbank98,gottesman97}. However such a construction assumes
the existence of a pure code of distance $d$. The families we
propose here are based on classical codes whose distance is low
compared to their quantum distance.
\begin{theorem}\label{th:impureprimepower}
Let $p$ be an odd prime and $\qrx{p}$. Let $t=\ord_p(q)$, and let $z$ be such that
$p^z \| q^t-1$. Then for $m>2z$, there exist degenerate $[[p^m,1,d]]_q$ quantum
codes  pure to $d'\leq p^z<d$ with $d^2\geq p^m $ and $d^2-d+1\geq p^m$ if $p\equiv
-1 \bmod 4$.
\end{theorem}
\begin{proof}
The existence of quantum stabilizer codes with these parameters
follows from Theorems~\ref{th:quantumduadic1}, which combined cover
the two cases $p\equiv \pm 1 \bmod 4$.

But $d'$, the minimum distance of the underlying classical even-like duadic codes,
is upper bounded by $p^z$, see~\cite[Theorem~6]{smid87}. For $m>2z$, the minimum
distance $d$ of the quantum code satisfies $d \geq p^{m/2}>p^z\geq d'$; thus, we
have a degenerate quantum code.
\end{proof}
 Our next goal is to find a generalization of
Theorem~\ref{th:impureprimepower} to lengths that are not necessarily prime powers.

\begin{lemma}\label{th:duadicevendist}
Let $n=\prod p_i^{m_i}$ be an odd integer and $\qrx{p_i}$. If $t_i=\ord_{p_i}(q)$
and $p_i^{z_i}\| q^{t_i}-1$, and $m_i>2z_i$, then there exists a duadic  code of
length $n$ and (even-like) minimum distance $\leq \min\{ p_i^{z_i}\} < \sqrt{n}$.
\end{lemma}
\begin{proof}
By Theorem~\ref{th:duadicexist} there exist duadic codes of lengths $p_i^{m_i}$ and
by \cite[Theorem~6]{smid87} their minimum distance, $d_i'$ is less than $p_i^{z_i}$.
Since we know that the odd-like distance is $\geq p_i^{m_i/2} > p_i^{z_i}$, the
minimum distance must be even-like. By \cite[Theorem~4]{smid87}, there exists duadic
codes of length $n=\prod p_i^{m_i}$ whose minimum distance $d'\leq \min\{d_i' \}
\leq \min \{p_i^{z_i} \} < \prod p_i^{m_i/2} =\sqrt{n}$. Since this is less than the
minimum odd-like distance, the minimum distance is even-like.~\end{proof}

\begin{theorem}\label{th:impurecss}
Let $n=\prod p_i ^{m_i}$ be an odd integer and $\qrx{p_i}$. Let
$t_i=\ord_{p_{i}}(q)$, and let $z_i$ be such that $p_i^{z_i} \| q^{t_i}-1$. Then for
$m_i>2z_i$, there exists a degenerate $[[n,1,d]]_q$ quantum code pure to $d'\leq
\min \{p_i^{z_i} \} <d$ with $d^2\geq n$. If $p_i\equiv -1\bmod 4$, then
$d^2-d+1\geq n$.
\end{theorem}
\begin{proof}
From Lemma~\ref{th:duadicevendist}, we know that there exist duadic codes of length
$n$ and minimum (even-like) distance $d'\leq \min\{p_i^{z_i}\}< \sqrt{n}$.  From
Theorem~\ref{th:quantumduadic1}, we know there exists a quantum duadic code with
parameters $[[n,1,d]]$, where $d\geq \sqrt{n}>d'$. Hence, the quantum code is
degenerate.

If $p_i\equiv -1 \bmod 4$, then by \cite[Theorem~8]{smid87}, the permutation
$\mu_{-1}$ gives a splitting for this code. Hence the odd-like distance must satisfy
$d^2-d+1$.
\end{proof}

 Note that the previous result does not specify whether these duadic codes have
a splitting given by $\mu_{-1}$. \nix{So depending upon whether the
permutation $\mu_{-1}$ gives a splitting or fixes these duadic
codes, we have the following result.

\begin{theorem}\label{th:impurecss}
Let $n=\prod p_i ^{m_i}$ be an odd integer and $\qr$ such that every $p_i\equiv
-1\bmod 4$. Let $t_i=\ord_{p_{i}}(q)$, and let $z_i$ be such that $p_i^{z_i} \|
q^{t_i}-1$. Then for $m_i>2z_i$, there exists a degenerate $[[n,1,d]]_q$ quantum
code pure to $d'\leq \min \{p_i^{z_i} \} <d$ with $d^2-d+1\geq n$.
\end{theorem}
\begin{proof}
From Lemma~\ref{th:duadicevendist}, we know that there exist duadic codes of length
$n$ and minimum (even-like) distance $d'\leq \min\{p_i^{z_i}\}$.  By
\cite[Theorem~8]{smid87}, the permutation $\mu_{-1}$ gives a splitting for this
code. Using Theorem~\ref{th:quantumduadic1}, we can infer that this gives a quantum
duadic code $[[n,1,d]]_q$. Since $d\geq \sqrt{n} > \min \{p_i^{z_i} \} \geq d'$, the
quantum code is degenerate.
\end{proof}
} Next we  consider  duadic codes when $\mu_{-1}$ leaves them
invariant.
\begin{theorem}\label{th:impurecss2}
Let $\qr$ such $n|(q^b+1)$ for some $b$.  Let $t_i=\ord_{p_{i}}(q)$, and let $z_i$
be such that $p_i^{z_i} \| q^{t_i}-1$. Then for $m_i>2z_i$, there exists a
degenerate $[[n,1,d]]_q$ quantum code pure to $d'\leq \min \{p_i^{z_i} \} <d$ with
$d^2\geq n $.
\end{theorem}
\begin{proof}
By Lemma~\ref{th:duadicevendist}, there exists a duadic code with
minimum even-like distance $d'\leq \min \{p^{z_i} \}$. But
Theorem~\cite[Theorem~3.2.10]{smid87} tells us that this code is
fixed by $\mu_{-1}$. Now Theorem~\ref{th:quantumduadic1} implies
that we can construct a $[[n,1,d\geq \sqrt{n}]]_q$ quantum code.  As
$d'\leq \min \{ p_i^{z_i}\} < \sqrt{n}\leq d$, we conclude that the
quantum code is degenerate.
\end{proof}

\begin{exampleX}
Let us consider binary quantum duadic codes of length $7^m$. Note
that $2$ is a quadratic residue modulo $7$ as $4^2\equiv 2 \mod 7$.
Since $\ord_7(2)=3$ and $7\|2^3-1$, we have $z=1$. By
Theorem~\ref{th:impurecss} for $m\geq 2$ there exist quantum codes
with the parameters $[[7^m,1,d]]_2$. As $p=7\equiv -1 \mod 4$  we
have with $d^2-d+1\geq 7^{m}$. But, $d'$, the distance of the
(even-like) duadic codes is upper bounded by $p^z=7$. Hence these
codes are pure to $d'\leq 7$. Actually, using the fact that the true
distance of the even-like codes is $4$ \cite{smid87} we can show
that the quantum codes are pure to $4$.
\end{exampleX}

\nix{
\begin{example}
Let us consider binary quantum duadic codes of length $7^m$. Note that $2$ is a
quadratic residue modulo $7$ as $4^2\equiv 2 \mod 7$. As $p=7\equiv -1 \mod 4$
Theorem~\ref{th:quantumduadic1} implies the existence of $[[7^m,1,d]]_2$ quantum
codes for all $m$ with $d^2-d+1\geq 7^{m}$.

Since $\ord_7(2)=3$ and $7\|2^3-1$, we have $z=1$. Hence, $d'$, the distance of the
(even-like) duadic codes is upper bounded by $7$. For $m\geq 2$ we have $d \geq
\sqrt{7^m+d-1} >7\geq d'$ and there exist degenerate quantum codes $[[7^m,1,d]]_2$
with $d^2-d+1\geq 7^m$ pure to $d'\leq 7$. Actually, using the fact that the true
distance of the even-like codes is $4$ we can show that the quantum codes are pure
to $4$.
\end{example}
}

\section{Quantum Duadic Codes -- Hermitian Case}\label{sec:hermitian}
Recall that if there exists an $\F_{q^2}$-linear $[n,k,d]_{q^2}$
code $C$ such that $C^{\hdual}\subseteq C$, then there exists an
$[[n,2k-n,\ge d]]_q$ quantum stabilizer code that is pure to $d$.
In this section, we construct duadic quantum codes using this
construction.  Since $q^2\equiv \square\bmod n$, duadic codes exist
over $\F_{q^2}$ for all $n$, when $\gcd(n,q^2)=1$. In this case, the
splitting $\mu_{-q}$ plays a role analogous to that of $\mu_{-1}$ in
the previous section.

\subsection{Basic Code Constructions}
\begin{lemma}\label{th:duadichdual}
Let $C_i$ and $D_i$ respectively be the even-like and odd-like duadic codes over
$\F_{q^2}$, where $i\in \{0,1\}$.  Then $C_i^\hdual = D_i$ if and only if there is a
$q^2$-splitting of $n$ given by $\mu_{-q}$, that is, $-qS_i \equiv S_{(i+1\bmod
2)}\bmod n$.
\end{lemma}
\begin{proof}
See~\cite[Theorem~4.4]{rushanan86}.
\end{proof}

\begin{lemma}\label{th:hermitiansplitting1}
Let $n=\prod p_i^{m_i}$ be an odd integer such that $\ord_n(q)$ is
odd. Then $\mu_{-q}$ gives a splitting of $n$ over $\F_{q^2}$. In
fact $\mu_{-1}$ and $\mu_{-q}$ give the same splitting. Otherwise
$\mu_{q}$ gives  a splitting of $n$.
\end{lemma}
\begin{proof}
Suppose that $\{S_0,S_1,a\}$ be a splitting. We know that each $S_i$ is an union of
some $q^2$-ary cyclotomic cosets, so $q^2S_i \equiv S_i\bmod n$. Now
$q^{\ord_n(q)}S_i \equiv S_i\bmod n$. If $\ord_n(q)=2k+1$, then $q^{2k+1} S_i \equiv
qS_i \equiv S_i\bmod n$; hence, $\mu_q$ fixes each $S_i$ if the multiplicative order
of $q$ modulo $n$ is odd.

Notice that if $\ord_n(q)$ is odd, then $\ord_n(q^2)$ is also odd. By
\cite[Lemma~5]{rushanan91}, we know that there exists a $q^2$-splitting of $n$ given
by $\mu_{-1}$ if and only if $\ord_n(q^2)$ is odd.  Hence $-S_i \equiv S_{(i+1 \bmod
2)}\bmod n$. Since $\mu_q$ fixes $S_i$ we have $-qS_i \equiv S_{(i+1\bmod 2)}\bmod
n$; hence, $\mu_{-q}$ gives a $q^2$-splitting of $n$.

Conversely, if $\mu_{-q}$ gives a splitting of $n$, then
$-qS_i\equiv S_{(i+1\bmod 2)} \bmod n$.  But as $\mu_{q}$ fixes
$S_i$ we have $-S_i\equiv S_{(i+1\bmod 2)}\bmod n$. Therefore
$\mu_{-1}$ gives the same splitting as $\mu_{-q}$.   If
$\ord_n(q)=2k$, then $q^{k}=-1$. Hence, $q^kS_i\bmod n=-S_i\bmod
n=S_{(i+1\bmod 2)}$ because $\mu_{-1}$ gives a splitting of $n$.
Because $\mu_{q^{2r}}$ fixes $S_i$, $k=2w+1$ for some $w$. And
$q^{2w+1}S_i\bmod n=qS_i \bmod n=-S_i=S_{(i+1\bmod 2)}$. Thus
$\mu_{q}$ gives a splitting of $n$.
\end{proof}

\begin{theorem}\label{th:hermitianduadic1}
Let $n$ be an odd integer such that $\ord_n(q)$ is odd. Then there exists an
$[[n,1,d]]_q$ quantum code with $d^2-d+1\geq n$.
\end{theorem}
\begin{proof}
By Lemma~\ref{th:hermitiansplitting1}, there exist duadic codes $C_i\subset D_i$
with splitting given by $\mu_{-q}$ and $\mu_{-1}$. This means that the $C_i\subseteq
C_i^\hdual =D_i$ by Lemma~\ref{th:duadichdual}. Hence there exists an
$[[n,n-(n-1),d]]_q$ quantum code with $d=\wt(D_i\setminus C_i)$. As $\mu_{-1}$ gives
a splitting, we have $d^2-d+1\geq n$ by Lemma~\ref{th:duadicdist}.~\end{proof}

\subsection{Degenerate Codes}
We construct a family of degenerate quantum codes that has a large minimum distance.
\begin{theorem}\label{th:impurehermitian}
Let $n=\prod p_i ^{m_i}$ be an odd integer with $\ord_n(q)$ odd and every $p_i\equiv
-1 \bmod 4$. Let $t_i=\ord_{p_i}(q^2)$, and $p_i^{z_i} \| q^{2t_i}-1$. Then for
$m_i>2z_i$, there exist degenerate quantum codes with parameters $[[n,1,d]]_q$ pure
to $d'\leq \min \{p_i^{z_i} \} <d$ with $d^2-d+1\geq n$.
\end{theorem}
\begin{proof}
From Lemma~\ref{th:duadicevendist} we know that there exists an even-like duadic
code with parameters $[n,(n-1)/2,d']_{q^2}$ and $d'\leq \min \{p_i^{z_i}\}$.

Then by \cite[Theorem~8]{smid87}, we know that for this code
$\mu_{-1}$ gives a splitting. By Lemma~\ref{th:hermitiansplitting1},
$\mu_{-q}$ also gives a splitting for this code. Hence by
Theorem~\ref{th:hermitianduadic1} this duadic code gives  a quantum
duadic code $[[n,1,d]]_q$, which is impure as $d'\leq \min
\{p_i^{z_i}\}< \sqrt{n}< d$. \nix{ By
Lemma~\ref{th:hermitianduadic1} this gives us an $[[n,1,d]]_q$
quantum code where $d^2-d+1\geq n$. Since $d> \sqrt{n}> \min \{
p_i^{z_i}\}\geq d'$, this is a degenerate quantum code. }
\end{proof}
Finally, one can construct more quantum codes, for instance when $\ord_n(q)$ is
even, by finding the conditions under which $\mu_{-q}$ gives  a splitting of $n$.

\begin{lemma}\label{lem:hermitian3}
Let $n$ be an odd integer such that $gcd(n,q^{2i-1}+1)=1$ for some integer $1\leq i
\leq ord_n(q)$. Then $\mu_{-q}$ gives a splitting of $n$ over $\F_{q^2}$.
\end{lemma}

\begin{proof}
 Assume w.l.g. that there exists $C_x \in S_0$ such that $-qC_x \bmod n \equiv C_x$ with $x \neq 0$. The proof is by contraction. Let $C_x=\{x, xq^2, xq^4,
..., xq^{2i}\}$, so, $-qx \equiv xq^{2i} \bmod n$. Hence, $-qx - xq^{2i} \bmod n
\equiv 0$ or $-xq(1+ q^{2i-1})  \bmod n \equiv 0$. Since
$gcd(n,q^{2i-1}+1)=1=gcd(n,q)$ and $x <n$, then there is no integer solution for the
last equation unless $x=0$ that contradicts out assumption. Therefore, $-qC_x \bmod
n \equiv C_y$. consequently, the lemma holds.
\end{proof}

\begin{lemma}
Let $n$ be an odd integer such that $gcd(n,q^{2i-1}+1)=1$ for some integer $1\leq i
\leq ord_n(q)$. Then there exists an $[[n,1,d]]_q$ quantum code with $d^2-d+1\geq
n$.
\end{lemma}

\begin{proof}
Direct conclusion and similar proof as Lemma \ref{th:hermitianduadic1} by using
Lemma \ref{lem:hermitian3} and Lemma \ref{th:duadichdual}.
\end{proof}
Now, we relax the condition in lemma \ref{lem:hermitian3} by studying the case where
$ord_n(q)$ is even.

\begin{lemma}\label{lem:hermitiansplitting5}
Let $n=\prod p_i^{m_i}$ be an odd integer such that every $p_i\equiv 1 \bmod 4$ or
$\ord_n(q)$ is even. If $n|(q^{2b}+1)$ for some integer b, Then $\mu_{-q}$ gives a
splitting of $n$ over $\F_{q^2}$ if $\mu_{-1}$ fixes $S_i \bmod n$.
\end{lemma}
\begin{proof}
Let w.l.g. $1 \in S_0$. We show that $-q \not \in S_0$. Suppose $-q\in S_0$, then
$-qS_0 \equiv -q^{2i+1}S_0 \bmod n=S_0=-S_0$ because $\mu_{-1}$ fixes $S_0$ and $1
\in S_0$. So, $q^{2i+1}S_0 \bmod n=S_0$ but this is contradiction since $ord_n(q)$
is even. Now, we construct all elements of $S_0$ and $S_1$ such that $S_0 \cap S_1 =
\phi$.

 Assume w.l.g. that there exist $C_x \in S_0$ and $C_y \in S_1$ such that $-qC_x
\bmod n \equiv C_y$.  let $C_x=\{x, xq^2, xq^4, ..., xq^{2i}\}$, so, $-qxq^{2i}
\bmod n \equiv y \bmod n$ or $-xq^{2i+1} \bmod n \equiv y \bmod n$. Since $x \in C_x
\in S_0$ and $y \in C_y \in S_1$ and consequently $q^{2i}=-1 \bmod n$. Using Lemma
\cite[Lemma 3.2.6.]{smid86} and the fact that $ord_n(q)$ is even then $n|
(q^{2b}+1)$ for some integer b. Indeed, $\mu_{-q}$ gives a splitting of $n$ over
$F_{q^2}$.
\end{proof}

\section{Conclusion}\label{sec:conclusion}
The motivation for this work was that many good quantum
error-correcting codes, such as quantum MDS codes, are typically
pure and thus require active corrective steps for all errors of
small Hamming weight. At the other extreme are decoherence free
subspaces~(see~\cite{lidar98,zanardi97}) that do not require any
active error correction at all, but perform poorly in terms of
minimum distance. We pointed out that degenerate quantum codes can
form a compromise, namely they can reach larger minimum distances
while allowing at least some nice errors of low weight that do not
require active error correction.

We have constructed two families of quantum duadic codes with the parameters
$[[n,1,\geq\sqrt{n}]]_q$ and have shown that they contain large subclasses of
degenerate quantum codes. Although these codes encode only one qubit, they are
interesting because they demonstrate that there exist families of classical codes
which can give rise to remarkable degenerate quantum codes.   A more
detailed study of the weight distribution of classical duadic codes can reveal which
codes are particularly interesting for quantum error correction.  We note that
generalizations of duadic codes, such as triadic and polyadic codes, can be used to
obtain degenerate quantum codes with higher rates.

\chapter{Quantum Projective Geometry Codes}\label{ch_QBC_QPRM}
\newcommand{\prm}{{\mathcal{P}_q}}
\newcommand{\ontop}[2]{\genfrac{}{}{0pt}{}{#1}{#2}}

\newcommand{\Sp}{{\rm {Span}}}
In this chapter I study projective geometry codes over finite
fields. I settle down conditions when these codes contain their dual
codes, $C^\perp \subseteq C$. Consequently, using the CSS
construction, I construct families of quantum error-correcting codes
based on projective geometry codes. For further details
see the joint paper with Klappenecker and Sarvepalli~\cite{klappenecker065}.

Lachaud~\cite{lachaud88,lachaud90,lachaud00} introduced projective
Reed-Muller codes (PRM) over finite fields in 1988. Projective
Reed-Muller (PRM) codes are a well-known class of projective
geometry codes. I establish conditions when Projective Reed-Muller
codes are self-orthogonal, hence I construct their corresponding
quantum PRM codes. In addition, I  study puncturing of these quantum
PRM codes.

\medskip

\noindent  \textbf{Notation:} Let us denote by
$\textbf{F}_q[X_0,X_1,...,X_m]$ the polynomial ring in
$X_0,X_1,...,X_m$ with coefficients in $\textbf{F}_q$. Furthermore,
let $\textbf{F}_q[X_0,X_1,...,X_m]_h^\nu \cup \{0 \}$ be the vector
space of homogeneous polynomials in $X_0,X_1,...,X_m$ with
coefficients in $\textbf{F}_q$ with degree $\nu$ (cf.
\cite{assmus98}, \cite{lachaud88}, \cite{sorensen91}). Let
$P^m(\textbf{F}_q)$ be the m-dimensional projective space over
$\textbf{F}_q$.
 We evaluate the function $f(P_i)$ at the
projective points $P_i \in P^m(\textbf{F}_q)$.\\

\medskip

\section{Projective Reed-Muller Codes}

A Generalized Reed-Muller code (GRM), $C_\nu(m,q)$ over
$\textbf{F}_q$ of order $1 \leq \nu \leq m(q-1)$ and length $q^m$ is
defined as
\begin{eqnarray}
C_\nu(m,q) &=&\{ \left( f(0), f(p_1),...,f(P_{q^m-1} \right) |
f(X_1,...,X_m) \nonumber \\ && \in \textbf{F}_q[X_1,...,X_m] ,
deg(f) \leq \nu\}.
\end{eqnarray}

\begin{lemma}\label{lem:GRM}
Generalized Reed-Muller (GRM) codes $C_{\nu}(m,q)$ over
$\textbf{F}_q$ of order $1 \leq \nu \leq (q-1)m$ have length
$n=q^m$, dimension
\begin{eqnarray}\label{GRM}
k(\nu)= \sum_{t=0}^{\nu} \sum_{j=0}^n (-1)^j \left(
\begin{array}{c}
           m\\
           j \\
         \end{array}
       \right)\left(
                \begin{array}{c}
                  t+m-jq-1 \\
                  t-jq \\
                \end{array}
              \right)\end{eqnarray}
and minimum distance $d(\nu)=(q-s)q^{m-r-1}$, where $\nu = (q-1)r+s$
,  $0\leq s < (q-1)$ and $0 \leq r \leq m-1 $.
\end{lemma}
\begin{proof}
See for instance \cite{sorensen91} and \cite[chapter 16 ]{assmus98}.\\
\end{proof}

\medskip
 The Projective Reed-Muller code (PRM) over $\textbf{F}_q$ of integer
order $\nu$ and length $n =(q^{m+1}-1)/(q-1)$  is denoted by
$\prm(\nu,m)$ and defined as
\begin{eqnarray} \prm(\nu,m) &=& \{ \left( f(P_1),...,f(P_{n} \right) | f(X_0,...,X_m) \in
\textbf{F}_q[X_0,...,X_m]_h^\nu \cup \{0\} \},\nonumber \\ && \text{ and } P_i \in
P^m(\textbf{F}_q) \text{ for } 1\leq i \leq n.\end{eqnarray}

\begin{lemma}\label{thm:PRM}
The projective Reed-Muller code  $\prm(\nu,m)$, $1\leq \nu \leq m(q-1)$, is an
$[n,k,d]_q$ code with length $n =(q^{m+1}-1)/(q-1)$,  dimension
\begin{eqnarray}\label{dimension-PRM}
k(\nu)= \sum_{ \ontop{t=\nu \bmod (q-1)}{t\leq \nu}}
\sum_{j=0}^{m+1} (-1)^j
 \binom{m+1}{j}\binom{t-jq+m}{t-jq}
\end{eqnarray}
and minimum distance $d(\nu)=(q-s)q^{m-r-1}$ where $\nu=r(q-1)+s+1$,  $0 \leq s <
q-1$
\end{lemma}
\begin{proof}{}
See  \cite [Theorem 1]{sorensen91}.
\end{proof}

The duals of PRM codes are also known and under some conditions they are also PRM
codes. The following result gives more precise details.

\begin{lemma}\label{thm:dual-PRM}
Let $\nu^\perp = m(q-1)-\nu$, then the dual of $\prm(\nu,m)$ is given by
\begin{eqnarray}
\prm(\nu,m)^{\bot} &=&\left\{ \begin{array}{ll}\prm(\nu^\perp,m) &\nu \not\equiv 0 \bmod (q-1)\\
\Sp_{\textbf{F}_q}\{1,\prm(\nu^\perp,m)\} & \nu \equiv 0 \bmod (q-1)
\end{array}
\right.
\end{eqnarray}
\end{lemma}
\begin{proof}{}
See~\cite[Theorem 2]{sorensen91}.
\end{proof}

As mentioned earlier our main methods of constructing quantum codes are the CSS
construction and the Hermitian construction. This requires us to identify nested
families of codes and/or self-orthogonal codes. First we identify when the PRM codes
are nested i.e., we find out when a PRM code contains other PRM codes as subcodes.

\begin{lemma}\label{lemma:C-sub-E}
If $\nu_2=\nu_1+k(q-1)$, where $k >0$, then $\prm(\nu_1,m) \subseteq \prm(\nu_2,m)$
and $\wt(\prm(\nu_2,m)\setminus \prm(\nu_1,m)) = \wt(\prm(\nu_2,m))$.
\end{lemma}
\begin{proof}{}
In the finite field $\F_q$, we can replace any variable $x_i$ by
$x_i^q$, hence every function in $\F_q[x_0,x_1,\ldots,x_m]_\nu^h$ is
present in $\F_q[x_0,x_1,\ldots,x_m]_{\nu+k(q-1)}^h$. Hence
$\prm(\nu_1,m) \subseteq \prm(\nu_2,m)$. Let $\nu_1=r(q-1)+s+1$,
then $\nu_2=(k+r)(q-1)+s+1$. By Lemma~\ref{thm:PRM},
$d(\nu_1)=(q-s)q^{m-r-1} > (q-s)q^{m-r-k-1}=d(\nu_2)$. This implies
that there exists a vector of weight $d(\nu_2)$ in $\prm(\nu_2,m)$
and $\wt(\prm(\nu_2,m)\setminus \prm(\nu_1,m)) =
\wt(\prm(\nu_2,m))$.
\end{proof}

\begin{example} Let $m=1$, $q=5$, so $n
=(q^{m+1}-1)/(q-1)=6$. There are 6 points in this space $\{(0,1),
(1,0), (1,1), (1,2), (1,3), (1,4)   \}$. Therefore, in
$\mathcal{P}_{5}(1,1)$, there are two codewords $ \{(011111) , (
101234) \}$. Also, in $\mathcal{P}_{5}(5,1)$, there are 6 codewords
\[\{ (011111) ,(001234) ,(001441) , (001324), (001111),(101234)
\},\] Hence, the $\mathcal{P}_{5}(1,1) \subset \mathcal{P}_{5}(5,1)$
as shown in Lemma \ref{lemma:C-sub-E}. Clearly, the code
$\mathcal{P}_{5}(1,1)$ is not contained in $\mathcal{P}_{5}(2,1)$,
$\mathcal{P}_{5}(3,1)$, or $\mathcal{P}_{5}(4,1)$.
\end{example}

\medskip

\section{Quantum Projective Reed-Muller Codes} We now construct
stabilizer codes using the CSS and hermitian constructions.
\begin{lemma}\label{CSS-construction}\textbf{(CSS Construction)}
Suppose given two classical linear codes $C = [n,k_C,d_C]_q$ and
$E=[n,k_E,d_E]_q$ over $\mathbf{F}_q$ with $C \subseteq E$.
Furthermore, let the minimum distance be $d= \min wt \{ (E
\backslash C) \cup (C^{\bot} \backslash  E^{\bot}) \}$ if $C \subset
E$ and $d=\min wt \{C \cup C^\bot \}$ if $C=E$, then there exists a
$[[n,k_E - k_C, d ]]_q$ quantum code.
\end{lemma}
\begin{proof}
See for instance~\cite[Lemma 2]{sarvepalli05}.\\
\end{proof}

\begin{theorem}\label{PQRM-euclidean}
Let $n =(q^{m+1}-1)/(q-1)$ and $1 \leq \nu_1 < \nu_2 \leq m(q-1)$ such that $ \nu_2
= \nu_1+ l(q-1)$ with  $\nu_1 \not\equiv 0 \bmod (q-1)$. Then there exists an  $[[
n, k(\nu_2)- k(\nu_1), \min\{d(\nu_2), d(\nu_1^\perp) \} ]]_q$ stabilizer code,
where the parameters $k(\nu)$ and $d(\nu)$ are given in Theorem~\ref{thm:PRM}.
\end{theorem}
\begin{proof}{}
A direct application of the CSS construction in conjunction with
Lemma~\ref{lemma:C-sub-E}.
\end{proof}
We do not need to use two pairs of codes as we had seen in the previous two cases,
we could use a single self-orthogonal code for constructing a quantum code. We will
illustrate this idea by finding self-orthogonal PRM codes.

\begin{corollary}
Let $0\leq \nu\leq \lfloor m(q-1)/2\rfloor$ and $2\nu\equiv 0\bmod q-1$, then
$\prm(\nu,m) \subseteq \prm(\nu,m)^\perp$. If $\nu \not\equiv 0 \bmod q-1$ there
exists an $[[n,n-2k(\nu),d(\nu^\perp)]]_q$ quantum code where $n=(q^{m+1}-1)/(q-1)$.
\end{corollary}
\begin{proof}{}
We know that $\nu^\perp=m(q-1)-\nu$ and if $\prm(\nu,m)\subseteq \prm(\nu,m)^\perp$,
then $\nu\leq \nu^\perp$ and by Lemma~\ref{lemma:C-sub-E} $\nu^\perp=\nu+k(q-1)$ for
some $k\geq 0$. It follows that $2\nu\leq \lfloor m(q-1)/2 \rfloor$ and
$2\nu=(m-k)(q-1)$, i.e., $2\nu\equiv 0\bmod q-1$. The quantum code then follows from
Theorem~\ref{PQRM-euclidean}.
\end{proof}

\bigskip

\noindent \textbf{Hermitian Constructions.} We can study Projective
Reed-Muller codes generated over $\textbf{F}_{q^2}$. We show that if
a code is contained in its hermitian dual code, then there is a
corresponding quantum PRM code. We define the hermitian inner
product of two codewords $c$ and $c'$ as
\begin{eqnarray}
  \scal{c}{c'} = X.\overline{Y} = \sum_{i=1}^n x_i \overline{y_i}
   = \sum_{i=1}^n x_i y_i^q
\end{eqnarray}

We say the code $C$ is  hermitian self-orthogonal if $C \subseteq
C^{\hdual}$ such that $\scal{c}{ c'} =0$ for all
codewords $c\in C$ and $c'\in C^{\hdual}$.\\

\begin{lemma}\label{hermition-lemmaI}
Let $[n,k,d]_{q^2}$ be a linear PRM  code such that $1 \leq \nu \leq
m(q-1)$ , then its contained in its hermitian dual (i.e.
$PC_{q^2}(\nu,m) \subseteq PC_{q^2}(\nu,m)^{\hdual}$).
\end{lemma}

\begin{lemma}\label{hermition-lemmaII}
Given a PRM $PC_{q^2}(\nu,m)$ that is contained in its hermitian
dual code $PC_{q^2}(\nu,m)^{\hdual}$ with minimum distance $d= \min
\{wt (C^{\hdual} \backslash C) \}$, then there exists an $[[n ,
n-2k,d ]]_q$ quantum stabilizer code.
\end{lemma}

\begin{proof}
See for instance~\cite[Corollary 2]{grassl04} and~\cite[Corollary
1]{ashikhmin01}.
\end{proof}

\begin{theorem}
Let $0 \leq \nu \leq m(q-1)$ and $ \nu \not\equiv 0 \mod (q-1)$,
there exist a quantum PRM code $[[n,n-2k(\nu),d(\nu^\perp) ]]_q$
with $n=(q^{2(m+1)}-1)/(q^2-1)$, where
\begin{eqnarray}\label{dimension-PRM}
k(\nu)= \sum_{ %
\begin{array}{c}
  t=\nu \mod (q^2-1) \\
   t\leq \nu\\
\end{array}%
 } \left( \sum_{j=0}^{m+1} (-1)^j \left(
\begin{array}{c}
  m+1 \\
  j \\
\end{array}\right)\left(
\begin{array}{c}
  t+m-jq^2 \\
  t-jq^2 \\
\end{array} \right)\right)
\end{eqnarray}
and \begin{eqnarray}d(\nu^\perp)= (q^2-s)q^{2(m-r-1)}\end{eqnarray}
such that $\nu-1=r(q^2-1)+s$,  $0 \leq s < q^2-1$
\end{theorem}
\begin{proof}
We note that this code is constructed over $\textbf{F}_{q^2}$, and
$wt(PC_{q^2}(\nu,m)^\perp) =
wt(PC_{q^2}(\nu,m)^{\hdual})=d(\nu^\perp)$. Applying Lemma
\ref{hermition-lemmaI} and Lemma \ref{hermition-lemmaII}, we
construct a quantum code with parameters $[[n,n-2k(\nu),d(\nu^\perp)
]]_q$.
\end{proof}

\bigskip

\section{Puncturing Quantum Codes}

Finally we will briefly touch upon another important aspect of
quantum code construction, which is the topic of shortening quantum
codes. In the literature on quantum codes, there is not much
distinction made between puncturing and shortening of quantum codes
and often the two terms are used interchangeably. Obtaining a new
quantum code from an existing one is more difficult task than in the
classical case, the main reason being that the code must be so
modified such that the resulting code is still self-orthogonal.
Fortunately, however there exists a method due to
Rains~\cite{rains99} that can solve this problem.

From Lemma~\ref{th:css} we know that with every quantum code constructed using
the CSS construction, we can associate two classical codes, $C_1$ and $C_2$.
Define $C$ to be the direct product of $C_1^\perp$ and $C_2^\perp$ viz.
$C=C_1^\perp\times C_2^\perp$. Then we can associate a puncture code $P(C)$
\cite[Theorem 12]{grassl03} which is defined as
\begin{eqnarray}
P(C)&=&\{ (a_ib_i)_{i=1}^{n} \mid a \in C_1^\perp,b\in C_2^\perp\}^\perp.
\label{eq:puncdef1}
\end{eqnarray}
Surprisingly, $P(C)$ provides information about the lengths to which we can
puncture the quantum codes. If there exists a vector of nonzero weight $r$ in
$P(C)$, then the corresponding quantum code can be punctured to a length $r$
and minimum distance greater than or equal to distance of the parent code.

\begin{theorem}
Let $0\leq \nu_1 < \nu_2 \leq m(q-1)-1$ where $\nu_2\equiv \nu_1\bmod q-1$.
Also let  $0\leq \mu\leq \nu_2-\nu_1$ and $\mu\equiv 0\bmod q-1$. If
$\prm(\mu,m)$ has codeword of weight $r$, then there exists an $[[r,\geq
(k(\nu_2)-k(\nu_1)-n+r),\geq d]]_q$ quantum code, where $n=(q^m-1)/(q-1)$
$d=\min\{d(\nu_2),d(\nu_1^\perp) \}$. In particular, there exists a
$[[d(\mu),\geq (k(\nu_2)-k(\nu_1)-n+d(\mu)),\geq d]]_q$ quantum code.
\end{theorem}

\begin{proof}{}
Let $C_i=\prm(\nu_i,m)$ with $\nu_i$ as stated. Then by
Theorem~\ref{PQRM-euclidean}, an $[[n,k(\nu_2)-k(\nu_1),d]]_q$ quantum code $Q$
exists where $d=\min\{d(\nu_2),d(\nu_1^\perp)\}$. From equation
(\ref{eq:puncdef1}) we find that $P(C)^\perp=\prm(\nu_1+\nu_2^\perp,m)$, so
\begin{eqnarray}
P(C)&=&\prm(m(q-1)-\nu_1-\nu_2^\perp,m),\nonumber\\
&=&\prm(\nu_2-\nu_1,m).
\end{eqnarray}
By \cite[Theorem 11]{grassl03}, if there exists a vector of weight $r$ in
$P(C)$, then there exists an $[[r,k',d']]_q$ quantum code, where $k'\geq
(k(\nu_2)-k(\nu_1)-n+r)$ and distance $d'\geq d$. obtained by puncturing $Q$.
Since $P(C) = \prm(\nu_2-\nu_1,m) \supseteq \prm(\mu,m)$ for all $0\leq\mu\leq
\nu_2-\nu_1$ and $\mu\equiv \nu_2-\nu_1\equiv 0 \bmod q-1$, the weight
distributions of $\prm(\mu,m)$ give all the lengths to which $Q$ can be
punctured. Moreover $P(C)$ will certainly contain vectors whose weight
$r=d(\mu)$, that is the minimum weight of $PC(\mu,m)$. Thus there exist
punctured quantum codes with the parameters $[[d(\mu),\geq
(k(\nu_2)-k(\nu_1)-n+d(\mu)),\geq d]]_q$.
\end{proof}

\section{Conclusion and Discussion}
In this chapter, I drove families of quantum codes based on
Projective Reed-Muller codes. In addition, I showed how to puncture
the constructed quantum codes.

One can study similar classes of Euclidean geometry codes to derive
new families of quantum error-correcting codes. For example, cyclic
 Reed-Muller~\cite{berger01}, non-primitive Reed-Muller~\cite{blahut03}, Euclidean geometry codes~\cite[Chapter 13]{macwilliams77},\cite{assmus98} over finite fields are obvious
extensions of the families given in this
 chapter. In addition one can investigate polynomial codes to derive
 a family of quantum codes based on polynomial codes~\cite{kasami68}.

\part{Subsystem Codes}
\chapter{Subsystem Codes}\label{ch_subsys_basic}

Subsystem codes are a relatively new construction of quantum error
control codes. Subsystem codes combine the features of decoherence
free subspaces, noiseless subsystems, and quantum error-correcting
codes. Such codes promise to offer appealing features, such as
simple syndrome calculation and a wide variety of easily
implementable fault-tolerant operations.

In this chapter  I give an introduction to subsystem codes. I will
show how to derive subsystem codes from classical codes that are not
necessarily self-orthogonal (or dual-containing).  I will establish
the relationships between stabilizer and subsystem codes. Some of this work  with further details  was appeared in~\cite{aly06c,aly08a,aly08f} that is based on  a joint work with A. Klappenecker  and P.
Sarvepalli.

\section{Introduction}\label{Sec:introduction}
Subsystem codes are a relatively new construction of quantum codes.
Subsystem codes generalize the known constructions of active and
passive quantum error control codes such as decoherence free
subspaces, noiseless subsystems, and quantum stabilizer codes,
see~\cite{zanardi97,lidar98,kempe01,shabani05}. The stabilizer
formalism of subsystem codes can be found
in~\cite{knill06,kribs05,poulin05}. Errors in subsystem codes not
only can be corrected but also can be avoided. Subsystem codes
promise to be useful for fault-tolerant quantum computation in
comparison to stabilizer codes~\cite{aliferis06,aly06c}.

The main purpose of subsystem codes is to simplify the known quantum
codes specifically the stabilizer codes. The subsystem codes do not
need the underlying classical codes to be self-orthogonal or dual
containing as in the case of stabilizer codes. Furthermore, errors
can be isolated into two subsystems. Therefore, they have less
syndrome measurement and more efficient error
corrections~\cite{bacon06,poulin05}. We will show that many
subsystem codes can be constructed easily from existing stabilizer
codes that are available in~\cite{magma,calderbank98}.

An $((n,K,R,d))_q$ subsystem code is a $KR$-dimensional subspace $Q$
of $\C^{q^n}$ that is decomposed into a tensor product $Q=A\otimes
B$ of a $K$-dimensional vector space $A$ and an $R$-dimensional
vector space $B$ such that all errors of weight less than~$d$ can be
detected by~$A$. The vector spaces $A$ and $B$ are respectively
called the subsystem $A$ and the co-subsystem $B$. For some
background on subsystem codes, see for
instance~\cite{aly06c,klappenecker0608,poulin05}.

 Assume that we have a
$[[n,k,r,d]]_q$ subsystem code $Q$ that decomposes as $Q=A\otimes
B$. In general $Q$ is a subspace in the $q^n$-dimensional Hilbert
space, $\C^{q^n}$, the information is stored on the correlations
between all the $n$-qudits, and there is not necessarily a one to
one correspondence between the logical qudits and the physical
qudits. Similarly for the gauge qudits, i.e., co-subsystem $B$. But
if there is a one to one correspondence between the physical qudits
and the gauge qudits, say $r'$ of them, then  the subsystem $A$ is
essentially in the Hilbert space of $n-r'$ qudits, and we can
discard the $r'$ gauge qudits to obtain a $[[n-r',k,r-r',d]]_q$
subsystem code. We call those gauge qudits trivial gauge qudits. If
all the gauge qudits can be identified with physical qudits, then we
call such a subsystem code a \textsl{trivial subsystem code}. Such
codes are no different from padding a stabilizer code with random
qudits; nothing is to be gained from them. Further, we will assume
that a nontrivial subsystem code has no trivial gauge qudits. We aim
in this study to judge whether stabilizer codes are superior to
subsystem codes.

There have been many families of stabilizer codes derived from
classical self-orthogonal codes over $\F_{q}$ and $\F_{q^2}$, see
for example~\cite{aly07a,ketkar06,calderbank98}. But in the other
hand, there are not many families of subsystem codes constructed
yet, except~\cite{bacon06b}. This is because the theory is recently
developed and it is a challenging  task to find two classical codes
such that dual of their intersection can lead to a subsystem code.
Subsystem codes exist given particular stabilizer codes over
$\F_{q}$.

{\em Notation:}  Let $q$ be a power of a prime integer $p$.  For vectors $x,y$
in $\F_{q}^n$, we define the Euclidean inner product $\langle x|y\rangle
=\sum_{i=1}^nx_iy_i$ and the Euclidean dual of $C\subseteq \F_{q}^n$ as
$C^\perp = \{x\in \F_{q}^n\mid \langle x|y \rangle=0 \mbox{ for all } y\in C
\}$. We also define the hermitian inner product for vectors $x,y$ in
$\F_{q^2}^n$ as $\langle x|y\rangle_h =\sum_{i=1}^nx_i^qy_i$ and the hermitian
dual of $C\subseteq \F_{q^2}^n$ as $C^\hdual= \{x\in \F_{q^2}^n\mid \langle x|y
\rangle_h=0 \mbox{ for all } y\in C \}$. The trace-symplectic product of two
elements $u=(a|b),v=(a'|b')$ in $\F_q^{2n}$ is defined as $\langle u|v
\rangle_s = \tr_{q/p}(a'\cdot b-a\cdot b')$, where $x\cdot y$ is the usual
Euclidean inner product.The trace-symplectic dual of a code $C\subseteq
\F_q^{2n}$ is defined as $C^\sdual=\{ v\in \F_q^{2n}\mid \langle v|w \rangle_s
=0 \mbox{ for all } w\in C\}$.


\begin{figure}[t]
  \begin{center}
  \includegraphics[scale=0.6]{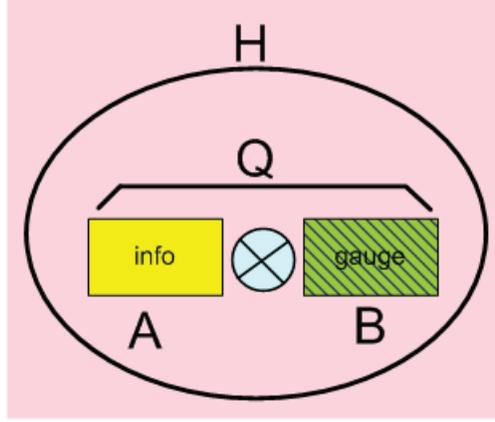}
  \caption{A quantum code Q is decomposed into two subsystem A (info) and B (gauge)}\label{fig:subsys1}
  \end{center}
\end{figure}

\section{Subsystem Codes}\label{sec:background}
 Let $\mathcal{H}$ be the Hilbert space $\mathcal{H}=\C^{q^n}=\C^q \otimes \C^q \otimes ...
\otimes \C^q$. Let $\ket{x}$ be the vectors of orthonormal basis of
$\C^q$, where the labels $x$ are  elements in the finite field
$\F_q$. For $a,b \in \F_q$, we define the unitary operators $X(a)$
and $Z(b)$ in $\C^q$ as follows:
\begin{eqnarray}X(a)\ket{x}=\ket{x+a},\qquad
Z(b)\ket{x}=\omega^{\tr(bx)}\ket{x},\end{eqnarray} where
$\omega=\exp(2\pi i/p)$ is a primitive $p$th root of unity and $\tr$
is the trace operation from $\F_q$ to $\F_p$

Now, we can define the set of error operators $E=\{X(a)Z(b)\,|\,
a,b\in \F_q\}$ in an error group. Let $\mathbf{a}=(a_1,\dots,
a_n)\in \F_q^n$ and $\mathbf{b}=(b_1,\dots, b_n)\in \F_q^n$. Let us
denote by $$ X(\mathbf{a}) = X(a_1)\otimes\, \cdots \,\otimes X(a_n)
\mbox { and },$$ $$ Z(\mathbf{b}) = Z(b_1)\otimes\, \cdots \,\otimes
Z(b_n)$$  the tensor products of $n$ error operators.  The set
$\textbf{E}=\{X(\mathbf{a})Z(\mathbf{b})\mid \mathbf{a, b} \in
\F_q^n\}$ form an error basis on $\C^{q^n}$. We can define the error
group $\mathbf{G}$ as follows \begin{eqnarray} \mathbf{G} = \{
\omega^{c}\textbf{E}=\omega^{c}X(\mathbf{a})Z(\mathbf{b})\,|\,
\mathbf{a, b} \in \F_q^n, c\in \F_p\}.\end{eqnarray}

Let $Q$ be a quantum code  such that $\mathcal{H}=Q\oplus Q^\perp$,
where $Q^\perp$ is the orthogonal complement of $Q$. We can define
the subsystem code $QA\otimes B$, see Fig.\ref{fig:subsys1}, as
follows
\begin{defn}
An $[[n,k,r,d]]_q$ subsystem code is a decomposition of the subspace $Q$ into a
tensor product of two vector spaces A and B such that $Q=A\otimes B$, where
$\dim A=k$ and $\dim B= r$. The code $Q$ is able to detect all errors  of
weight less than $d$ on subsystem $A$.
\end{defn}
Subsystem codes can be constructed  from the classical codes  over
$\F_q$ and $\F_{q^2}$. Such codes do not need the classical codes to
be self-orthogonal (or dual-containing) as shown in the following
theorem.

\begin{theorem}\label{th:subsys-main}
Let $C$ be a classical additive subcode of\/ $\F_q^{2n}$ such that
$C\neq \{0\}$ and let $D$ denote its subcode $D=C\cap C^\sdual$. If
$x=|C|$ and $y=|D|$, then there exists a subsystem code $Q= A\otimes
B$ such that
\begin{compactenum}[i)]
\item $\dim A = q^n/(xy)^{1/2}$,
\item $\dim B = (x/y)^{1/2}$.
\end{compactenum}
The minimum distance of subsystem $A$ is given by
\begin{compactenum}[(a)]
\item $d=\swt((C+C^\sdual)-C)=\swt(D^\sdual-C)$ if $D^\sdual\neq C$;
\item $d=\swt(D^\sdual)$ if $D^\sdual=C$.
\end{compactenum}
Thus, the subsystem $A$ can detect all errors in $E$ of weight less
than $d$, and can correct all errors in $E$ of weight $\le \lfloor
(d-1)/2\rfloor$.
\end{theorem}

Many subsystem codes can be derived based on the previous theorem as
we will show in the next chapters.
\section{Bounds on Pure Subsystem Code Parameters}

We want to investigate some bounds and limitations on subsystem
codes that can be constructed with the help of
Theorem~\ref{th:subsys-main}. It will be convenient to introduce
first some standard notations for the parameters of the codes.

All stabilizer codes obey the quantum Singleton bound and all pure stabilizer
codes also saturate the quantum Hamming bound. The conjecture where impure
stabilizer codes obey or disobey quantum Hamming bound has been an open
question. We will show that also pure subsystem codes obey Singleton and
Hamming bounds.

Let $X$ be an additive subcode of $\F_{q}^{2n}$ and $Y=X\cap
X^\sdual$.  By Theorem~\ref{th:subsys-main}, we can obtain an
$((n,K,K',d))_q$ subsystem code $Q$ from $X$ that has minimum
distance $d=\swt(Y^\sdual - X)$.  The set difference involved in the
definition of the minimum distance make it harder to compute the
minimum distance. Therefore, we introduce pure codes that are easier
to analyze. Let $d_p$ denote the minimum distance of the code $X$,
that is, $d_p=\swt(X)$.  Then we say that the associated subsystem
code is \textit{pure to $d_p$}. Furthermore, we call $Q$ a pure code
if $d_p\ge d$, and an impure code otherwise.

\begin{lemma}\label{th:stabcode}
If Theorem~\ref{th:subsys-main} allows one to construct a pure
$((n,K,K',d))_q$ subsystem code $Q$, then there exists a pure
$((n,KK',d))_q$ stabilizer code.
\end{lemma}
\begin{proof}
Let $X$ be a classical additive subcode of $\F_q^{2n}$ that defines
$Q$, and let $Y=X\cap X^\sdual$.  Furthermore,
Theorem~\ref{th:subsys-main} implies that $KK'=q^n/|Y|$. Since
$Y\subseteq Y^\sdual$, there exists an $((n,q^n/|Y|,d')_q$
stabilizer code with minimum distance $d'=\wt(Y^\sdual - Y)$. The
purity of $Q$ implies that $\swt(Y^\sdual - X) = \swt(Y^\sdual)=d$.
As $Y\subseteq X$, it follows that $d'=\swt(Y^\sdual -
Y)=\swt(Y^\sdual)=d$; hence, there exists a pure $((n,KK',d))_q$
stabilizer code.
\end{proof}
In Chapter~\ref{ch_subsys_construction}, we generalize
Lemma~\ref{th:stabcode} and also derive the converse.

\subsection{Quantum Singleton Bound} The quantum Singleton bound for
pure subsystem codes, not necessarily linear, can be stated as
follows.

\begin{theorem}[Singleton Bound.]\label{th:pureBound}
Any pure $((n,K,K',d))_q$ subsystem code that is constructed using
Theorem~\ref{th:subsys-main} satisfies the bound
\begin{eqnarray}KK'\leq
q^{n-2d+2}. \end{eqnarray}
\end{theorem}
\begin{proof}
By Lemma~\ref{th:stabcode}, there exists a pure $((n,KK',d))_q$ stabilizer
code. By the quantum Singleton bound, we have $KK'\leq q^{n-2d+2}$.
\end{proof}
\begin{corollary}
A pure $[[n,k,r,d]]_q$ code satisfies $ k+ r\leq n-2d+2$.
\end{corollary}

Our next goal is to show that in fact all $((n,q^{n-2d+2},K',d))_q$
subsystem codes are pure. Note that $((n,q^{n-2d+2},d))$ are the
parameters of a quantum MDS code. An $[[n,k,r,d]]_q$ subsystem code
derived from an $\F_q$-linear classical code $C\le \F_q^{2n}$
satisfies the Singleton bound $k+r\le n-2d+2$. A subsystem code
attaining the Singleton bound with equality is called an MDS
subsystem code.

An important consequence of the previous theorems is the following
simple observation which yields an easy construction of subsystem
codes that are optimal among the $\F_q$-linear Clifford subsystem
codes.

\begin{theorem}\label{th:mdsPurity}
Any  $[[n,n-2d+2,r,d]]_q$ subsystem code is pure.
\end{theorem}
\begin{proof}
Assume that there exists an $[[n,n-2d+2,r,d]]_q$ subsystem code that is impure.
Then there exists an $(n,q^{n-k+r})_{q^2}$ classical code $X\subseteq
\F_{q^2}^n$ and an $(n,q^{n-k-r})_{q^2}$ code $Y=X\cap X^\adual$ such that
$k=n-2d+2=\dim_{\F_{q^2}}  Y^\adual-\dim_{\F_{q^2}} X$ and
$\wt(Y^\adual\setminus X)= d$ and $\wt(X)=d' < d$. Then it is possible to
construct a stabilizer code with distance $\geq d$ that is impure to $d'$ by
considering a self-orthogonal subcode  $X\cap X^\adual \subseteq X'\subseteq X$
that includes a vector of weight $d'$ such that $|X'| = q^{n-k}$. Such a
subcode will always exist. Then the resulting stabilizer code is of parameters
$[[n,n-2d+2,d]]_q$ and is impure. But we know that all quantum MDS codes are
pure \cite{rains99}, see also \cite[Corollary~60]{ketkar06}. This implies that
$d'\geq d$ contradicting that $d'<d$. Hence every $[[n,n-2d+2,r,d]]_q$
subsystem code is pure.
\end{proof}

A very straightforward consequence of Theorems~\ref{th:pureBound} and
\ref{th:mdsPurity} is the following corollary:
\begin{lemma}
There exists no $[[n,n-2d+2,r,d]]_q$ subsystem code with $r>0$.
\end{lemma}
This still leaves a room for subsystem codes  being superior to
quantum block codes. For instance if a $[[11,1,8,3]]_2$ code exists,
then it is equivalent to a $[[3,1,3]]_2$ code which is superior to
$[[5,1,3]]_2$ code. In addition, there does not exist an
$[[11,9,3]]_2$ stabilizer code.

\begin{theorem}\label{th:pureMDS}
If there exists an $\F_q$-linear $[[n,k,d]]_q$ MDS stabilizer code,
then there exists a pure $\F_q$-linear $[[n,k-r,r,d]]_q$ MDS
subsystem code for all $r$ in the range $0\le r< k$.
\end{theorem}
\begin{proof}
From Lemma~\ref{th:mdsPurity}, we know that the MDS stabilizer code
with parameters $[[n,k,d]]_q$ exists and  must be pure. Therefore it
obey the quantum Singleton bound with equality. Therefore the pure
subsystem code exists with parameters $[[n,k-r,r,d]]_q$ for $0\le r<
k$ and it must be an MDS code since it obeys the same bound with
equality.
\end{proof}
%
%
%
%

\subsection{Quantum Hamming Bound} We can also derive the quantum
Hamming bound on subsystem code parameters.  We can show that It is
easy to derive a Hamming like bound for pure subsystem codes as
stated in the following lemma.
\begin{lemma}[Hamming Bound.]\label{lem:hammingbound}
A pure $((n,K,K',d))_q$ code satisfies \begin{eqnarray}
\sum_{j=0}^{\lfloor\frac{d-1}{2} \rfloor}\binom{n}{j}(q^2-1)^j \leq
q^n/KK'.\end{eqnarray}
\end{lemma}
\begin{proof}
By Lemma~\ref{th:stabcode} a pure subsystem $((n,K,K',d))_q$ code
implies the existence of a pure $((n,KK',d))_q$ code.  But this
obeys the quantum Hamming bound \cite{feng04}. Therefore it follows
that \begin{eqnarray} \sum_{j=0}^{\lfloor\frac{d-1}{2}
\rfloor}\binom{n}{j}(q^2-1)^j \leq q^n/KK'.\end{eqnarray}
\end{proof}

Recall that a pure subsystem code is called perfect if and only if
it attains the Hamming bound with equality. We conclude this section
with the following consequence lemma:
\begin{lemma}
If there exists an $\F_q$-linear pure $[[n,k,d]]_q$ stabilizer code
that is perfect, then there exists a pure $\F_q$-linear
$[[n,k-r,r,d]]_q$ perfect subsystem code for all $r$ in the range
$0\leq r \leq k$.
\end{lemma}
\begin{proof}
Existence of an  $\F_q$-linear pure stabilizer code with parameters
$[[n,k,d]]_q$ implies  existence of a subsystem code with parameters
$[[n,k-r,r,d]]_q$ for $0 \leq r < k$. But we know that the
stabilizer code is perfect then

\begin{eqnarray}
\sum_{j=0}^{\lfloor (d-1)/2 \rfloor} \binom{n}{j}(q^2-1)^j =q^{n-k}
\end{eqnarray}
By Lemma~\ref{lem:hammingbound}, it is a direct consequence that the
subsystem code obeys this bound with equality.
\end{proof}

In the following chapters, we will give various methods to construct subsystem
codes. In addition, we will derive many families of subsystem codes. We will
give tables of upper and lower bounds on subsystem code parameters.

\chapter{Subsystem Code Constructions}\label{ch_subsys_construction}
Subsystem codes are the most versatile class of quantum
error-correcting codes known to date that combine the best features
of all known passive and active error-control schemes.  The
subsystem code is a subspace of the quantum state space that is
decomposed into a tensor product of two vector spaces: the subsystem
and the co-subsystem.  In this chapter, A generic method to derive
subsystem codes from existing subsystem codes is given that allows
one to trade the dimensions of subsystem and co-subsystem while
maintaining or improving the minimum distance. As a consequence, it
is shown that all pure MDS subsystem codes are derived from MDS
stabilizer codes. The existence of numerous families of MDS
subsystem codes is established.

\section{Introduction}
Subsystem codes are a relatively new construction of quantum codes
that combine the features of decoherence free
subspaces~\cite{lidar98}, noiseless subsystems~\cite{zanardi97}, and
quantum error-correcting codes~\cite{calderbank98,gottesman96}. Such
codes promise to offer appealing features, such as simplified
syndrome calculation and a wide variety of easily implementable
fault-tolerant operations,
see~\cite{aliferis06,aly06c,bacon06,kribs05}.

An $((n,K,R,d))_q$ subsystem code is a $KR$-dimensional subspace $Q$
of $\C^{q^n}$ that is decomposed into a tensor product $Q=A\otimes
B$ of a $K$-dimensional vector space $A$ and an $R$-dimensional
vector space $B$ such that all errors of weight less than~$d$ can be
detected by~$A$. The vector spaces $A$ and $B$ are respectively
called the subsystem $A$ and the co-subsystem $B$. For some
background on subsystem codes, see for
instance~\cite{klappenecker0608,poulin05,aly06c}.

A special feature of subsystem codes is that any classical additive
code $C$ can be used to construct a subsystem code. One should
contrast this with stabilizer codes, where the classical codes are
required to satisfy a self-orthogonality condition.
\medskip

We assume that the reader is familiar with the relation between
classical and quantum stabilizer codes,
see~\cite{calderbank98,rains99}. In \cite{aly06c,klappenecker0608},
the authors gave an introduction to subsystem codes, established
upper and lower bounds on subsystem code parameters, and provided
two methods for constructing subsystem codes. The main results on
this chapter are as follows:
\begin{compactenum}[i)]
\item If $q$ is a power of a prime $p$, then we show that a subsystem
code with parameters $((n,K/p,pR,\geq d))_q$ can be obtained from a
subsystem code with parameters $((n,K,R,d))_q$. Furthermore, we show
that the existence of a pure $((n,K,R,d))_q$ subsystem code implies
the existence of a pure $((n,pK,R/p,d))_q$ code.
\item We show that all pure MDS subsystem codes are derived from MDS
stabilizer codes. We establish here for the first time the existence
of numerous families of MDS subsystem codes.

\end{compactenum}

\section{Subsystem Code Constructions}
First we recall the following fact that is key to most constructions
of subsystem codes (see below for notations):
\begin{theorem}\label{th:oqecfq}
Let $C$ be a classical additive subcode of\/ $\F_q^{2n}$ such that
$C\neq \{0\}$ and let $D$ denote its subcode $D=C\cap C^\sdual$. If
$x=|C|$ and $y=|D|$, then there exists a subsystem code $Q= A\otimes
B$ such that
\begin{compactenum}[i)]
\item $\dim A = q^n/(xy)^{1/2}$,
\item $\dim B = (x/y)^{1/2}$.
\end{compactenum}
The minimum distance of subsystem $A$ is given by
\begin{compactenum}[(a)]
\item $d=\swt((C+C^\sdual)-C)=\swt(D^\sdual-C)$ if $D^\sdual\neq C$;
\item $d=\swt(D^\sdual)$ if $D^\sdual=C$.
\end{compactenum}
Thus, the subsystem $A$ can detect all errors in $E$ of weight less
than $d$, and can correct all errors in $E$ of weight $\le \lfloor
(d-1)/2\rfloor$.
\end{theorem}

A subsystem code that is derived with the help of the previous
theorem is called a Clifford subsystem code. We will assume
throughout this work that all subsystem codes are Clifford
subsystem codes. In particular, this means that the existence of an
$((n,K,R,d))_q$ subsystem code implies the existence of an additive
code $C\le \F_q^{2n}$ with subcode $D=C\cap C^\sdual$ such that
$|C|=q^nR/K$, $|D|=q^n/(KR)$, and $d=\swt(D^\sdual - C)$, see
Fig.~\ref{fig:stab_subssys3}.

A subsystem code derived from an additive classical code $C$ is
called pure to $d'$ if there is no element of symplectic weight less
than $d'$ in $C$. A subsystem code is called pure if it is pure to
the minimum distance $d$. We require that an $((n,1,R,d))_q$
subsystem code must be pure.

We also use the bracket notation $[[n,k,r,d]]_q$ to write the
parameters of an $((n,q^k,q^r,d))_q$ subsystem code in simpler form.
Some authors say that an $[[n,k,r,d]]_q$ subsystem code has $r$
gauge qudits, but this terminology is slightly confusing, as the
co-subsystem typically does not correspond to a state space of $r$
qudits except perhaps in trivial cases. We will avoid this
misleading terminology. An $((n,K,1,d))_q$ subsystem code is also an
$((n,K,d))_q$ stabilizer code and vice versa.

\medskip

{\em Notation.} Let $q$ be a power of a prime integer $p$. We denote
by $\F_q$ the finite field with $q$ elements. We use the notation
$(x|y)=(x_1,\dots,x_n|y_1,\dots,y_n)$ to denote the concatenation of
two vectors $x$ and $y$ in $\F_q^n$. The symplectic weight of
$(x|y)\in \F_q^{2n}$ is defined as $$\swt(x|y)=\{(x_i,y_i)\neq
(0,0)\,|\, 1\le i\le n\}.$$ We define $\swt(X)=\min\{\swt(x)\,|\,
x\in X, x\neq 0\}$ for any nonempty subset $X\neq \{0\}$ of
$\F_q^{2n}$.

The trace-symplectic product of two vectors $u=(a|b)$ and
$v=(a'|b')$ in $\F_q^{2n}$ is defined as
$$\langle u|v \rangle_s = \tr_{q/p}(a'\cdot b-a\cdot b'),$$ where
$x\cdot y$ denotes the dot product and $\tr_{q/p}$ denotes the trace
from $\F_q$ to the subfield $\F_p$.  The trace-symplectic dual of a
code $C\subseteq \F_q^{2n}$ is defined as $$C^\sdual=\{ v\in
\F_q^{2n}\mid \langle v|w \rangle_s =0 \mbox{ for all } w\in C\}.$$
We define the Euclidean inner product $\langle x|y\rangle
=\sum_{i=1}^nx_iy_i$ and the Euclidean dual of $C\subseteq \F_{q}^n$
as $$C^\perp = \{x\in \F_{q}^n\mid \langle x|y \rangle=0 \mbox{ for
all } y\in C \}.$$ We also define the Hermitian inner product for
vectors $x,y$ in $\F_{q^2}^n$ as $\langle x|y\rangle_h
=\sum_{i=1}^nx_i^qy_i$ and the Hermitian dual of $C\subseteq
\F_{q^2}^n$ as
$$C^\hdual= \{x\in \F_{q^2}^n\mid \langle x|y \rangle_h=0 \mbox{ for all } y\in
C \}.$$

\begin{figure}[t]
  \includegraphics[scale=0.7]{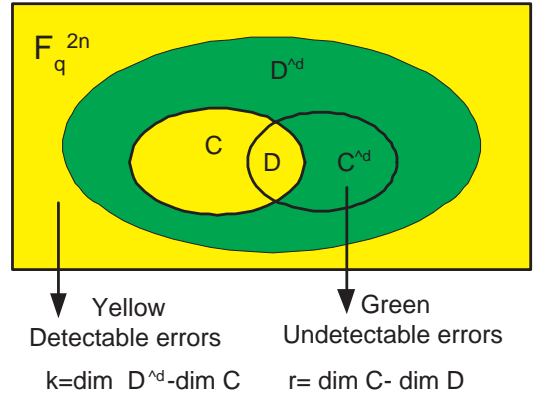}
 \centering
  \caption{Subsystem code parameters from classical codes}\label{fig:stab_subssys3}
\end{figure}

\medskip

\section{Trading Dimensions of Subsystem  Codes}\label{sec:dimensions}
In this section we show how one can trade the dimensions of
subsystem and co-subsystem to obtain new codes from a given
subsystem or stabilizer code. The results are obtained by exploiting
the symplectic geometry of the space. A remarkable consequence is
that nearly any stabilizer code yields a series of subsystem codes.

Our first result shows that one can decrease the dimension of the
subsystem and increase at the same time the dimension of the
co-subsystem while keeping or increasing the minimum distance of the
subsystem code.

\begin{theorem}\label{th:shrinkK}
Let $q$ be a power of a prime~$p$. If there exists an
$((n,K,R,d))_q$ subsystem code with $K>p$ that is pure to $d'$, then
there exists an $((n,K/p,pR,\geq d))_q$ subsystem code that is pure
to $\min\{d,d'\}$. If a pure $((n,p,R,d))_q$ subsystem code exists,
then there exists a $((n,1,pR,d))_q$ subsystem code.
\end{theorem}
\begin{proof}
By definition, an $((n,K,R,d))_q$ Clifford subsystem code is
associated with a classical additive code $C \subseteq \F_q^{2n}$
and its subcode $D=C\cap C^\sdual$ such that $x=|C|$, $y=|D|$,
$K=q^n/(xy)^{1/2}$, $R=(x/y)^{1/2}$, and $d=\swt(D^\sdual - C)$ if
$C\neq D^\sdual$, otherwise $d=\swt(D^\sdual)$ if $D^\sdual=C$.

We have $q=p^m$ for some positive integer $m$. Since $K$ and $R$ are
positive integers, we have $x=p^{s+2r}$ and $y=p^s$ for some
integers $r\ge 1$, and $s\ge 0$. There exists an $\F_p$-basis of $C$
of the form
$$ C = \spann_{\F_p}\{z_1,\dots,z_s,x_{s+1},z_{s+1},\dots,
x_{s+r},z_{s+r}\}$$ that can be extended to a symplectic basis
$\{x_1,z_1,\dots,x_{nm},z_{nm}\}$ of $\F_q^{2n}$, that is,
$\scal{x_k}{x_\ell}=0$, $\scal{z_k}{z_\ell}=0$,
$\scal{x_k}{z_\ell}=\delta_{k,\ell}$ for all $1\le k,\ell \le nm$,
see~\cite[Theorem 8.10.1]{cohn05}.

Define an additive code $$C_m =
\spann_{\F_p}\{z_1,\dots,z_s,x_{s+1},z_{s+1},\dots,
x_{s+r+1},z_{s+r+1}\}.$$ It follows that
$$C^\sdual_m=\spann_{\F_p}\{z_1,\dots,z_s,x_{s+r+2},z_{s+r+2}, \dots,
x_{nm},z_{nm}\}$$ and
$$D=C_m\cap C_m^\sdual =
\spann_{\F_p}\{z_1,\dots,z_s\}.$$ By definition, the code $C$ is a
subset of $C_m$.

The subsystem code defined by $C_m$ has the parameters
$(n,K_m,R_m,d_m)$, where $K_m=q^n/(p^{s+2r+2}p^s)^{1/2}=K/p$ and
$R_m=(p^{s+2r+2}/p^s)^{1/2}=pR$. For the claims concerning minimum
distance and purity, we distinguish two cases:
\begin{compactenum}[(a)]
\item If $C_m\neq D^\sdual$, then $K>p$ and $d_m=\swt(D^\sdual -
C_m)\ge \swt(D^\sdual-C)=d$. Since by hypothesis
$\swt(D^\sdual-C)=d$ and $\swt(C)\ge d'$, and $D\subseteq C\subset
C_m\subseteq D^\sdual$ by construction, we have $\swt(C_m)\ge \min\{
d,d'\}$; thus, the subsystem code is pure to $\min\{d,d'\}$.

\item If $C_m=D^\sdual$, then $K_m=1=K/p$, that is, $K=p$;  it follows from
the assumed purity that $d=\swt(D^\sdual-C)=\swt(D^\sdual)=d_m$.
\end{compactenum}
This proves the claim.
\end{proof}

For $\F_q$-linear subsystem codes there exists a variation of the
previous theorem which asserts that one can construct the resulting
subsystem code such that it is again $\F_q$-linear.

\begin{theorem}\label{th:FqshrinkK}
Let $q$ be a power of a prime~$p$. If there exists an $\F_q$-linear
$[[n,k,r,d]]_q$ subsystem code with $k>1$ that is pure to $d'$, then
there exists an $\F_q$-linear $[[n,k-1,r+1,\geq d]]_q$ subsystem
code that is pure to $\min\{d,d'\}$.  If a pure $\F_q$-linear
$[[n,1,r,d]]_q$ subsystem code exists, then there exists an
$\F_q$-linear $[[n,0,r+1,d]]_q$ subsystem code.
\end{theorem}
\begin{proof}
The proof is analogous to the proof of the previous theorem, except
that $\F_q$-bases are used instead of $\F_p$-bases.
\end{proof}

There exists a partial converse of Theorem~\ref{th:shrinkK}, namely
if the subsystem code is pure, then it is possible to increase the
dimension of the subsystem and decrease the dimension of the
co-subsystem while maintaining the same minimum distance.

\begin{theorem}\label{th:shrinkR}
Let $q$ be a power of a prime $p$. If there exists a pure
$((n,K,R,d))_q$ subsystem code with $R>1$, then there exists a pure
$((n,pK,R/p,d))_q$ subsystem code.
\end{theorem}
\begin{proof}
Suppose that the $((n,K,R,d))_q$ Clifford subsystem code is
associated with a classical additive code
$$ C_m = \spann_{\F_p}\{z_1,\dots,z_s,x_{s+1},z_{s+1},\dots,
x_{s+r+1},z_{s+r+1}\}.$$ Let $D=C_m\cap C_m^\sdual$. We have
$x=|C_m|=p^{s+2r+2}$, $y=|D|=p^s$, hence $K=q^n/p^{r+s}$ and
$R=p^{r+1}$. Furthermore, $d=\swt(D^\sdual)$.

The code $$C=\spann_{\F_p}\{z_1,\dots,z_s,x_{s+1},z_{s+1},\dots,
x_{s+r},z_{s+r}\}$$ has the subcode $D=C\cap C^\sdual$. Since
$|C|=|C_m|/p^2$, the parameters of the Clifford subsystem code
associated with $C$ are $((n,pK,R/p,d'))_q$. Since $C\subset C_m$,
the minimum distance $d'$ satisfies $$d'=\swt(D^\sdual-C)\le
\swt(D^\sdual - C_m)=\swt(D^\sdual)=d.$$  On the other hand,
$d'=\swt(D^\sdual-C)\ge \swt(D^\sdual)=d$, whence $d=d'$.
Furthermore, the resulting code is pure since
$d=\swt(D^\sdual)=\swt(D^\sdual-C)$.
\end{proof}

Replacing $\F_p$-bases by $\F_q$-bases in the proof of the previous
theorem yields the following variation of the previous theorem for
$\F_q$-linear subsystem codes.
\begin{theorem}\label{th:FqshrinkR}
Let $q$ be a power of a prime $p$. If there exists a pure
$\F_q$-linear $[[n,k,r,d]]_q$ subsystem code with $r>0$, then there
exists a pure $\F_q$-linear $[[n,k+1,r-1,d]]_q$ subsystem code.
\end{theorem}

The purity hypothesis in Theorems~\ref{th:shrinkR}
and~\ref{th:FqshrinkR} is essential, as the next remark shows.

\begin{remark}
The Bacon-Shor code is an impure $[[9,1,4,3]]_2$ subsystem code.
However, there does not exist any $[[9,5,3]]_2$ stabilizer code.
Thus, in general one cannot omit the purity assumption from
Theorems~\ref{th:shrinkR} and~\ref{th:FqshrinkR}, see also
Fig.~\ref{fig:stab_subsys4}.
\end{remark}

An $[[n,k,d]]_q$ stabilizer code can also be regarded as an
$[[n,k,0,d]]_q$ subsystem code. We record this important special
case of the previous theorems in the next corollary.

\goodbreak
\begin{corollary}\label{cor:generic}
If there exists an ($\F_q$-linear) $[[n,k,d]]_q$ stabilizer code
that is pure to $d'$, then there exists for all $r$ in the range
$0\le r<k$ an ($\F_q$-linear) $[[n,k-r,r,\ge d]]_q$ subsystem code
that is pure to $\min\{d,d'\}$ .  If a pure ($\F_q$-linear)
$[[n,k,r,d]]_q$ subsystem code exists, then a pure ($\F_q$-linear)
$[[n,k+r,d]]_q$ stabilizer code exists.
\end{corollary}

This result makes it very easy to obtain subsystem codes from
stabilizer codes. For example, if there is a stabilizer code with
parameters $[[9,3,3]]_2$, then there are subsystem codes with
parameters $[[9,1,2,3]]_2$ and $[[9,2,1,3]]_2$. The optimal
stabilizer codes derived in~\cite{grassl04,ketkar06} can all be
converted to subsystem codes. These code families satisfy Singleton
bound $k+2d=n+2$. An illustration of this corollary and families of
subsystem codes based on RS codes are given in the next chapter.

\medskip


\noindent \textbf{From Subsystem to Stabilizer Codes.}  We have
established a connection from stabilizer codes to subsystem codes as
well as trading the dimensions between subsystem codes and
co-subsystem codes. This result is applicable for both pure and
impure stabilizer codes. Here  we show that not all subsystem
(co-subsystem) codes can be reduced to stabilizer codes.  We gave a
partial answer to this statement in~\cite{aly06c}. We showed that
pure subsystem codes can be converted to pure stabilizer codes as
stated in Lemma~\ref{th:puresubsysTOstabcode}.

\begin{lemma}\label{th:puresubsysTOstabcode}
If a pure $((n,K,R,d))_q$  subsystem code~$Q$ exists, then there exists a pure
$((n,KR,d))_q$ stabilizer code.
\end{lemma}
\begin{proof}
Let $C$ be a classical additive subcode of $\F_q^{2n}$ that defines $Q$. The
code
$$C=\spann_{\F_p}\{z_1,\dots,z_s,x_{s+1},z_{s+1},\dots, x_{s+r},z_{s+r}\}$$ has
subcode $D=C\cap C^\sdual$. We have $|C|=p^{s+2r}$ and $|D|=p^s$ for some
integers $r\ge 1$, and $s\ge 0$.  Furthermore, we know that $K=q^n/(|C|
|D|)^{1/2}$ and $R=\sqrt{|C|/|D|}$,  then  $KR=q^n/|D|$. Since $D\subseteq
D^\sdual$, there exists an $((n,q^n/|D|,d'))_q$ stabilizer code with minimum
distance $d'=\wt(D^\sdual - D)$. The purity of $Q$ implies that $\swt(D^\sdual
- C) = \swt(D^\sdual)=d$. As $D\subseteq C$, it follows that $d'=\swt(D^\sdual
- D)=\swt(D^\sdual)=d$; hence, there exists a pure $((n,KR,d))_q$ stabilizer
code.
\end{proof}
Now, what  we can say about the impure subsystem codes. It turns out that not
every impure subsystem code can be transferred to a stabilizer code as shown in
the following Lemma.
\begin{lemma}\label{th:impuresubsysTOstabcode}
If an impure $((n,K,R,d))_q$ subsystem code~$Q$ exists, then there not
necessarily exists an impure $((n,KR,d))_q$ stabilizer code.
\end{lemma}
\begin{proof}
Let an impure $((n,K,R,d))_q$ subsystem code~$Q$ exists. We prove by
contradiction that there is no impure $((n,KR,d))_q$ stabilizer code in
general. The proof is shown by an example. We know that $[[9,1,4,3]]_2$
Becan-shor code is an impure code, which beats quantum Hamming bound for
subsystem codes. If an $[[9,5,3]]_2$ stabilizer code exists, then it would not
obey the quantum Hamming bound for quantum block codes. But, from the linear
programming upper bound, there is no such $[[9,5,3]]$ over the binary field,
see~\cite{calderbank98}. Therefore, not every impure subsystem code gives
stabilizer code.
\end{proof}

\medskip


\noindent \textbf{Subsystem versus Stabilizer Codes.} There is a
tradeoff between stabilizer and subsystem codes. We showed that one
can reduce subsystem codes with parameters $[[n,k,r,d]]_q$ for
$0\leq r <k$ to stabilizer codes with parameters $[[n-r,k,d]]_q$.
Also, pure subsystem codes with parameters $[[n,k,r,d]]_q$ give
raise to stabilizer codes with parameters $[[n,k+r,d]]_q$. In the
other hand, one can start with a stabilizer code with parameters
$[[n,k,d]]_q$ and obtain a subsystem code with parameters
$[[n,k-r,r,d]]_q$, for $0\leq r<k$, see Corollary~\ref{cor:generic}.
The comparison between subsystem codes and stabilizer codes can be
viewed as follows.
\begin{compactitem}
\item Syndrome measurements.
One way is to look at the number of syndrome measurements. Stabilizer codes
need $n-k$ syndrome measurements while subsystem codes need $n-k-r$ for fixed
$n$ and $d$, as for example, the short subsystem code $[[8,2,1,3]]_2$ (or
$[[8,1,2,3]]_2$).

\item Subsystem codes  may beat the Singleton and Hamming bound.
There might exist subsystem codes that beat the quantum Singleton bound $k+r
\leq n-2d+2$ and the quantum Hamming bound $\sum_{i=0}^{\lfloor (d-1)/2
\rfloor} \binom{n}{i}(q^2-1)^i \leq q^n/KR$. We have not found any codes for
small length $n \leq 50$, using MAGMA computer algebra,  that beat the
Singleton bound. Most likely there are no codes that beat this bound as we
showed in case of linear pure subsystem codes in~\cite{aly06c}. Pure subsystem
codes obey the quantum Hamming bound. In the other hand, there are some impure
subsystem codes that beat the quantum Hamming bound. For example, subsystem
codes with parameters $[[9,1,4,3]]_2$, $[[ 25, 1, 16, 5 ]]_2$, and $[[30, 1,
20, 5 ]]_2$ do not obey the quantum Hamming bound. They are constructed using
Bacon-Shor code constructions over $\F_2$. In fact, we found many subsystem
codes that do not obey this bound and be easily derived from this construction.

\item
Encoding and decoding circuits. It has been shown that the encoding and
decoding circuits of stabilizer codes can also be used in subsystem codes. The
conjecture is that subsystem codes might have better efficient encoding and
decoding circuits using  benefit of the gauge qubits, see~\cite{bacon06b}.
\item Fault tolerant and subsystem codes. It has been shown recently that subsystem
codes are suitable to protect quantum information since they have a good
strategy of fault tolerant and  high threshold values, see~\cite{aliferis06}.
\end{compactitem}

\begin{figure}[t]
  \includegraphics[scale=0.7]{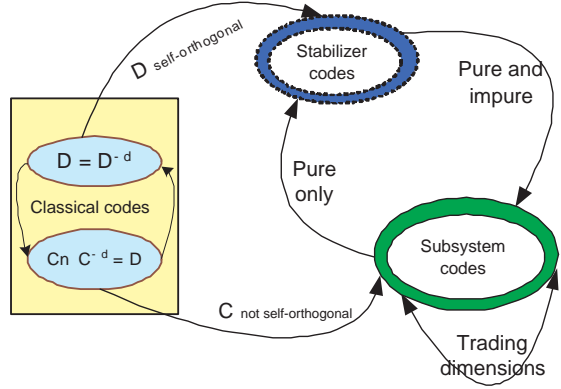}
 \centering
  \caption{Stabilizer and subsystem codes based on classical codes}\label{fig:stab_subsys4}
\end{figure}

\bigskip

\section{MDS Subsystem Codes}
In this section we derive all MDS subsystem codes.  Recall that an
$[[n,k,r,d]]_q$ subsystem code derived from an $\F_q$-linear
classical code $C\le \F_q^{2n}$ satisfies the Singleton bound
$k+r\le n-2d+2$. A subsystem code attaining the Singleton bound with
equality is called an MDS subsystem code. An important consequence
is the following simple observation which yields an easy
construction of subsystem codes that are optimal among the
$\F_q$-linear Clifford subsystem codes.

\begin{theorem}\label{th:pureMDS}
If there exists an $\F_q$-linear $[[n,k,d]]_q$ MDS stabilizer code,
then there exists a pure $\F_q$-linear $[[n,k-r,r,d]]_q$ MDS
subsystem code for all $r$ in the range $0\le r\le k$.
\end{theorem}
\begin{proof}
An MDS stabilizer code must be pure, see~\cite[Theorem~2]{rains99}
or \cite[Corollary 60]{ketkar06}. By Corollary~\ref{cor:generic}, a
pure $\F_q$-linear $[[n,k,d]]_q$ stabilizer code implies the
existence of an $\F_q$-linear $[[n,k-r,r, d_r\ge d]]_q$ subsystem
code that is pure to~$d$ for any $r$ in the range $0\le r\le k$.
Since the stabilizer code is MDS, we have $k=n-2d+2$. By the
Singleton bound, the parameters of the resulting $\F_q$-linear
$[[n,n-2d+2-r,r,d_r]]_q$ subsystem codes must satisfy
$(n-2d+2-r)+r\le n-2d_r+2$, which shows that the minimum distance
$d_r=d$, as claimed.
\end{proof}

\medskip

\medskip

\begin{remark}
We conjecture that $\F_q$-linear MDS subsystem codes are actually optimal among
all subsystem codes, but a proof that the Singleton bound holds for general
subsystem codes remains elusive.
\end{remark}

\medskip

We recall that the Hermitian construction of stabilizer codes yields
$\F_q$-linear stabilizer codes, as can be seen from our reformulation
of~\cite[Corollary~2]{grassl04}.

\medskip

\begin{lemma}[\cite{grassl04}]\label{l:hermitian-linear}
If there exists an $\F_{q^2}$-linear code $X\subseteq \F_{q^2}^n$ such that
$X\subseteq X^\hdual$, then there exists an $\F_q$-linear code $C\subseteq
\F_q^{2n}$ such that $C\subseteq C^\sdual$, $|C|=|X|$, $\swt(C^\sdual -
C)=\wt(X^\hdual - X)$ and $\swt(C)=\wt(X)$.
\end{lemma}
\begin{proof}
Let $\{1,\beta\}$ be a basis of $\F_{q^2}/\F_q$. Then
$\tr_{q^2/q}(\beta)=\beta+\beta^q$ is an element $\beta_0$ of $\F_q$; hence,
$\beta^q=-\beta+\beta_0$. Let $$C=\{ (u|v)\,|\, u,v\in \F_q^n, u+\beta v\in
X\}.$$ It follows from this definition that $|X|=|C|$ and that
$\wt(X)=\swt(C)$. Furthermore, if $u+\beta v$ and $u'+\beta v'$ are elements of
$X$ with $u,v,u',v'$ in $\F_q^n$, then
$$
\begin{array}{lcl}
0&=&(u+\beta v)^q\cdot (u'+\beta v') \\
&=& u\cdot u' + \beta^{q+1} v\cdot v' + \beta_0 v \cdot u' + \beta (u\cdot v'
-v \cdot u').
\end{array}
$$ On the right hand side, all terms but the last are in $\F_q$; hence
we must have $(u\cdot v' -v \cdot u')=0$, which shows that $(u|v) \,\sdual\,
(u'|v')$, whence $C\subseteq C^\sdual$. Expanding $X^\hdual$ in the basis
$\{1\,\beta\}$ yields a code $C'\subseteq C^\sdual$, and we must have equality
by a dimension argument. Since the basis expansion is isometric, it follows
that $$\swt(C^\sdual - C)=\wt(X^\hdual - X).$$ The $\F_q$-linearity of $C$ is a
direct consequence of the definition of $C$.
\end{proof}

\medskip

In corollary~\ref{cor:examplesMDS}, we give a few examples of MDS subsystem
codes that can be obtained from Theorem~\ref{th:pureMDS}.

\medskip

\begin{corollary}\label{cor:examplesMDS}
\begin{enumerate}[i)]
\item An $\F_q$-linear pure $[[n,n-2d+2-r,r,d]]_q$ MDS subsystem code exists
for all $n$, $d$, and $r$ such that $3\le n\le q$, $1\le d\le n/2+1$, and\/
$0\le r\le n-2d+1$.
\item An $\F_q$-linear pure $[[(\nu+1)q,(\nu+1)q-2\nu-2-r,r,\nu+2]]_q$ MDS subsystem code exists for all $\nu$ and $r$ such that $0\le \nu\le q-2$ and
$0\le r\le (\nu+1)q-2\nu-3$.
\item An $\F_q$-linear pure $[[q - 1, q-1 -2\delta -r,
r,\delta + 1]]_q$ MDS subsystem code exists for all $\delta$ and $r$ such that
$0 \leq \delta < (q -1)/2$ and $0\leq r \le q - 2\delta - 1$.
\item An $\F_q$-linear pure $[[q, q -
2\delta - 2-r',r', \delta + 2]]_q$ MDS subsystem code exists for all $0 \leq
\delta < (q -1)/2$ and $0\leq r' <q - 2\delta - 2$.
\item An $\F_q$-linear pure $[[q^2 - 1, q^2 - 2\delta - 1-r,r, \delta +
1]]_q$ MDS subsystem code exists for all $\delta$ and $r$ in the range $0 \leq
\delta < q-1$ and $0\leq r< q^2 - 2\delta - 1$.
\item An $\F_q$-linear pure $[[q^2, q^2 - 2\delta -
2-r',r', \delta + 2]]_q$ MDS subsystem code exists for all $\delta$ and $r'$ in
the range $0 \leq \delta < q-1$ and $0\leq r' <q^2 - 2\delta - 2$.
\end{enumerate}
\end{corollary}
\begin{proof}
\begin{enumerate}[i)]
\item[i)] By \cite[Theorem~14]{grassl04}, there exist $\F_q$-linear
$[[n,n-2d+2,d]]_q$ stabilizer codes for all $n$ and $d$ such that $3\le n\le q$
and $1\le d\le n/2+1$. The claim follows from Theorem~\ref{th:pureMDS}.

\item[ii)] By \cite[Theorem~5]{sarvepalli05}, there exist a
$[[(\nu+1)q,(\nu+1)q-2\nu-2,\nu+2]]_q$ stabilizer code. In this case, the code
is derived from an $\F_{q^2}$-linear code $X$ of length $n$ over $\F_{q^2}$
such that $X\subseteq X^\hdual$. The claim follows from
Lemma~\ref{l:hermitian-linear} and Theorem~\ref{th:pureMDS}.

\item [iii)], iv) There exist $\F_q$-linear stabilizer codes with
parameters $[[q - 1, q - 2\delta - 1,\delta + 1]]_q$ and $[[q, q - 2\delta - 2,
\delta + 2]]_q$ for $0 \leq \delta < (q -1)/2$, \
see~\cite[Theorem~9]{grassl04}. Theorem~\ref{th:pureMDS} yields the claim.

\item [v)], vi) There exist $\F_q$-linear stabilizer codes with
parameters $[[q^2 - 1, q^2 - 2\delta - 1, \delta + 1]]_q$ and $[[q^2, q^2 -
2\delta - 2, \delta + 2]]_q$.  for $0 \leq \delta < q-1$
by~\cite[Theorem~10]{grassl04}. The claim follows from
Theorem~\ref{th:pureMDS}.
\end{enumerate}
\end{proof}
The existence of the codes in i) are merely established by a non-constructive
Gilbert-Varshamov type counting argument.  However, the result is interesting,
as it asserts that there exist for example $[[6,1,1,3]]_q$ subsystem codes for
all prime powers $q\ge 7$, $[[7,1,2,3]]_q$ subsystem codes for all prime powers
$q\ge 7$, and other short subsystem codes that one should compare with a
$[[5,1,3]]_q$ stabilizer code. If the syndrome calculation is simpler, then
such subsystem codes could be of practical value.

The subsystem codes given in ii)-vi) of the previous corollary are
constructively established. The subsystem codes in ii) are derived from
Reed-Muller codes, and in iii)-vi) from Reed-Solomon codes. There exists an
overlap between the parameters given in ii) and in iv), but we list here both,
since each code construction has its own merits.

\medskip

\begin{remark}
By Theorem~\ref{th:FqshrinkR}, pure MDS subsystem codes can always be derived
from MDS stabilizer codes. Therefore, one can derive in fact all possible
parameter sets of pure MDS subsystem codes with the help of
Theorem~\ref{th:pureMDS}.
\end{remark}

\medskip

\begin{remark}
In the case of stabilizer codes, all MDS codes must be pure. For subsystem
codes this is not true, as the $[[9,1,4,3]]_2$ subsystem code shows. Finding
such impure $[[n,k,r,d]]_q$ MDS subsystem codes with $k+r> n-2d+2$ is a
particularly interesting challenge.
\end{remark}

\bigskip

\section{Conclusion and Discussion}
Subsystem codes -- or operator quantum error-correcting codes as some authors
prefer to call them -- are among the most versatile tools in quantum
error-correction, since they allow one to combine the passive error-correction
found in decoherence free subspaces and noiseless subsystems with the active
error-control methods of quantum error-correcting codes. The subclass of
Clifford subsystem codes that was studied in this chapter is of particular
interest because of the close connection to classical error-correcting codes.
As Proposition~\ref{th:oqecfq} shows, one can derive from each additive code
over $\F_q$ an Clifford subsystem code. This offers more flexibility than the
slightly rigid framework of stabilizer codes. However, there exist few
systematic constructions of good families subsystem codes and much of the
theory remains to be developed. For instance, more bounds are needed for the
parameters of subsystem codes.

In this chapter, we showed that any $\F_q$-linear MDS stabilizer code yields a
series of pure $\F_q$-linear MDS subsystem codes. These codes are known to be
optimal among the $\F_q$-linear Clifford subsystem codes. We conjecture that
the Singleton bound holds in general for subsystem codes. There is quite some
evidence for this fact, as pure Clifford subsystem codes and $\F_q$-linear
Clifford subsystem codes are known to obey this bound.

We used Reed-Muller and Reed-Solomon codes to derive pure $\F_q$-linear MDS
subsystem codes. In a similar fashion, one can derive other interesting
subsystem codes from BCH stabilizer codes, see for instance~\cite{aly07a}.

\chapter{Families of Subsystem Codes}\label{ch_subsys_families}
In this chapter I construct families of subsystem codes over finite fields. I will
derive cyclic subsystem codes, as well as BCH and RS subsystem
codes. I will present an optimal family of subsystem codes in a
sense that this family obeys quantum Singleton bound with equality. This chapter and next one are  appeared in a joint work with A. Klappenecker in~\cite{aly08a}.

\section{Introduction}

Let $Q$ be a quantum code  such that $\mathcal{H}=Q\oplus Q^\perp$,
where $Q^\perp$ is the orthogonal complement of $Q$. Recall
definition of the error model acting in  qubits as shown in
Chapter~\ref{ch_QBC_basics}. We can define the subsystem code $Q$ as
follows.
\begin{definition}
An $[[n,k,r,d]]_q$ subsystem code is a decomposition of the subspace
$Q$ into a tensor product of two vector spaces A and B such that
$Q=A\otimes B$, where  $\dim A=q^k$ and $\dim B=q^r$. The code $Q$
is able to detect all errors  of weight less than $d$ on subsystem
$A$.
\end{definition}

Subsystem codes can be constructed  from  classical codes  over
$\F_q$ and $\F_{q^2}$. We recall the Euclidean and Hermitian
construction from~\cite{aly06c}.
\begin{lemma}[Euclidean Construction]\label{lem:css-Euclidean-subsys}
If $C$ is a $k'$-dimensional $\F_q$-linear code of length $n$ that
has a $k''$-dimensional subcode $D=C\cap C^\perp$ and $k'+k''<n$,
then there exists an
$$[[n,n-(k'+k''),k'-k'',\wt(D^\perp\setminus C)]]_q$$
subsystem code.
\end{lemma}
\begin{proof}
Let us define the code $X=C\times C \subseteq \F_q^{2n}$, therefore
$X^\sdual=(C\times C)^\sdual=C^\sdual\times C^\sdual$. Hence $Y=X
\cap X^\sdual=(C\times C)\cap (C^\sdual\times C^\sdual)=C\cap
C^\sdual$. Let $\dim_{\F_q} Y=k''$. Hence $|X||Y|=q^{k'+k''}$ and
$|X|/|Y|=q^{k'-k''}$. By Theorem~\cite[Theorem 1]{aly06c}, there
exists a subsystem code $Q=A\otimes B$ with parameters $[[n,\dim
A,\dim B, d]]_q$ such that
\begin{compactenum}[i)]
\item $\dim A=q^n/(|X||Y|)=q^{n-k'-k''}$.
\item $\dim B=|X|/|Y|=q^{k'-k''}$.
\item $d= swt(Y^\sdual \backslash X)=\wt(D^\perp\setminus C)$.
\end{compactenum}
\end{proof}

Also, subsystem codes can be constructed from two classical codes
using the Euclidean construction as shown in the following lemma.

\begin{lemma}[Euclidean Construction]\label{lem:css-Euclidean-subsys}
Let $C_i \subseteq \F_q^n$, be $[n,k_i]_q$ linear codes where $i\in
\{1,2\}$. Then there exists an $[[n,k,r,d]]_q$ subsystem code with
\begin{compactitem}
\item $k=n-(k_1+k_2+k')/2$,
\item $r=(k_1+k_2-k')/2$, and
\item
$d=\min \{ \wt((C_1^\perp\cap C_2)^\perp\setminus C_1),
\wt((C_2^\perp\cap C_1)^\perp\setminus C_2) \}$,
\end{compactitem}
where $k'= \dim_{\F_q}(C_1\cap C_2^\perp)\times (C_1^\perp\cap
C_2)$.
\end{lemma}
\smallskip

 Also, the subsystem codes can be derived from classical codes,
that are defined over $\F_{q^2}$, using the Hermitian construction.

\begin{lemma}[Hermitian Construction]\label{lem:css-Hermitina-subsys}
Let $C \subseteq \F_{q^2}^n$ be an $\F_{q^2}$-linear $[n,k,d]_{q^2}$
code such that $D=C\cap C^\hdual$ is of dimension
$k'=\dim_{\F_{q^2}} D$. Then there exists an
$$[[n,n-k-k',k-k',\wt(D^\hdual \setminus C)]]_q$$ subsystem code.
\end{lemma}

\noindent {\em Notation.} If $S$ is a set, then $|S|$ denotes the
cardinality of the set $S$. Let $q$ be a power of a prime integer
$p$. We denote by $\F_q$ the finite field with $q$ elements. We use
the notation $(x|y)=(x_1,\dots,x_n|y_1,\dots,y_n)$ to denote the
concatenation of two vectors $x$ and $y$ in $\F_q^n$. The symplectic
weight of $(x|y)\in \F_q^{2n}$ is defined as
$$\swt(x|y)=\{(x_i,y_i)\neq (0,0)\,|\, 1\le i\le n\}.$$ We define
$\swt(X)=\min\{\swt(x)\,|\, x\in X, x\neq 0\}$ for any nonempty
subset $X\neq \{0\}$ of $\F_q^{2n}$. The trace-symplectic product of
two vectors $u=(a|b)$ and $v=(a'|b')$ in $\F_q^{2n}$ is defined as
$$\langle u|v \rangle_s = \tr_{q/p}(a'\cdot b-a\cdot b'),$$ where
$x\cdot y$ denotes the dot product and $\tr_{q/p}$ denotes the trace
from $\F_q$ to the subfield $\F_p$.  The trace-symplectic dual of a
code $C\subseteq \F_q^{2n}$ is defined as $$C^\sdual=\{ v\in
\F_q^{2n}\mid \langle v|w \rangle_s =0 \mbox{ for all } w\in C\}.$$
We define the Euclidean inner product $\langle x|y\rangle
=\sum_{i=1}^nx_iy_i$ and the Euclidean dual of $C\subseteq \F_{q}^n$
as $$C^\perp = \{x\in \F_{q}^n\mid \langle x|y \rangle=0 \mbox{ for
all } y\in C \}.$$ We also define the Hermitian inner product for
vectors $x,y$ in $\F_{q^2}^n$ as $\langle x|y\rangle_h
=\sum_{i=1}^nx_i^qy_i$ and the Hermitian dual of $C\subseteq
\F_{q^2}^n$ as
$$C^\hdual= \{x\in \F_{q^2}^n\mid \langle x|y \rangle_h=0 \mbox{ for all } y\in
C \}.$$

\section{Cyclic Subsystem Codes}\label{Sec:CyclicSubsys}

In this section we shall derive subsystem codes from classical
cyclic codes. We first recall some definitions before embarking on
the construction of subsystem codes.  For further details concerning
cyclic codes see for instance~\cite{huffman03} and
\cite{macwilliams77}.

Let $n$ be a positive integer and $\F_q$ a finite field with $q$
elements such that $\gcd(n,q)=1$. Recall that a linear code
$C\subseteq \F_q^n$ is called \textit{cyclic} if and only if
$(c_0,\dots,c_{n-1})$ in $C$ implies that $(c_{n-1},c_0,\dots,
c_{n-2})$ in $C$.

For $g(x)$ in $\F_q[x]$, we write $(g(x))$ to denote the principal
ideal generated by $g(x)$ in $\F_q[x]$. Let $\pi$ denote the vector
space isomorphism $\pi\colon \F_q^n\rightarrow R_n=\F_q[x] /
(x^n-1)$ given by
$$ \pi((c_0,\dots,c_{n-1})) =
c_0+c_1x+\cdots+c_{n-1}x^{n-1}+(x^n-1).$$ A cyclic code $C\subseteq
\F_q^n$ is mapped to a principal ideal $\pi(C)$ of the ring $R_n$.
For a cyclic code $C$, the unique monic polynomial $g(x)$ in
$\F_q[x]$ of the least degree such that $(g(x))=\pi(C)$ is called
the \textit{generator polynomial} of $C$. If $C\subseteq \F_q^n$ is
a cyclic code with generator polynomial $g(x)$, then
$$\dim_{\F_q} C = n-\deg g(x).$$

Since $\gcd(n,q)=1$, there exists a primitive $n^\text{th}$ root of
unity $\alpha$ over $\F_q$; that is, $\F_q[\alpha]$ is the splitting
field of the polynomial $x^n-1$ over $\F_q$. Let us henceforth fix
this primitive $n^\text{th}$ primitive root of unity $\alpha$.
Since the generator polynomial $g(x)$ of a cyclic code $C\subseteq
\F_q^n$ is of minimal degree, it follows that $g(x)$ divides the
polynomial $x^n-1$ in $\F_q[x]$.  Therefore, the generator
polynomial $g(x)$ of a cyclic code $C\subseteq \F_q^n$ can be
uniquely specified in terms of a subset $T$ of $\{0,\dots,n-1\}$
such that
$$ g(x) = \prod_{t\in T} (x-\alpha^t).$$ The set $T$ is called the
\textit{defining set} of the cyclic code $C$ (with respect to the
primitive $n^\text{th}$ root of unity $\alpha$). A defining set is
the union of cyclotomic cosets modulo $n$. The following lemma
recalls some well-known and easily proved facts about defining sets
(see e.g.~\cite{huffman03}).

\begin{lemma}\label{lem:definingsets}
Let $C_i$ be a cyclic code of length $n$ over $\F_q$ with defining
set a $T_i$ for $i=1,2$. Let $N=\{0,1,\dots,n-1\}$ and
$T_1^{a}=\{at\bmod n\,|\, t\in T\}$ for some integer $a$. Then
\begin{compactenum}[i)]
\item $C_1\cap C_2$ has defining set $T_1 \cup T_2$.
\item $C_1+C_2$ has defining set $T_1 \cap T_2$.
\item $C_1 \subseteq C_2$ if and only if $T_2 \subseteq T_1$.
\item $C_1^\perp$ has defining set $N \setminus T_1^{-1}$.
\item $C_1^\hdual$ has defining set $N\setminus T_1^{-r}$ provided
that $q=r^2$ for some positive integer $r$.
\end{compactenum}
\end{lemma}
\textit{Notation.} If $T$ is a defining set of a cyclic code of
length $n$, then we denote henceforth by $T^a$ the set
$$T^a = \{ at\bmod n\,|\, t\in T\},$$ as in the previous lemma. We use
a superscript, since this notation will be frequently used in set
differences, and arguably $N\setminus T^{-q}$ is more readable than
$N\setminus -qT$.

Now, we shall give a general construction for subsystem cyclic
codes. We say that a code $C$ is self-orthogonal if and only if
$C\subseteq C^\perp$. We show that if a classical cyclic code is
self-orthogonal, then one can easily construct cyclic subsystem
codes.

\begin{proposition}\label{lem:cyclic-subsysI}
Let $D$ be a self-orthogonal cyclic code of length $n$ over $\F_q$
with defining set $T_D$. Let $T_D$ and $T_{D^\perp}$ respectively
denote the defining sets of $D$ and $D^\perp$. If $T$ is a subset of
$T_D \setminus T_{D^\perp}$, then one can define a cyclic code $C$
of length $n$ over $\F_q$ by the defining set $T_C= T_D \setminus (T
\cup T^{-1})$.  If $n-k=|T_{D}|$, $r=|T\cup T^{-1}|$ with $0\le r<
n-2k$, and $d= \min \wt(D^\perp \setminus C)$, then there exists a
subsystem code with parameters $[[n,n-2k-r,r,d]]_q$.
\end{proposition}
\begin{proof}
Since $D$ is a self-orthogonal cyclic code, we have $D\subseteq
D^\perp$, whence $T_{D^\perp} \subseteq T_{D}$ by
Lemma~\ref{lem:definingsets}~iii).  Observe that if $s$ is an
element of the set  $S= T_D\setminus T_{D^\perp} = T_D \setminus
(N\setminus T_D^{-1})$, then $-s$ is an element of $S$ as well. In
particular, $T^{-1}$ is a subset of $T_D\setminus T_{D^\perp}$.

By definition, the cyclic code $C$ has the defining set $T_C= T_D
\setminus (T \cup T^{-1})$; thus, the dual code $C^\perp$ has the
defining set
$$T_{C^\perp}=N\setminus T_{C}^{-1} =
T_{D^\perp}\cup (T\cup T^{-1}).$$ Furthermore, we have
$$T_C \cup T_{C^\perp}=(T_D \setminus (T \cup T^{-1})) \cup (T_{D^\perp}\cup
T\cup T^{-1})=T_D;$$ therefore, $C \cap C^{\perp}=D$ by
Lemma~\ref{lem:definingsets}~i).

Since $n-k=|T_D|$ and $r=|T\cup T^{-1}|$, we have $\dim_{\F_q}
D=n-|T_D| = k$ and $\dim_{\F_q} C = n-|T_C|=k+r$. Thus, by
Lemma~\ref{lem:css-Euclidean-subsys} there exists an $\F_q$-linear
subsystem code with parameters $[[n,\kappa,\rho,d]]_q$, where
\begin{compactenum}[i)]
\item $\kappa = \dim D^\perp -\dim C=n-k-(k+r)=n-2k-r$,
\item $\rho = \dim C -\dim D= k+r-k=r$,
\item $d= \min \wt(D^\perp \setminus C)$,
\end{compactenum}
as claimed.
\end{proof}

We notice that if $\wt(D)\leq \wt(D^\perp)$, then the constructed
cyclic subsystem codes are impure. In addition, if
$d=\wt(D^\perp)=\wt(D^\perp\backslash D) $, then the constructed
codes are pure up to d.

We can also derive subsystem codes from cyclic codes over $\F_{q^2}$
by using cyclic codes that are self-orthogonal with respect to the
Hermitian inner product.

\begin{proposition}\label{lem:cyclic-subsysII}
Let $D$ be a cyclic code of length $n$ over $\F_{q^2}$ such that
$D\subseteq D^\hdual$. Let $T_D$ and $T_{D^\hdual}$ respectively be
the defining set of $D$ and $D^\hdual$. If $T$ is a subset of $T_D
\setminus T_{D^\hdual}$, then one can define a cyclic code $C$ of
length $n$ over $\F_{q^2}$ with defining set $T_C= T_D \setminus (T
\cup T^{-q})$.  If $n-k=|T_D|$ and $r=|T\cup T^{-q}|$ with $0\le
r<n-2k$, and $d=\wt(D^\hdual\setminus C)$, then there exists an
$[[n,n-2k-r,r, d]]_q$ subsystem code.
\end{proposition}
\begin{proof}
Since $D \subseteq D^\hdual$, their defining sets satisfy
$T_{D^\hdual} \subseteq T_{D}$ by Lemma~\ref{lem:definingsets}~iii).
If $s$ is an element of $T_{D}\setminus T_{D^\hdual}$, then one
easily verifies that $-qs \pmod n$ is an element of $T_{D}\setminus
T_{D^\hdual}$.

Let $N=\{0,1,\dots,n-1\}$. Since the cyclic code $C$ has the
defining set $T_C= T_D \setminus (T \cup T^{-q})$, its dual code
$C^\hdual$ has the defining set $T_{C^\hdual}= N\setminus T_C^{-q} =
T_{D^\hdual}\cup (T\cup T^{-q}).$ We notice that
$$T_C \cup T_{C^\hdual}=(T_D \setminus
(T \cup T^{-q})) \cup (T_{D^\hdual}\cup T\cup T^{-q})=T_D;$$ thus,
$C \cap C^{\hdual}=D$ by Lemma~\ref{lem:definingsets}~i).

Since $n-k=|T_D|$ and $r=|T\cup T^{-q}|$, we have $\dim D=n-|T_D|=k$
and $\dim C=n-|T_C|=k+r$. Thus, by
Lemma~\ref{lem:css-Hermitina-subsys} there exists an
$[[n,\kappa,\rho,d]]_q$ subsystem code with
\begin{compactenum}[i)]
\item $\kappa = \dim D^\hdual -\dim C=(n-k)-(k+r)=n-2k-r$,
\item $\rho=\dim C -\dim D= k+r-k=r$,
\item $d= \min \wt(D^\hdual \setminus C)$,
\end{compactenum}
as claimed.
\end{proof}

We notice that if $\wt(D)\leq \wt(D^\hdual)$, then the constructed
cyclic subsystem codes are impure. In addition, if
$d=\wt(D^\perp)=\wt(D^\hdual \backslash D) $, then the constructed
codes are pure up to d.

The previous two propositions allow one to easily construct
subsystem codes from classical cyclic codes. We will illustrate this
fact by deriving cyclic subsystem codes from BCH and Reed-Solomon
codes. Also, one can derive subsystem codes from classical cyclic
codes if the generator polynomial is known.
%

%
\section{Subsystem BCH Codes}
In this section we consider an important class of cyclic codes that
can be constructed with arbitrary designed distance $\delta$. We
will construct families of subsystem BCH codes.

Let $n$ be a positive integer,  $\F_q$ be a finite field with $q$
elements, and $\alpha$ is a primitive $n$th root of unity.  A
primitive narrow-sense BCH code $C$ of designed distance $\delta$
and length $n$ is a cyclic code with generator monic polynomial
$g(x)$ over $\F_q$ that has $\alpha, \alpha^2, \ldots,
\alpha^{\delta-1}$ as zeros. $c$ is a codeword in $C$ if and only if
$c(\alpha)=c(\alpha^2)=\ldots=c(\alpha^{\delta-1})=0$. The parity
check matrix of this code  can be defined as
\begin{eqnarray}
 H =\left[ \begin{array}{ccccc}
1 &\alpha &\alpha^2 &\cdots &\alpha^{n-1}\\
1 &\alpha^2 &\alpha^4 &\cdots &\alpha^{2(n-1)}\\
\vdots& \vdots &\vdots &\ddots &\vdots\\
1 &\alpha^{\delta-1} &\alpha^{2(\delta-1)} &\cdots
&\alpha^{(\delta-1)(n-1)}
\end{array}\right]
\end{eqnarray}

We have shown in~\cite{aly07a,aly06a} that narrow sense BCH codes,
primitive and non-primitive, with length $n$ and designed distance
$\delta$ are Euclidean dual-containing codes if and only if $2\le
\delta\le \delta_{\max}=\frac{n}{q^{m}-1} (q^{\lceil
m/2\rceil}-1-(q-2)[m \textup{ odd}])$.  We use this result and~\cite[Theorem 2]{aly08a} to
derive primitive subsystem BCH codes from classical BCH codes over
$\F_q$ and $\F_{q^2}$~\cite{aly06c,aly06a}.

\begin{lemma}\label{lem:BCHExistFq}
If $q$ is a power of a prime, $m$ is a positive integer, and $2\leq
\delta\leq q^{\lceil m/2\rceil}-1 -(q-2)[m \text{ odd }]$. Then
there exists a subsystem BCH code with parameters
$[[q^m-1,n-2m\lceil(\delta-1)(1-1/q) \rceil -r, r,\geq \delta ]]_q$
where $0 \leq r< n-2m\lceil(\delta-1)(1-1/q) \rceil$.
\end{lemma}
\begin{proof}
We know that if $2\leq \delta\leq q^{\lceil m/2\rceil}-1 -(q-2)[m
\text{ odd }]$, then there exists a stabilizer code with parameters
$[[q^m-1,n-2m\lceil(\delta-1)(1-1/q) \rceil, \geq \delta ]]_q$. Let
r be an integer in the range $0 \leq r< n-2m\lceil(\delta-1)(1-1/q)
\rceil$. From~\cite[Theorem 2]{aly08a}, then there must exist a
subsystem BCH code with parameters $
[[q^m-1,n-2m\lceil(\delta-1)(1-1/q) \rceil -r, r,\geq \delta ]]_q$.
\end{proof}

\medskip

\begin{lemma}\label{lem:BCHExistFq2}
If $q$ is a power of a prime, $m$ is a positive integer, and
$\delta$ is an integer in the range $2\le \delta \le
\delta_{\max}=q^{m+[m \textup{ even}]}-1 -(q^2-2)[m \textup{
even}]$, then there exists a subsystem code $Q$ with parameters
$$ [[q^{2m}-1, q^{2m}-1-2m\lceil(\delta-1)(1-1/q^2)\rceil -r,r, d_Q\ge
\delta]]_q$$ that is pure up to $\delta$, where $0 \leq r
<q^{2m}-1-2m\lceil(\delta-1)(1-1/q^2)\rceil$.
\end{lemma}
\begin{proof}
If $2\le \delta \le \delta_{\max}=q^{m+[m \textup{ even}]}-1
-(q^2-2)[m \textup{ even}]$, then exists a classical BCH code with
parameters $[q^m-1,q^m-1-m\lceil(\delta-1)(1-1/q)\rceil,\ge
\delta]_q$ which contains its dual code. From~\cite[Theorem
2]{aly08a},\cite{aly08b}, then there must exist a subsystem code
with the given parameters.
\end{proof}

Instead of constructing subsystem codes from stabilizer BCH codes as
shown in Lemmas~\ref{lem:BCHExistFq}, \ref{lem:BCHExistFq2},  we can
also construct subsystem codes from classical BCH code over $\F_q$
and $\F_{q^2}$ under some restrictions on the designed distance. Let
$C_i$ be a cyclotomic coset defined as $\{ iq^j \mod n \mid j \in Z
\}$.

%
\begin{lemma}\label{lem:subsysBCHq}
If $q$ is a power of a prime, $m$ is a positive integer, and $2\leq
\delta\leq q^{\lceil m/2\rceil}-1 -(q-2)[m \text{ odd }]$.  Let $D$
be a BCH code with length $n=q^m-1$ and defining set
$T_D=\{C_0,C_1,\ldots, C_{n-\delta} \}$, such that $\gcd(n,q)=1$.
Let $T \subseteq \{ 0\}\cup\{C_\delta, \ldots, C_{n-\delta} \}$ be a
nonempty set. Assume $C \subseteq \F_q^n$ be a BCH code with the
defining set $T_C=\{C_0,C_1,\ldots, C_{n-\delta} \}\setminus (T\cup
T^{-1})$ where $T^{-1} =\{ -t \bmod n\mid t\in T\}$. Then there
exists a subsystem BCH code with the parameters $[[n,n-2k-r,r,\geq
\delta]]_q$, where $k=m\lceil(\delta-1)(1-1/q) \rceil$ and $r=|T\cup
T^{-1}|$.
\end{lemma}
\begin{proof}
The proof can be divide into the following parts:

\begin{enumerate}[i)]
\item  We know that $T_D=\{C_0,C_1,\ldots, C_{n-\delta} \}$ and $T \subseteq \{ 0\}\cup\{C_\delta, \ldots, C_{n-\delta} \}$ be a
nonempty set. Hence $T_D^\perp=\{C_1,\ldots, C_{\delta-1} \}$.
Furthermore, if $2\leq \delta\leq q^{\lceil m/2\rceil}-1 -(q-2)[m
\text{ odd }]$, then $D \subseteq D^\perp$. Furthermore, let
$k=m\lceil(\delta-1)(1-1/q) \rceil$, then $\dim D^\perp=n-k$ and
$\dim D= k$.

\item  We know that $C \in \F_q^n$ is a BCH code with defining set
$T_C=T_D \setminus (T\cup T^{-1})=\{C_0,C_1,\ldots, C_{n-\delta}
\}\setminus (T\cup T^{-1})$ where $T^{-1} =\{ -t \bmod n\mid t\in
T\}$. Then the dual code $C^\perp$ has defining set
$T_C^\perp=\{C_1,\ldots, C_{\delta-1} \}\cup T\cup
T^{-1}=T_{D^\perp}\cup T\cup T^{-1}$. We can compute the union set
$T_D$  as $T_C \cup T_C^\perp=\{C_0,C_1,\ldots, C_{n-\delta}
\}=T_D$. By Lemma~\ref{lem:definingsets}, therefore, $C \cap
C^\perp=D$. Furthermore, if $r=|T\cup T^{-1}|$, then $\dim C= k+r$.

\item
From step (i) and (ii), and for $0\leq r <n-2k$, and by
Lemma~\ref{lem:css-Euclidean-subsys},  there exits a subsystem code
with parameters $[[n,\dim D-\dim C,\dim C-\dim D,
d]]_q=[[n,n-2k-r,r,d]]_q$, $d=\min wt(D^\perp - C)\geq \delta$.
\end{enumerate}
\end{proof}
 Also, we can derive subsystem BCH codes from classical BCH codes
over $\F_{q^2}$ as shown in the following Lemma,
see~\cite{aly06a,aly07a,aly08b}.

\begin{lemma}\label{lem:subsysBCHq2}
If $q$ is a power of a prime, $n,m$ are  positive integers, and
$\gcd(n,q)=1$. Let $n=(q^2)^m-1$, $2\leq \delta\leq q^{ m}-1
-(q-2)[m \text{ odd }]$ and $T \subseteq \{ 0\}\cup\{C_\delta,
\ldots, C_{n-\delta} \}$. Let $C \subseteq \F_{q^2}^n$ be a cyclic
code with the defining set $T_C=\{C_0,C_1,\ldots, C_{n-\delta}
\}\setminus (T\cup T^{-q})$ where $T^{-q} =\{ -qt \bmod n\mid t\in
T\}$. Then there exists a cyclic subsystem code with the parameters
$[[n,n-2k-r,r,\geq \delta]]_q$, where $k=m\lceil(\delta-1)(1-1/q^2)
\rceil$ and $0\leq r=|T\cup T^{-q}| <n-2k$.
\end{lemma}
\begin{proof}

The proof is very similar to the proof shown in
Lemma~\ref{lem:subsysBCHq} taking in consideration that the
classical BCH codes are over $\F_{q^2}$.
\begin{enumerate}[i)]
\item
We know that the BCH code contains its Hermitian dual code if $2\leq
\delta\leq q^{ m}-1 -(q-2)[m \text{ odd }]$. Let $n=(q^2)^m-1$ and
$D^\hdual \subseteq \F_{q^2}^n$ be a BCH code defined with a
designed distance $\delta$. The dual code $D^{\hdual}$ has defining
set $T_{D^\hdual}=\{C_1,\ldots, C_{\delta-1} \}$. Consequently, the
code $D$ has defining set $\{C_0,C_1,\ldots, C_{n-\delta} \}$ and it
is self-orthogonal, i.e., $D\subseteq D^\hdual$. Furthermore, if
$k=m\lceil(\delta-1)(1-1/q^2) \rceil$, then $\dim D^\hdual=n-k$ and
$\dim =k$.

\item  We know that $C \subseteq \F_{q^2}^n$ is a BCH code with defining set $T_C=\{C_0,C_1,\ldots,
C_{n-\delta} \}\setminus (T\cup T^{-q})$ where $T^{-q} =\{ -qt \bmod
n\mid t\in T\}$. Then the dual code $C^\hdual$ has defining set
$T_{C^\hdual}=\{C_1,\ldots, C_{\delta-1} \}\cup T\cup T^{-q}$. We
can compute the union set $T_D$  as $T_C \cup
T_{C^\hdual}=\{C_0,C_1,\ldots, C_{n-\delta} \}$. Therefore, $C \cap
C^{\hdual}=D$. Assume $r=|T\cup T^{-q}|$, then $\dim C = k+r$

\item
From step (i) and (ii), and by Lemma~\ref{lem:css-Hermitina-subsys}
for $0\leq r <n-2k$, there exits a subsystem code with parameters
$[[n,n-2k-r,r, d]]_q$, where $k=m\lceil(\delta-1)(1-1/q^2) \rceil$
and $0\leq r=|T\cup T^{-q}|<n-2k$, $d =\min wt(D^\perp - C) \geq
\delta$.
\end{enumerate}
\end{proof}
Tables~\ref{table:bchtable} and~\ref{table:bchtableII}show some
families of subsystem BCH codes derived from classical BCH codes.
The subsystem code $[[21,18,1,2]]_2$ constructed using BCH codes,
but the stabilizer code $[[21,19,2]]_2$ does not exist using the
linear programming bound~\cite{calderbank98}.

\begin{table}[ht]
\caption{Subsystem BCH codes that are derived using the Euclidean
construction} \label{table:bchtable}
\begin{center}
\begin{tabular}{|l|l|c|}
\hline
\text{Subsystem Code} &  \text{Parent BCH} & \text {Designed}  \\
 &  \text{Code} $C$ & \text{distance }  \\
 \hline
 &&\\
 $[[15 ,4 ,3 ,3 ]]_2$   &$[15 ,7 ,5 ]_2$  & 4\\
 $ [[15 ,6 ,1 ,3 ]]_2  $ &$[15 ,5 ,7 ]_2 $ & 6\\
  $ [[31 ,10,1 ,5 ]]_2 $  &$[31 ,11,11]_2 $ & 8\\
  $  [[31 ,20,1 ,3 ]]_2  $ &$[31 ,6 ,15]_2 $ & 12\\
   $  [[63 ,6 ,21,7 ]]_2 $  &$[63 ,39,9 ]_2 $ & 8\\
$ [[63 ,6 ,15,7 ]]_2 $  &$[63 ,36,11]_2$  & 10\\
 $ [[63 ,6 ,3 ,7 ]]_2 $  &$[63 ,30,13]_2$  & 12\\
$ [[63 ,18,3 ,7 ]]_2$   &$[63 ,24,15]_2$  & 14\\
$  [[63 ,30,3 ,5 ]]_2$   &$[63 ,18,21]_2 $ & 16\\
  $ [[63 ,32,1 ,5 ]]_2 $  &$[63 ,16,23]_2 $ & 22\\
  $  [[63 ,44,1 ,3 ]]_2  $ &$[63 ,10,27]_2 $ & 24\\
  $  [[63 ,50,1 ,3 ]]_2  $ &$[63 ,7 ,31]_2  $& 28\\
  \hline
  &&\\
  $[[15 ,2 ,5 ,3 ]]_4$   &$[15 ,9 ,5 ]_4$  & 4\\
   $[[15 ,2 ,3 ,3 ]]_4  $ &$[15 ,8 ,6 ]_4$  & 6\\
    $[[15 ,4 ,1 ,3 ]]_4  $ &$[15 ,6 ,7 ]_4$  & 7\\
     $[[15 ,8 ,1 ,3 ]]_4  $ &$[15 ,4 ,10]_4$  & 8\\

$[[31 ,10,1 ,5 ]]_4  $ &$[31 ,11,11]_4$  & 8\\
 $[[31 ,20,1 ,3 ]]_4  $ &$[31 ,6 ,15]_4$  & 12\\
  $[[63 ,12,9 ,7 ]]_4  $ &$[63 ,30,15]_4$  & 15\\
   $[[63 ,18,9 ,7 ]]_4  $ &$[63 ,27,21]_4$  & 16\\
    $[[63 ,18,7 ,7 ]]_4  $ &$[63 ,26,22]_4$  & 22\\

 \hline
\end{tabular}
\\$*$ punctured code\\
$+$ Extended code
\end{center}
\end{table}

\begin{table}[ht]
\caption{Subsystem BCH codes that are derived with the help of the
Hermitian construction} \label{table:bchtableII}
\begin{center}
\begin{tabular}{|c|c|c|}
\hline
\text{Subsystem Code} &  \text{Parent BCH} & \text {Designed}  \\
 &  \text{Code $C$} & \text{ distance }  \\
 \hline
 $[[14 ,1 ,3 ,4 ]]_2$ &$[14 ,8 ,5 ]_{2^2}$&$6^*$ \\
$[[15,1,2,5]]_2$ & $[15,8,6]_{2^2}$ &6 \\
{} $[[15,5,2,3]]_2$&$[15,6,7]_{2^2}$&7\\
$[[16 ,5 ,2 ,3 ]]_2$ &$ [16 ,6 ,7 ]_{2^2}$&$7^+$ \\
 \hline
$[[17,8,1,4]]_2 $&$ [17,5,9]_{2^2}$ &4 \\
\hline
 $[[21,6,3,3]]_2$&$ [21,9,7]]_{2^2}$&6\\{}
 $[[21 ,7 ,2 ,3 ]]_2$& $ [21 ,8 ,9 ]_{2^2}$&8\\
\hline $[[31,10,1,5]]_2$&$[31,11,11]_{2^2} $&8\\{}
$[[31 ,20,1 ,3 ]]_2$&$ [31 ,6 ,15]_{2^2}$&12\\
$[[32 ,10,1 ,5 ]]_2$ &$ [32 ,11,11]_{2^2}$&$8^+$\\
$[[32 ,20,1 ,3 ]]_2$ &$[32 ,6 ,15]_{2^2}$&$12^+$\\
 \hline $[[25 ,12,3 ,3 ]]_3$ & $[25 ,8
,12]_{3^2}$&$9^*$
\\
$[[26 ,6 ,2 ,5 ]]_3$&$[26 ,11,8 ]_{3^2}$&8\\
$[[26 ,12,2 ,4 ]]_3 $&$[26 ,8 ,13]_{3^2} $&9\\
$[[26 ,13,1 ,4 ]]_3$&$[26 ,7 ,14]_{3^2}$&14\\
\hline
$[[80 ,1 ,17,20]]_3$&$[80 ,48,21]_{3^2} $&21\\
$[[80 ,5 ,17,17]]_3 $& $[80 ,46,22]_{3^2}$&22\\
 \hline
\end{tabular}
\\$*$ punctured code\\
$+$ Extended code
\end{center}
\end{table}

It may be useful to end up this section with an example
\begin{example}
Consider a BCH code $D^\perp$ with designed distance $d=5$ and
length $n=2^5-1$ over $\F_4$. Then $C_1=\{1,2,4,8,16\}$,
$C_2=\{3,6,12,24,17\}$, and $C_5=\{5,10,20,9,18\}$. Then
$T_{D^\hdual}=C_1\cup C_3$. Hence $\dim D=10$ and $\dim
D^\hdual=21$. Now, let $T=C_5$, so,
$T^{-q}=C_{11}=\{11,13,21,22,26\}$ and $T_{C^\hdual}=T_{D^\hdual}
\cup T \cup T^{-q}$. We have $|T_{C^\hdual}=20|$, therefore $\dim
C=20$. Conseqeuntly, there exists a subsystem BCH codes with
parameters $[[n,\dim D^\hdual-\dim C,\dim C -\dim D, \geq
\delta]]_q=[[31,1,10,\geq 5]]_2$. Some subsystem BCH codes are shown
in Tables~\ref{table:bchtable} and \ref{table:bchtableII}.
\end{example}

\section{Subsystem RS Codes}\label{Sec:RSSubsys}

In this section we will derive cyclic subsystem codes based on
Reed-Solomon codes. Also, we show that given optimal stabilizer
codes, one can construct optimal subsystem  codes. Recall that a
Reed-Solomon code over $\F_q$ is a BCH code with length $n=q-1$ and
minimum distance equals to its designed distance $\delta$.
Therefore, the RS code $C$ with designed distance $\delta$ has
defining set $T$ with size $\delta-1$. This can be seen as all roots
lie in different cyclotomic cosets. The dimension of a RS code is
given by $n-\delta+1$. RS codes are an important class of optimal
cyclic codes. They are  MDS codes, in  which  Singleton bound is
satisfied with equality.

Grassl \emph{et al.} in \cite{grassl04} showed that optimal
stabilizer codes with maximal minimum distance exist with parameters
$[[n,n-2d+2,d]]_q$ over $\F_q$ for  $3 \leq n \leq q$ and $1 \leq d
\leq n/2+1$. Also, optimal stabilizer codes exist with parameters
$[[q^2,q^2-2d+2,d]]_q$ for $1 \leq d \leq q$ over $\F_q$,
see~\cite[Theorems 9, 10]{grassl04}. These codes satisfy the quantum
Singleton bound $k+2d=n+2$.
 The following subsystem codes are optimal since they obey the singleton bound $k+r+2d=n+2$ as shown
 in~\cite[Theorem 21]{aly06c}.

\begin{lemma}[Reed-Solomon Subsystem codes]\label{lem:rsExist}
Let q be power of a prime.
\begin{enumerate}[i)]
\item If $0 \leq \delta < (q -1)/2$ there
exist subsystem codes with parameters $[[q - 1, q - 2\delta - 1-r,
r,\delta + 1]]_q$ and $[[q, q - 2\delta - 2-r,r, \delta + 2]]_q$.

\item If $0 \leq \delta < q-1$ there exist subsystem codes with parameters
  $[[q^2 - 1, q^2 - 2\delta - 1-r,r, \delta +
1]]_q$ and $[[q^2, q^2 - 2\delta - 2-r,r, \delta + 2]]_q$
\end{enumerate}
\end{lemma}
\begin{proof}
\begin{enumerate}[i)]
\item
We know that if $0 \leq \delta < (q -1)/2$, then there are
stabilizer codes with parameters $[[q - 1, q - 2\delta - 1,\delta +
1]]_q$ and $[[q, q - 2\delta - 2, \delta + 2]]_q$, see~\cite[Theorem
9]{grassl04}. Now, let $ 0 \leq r < q - 2\delta - 1$, then
using~\cite[Corollary 6]{aly08a}, there are subsystem codes with
parameters $[[q - 1, q - 2\delta - 1-r, r,\delta + 1]]_q$ and $[[q,
q - 2\delta - 2-r,r, \delta + 2]]_q$.

\item Similarly, if $0 \leq \delta < q-1$, then from \cite[Theorem 10]{grassl04}, there
exist  stabilizer codes with parameters $[[q^2 - 1, q^2 - 2\delta -
1, \delta + 1]]_q$ and $[[q^2, q^2 - 2\delta - 2-r,r, \delta +
2]]_q$. Assuming $ 0 \leq r < q^2 - 2\delta - 1$, then
from~\cite[Corollary 6]{aly08a}, there exist subsystem codes with
parameters $[[q^2 - 1, q^2 - 2\delta - 1-r,r, \delta + 1]]_q$ and
$[[q^2, q^2 - 2\delta - 2-r,r, \delta + 2]]_q$.
\end{enumerate}

\end{proof}
Instead of extending the subsystem code that we constructed, one can
start with a subsystem code with length $n=q$ and shorten it to a
subsystem code with length $n=q-1$.
These subsystem codes are all $\F_{q^2}$-linear. Therefore they
satisfy $k+r = n-2d+2$. As a consequence the subsystem codes in
Lemma~\ref{lem:rsExist} are optimal. The subsystem codes that we
derive are not necessarily cyclic. In order to derive cyclic codes
we need to make further restrictions on the codes. The following
lemma gives an explicit construction for cyclic subsystem codes
based on the Reed-Solomon codes over $\F_q$.

\begin{lemma}\label{lem:subsysRSq}
Let $q$ be a prime power, and $n=q-1$, $2\leq \delta < (q-1)/2$ and
$T\subseteq \{ 0\}\cup\{\delta, \ldots, n-\delta \}$. Let $C
\subseteq \F_q^n$ be a cyclic code with the defining set
$T_C=\{0,1,\ldots, n-\delta \}\setminus (T\cup T^{-1})$ where
$T^{-1} =\{ -t \bmod n\mid t\in T\}$. Then there exists a cyclic
subsystem RS code with the parameters $[[n,n-2\delta+2-r,r,\geq
\delta]]_q$, where $0 \leq r=|T\cup T^{-1}| < n-2(\delta+1)$.
\end{lemma}
\begin{proof}
We divide the proof to the following parts
\begin{enumerate}[i)]
\item We know that if $2\leq \delta < (q-1)/2$, then there  exists  classical cyclic code $D^\perp$
that contains its dual code $D$, i.e., $D \subseteq D^\perp$. The
code $D^\perp$ has defining set $T_{D^\perp}=\{1,2,...,\delta-1\}$.
Therefore  the defining set of $D$ is given by $T_D=\{0\} \cup
\{1,\cdots,n-\delta\}$ and $D =C\cap C^\perp$. Also, $\dim
D^\perp=n-(\delta-1)$ and $\dim D=\delta-1$.

\item Let $T \subseteq T_D$ be a nonempty set and $T^{-1} =\{ -t \bmod n\mid t\in T\}$. Let $C \subseteq \F_q^n$  be a cyclic code with the defining set
$T_C=T_D\setminus (T\cup T^{-1})$. We can actually compute the
defining set of the dual code $C^\perp$ as $T_{C^\perp}=T_{D^\perp}
\cup T \cup T^{-1}$. We notice that $T_C \cup T_{C^\perp}=
\{1,2,\cdots,n-\delta \}\cup \{0\}=T_D$. Let $k=\delta-1$ and $0\leq
r=|T\cup T^{-1}|< n-2k$.

\item From steps (i), (ii) and by using Lemma~\ref{lem:css-Euclidean-subsys}, there is a subsystem  code with
$[[n,k,r,\geq \delta]]_q$, where $k=n-2\delta+2-r$ and $0\leq
r=|T\cup T^{-1}|<n-2(\delta-1)$.
\end{enumerate}
\end{proof}
Also, cyclic subsystem codes, based on RS codes over $\F_{q^2}$, can
be derived  as shown in the following lemma. Some codes are shown in
Table~\ref{table:opttable}.
\begin{lemma}\label{lem:subsysRSq2}
Let $q$ be a prime power,  $n=q^2-1$, and $2\leq \delta < (q-1)$.
Let $T\subseteq \{ 0\}\cup\{q\delta, \ldots, q(n-\delta) \}$ be a
nonempty set. Let $C \subseteq \F_{q^2}^n$ be a cyclic code with the
defining set $T_C=\{0,q,\ldots, q(n-\delta) \}\setminus (T\cup
T^{-q})$ where $T^{-q} =\{ -qt \bmod n\mid t\in T\}$. Then there
exists a cyclic subsystem RS code with the parameters
$[[n,n-2(\delta-1)-r,r, \geq \delta]]_q$, where $0\leq r=|T\cup
T^{-q}|< n-2(\delta-1)$.
\end{lemma}
\begin{proof} The proof is a direct consequence as shown in the
previous lemmas.

We know that if $2\leq \delta < (q-1)$, then there exists a cyclic
code $D^\perp$ over $\F_{q^2}$ that contains it is dual code $D$.
The code $D^\hdual$ has length $n$, and minimum distance $\delta$.
The defining set of the code $D$ is given by $
T_D=\{q,2q,\cdots,q(n-\delta) \}\cup \{0\}$

 We just notice that the
defining set of the dual code $C^\hdual$ is given by
$T_{C^\hdual}=\{q,2q,...,q(\delta-1) \} \cup T \cup T^{-q}$.
Furthermore, $T_C \cup T_{C^\hdual}= \{q,2q,\cdots,q(n-\delta)
\}\cup \{0\}=T_D$.   Hence,  $D \subseteq C$, $D\subseteq C^\hdual$,
and $D =C \cap C^\hdual$. From Lemma~\ref{lem:css-Hermitina-subsys},
there must exist a cyclic subsystem RS code with parameters
$[[n,k,r,\geq \delta]]_q$, where $k=n-2(\delta-1)-r$ and $0\leq
r=|T\cup T^{-q}|< n-2(\delta+1)$.
\end{proof}

\begin{table}[ht]
\caption{Optimal pure subsystem codes} \label{table:opttable}
\begin{center}
\begin{tabular}{|c|c|c|}
\hline
\text{Subsystem Codes} &  \text{Parent}  \\
 &  \text{Code (RS Code)}  \\
\hline $[[8  ,1  ,5  ,2  ]]_3$ & $[8  ,6  ,3  ]_{3^2}$ \\{}

$[[8  ,4  ,2  ,2  ]]_3$&$[8  ,3  ,6  ]_{3^2}$\\{} $[[8  ,5  ,1  ,2
]]_3$&$[8 ,2 ,7 ]_{3^2}$\\{}

$[[9  ,1 ,4 ,3 ]]_3$&$[9  ,6 ,4 ]_{3^2}^{\dag}, \delta=3$\\
$[[9  ,4 ,1 ,3 ]]_3$&$[9  ,3 ,7 ]_{3^2}^{\dag}, \delta=6$ \\
 \hline $[[15,1,10,3]]_4$ &
$[15 ,12 ,4 ]_{4^2}$  \\{}
$[[15,9,2,3]]_4$&$[15,4,12]_{4^2}$\\{} $[[15,10,1,3]]_4$&$[15,3,13]_{4^2}$\\
$[[16 ,1 ,9 ,4 ]]_4$& $[16 ,12,5 ]_{4^2}^{\dag},\delta= 4 $\\

  \hline $[[24,1,17,4]]_5$
&$[24,20,5]_{5^2}$
\\{}
$ [[24,16,2,4]]_5$ &$[24,5,20]_{5^2}$\\{} $[[24,17 ,1,4 ]]_5
$&$[24,4,21]_{5^2}$\\{}
 $[[24,19,1,3]]_5$ &$[24,3,22]_{5^2}$\\{}
$[[24 ,21 ,1  ,2  ]]_5$ & $ [24 ,2  ,23 ]_{5^2}$ \\{} $[[23 ,1 ,18,3
]]_5$&$[23 ,20,4 ]_{5^2}^{*}, \delta=5$\\{}
$[[23 ,16,3 ,3 ]]_5$&$[23 ,5 ,19]_{5^2}^{*}, \delta=20$\\
 \hline
 $[[48 ,1  ,37 ,6  ]]_7$  &$[48 ,42 ,7 ]_{7^2}$\\
 \hline
\end{tabular}
\\
* Punctured code\\
$\dag$ Extended code
\end{center}
 \end{table}

In table~\ref{table:opttable} we show various optimal subsystem
codes derived from RS codes. Some of these codes have been derived
by puncture existing subsystem codes. It is also possible to derive
some optimal impure subsystem codes. For instance $[[9,1,4,3]]_2$ is
an optimal impure subsystem codes.
\paragraph{Puncture Subsystem Codes}
The MDS subsystem codes  constructed from RS codes can also be
punctured to other subsystem codes. Recall that if there is a
subsystem code with parameters $[[n,k,r,d]]_q$ then there is a
subsystem code with parameters $[[n-1,k,r,\geq d-1]_q$. This is
known as the propagation rules of quantum code constructions.

We end up this section by presenting two examples to illustrate the
previous construction.
\begin{example}
Let $C$ be a RS code with length $n=q-1=6$ over $\F_q$. Define
$N=\{0,1,2,3,4,5\}$. We can construct subsystem code from RS codes
with parameters $[6,4,3]_7$. This code is a subcode-subfield in BCH
codes with deigned distance $\delta=3$. So, $T_{D^\perp}=\{1,2\}$, $
T_D=\{0,1,2,3\}$ , $T_C=\{1,2,3\}$ and $T_{C^\perp}=\{0,1,2\}$.
 We notice  that $T_D=T_C \cup T_{C^\perp}$ and $\dim C=3$, $\dim D=2$ and $\dim D^\perp=4$.
So, we have k=4-3=1 and r=3-2=1. Consequently, there exists a
subsystem code with parameters $[6,1,1,3]$ over $\F_7$
\end{example}
The previous example shows the shortest subsystem codes with length
$n=6$. However, it is not necessarily that this code exists only
over $\F_7$. In fact, as we were able to show that there exists a
subsystem code with length $n=6$ over $\F_3$.
\begin{example}
Let $F_{13}$ be the finite field with $q=13$ elements. Let $D^\perp$
be the narrow-sense Reed-Solomon code of length $n=12$ and designed
distance $\delta=5$ over $F_{13}$. So, $D^\perp$ has defining set
$T_{D^\perp}=\{1,2,3,4 \}$. Therefore, $D^\perp$  is  an MDS code
with parameters $[12,8,5]$. The dual of $D^\perp$ is a RS code $D$
with defining set $T_D=\{ 0,1,2,3,4,5,6,7 \}$. Also, $D$ is an MDS
code with parameters $[12,4,9]$. Clearly, from our construction, $$D
\subseteq D^\perp \Longleftrightarrow T_{D^\perp} \subseteq T_D$$
Now, let us define the code $C$ by choosing a defining set
$T_C=\{1,2,3,4,7 \}$. So, $D \subseteq C \Longleftrightarrow T_C
\subseteq T_D$. Also compute the defining set of $C^\perp$ as
$T_{C^\perp}=\{ 0,1,2,3,4,6,7 \}$. So, $D \subseteq C^\perp
\Longleftrightarrow T_{C^\perp} \subseteq T_D$. We see from our
construction of these codes that $$C \cap C^\perp=D
\Longleftrightarrow T_C \cup T_{C^\perp}=T_D.$$ Hence, we can
compute the parameters of the subsystem code as follows. The minimum
distance is given by   $d_{min}=D^\perp \backslash C =5$, dimension
$k=\dim D^\perp-\dim C=8-7=1$, and gauge qubits $r=\dim C-\dim
D=7-4=3$. Therefore, we have a subsystem code with parameters
$[[12,1,3, 5]]$, which is also an MDS code obeying Singleton bound
$k+r+2d=n+2$.
\end{example}
Actually, if we choose the defining set of $C$ to be
$T_C=\{1,2,3,4,6,7 \}$, then the defining set of $C^\perp$ is
$T_{C^\perp}=\{ 0,1,2,3,4,7 \}$, then we get a subsystem code with
parameters  $d_{min}=D^\perp \backslash C =5$,
  $k=\dim D^\perp-\dim C=8-6=2$,
   $r=\dim C-\dim D=6-4=2$. Therefore, we have a subsystem code with parameters $[[12,2,2,5]]$,
which is also an MDS code. Some of subsystem RS codes are listed in
Table~\ref{table:rstable}.
\begin{table}[ht]
\caption{Reed-Solomon(RS)  subsystem codes} \label{table:rstable}
\begin{center}
\begin{tabular}{|c|c|c|}
\hline\hline
\text{Subsystem Codes} &  \text{Parent}  \\
 &  \text{RS Code}  \\
\hline $[[15,1,10,3]]_4$ & $[15 ,12 ,4 ]_{4^2}$  \\{}
 $[[15 ,1  ,8  ,3 ]]_4$ & $[15 ,11 ,5]_{4^2}$ \\{}  $[[15 ,1  ,6  ,3
]]_4$&$[15 ,10 ,6]_{4^2}$\\{} $[[15 ,2  ,5  ,3  ]] _4$&$ [15 ,9  ,7]_{4^2}$\\
\hline $[[24,1,17,4]]_5$ &$[24,20,5]_{5^2}$
\\{} $[[24,2,10,4]]_5 $&$[24,16,9]_{5^2}$\\{}
$[[24 ,4 ,10,4 ]]_5$ &$[24,15,10]_{5^2}$\\{}$ [[24,16,2,4]]_5$
&$[24,5,20]_{5^2}$\\{}
$[[24,17 ,1,4 ]]_5 $&$[24,4,21]_{5^2}$\\{} $[[24,19,1,3]]_5$ &$[24,3,22]_{5^2}$\\
 \hline
 $[[48 ,1  ,37 ,6  ]]_7$  &$[48 ,42 ,7 ]_{7^2}$\\{}
  $[[48 ,2  ,26 ,6  ]]_7 $ &$[48 ,36 ,13 ]_{7^2}$\\
 \hline
\end{tabular}
\end{center}\label{tab:RS}
\end{table}

\section{Subsystem Codes $[[8,1,2,3]]_2$ and $[[6,1,1,3]]_3$}\label{Sec:ShortSubsys}
In this section we present the generator matrices of two short subsystem codes over $\F_2$
and $\F_3$ fields.  Corollary~\ref{cor:generic} implies that a
stabilizer code with parameters $[[n,k,d]]_q$ gives subsystem codes
with parameters $[[n,k-r,r,d]]_q$, see Tables~ \ref{table:bchtable},
\ref{table:bchtableII}, \ref{table:opttable}, \ref{table:rstable},
\ref{table:families}.


  Consider a stabilizer code with parameters $[[8,3,3]]_2$.
This code can be used to derive $[[8,2,1,3]]_2$ and $[[8,1,2,3]]_2$ subsystem
codes. We give an explicit construction of these codes. We obtain these codes using MAGMA computer algebra search . It
remains to study properties of these codes and whether they have nice error
correction capabilities.
 We show the stabilizer and normalizer matrices for these codes. Also, we prove
their minimum distances using the weight enumeration of  these
codes. It was known that the $[[9,1,4,3]]_2$ Becan-Shor code is the
shortest subsystem code constructed via graphs,  in which it
tolerates 4 gauge qubits. We present two codes with less length,
however we can not tolerate more than 2 gauge qubits. The following
example shows $[[8,1,2,3]]$  subsystem code over $\F_2$.
\begin{exampleX}

\begin{eqnarray}
D_S= \left[ \begin{array} {cccccccc}
X & I & Y & I & Z & Y & X & Z \\
Y & I & Y & X & I & Z & Z & X \\
I & X & Y & Y & Z & X & Z & I \\
I & Y & I & Z & Y & X & X & Z \\
I & I & X & Z & X & Y & Z & Y \\
\end{array} \right]
\end{eqnarray}

\begin{eqnarray}
D^\perp_S= \left[ \begin{array} {cccccccc}
 X & I & I & I & I & I & Z & Y\\
 Y & I & I & I & I & Y & X & X \\
 I & X & I & I & I & Y & Y & X \\
 I & Y & I & I & I & I & X & Z \\
 I & I & X & I & I & Y & Z & I\\
 I & I & Y & I & I & I & Z & X \\
 I & I & I & X & I & Y & I & Z \\
 I & I & I & Y & I & Y & Y & Y\\
 I & I & I & I & X & I & Y & Z\\
 I & I & I & I & Y & Y & Z & Z\\
 I & I & I & I & I & Z & X & Y\\
\end{array} \right]
\end{eqnarray}

\begin{eqnarray}
C_S= \left[ \begin{array} {cccccccc}
X & I & Y & I & Z & Y & X & Z \\
Y & I & Y & X & I & Z & Z & X \\
I & X & Y & Y & Z & X & Z & I \\
I & Y & I & Z & Y & X & X & Z \\
I & I & X & Z & X & Y & Z & Y \\
\hline
 Y & I & I & I & I & Y & X & X \\
 I & X & I & I & I & Y & Y & X \\
\end{array} \right]
\end{eqnarray}

\begin{eqnarray}
C^\perp_S= \left[ \begin{array} {cccccccc}
X & I & Y & I & Z & Y & X & Z \\
Y & I & Y & X & I & Z & Z & X \\
I & X & Y & Y & Z & X & Z & I \\
I & Y & I & Z & Y & X & X & Z \\
I & I & X & Z & X & Y & Z & Y \\
\hline
X & I & I & I & I & I & Z & Y\\
 I & I & I & Y & I & Y & Y & Y\\
\end{array} \right]
\end{eqnarray}

We notice that the matrix $D_S$ generates the code $D=C \cap
C^{\perp_s}$. Furthermore, dimensions of the subsystems $A$ and $B$
are given by $k=\dim D^{\perp_s}- \dim C=(11-7)/2=2$ and $r= \dim C
- \dim D=(7-5)/2=1$. Hence we have $[[8,2,1,3]]_2$ and
$[[8,1,2,3]]_2$ subsystem codes.
\end{exampleX}

We show that the subsystem codes $[[8,1,2,3]]_2$ is not better than the
stabilizer code $[[8,3,3]]_2$ in terms of syndrome measurement. The reason is
that the former needs $8-1-2=5$ syndrome measurements, while the later needs
also $8-3=5$ measurements. This is an obvious example where subsystem codes
have no superiority in terms of syndrome measurements.

We post an open question regarding the threshold value and fault tolerant gate
operations for this code. We do not know at this time if the code
$[[8,1,2,3]]_2$ has better threshold value and less fault-tolerant operations.
Also, does the subsystem code with parameters  $[[8,1,3,3]]_2$ exist?

\textbf{No nontrivial $[[7,1,1,3]]_2$ exists.} There exists a trivial
$[[7,1,1,3]]_2$ code obtained by simply extending the $[[7,1,3]]_2$ code as the
$[[5,1,3]]_2$ code. We show the smallest subsystem code with length $7$ must
have at most  minimum weight equals to 2. Since $[[7,2,2]]_2$ exists, then we
can construct the stabilizer and normalizer matrices as follows.

\begin{eqnarray}
D_S= \left[ \begin{array} {ccccccc}
 X & X & X & X & I & I & I\\
 Y & Y & Y & Y & I & I & I \\
 I & I & I & I & X & I & I \\
 I & I & I & I & I & X & I \\
 I & I & I & I & I & I & X\\
\end{array} \right]
\end{eqnarray}

\begin{eqnarray}
D^\perp_S= \left[ \begin{array} {ccccccc}
X & I & I & X & I & I & I \\

Y & I & I & Y & I & I & I \\

I & X & I & X & I & I & I \\

I & Y & I & Y & I & I & I \\

I & I & X & X & I & I & I \\

I & I & Y & Y & I & I & I \\

I & I & I & I & X & I & I\\

I & I & I & I & I & X & I\\

I & I & I & I & I & I & X\\

\end{array} \right]
\end{eqnarray}
Clearly, from our construction and using Corollary
\ref{cor:generic}, there must exist a subsystem code with parameters
$k$ and $r$ given as follows. $\dim D^{\perp_s}=9/2$ and $\dim
C=7/2$. Also, $\dim D=5/2$ and $\ min (D^{\perp_s} \backslash C)=2$.
Therefore, , $k=(9-7)/2=1$ and $r=(7-5)/2=1$. Consequently, the
parameters of the subsystem code are $[[7,1,1,2]]_2$.

\smallskip

This example shows $[[6,1,1,3]]$  subsystem code over $\F_3$.
\begin{exampleX}
We give a nontrivial short subsystem code over $\F_3$. This is
derived from the $[[6,2,3]]_3$ graph quantum code, see~\cite{feng02}
for existence results and~\cite{grassl02} for a method to construct
the code. Also, we showed an example earlier for an $[[6,1,1,3]]$
subsystem code over $\F_7$. Consider the field $\F_3$ and let
$C\subseteq \F_3^{12}$ be a linear code defined by the following
generator matrix.
\begin{eqnarray*}
C=\left[ \begin{array}{rrrrrr|rrrrrr}
1&0&0&0&2&0&0&2&0&2&0&2\\
0&1&0&0&0&2&1&0&1&0&1&0\\
0&0&1&0&2&0&0&1&0&1&0&1\\
0&0&0&1&0&2&2&0&2&0&2&0\\
\hline
0&0&0&0&1&0&0&2&0&1&0&0\\
0&0&0&0&0&0&1&0&1&0&1&0\\
\end{array}\right]=\left[\begin{array}{c} S\\ \hline X_1\\Z_1\\
\end{array}\right].
\end{eqnarray*}
Let the symplectic inner product $\langle (a|b)|(c|d)\rangle_s = a\cdot d -
b\cdot c$. Then the symplectic dual of $C$ is generated by
\begin{eqnarray*}
C^\sdual=\left[\begin{array}{c}  S\\ \hline X_2\\Z_2\\
\end{array}\right],
\end{eqnarray*}
where $X_2=\big[ \begin{array}{rrrrrr|rrrrrr}
0&0&0&0&0&1&1&0&2&0&0&0\\
\end{array} \big]$ and\\ $Z_2=\big[ \begin{array}{rrrrrr|rrrrrr}
0&0&0&0&0&0&0&1&0&1&0&1\\
\end{array} \big]$. The matrix $S$ generates the code $D=C\cap
C^\sdual$. Now $D$ defines a $[[6,2,3]]_3$ stabilizer code
\cite[Theorem~3.1]{feng02} and \cite[Theorem~1 and
Equation~(15)]{grassl02}. Therefore, $\swt(D^\sdual\setminus D)=3$.
It follows that $\swt(D^\sdual\setminus C) \geq \swt(D^\sdual)=3$.
By \cite[Theorem~4]{aly06c}, we have a $[[6,(\dim D^\sdual-\dim
C)/2,(\dim C-\dim D)/2, 3]]_3$ viz. a $[[6,1,1,3]]_3$ subsystem
code.
\end{exampleX}

We can also have a trivial $[[6,1,1,3]]_2$ code.  This trivial extension seems
to argue against the usefulness of subsystem codes and if they will really lead
to improvement in performance. An obvious open question is if there exist
nontrivial $[[6,1,1,3]]_2$ or $[[7,1,1,3]]_2$ subsystem codes.

\section{Conclusion and Discussion} We constructed cyclic subsystem codes by using the defining sets of classical cyclic codes over $\F_q$ and $\F_{q^2}$.
 Also, we presented  a simple method to obtain subsystem
codes from stabilizer codes and derived optimal subsystem codes from
RS codes. In addition, we drove families of subsystem BCH and RS
codes. We introduced the short subsystem codes over binary and
ternary fields. We leave it as open questions to realize performance
and usefulness of these codes. Also, we pose the construction of a
nontrivial $[[6,1,1,3]]_2$ code and compare its performance with the
$[[5,1,3]]_2$ code as an open problem.

One can derive many other families of subsystem codes using the
Euclidean and Hermitian construction of subsystem codes. In
addition, one can design the encoding and decoding circuits of
cyclic subsystem codes.


\begin{table}[htb]
\caption{Families of subsystem codes  from stabilizer
codes}\label{table:families} \small
\begin{center}
\begin{tabular}{|@{}c@{}||c|@{}c@{}|@{}c@{}|}
\hline
Family & Stabilizer $[[n,k,d]]_q$ & Subsystem  $[[n,k-r,r,d]]_q$,  \\
&&$k>r\geq0$\\ \hline\hline
Short MDS & $[[n,n-2d+2,d]]_q$  & $[[n,n-2d+2-r,r,d]]_q$ \\
\hline
Hermitian  & $[[n,n-2m,3]]_q$ & $m\ge 2$,  $[[n,n-2m-r,r,3]]_q$ \\

 Hamming &  &  \\
\hline
 Euclidean  & $[[n,n-2m,3]]_q$ & $[[n,n-2m-r,r,3]]_q$ \\

 Hamming &  &  \\
\hline Melas  &$[[n,n-2m,\ge 3]]_q$  & $[[n,n-2m-r,r,\ge
3]]_q$\\
\hline Euclidean  & $[[n,n-2m\lceil(\delta-1)(1-1/q)\rceil,\geq \delta]]_q$& $[[n,n-2m\lceil(\delta-1)(1-1/q)\rceil-r,$\\
 BCH &  & $r,\geq \delta]]_q$ \\
\hline Hermitian  & $[[n,n-2m\lceil(\delta-1)(1-1/q^2)\rceil,\geq
\delta]]_q$ &
 $[[n,n-2m\lceil(\delta-1)(1-1/q^2)\rceil-r,$
  \\
 BCH &  & $r,\geq \delta]]_q$ \\
\hline Punctured  & $[[q^2-q\alpha, q^2-q\alpha-2\nu-2,\nu+2]]_q$  & $[[q^2-q\alpha, q^2-q\alpha-2\nu-2-r,$\\
 MDS &  &  $r,\nu+2]]_q$\\
\hline Euclidean  & $[[n,n-2d+2]]_q$& $[[n,n-2d+2-r,r]]_q$\\
 MDS &  &  \\
\hline Hermitian  & $[[q^2-s,q^2-s-2d+2,d]]_q$& $[[q^2-s,q^2-s-2d+2-r,r,d]]_q$\\
 MDS &  &  \\
\hline
Twisted & $[[q^r,q^r-r-2,3]]_q$ &$[[q^r,q^r-r-2-r,r,3]]_q$\\
\hline Extended  & $[[q^2+1,q^2-3,3]]_q$& $[[q^2+1,q^2-3-r,r,3]]_q$\\
 twisted &  &  \\
\hline Perfect & $[[n,n-s-2,3]]_q$& $[[n,n-s-2-r,r,3]]_q$
\\
&$[[n,n-s-2,3]]_q$& $[[n,n-s-2-r,r,3]]_q$\\ \hline
\end{tabular}
\end{center}
\end{table}

\chapter{Propagation Rules and Tables of Subsystem Code Constructions}\label{ch_subsys_rules_tables}

In this chapter I present tables of upper and lower bounds on
subsystem code parameters. I  derive new subsystem codes from
existing ones by extending and shortening  the length of the codes.
Also, I  trade the dimension of subsystem $A$ and co-subsystem $B$
to obtain new subsystem codes from known codes with the same
lengths.

\section{Introduction}
We investigate subsystem codes and study their properties. Given a
subsystem code with parameters $[[n,k,r,d]]_q$, we establish
propagation rules to derive new subsystem codes with possibly
parameters $[[n+1,k,r,\geq d]]_q$, $[[n-1,k-1,\geq r,d]]_q$, etc. We
construct tables of the upper bounds on the minimum distance and
dimension of subsystem codes  using linear programming  bounds over
$\F_2$ and $\F_3$. Also, we construct tables of lower bounds on
subsystem code parameters using Gilbert-Varshamov (GV) bound. We
show that our method gives all codes over $\F_2$ for small code
length and one can generate more tables over higher fields with
large alphabets. Our results provide us with better understanding of
subsystem codes in terms of comparing these codes with stabilizer
codes. Subsystem codes need $n-k-r$ syndrome measurements in
comparison to stabilizer codes that need $n-k$ syndrome
measurements. We show that some impure subsystem codes do not give
raise to stabilizer codes. Also, such codes do not obey the quantum
Hamming bound.

{\em Notation:} We assume that $q$ is a power of  prime $p$ and
$\F_q$ denotes a finite field with $q$ elements. By qudit we mean a
$q$-ary quantum bit.  The symplectic weight of an element
$w=(x_1,\ldots,x_n,y_1,\ldots,y_n)$ in $\F_q^{2n}$ is defined as
$\swt(w)=|\{(x_i,y_i)\neq (0,0)\mid 1\leq i\leq n \} |$. The
trace-symplectic product of two elements $u=(a|b),v=(a'|b')$ in
$\F_q^{2n}$ is defined as $\langle u|v \rangle_s = \tr_{q/p}(a'\cdot
b-a\cdot b')$, where $x\cdot y$ is the usual Euclidean inner
product. The trace-symplectic dual of a code $C\subseteq \F_q^{2n}$
is defined as $C^\sdual=\{ v\in \F_q^{2n}\mid \langle v|w \rangle_s
=0 \mbox{ for all } w\in C\}$. For vectors $x,y$ in $\F_{q^2}^n$, we
define the Hermitian inner product $\langle x|y\rangle_h
=\sum_{i=1}^nx_i^qy_i$ and the Hermitian dual of $C\subseteq
\F_{q^2}^n$ as $C^\hdual= \{x\in \F_{q^2}^n\mid \langle x|y
\rangle_h=0 \mbox{ for all } y\in C \}$.  The trace alternating form
of two vectors $u,w$ in $\F_{q^2}^n$ is defined as $\langle
u|v\rangle_a=\tr_{q/p}[(\langle u|v\rangle_h - \langle
v|u\rangle_h)/(\beta^{2}-\beta^{2q})]$, where $\{\beta,\beta^q\}$ is
a normal basis of $\F_{q^2}$ over $\F_q$. If $C\subseteq
\F_{q^2}^n$, then the trace alternating dual of $C$ is defined as
$C^\adual =\{x\in \F_{q^2}^n\mid \langle x|y\rangle_a =0\mbox{ for
all } y\in C\}$.

\section{Upper and Lower Bounds on Subsystem Code  Parameters} \label{sec:uBounds}
We want to investigate some limitations on subsystem codes that are
constructed in the previous chapters. Bounds on code parameters are
useful for many reasons such as the computer search can be
minimized. To that end, we will investigate some upper and lower
bounds on the parameters of subsystem codes.

\medskip

\noindent \textbf{Linear Programming Bounds.} We will show the
linear programming bound as an upper bound on subsystem code
parameters. We ensure that one can not hope to obtain subsystem
codes unless they obey this bound. This also means that if a
subsystem code  obeys this bound, it is not guaranteed that the code
itself will exist unless it can be constructed. Assume we have the
same notation as above.
\begin{theorem}\label{th:lp}
If an $((n,K,R,d))_q$ Clifford subsystem code with $K>1$ exists,
then there exists a solution to the optimization problem: maximize
$\sum_{j=1}^{d-1} A_j$ subject to the constraints
\begin{enumerate}
\item $A_0=B_0=1$ and $0\le B_j \le A_j$ for all $1\le j\le n$;
\item  $\ds\sum_{j=0}^n A_j = q^{n}R/K$; \quad $\ds\sum_{j=0}^n B_j = q^{n}/KR$;
\item $A_j^\sdual = \ds\frac{K}{q^{n}R} \sum_{r=0}^n K_j(r)A_r$ holds for all $j$ in the
range $0\le j \le n $;
\item $B_j^\sdual = \ds\frac{KR}{q^{n}} \sum_{r=0}^n K_j(r)B_r$ holds for all $j$ in the
range $0\le j \le n $;
\item $A_j=B_j^\sdual$ for all $j$ in $0\le j<d$ and $A_j\le B_j^\sdual$ for all $d\le j\le n$;
\item $B_j=A_j^\sdual$ for all $j$ in $0\le j<d$ and $B_j\le A_j^\sdual$ for all $d\le j\le n$;
\item $(p-1)$ divides $A_j$, $B_j$, $A_j^\sdual$, and $B_j^\sdual$ for
all $j$ in the range $1\le j\le n$;
\end{enumerate}
where the coefficients $A_j$ and $B_j$ assume only integer values,
and $K_j(r)$ denotes the Krawtchouk polynomial \begin{eqnarray}
K_j(r) = \sum_{s=0}^j (-1)^s
(q^2-1)^{j-s}\binom{r}{s}\binom{n-r}{j-s}.\end{eqnarray}
\end{theorem}
\begin{proof}
If an $((n,K,R,d))_q$ subsystem code exists, then the weight
distribution $A_j$ of the associated additive code~$C$ and the
weight distribution $B_j$ of its subcode $D =C\cap C^\sdual$
obviously satisfy~1).  By Lemma~\ref{lem:css-Euclidean-subsys}, we
have $K=q^n/\sqrt{|C||D|}$ and $R=\sqrt{|C|/|D|}$, which implies
$|C|=\sum A_j = q^nR/K$ and $|D|=\sum B_j =q^n/KR$, proving~2).
Conditions~3) and~4) follow from the MacWilliams relation for
symplectic weight distribution, see \cite[Theorem~23]{ketkar06}. As
$C$ is an $\F_p$-linear code, for each nonzero codeword $c$ in $C$,
$\alpha c$ is again in $C$ for all $\alpha$ in $\F_p^\times$; thus,
condition~7) must hold. Since the quantum code has minimum distance
$d$, all vectors of symplectic weight less than $d$ in $D^\sdual$
must be in $C$, since $D^\sdual-C$ has minimum distance $d$; this
implies~5). Similarly, all vectors in $C^\sdual\subseteq C+C^\sdual$
of symplectic weight less than $d$ must be contained in $C$, since
$(C+C^\sdual)-C$ has minimum distance $d$; this implies~6).
\end{proof}

We can use the previous theorem to derive bounds on the dimension of the
co-subsystem. If the optimization problem is not solvable, then we can
immediately conclude that a code with the corresponding parameter settings
cannot exist. We are able to solve this optimization problem and have
constructed Table~\ref{LPtable-q2} over $\F_2$. Also, Table~\ref{LPtable-q3}
shows code parameters of subsystem codes over $\F_3$. It is not necessary that
the short subsystem codes are binary. The linear programming indicates that
there is no subsystem code with parameters $[[6,1,1,3]]_2$. However, there is a
subsystem code with parameters $[[6,1,1,3]]_3$ constructed over graphs.

\noindent \textbf{Impure Subsystem Codes and Hamming Bound.} The
following Lemma shows that there exist some families of subsystem
codes that beat the quantum Hamming bound. For stabilizer Hamming
codes see the tables given in~\cite{ketkar06}.
\begin{lemma}\label{lem:beatHamming}
If there exists an $[[n,k,d]]_q$ stabilizer perfect code and  $d' \geq d+2$ ,
then there must be an $[[n,k-r,r,d']]_q$ subsystem code that beats the Hamming
bound.
\end{lemma}
\begin{proof}
We know that  the stabilizer code satisfies the Hamming bound
\begin{eqnarray}
\sum_{i=0}^{\lfloor (d-1)/2 \rfloor} \binom{n}{i}(q^2-1)^i \leq q^{n-k},
\end{eqnarray}
But the given code is perfect, then the inequality holds. From our
construction in Theorem~\ref{th:stab2sub}, there must exist a
subsystem code with the given parameters. Since $\lfloor
(d'-1)/2\rfloor \geq \lfloor (d-1)/2\rfloor$ then the result is a
direct consequence.
\end{proof}
One example to show this Theorem would be Hermitian stabilizer Hamming codes.
These codes have parameters $[[n,n-2m,3]]_q$, where $m \geq 2$,
$gcd(m,q^2-1)=1$ and $n=\frac{q^{2m}-1}{q^2-1}$. Let $q=2$, and $m=4$ such that
$gcd(m,q^2-1)=1$, then $n=(q^{2m}-1)/(q^2-1)=85$. So, there exists a perfect
stabilizer Hamming code with parameters $[[85,77,3]]_2$. Consequently, there
must be a subsystem code with parameters $[[85,77-r,r,\ge 5]]_2$ that beats
Hamming bound. Also, the code $[[341,331,3]]_2$ gives us the same result.

The quantum Hamming bound for impure nonbinary stabilizer codes has
not been proved for $d\geq 7$, see~\cite{aly07hamming}. Of course if
the underline stabilizer code beats Hamming bound, obviously, the
subsystem codes would also beat the Hamming bound.  The condition in the theorem can be
relaxed. It is not necessarily needed the stabilizer code to be
perfect but it seems to be hard to find a general theme in this
case.

\noindent \textbf{Lower Bounds for Subsystem Codes.}  We can also
present a lower bound of subsystem code parameters known as the
Gilbert-Varshamov bound. Our goal is to provide a table of a lower
bound on subsystem code parameters, for more details
see~\cite{aly06c}.

\begin{theorem}\label{th:gvoqec}
Let $\F_q$ be a finite field of characteristic $p$.  If $K$ and $R$ are powers
of $p$ such that $1<KR\le q^n$ and $d$ is a positive integer such that
$$
\sum_{j=1}^{d-1} \binom{n}{j}(q^{2}-1)^j (q^nKR-q^nR/K)<(p-1)(q^{2n}-1)$$
holds, then an $((n,K,R,\ge d))_q$ subsystem code exists.
\end{theorem}

\begin{proof}
See~\cite[Thoerem 7]{aly06c}.
\end{proof}
\section{Pure Subsystem Code Constructions}

\begin{lemma}
If there exists a pure $((n,K,R,d))_q$ Clifford subsystem code, then there also
exists an $((n,R,K, \mbox{$\ge$}\,d))_q$ Clifford subsystem code that is pure
to $d$.
\end{lemma}
\begin{proof}
By Theorem~\ref{th:oqecfq}, there exist classical codes $D\subseteq
C\subseteq \F_{q^2}^n$ with the parameters $(n,q^nR/K)_{q^2}$ and
$(n,q^n/KR)_{q^2}$. Furthermore, since the subsystem code is pure,
we have $\wt(D^\adual\setminus C) = \wt( D^\adual)= d$.
Let us interchange the roles of $C$ and $C^\adual$, that is, now we
construct a subsystem code from $C^\adual$. The parameters of the
resulting subsystem code are given by \begin{eqnarray}((n,
\sqrt{|D^\adual|/|C^\adual|},\sqrt{|C^\adual|/|D|},\wt(D^\adual\setminus
C^\adual) ))_q.\end{eqnarray} We note that
\begin{compactitem}
\item  $\sqrt{|D^\adual|/|C^\adual|} =\sqrt{|C|/|D|} =R$ and
\item $\sqrt{|C^\adual|/|D|}= \sqrt{|D^\adual|/|C|}=K$.
\end{compactitem}
The minimum distance $d'$ of the resulting code satisfies $d' =
\wt(D^\adual\setminus C^\adual) \geq \wt( D^\adual) = d$; the claim about the
purity follows from the fact that $\wt(D^\adual)=d$.
\end{proof}
The following Theorem shows that given a stabilizer code, one can
construct subsystem codes with the same length and distance. Various
methods of subsystem code constructions have been shown in the
previous two chapters.
\begin{theorem}\label{th:stab2sub}
Let $q$ and $R$ be powers of a prime $p$. If there exists an $((n,K,d))_q$
stabilizer code pure to $d'$, then there exists an $((n,K/R,R,\geq d))_q$
subsystem code that is pure to $d'$.
\end{theorem}
\begin{proof}
 Let $D \subseteq D^\sdual \subseteq \F_q^{2n}$ be a classical code
generated by the $\F_p$-basis $\beta_D=\{z_{1},z_{2},...,z_{s}\}$ where
$d=\swt( D^{\perp_s} \backslash D)$. We know that there exists a stabilizer
code $Q$ with parameters $((n,K,d))_q$ that it is pure to $d'=\swt(D)$. $\dim
Q=|D^{\perp_s}|/|D|=q^n/p^s=p^{nm-s}$, where $q=p^m$.

Let us construct the additive code $C \subseteq D^\sdual$ by
expanding the set $\beta_D$ as follows
\begin{eqnarray*}
C&=&span_{\F_p}(\beta_D, \{z_{s+1},x_{s+1},...,z_{s+r},x_{s+r} \})\\
&=&<z_1,...,z_s;z_{s+1},x_{s+1},...,z_{r+s},x_{s+r}>.
\end{eqnarray*}
From Lemma~\cite[Lemma 10]{aly06c},
$\scal{x_k}{x_\ell}=0=\scal{z_k}{z_\ell}$ and
$\scal{x_k}{z_\ell}=\delta_{k,\ell}$, therefore $D \subseteq C$. We
notice that the code $C$ does not contain its dual $C^{\perp_s}$
because the elements in $C$ does not commute with each other. The
dual code $C^{\perp_s}$ is generated by the set
\begin{eqnarray*}
C^{\perp_s} &=& span_{\F_p}(\beta_D, \{z_{r+s+1},x_{r+s+1},...,z_{n},x_n \})
\end{eqnarray*}
The symplectic inner product between any two elements in $C$ and
$C^{\perp_s}$ vanishes. We see that $D=C \cap C^\sdual=<
z_{1},z_{2},...,z_{s}>$. Therefore, using \cite[Theorem 1]{aly06c},
there exists a subsystem code $Q_s=A \otimes B$ such that $\dim A=
q^n/ (|C||D|)^{1/2}=q^n/ (p^{2r+s}q^{s} )^{1/2}=p^{mn-r-s}=K/R$.
Also, $\dim B=|C|/|D|=(p^{2r+s}/p^{s})^{1/2}=p^r=R$.

If weight of a codeword $c$ in $D^\sdual$ is $d$, then either $c \in
C$ or $c \in D^\sdual \backslash C$. If $c \in D^\sdual \backslash
C$, then the subsystem code $Q_s$ has minimum distance $d$. If $c
\in C$ and no other codewords in $D^\sdual \backslash C$ has weight
$d$, then the subsystem code $Q_s$ has minimum distance $\geq d$.
Let $wt(D)$ be $d'$,  since $D \subseteq C$ then the subsystem code
$Q_s$ is pure to $d'$.
\end{proof}

\section{Propagation Rules of Subsystem Codes}
In this section we present propagation rules of subsystem  code
constructions similar to propagation rules of stabilizer code
constructions. We show that given  a subsystem code with parameters
$[[n,k,r,d]]_q$, it is possible to construct new codes with either
increase or decrease the length and dimension of the code by one.
Also, we can construct new subsystem codes from known two subsystem
codes.

Recall Lemmas~\ref{lem:css-Euclidean-subsys} and
\ref{lem:css-Hermitina-subsys}, there exists a subsystem code $Q$
with parameters $[[n,k,r,d]]_q$ using the Euclidean and Hermitian
constructions. The code $Q$ is decomposed into two sub-systems, $Q=A
\otimes B$, where $|A|=q^k$ and $|B|=q^r$. From the previous
section, if there is an $[[n,k,r,d]]_q$ subsystem code, then there
are two classical codes $C, D \in F_{q^2}^n$ such that $D=C \cap
C^{\perp_s}$, $X=|C|=q^{n-k+r}$ and $Y=|D|=q^{n-k-r}$. The minimum
distance of $Q$ is $d= \min \swt (D^{\perp_s}\backslash C)$. We use
this note to show the following Lemmas.

Let $C_1 \le \F_q^n$ and $C_2 \F_q^n$ be two classical codes defined
over $F_q$. The direct sum of $C_1$ and $C_2$ is a code $C \le
\F_q^{2n}$ defined as follows

\begin{eqnarray}
C=C_1 \oplus C_2=\{uv \mid u \in C_1, v \in C_2\}.
\end{eqnarray}
In a matrix form the code $C$ can be described as
$$C = \Big(\begin{array}{cc} C_1&0\\0&C_2 \end{array}\Big)$$

An $[n,k_1,d_1]_q$ classical code $C_1$ is a subcode in an
$[c,k_2,d_2]_q$ if every codeword $v$ in $C_1$ is also a codeword in
$C_2$, hence $k_1\leq k_2$. We say that an $[[n,k_1,r_1,d_1]]_q$
subsystem code $Q_1$ is a subcode in an $[[n,k_2,r_2,d_2]]_q$
subsystem code $Q_2$ if  every codeword $\ket{v}$ in $Q_1$ is also a
codeword in $Q_2$ and $k_1+r_1 \leq k_2+r_1$.

\bigskip

{\em Notation.} Let $q$ be a power of a prime integer $p$. We denote
by $\F_q$ the finite field with $q$ elements. We use the notation
$(x|y)=(x_1,\dots,x_n|y_1,\dots,y_n)$ to denote the concatenation of
two vectors $x$ and $y$ in $\F_q^n$. The symplectic weight of
$(x|y)\in \F_q^{2n}$ is defined as $$\swt(x|y)=\{(x_i,y_i)\neq
(0,0)\,|\, 1\le i\le n\}.$$ We define $\swt(X)=\min\{\swt(x)\,|\,
x\in X, x\neq 0\}$ for any nonempty subset $X\neq \{0\}$ of
$\F_q^{2n}$.

The trace-symplectic product of two vectors $u=(a|b)$ and
$v=(a'|b')$ in $\F_q^{2n}$ is defined as
$$\langle u|v \rangle_s = \tr_{q/p}(a'\cdot b-a\cdot b'),$$ where
$x\cdot y$ denotes the dot product and $\tr_{q/p}$ denotes the trace
from $\F_q$ to the subfield $\F_p$.  The trace-symplectic dual of a
code $C\subseteq \F_q^{2n}$ is defined as $$C^\sdual=\{ v\in
\F_q^{2n}\mid \langle v|w \rangle_s =0 \mbox{ for all } w\in C\}.$$
We define the Euclidean inner product $\langle x|y\rangle
=\sum_{i=1}^nx_iy_i$ and the Euclidean dual of $C\subseteq \F_{q}^n$
as $$C^\perp = \{x\in \F_{q}^n\mid \langle x|y \rangle=0 \mbox{ for
all } y\in C \}.$$ We also define the Hermitian inner product for
vectors $x,y$ in $\F_{q^2}^n$ as $\langle x|y\rangle_h
=\sum_{i=1}^nx_i^qy_i$ and the Hermitian dual of $C\subseteq
\F_{q^2}^n$ as
$$C^\hdual= \{x\in \F_{q^2}^n\mid \langle x|y \rangle_h=0 \mbox{ for all } y\in
C \}.$$
\bigskip

\begin{theorem}\label{th:oqecfq}
Let $C$ be a classical additive subcode of\/ $\F_q^{2n}$ such that
$C\neq \{0\}$ and let $D$ denote its subcode $D=C\cap C^\sdual$. If
$x=|C|$ and $y=|D|$, then there exists a subsystem code $Q= A\otimes
B$ such that
\begin{compactenum}[i)]
\item $\dim A = q^n/(xy)^{1/2}$,
\item $\dim B = (x/y)^{1/2}$.
\end{compactenum}
The minimum distance of subsystem $A$ is given by
\begin{compactenum}[(a)]
\item $d=\swt((C+C^\sdual)-C)=\swt(D^\sdual-C)$ if $D^\sdual\neq C$;
\item $d=\swt(D^\sdual)$ if $D^\sdual=C$.
\end{compactenum}
Thus, the subsystem $A$ can detect all errors in $E$ of weight less
than $d$, and can correct all errors in $E$ of weight $\le \lfloor
(d-1)/2\rfloor$.
\end{theorem}

\bigskip

 \noindent
\textbf{Extending Subsystem Codes.} We derive new subsystem codes
from known ones by extending and shortening the length of the code.

\begin{theorem}\label{lemma_n+1k}
If there exists an  $((n,K,R,d))_q$ Clifford subsystem code with
$K>1$, then there exists an $((n+1, K, R, \ge d))_q$ subsystem code
that is pure to~1.
\end{theorem}
\begin{proof}
We first note that for any additive subcode $X\le \F_q^{2n}$, we can
define an additive code $X'\le \F_q^{2n+2}$ by
$$X'=\{ (a\alpha|b0)\,|\, (a|b)\in X, \alpha\in
\F_q\}.$$ We have $|X'|=q|X|$. Furthermore, if $(c|e)\in X^\sdual$,
then $(c\alpha|e0)$ is contained in $(X')^\sdual$ for all $\alpha$
in $\F_q$, whence $(X^\sdual)'\subseteq (X')^\sdual$.  By comparing
cardinalities we find that equality must hold; in other words, we
have
$$(X^\sdual)'= (X')^\sdual.$$

By Theorem~\ref{th:oqecfq}, there are two additive codes $C$ and $D$
associated with an $((n,K,R,d))_q$ Clifford subsystem code such that
$$|C|=q^nR/K$$ and $$|D|=|C\cap C^\sdual| = q^n/(KR).$$ We can derive from the
code $C$ two new additive codes of length $2n+2$ over $\F_q$, namely
$C'$ and $D'=C'\cap (C')^\sdual$. The codes $C'$ and $D'$ determine
a $((n+1,K',R',d'))_q$ Clifford subsystem code. Since
\begin{eqnarray*}
D'&=&C'\cap (C')^\sdual = C'\cap (C^\sdual)' \\&=&(C\cap C^\sdual)',
\end{eqnarray*}
 we have
$|D'|=q|D|$. Furthermore, we have $|C'|=q|C|$. It follows from
Theorem~\ref{th:oqecfq} that
\begin{compactenum}[(i)]
\item $K'= q^{n+1}/\sqrt{|C'||D'|}=q^n/\sqrt{|C||D|}=K$,
\item $R'=(|C'|/|D'|)^{1/2} = (|C|/|D|)^{1/2} = R$,
\item $d'= \swt( (D')^\sdual \setminus C')\ge \swt( (D^\sdual\setminus C)')=d$.
\end{compactenum}
Since $C'$ contains a vector $(\mathbf{0}\alpha|\mathbf{0}0)$ of
weight $1$, the resulting subsystem code is pure to~1.
\end{proof}


\begin{corollary}
If there exists an $[[n,k,r,d]]_q$ subsystem  code with $k>0$ and
$0\leq r <k$, then there exists an $[[n+1, k, r, \ge d]]_q$
subsystem code that is pure to~1.
\end{corollary}

\medskip
\noindent \textbf{Shortening Subsystem Codes.} We can also shorten
the length of a subsystem code and still trade the dimensions of the
new subsystem code and its co-subsystem code as shown in the
following Lemma.

\bigskip

\begin{theorem}\label{lem:n-1k+1rule}
If an $((n,K,R,d))_q$ pure subsystem code $Q$ exists, then there is
a pure subsystem code $Q_p$ with parameters $((n-1,qK,R,\geq
d-1))_q$.
\end{theorem}
\begin{proof}
We know that existence of the pure subsystem code $Q$ with
parameters $((n,K,R,d))_q$ implies existence of a pure stabilizer
code with parameters $((n,KR,\geq d))_q$ for $n \geq 2$ and $d\geq
2$ from~\cite[Theorem 2.]{aly08a}. By~\cite[Theorem 70]{ketkar06},
there exist a pure stabilizer code with parameters $((n-1,qKR,\geq
d-1))_q$. This stabilizer code can be seen as $((n-1,qKR,0,\geq
d-1))_q$ subsystem code. By using \cite[Theorem 2.]{aly08a}, there
exists a pure $\F_q$-linear subsystem code with parameters
$((n-1,qK,R,\geq d-1))_q$ that proves the claim.
\end{proof}
Analog of the previous Theorem is the following Lemma.

\begin{lemma}\label{lem:n-1k+1rule}
If an $\F_q$-linear $[[n,k,r,d]]_q$ pure subsystem code $Q$ exists,
then there is a pure subsystem code $Q_p$ with parameters
$[[n-1,k+1,r,\geq d-1]]_q$.
\end{lemma}
\begin{proof}
We know that existence of the pure subsystem code $Q$ implies
existence of a pure stabilizer code with parameters $[[n,k+r,\geq
d]]_q$ for $n \geq 2$ and $d\geq 2$ by using~\cite[Theorem 2. and
Theorem 5.]{aly08a}. By~\cite[Theorem 70]{ketkar06}, there exist a
pure stabilizer code with parameters $[[n-1,k+r+1,\geq d-1]]_q$.
This stabilizer code can be seen as an $[[n-1,k+r+1,0,\geq d-1]]_q$
subsystem code. By using \cite[Theorem 3.]{aly08a}, there exists a
pure $\F_q$-linear subsystem code with parameters $[[n-1,k+1,r,\geq
d-1]]_q$ that proves the claim.
\end{proof}

\bigskip
We can also prove the previous Theorem by defining a new code $C_p$
from the code $C$ as follows.
%
 \begin{theorem}\label{lemma_n-1k+1}
 If there exists a
pure subsystem code $ Q=A\otimes B$ with parameters $((n,K,R,d))_q$
with $n\geq 2$ and $d \geq 2$, then there is a subsystem code $Q_p$
with parameters $((n-1,K,q R, \geq d-1))_q$.
\end{theorem}
\begin{proof}
By Theorem~\ref{th:oqecfq}, if an $((n,K,R,d))_q$ subsystem code $Q$
exists for $K>1$ and $1\leq R<K$, then there exists an additive code
$C \in \F_q^{2n}$ and its subcode $D\leq \F_q^{2n}$ such that
$|C|=q^n R/K$ and $|D|=|C\cap C^{\perp_s}|=q^{n}/KR$. Furthermore,
$d=\min \swt(D^{\perp_s}\backslash C)$. Let $w=(w_1,w_2,\ldots,w_n)$
and $u=(u_1,u_2,\ldots,u_n)$ be two vectors in $\F_q^n$. W.l.g., we
can assume that the code $D^\sdual$ is defined as
$$D^\sdual=\{(u|w) \in \F_q^{2n} \mid w,u \in \F_q^n\}.$$
Let $w_{-1}=(w_1,w_2,\ldots,w_{n-1})$ and
$u_{-1}=(u_1,u_2,\ldots,u_{n-1})$ be two vectors in $\F_q^{n-1}$.
Also, let $D_p^\sdual$ be the code obtained by puncturing the first
coordinate of $D^\sdual$, hence
$$D_p^\sdual=\{(u_{-1}|w_{-1}) \in \F_q^{2n-2} \mid w_{-1},u_{-1} \in \F_q^{n-1}\}.$$
since the minimum distance of $D^\sdual$ is at least 2, it follows
that $|D_p^\sdual|=|D^\sdual|=K^2|C|=K^2q^nR/K=q^nRK$ and the
minimum distance of $D_p^\sdual$ is at least $d-1$. Now, let us
construct the dual code of $D_p^\sdual$ as follows.
\begin{eqnarray*}
(D_p^\sdual)^\sdual &=&\{(u_{-1}|w_{-1}) \in \F_q^{2n-2} \mid \\&&
(0u_{-1}|0w_{-1}) \in D,w_{-1},u_{-1} \in
\F_q^{n-1}\}.\end{eqnarray*}

Furthermore, if $(u_{-1}|w_{-1}) \in D_p$, then $(0u_{-1}|0w_{-1})
\in D$. Therefore, $D_p$ is a self-orthogonal code and it has size
given by $$|D_p|=q^{2n-2}/|D_p^\sdual|=q^{n-2}/RK.$$
We can also puncture the code $C$ to the code $C_p$ at the first
coordinate, hence \begin{eqnarray*} C_p&=&\{(u_{-1}|w_{-1}) \in
\F_q^{2n-2} \mid w_{-1},u_{-1} \in \F_q^{n-1},\\ &&
(aw_{-1}|bu_{-1}) \in C, a,b \in F_q\}.\end{eqnarray*} Clearly, $D
\subseteq C$  and if $a=b=0$, then the vector $(0u_{-1}|0w_{-1}) \in
D$, therefore, $(u_{-1},w_{-1}) \in D_p$. This gives us that $D_p
\subseteq C_p$. Furthermore, hence $|C|=|C_p|$. The dual code
$C_p^\sdual$  can be defined as \begin{eqnarray*} C_p^\sdual
&=&\{(u_{-1}|w_{-1}) \in \F_q^{2n-2} \mid w_{-1},u_{-1} \in
\F_q^{n-1},\\ && (ew_{-1}|fu_{-1}) \in C^\sdual, e,f \in
F_q\}.\end{eqnarray*} Also, if $e=f=0$, then $D_p \subseteq
C_p^\sdual$, furthermore, \begin{eqnarray}D_p^\sdual&=&C_p \cup
C_p^\sdual=\{ (u_{-1}|w_{-1}) \in \F_q^{2n-2} \mid \\&&
(0u_{-1}|0w_{-1}) \in D\}\end{eqnarray}

Therefore there exists a subsystem code $Q_p=A_p \otimes B_p$. Also,
the code $D_p^\sdual$ is pure and has minimum distance at least
$d-1$. We can proceed and compute the dimension of subsystem $A_p$
and co-subsystem $B_p$ from Theorem~\ref{th:oqecfq} as follows.

\begin{compactenum}[(i)]
\item $K_p= q^{n-1}/\sqrt{|C_p||D_p|}=q^{n-1}/\sqrt{(q^nR/K)(q^{n-2}/RK)}=K$,
\item $R_p=(|C_p|/|D_p'|)^{1/2} = ((q^nR/K)/(q^{n-2}/RK))^{1/2} = qR$,
\item $d_p= \swt( (D_p)^\sdual \setminus C_p)= \swt( (D^\sdual\setminus C_p))\geq d-1$.
\end{compactenum}

  Therefore, there exists a subsystem cod with parameters
$((n-1,K,q R, \geq d-1))_q$.

The minimum distance condition follows since the code $Q$ has $d=
\min \swt (D^{\perp_s}\backslash C)$ and the code $Q_p$ has minimum
distance as $Q$ reduced by one. So, the minimum weight of
$D_p^{\sdual} \backslash C_p$ is at least the minimum weight of
$(D^\sdual \backslash C)-1$
\begin{eqnarray*}
d_p&=&\min \swt (D{_p}^{\perp_s} \backslash C_p)  \nonumber \\
&\geq& \min \swt (D^{\perp_s} \backslash C)-1=d-1
\end{eqnarray*}
If the code $Q$ is pure, then $\min \swt (D^{\perp_s})=d$,
therefore, the new code $Q_p$ is pure since $d_p=\min \swt
(D_p^{\perp_s}) \geq d$.

We conclude that if there is a subsystem code with parameters
$((n-1,K,q R, \geq d-1))_q$, using ~\cite[Theorem 2.]{aly08a}, there
exists a code with parameters $((n-1,qK, R, \geq d-1))_q$.
\end{proof}

\medskip
\noindent \textbf{Reducing Dimension.} We also can reduce dimension
of the subsystem code for fixed length $n$ and minimum distance $d$,
and still obtain a new subsystem code with improved minimum distance
as shown in the following results.

\begin{theorem}\label{lem:reducingK}
If a (pure)$\F_q$-linear $[[n,k,r,d]]_q$ subsystem code $Q$ exists
for $d\geq 2$, then there exists an $\F_q$-linear
$[[n,k-1,r,d_e]]_q$ subsystem code $Q_e$ (pure to d) such that $d_e
\geq d$.
\end{theorem}
\begin{proof}
Existence of the $[[n,k,r,d]]_q$ subsystem code $Q$,  implies
existence of two additive codes $C\leq \F_q^{2n}$ and $D\leq
\F_q^{2n}$ such that $|C|=q^{n-k+r}$ and $|D|=|C \cap
C^{\perp_s}|=q^{n-k-r}$. Furthermore, $d=\min
\swt(D^{\perp_s}\backslash C)$ and $D \subseteq D^\sdual$.

The idea of the proof comes by extending the code $D$ by some
vectors from $D^{\sdual} \backslash (C\cup C^\sdual$). Let us choose
a code $D_e$ of size $|q^{n+1-r-k}|=q|D|$. We also ensure that the
code $D_e$ is self-orthogonal. Clearly extending the code $D$ to
$D_e$ will extend both the codes $C$ and $C^\sdual$ to $C_e$ and
$C_e^\sdual$, respectively. Hence $C_e=q|C|=q^{n+1+r-k}$ and $D_e =
C_e \cap C_e^\sdual$.

There exists a subsystem code $Q_e$ stabilized by the code $C_e$.
The result follows by computing parameters of the subsystem code
$Q_e=A_e \otimes B_e$.

\begin{compactenum}[(i)]
\item $K_e= q^{n}/\sqrt{|C_e||D_e|}=q^{n}/((q^{n+1+r-k})(q^{n+1-k-r}))^{1/2}=q^{k-1}$,
\item $R_e=(|C_e|/|D_e|)^{1/2} = ((q^{n+1}R/K)/(q^{n+1}/RK))^{1/2} = q^r$,
\item $d_e= \swt( (D_e)^\sdual \setminus C_e)\geq \swt( (D^\sdual \setminus C_e))=
d$. If the inequality holds, then the code is pure to $d$.
\end{compactenum}
Arguably, It follows that the set $(D_e^{\sdual}\backslash C_e)$ is
a subset of the set $D^{\sdual}\backslash C$ because $C \leq C_e$,
hence the minimum weight $d_e$ is at least $d$.
\end{proof}

\bigskip

\medskip

\begin{lemma}\label{lem:reducen-m}
Suppose an $[[n,k,r,d]]_q$ linear pure subsystem code $Q$ exists
generated by the two codes $C,D \leq \F_q^{2n}$. Then there exist
linear $[[n-m,k',r',d']]_q$ and $[[n-m,k'+r'-r'',r'',d']]_q$
subsystem codes with $k'\geq k-m$, $r' \geq r$, $0\leq r''< k'+r'$,
and $d'\geq d$ for any integer $m$ such that there exists a codeword
of weight $m$ in $(D^{\perp_s}\backslash C)$.
\end{lemma}
\begin{proof}{[Sketch]}
This lemma~\ref{lem:reducen-m} can be proved easily by mapping the
subsystem code $Q$ into a stabilizer code. By using \cite[Theorem
7.]{calderbank98}, and the new resulting stabilizer code can be
mapped again to a subsystem code with the required parameters.
\end{proof}

\medskip

\noindent \textbf{Combining Subsystem Codes} We can also construct
new subsystem codes from given two subsystem codes. The following
theorem shows that two subsystem codes can be merged together into
one subsystem code with possibly improved distance or dimension.
\medskip

\begin{theorem}\label{thm:twocodes_n1k1r1d1n2k2r2d2}
Let $Q_1$ and $Q_2$ be two pure binary subsystem codes with
parameters $[[n_1,k_1,r_1,d_1]]_2$ and $[[n_2,k_2,r_2,d_2]]_2$ for
$k_2+r_2\leq n_1$, respectively. Then there exists a subsystem code
with parameters $[[n_1+n_2-k_2-r_2,k_1+r_1-r,r,d]]_2$, where $d \geq
min \{d_1,d_1+d_2-k_2-r_2\}$ and $0 \leq r <k_1+r_1$.
\end{theorem}

\begin{proof}
Existence of an $[[n_i,k_i,r_i,d_i]]_2$ pure subsystem code $Q_i$
for $i\in\{1,2\}$ , implies existence of a pure stabilizer code
$S_i$ with parameters $[[n_i,k_i+r_i,d_i]]_2$ with $k_2+r_2 \leq
n_1$, see~\cite{aly08a}. Therefore, by~\cite[Theorem
8.]{calderbank98}, there exists  a stabilizer code with parameters
$[[n_1+n_2-k_2-r_2,k_1+r_1,d]]_2$, $d \geq \min
\{d_1,d_1+d_2-k_2-r_2\}$. But this code gives us a subsystem code
with parameters $[[n_1+n_2-k_2-r_2,k_1+r_1-r,r,\geq d]]_2$ with
$k_2+r_2 \leq n_1$ and $0 \leq r <k_1+r_1$ that proves the claim.
\end{proof}

\medskip

\begin{theorem}\label{lem:twocodes_nk1r1d1nk2r2d2} Let $Q_1$ and $Q_2$ be two
pure subsystem codes with parameters $[[n,k_1,r_1,d_1]]_q$ and
$[[n,k_2,r_2,d_2]]_q$, respectively. If $Q_2\subseteq Q_1$, then
there exists an $[[2n,k_1+k_2+r_1+r_2-r,r,d]]_q$ pure subsystem code
 with minimum distance $d \geq \min \{d_1,2d_2\}$ and $0\leq r <k_1+k_2+r_1+r_2$.
\end{theorem}
\begin{proof}
Existence of a pure subsystem code with parameters
$[[n,k_i,r_i,d_i]]_q$ implies existence of a pure stabilizer code
with parameters $[[n,k_i+r_i,d_i]]_q$ using~\cite[Theorem
4.]{aly08a}. But by using~\cite[Lemma 74.]{ketkar06}, there exists a
pure stabilizer code with parameters $[[2n,k_1+k_2+r_1+r_2,d]]_q$
with $d \geq \min
 \{2d_2,d_1\}$. By~\cite[Theorem 2., Corollary 6.]{aly08a}, there
 must exist a pure subsystem code with parameters
 $[[2n,k_1+k_2+r_1+r_2-r,r,d]]_q$ where $d \geq \min
 \{2d_2,d_1\}$ and $0\leq r <k_1+k_2+r_1+r_2$, which proves the claim.
\end{proof}
\bigskip

We can recall the trace alternative product between  two codewords
of a classical code and the proof
of~Theorem~\ref{lem:twocodes_nk1r1d1nk2r2d2} can be stated as
follows.

\begin{lemma}\label{lem:twocodes_nk1r1d1nk2r2d2another} Let $Q_1$
and $Q_2$ be two pure subsystem codes with parameters
$[[n,k_1,r_1,d_1]]_q$ and $[[n,k_2,r_2,d_2]]_q$, respectively. If
$Q_2\subseteq Q_1$, then there exists an
$[[2n,k_1+k_2,r_1+r_2,d]]_q$ pure subsystem code
 with minimum distance $d \geq \min \{d_1,2d_2\}$.
\end{lemma}

\begin{proof}
Existence of the code $Q_i$ with parameters $[[n,K_i,R_i,d_i]]_q$
implies existence of two additive  codes $C_i$ and $D_i$ for $i \in
\{1,2\}$ such that $ |C_i|=q^nR_i/K_i$ and $|D_i|=|C\cup
C^{\perp_s}|=q^n/R_iK_i$.

 We know that there exist additive linear
codes $D_i \subseteq D_i^\adual$, $D_i \subseteq C_i$, and $D_i
\subseteq C_i^\adual$. Furthermore, $D_i=C_i \cap C_i^\adual$ and
$d_i = wt (D_i^\adual \backslash C_i)$. Also, $C_i = q^{n+r_i-k_i}$
and $|D|=q^{n-r_i-k_i}$.

Using the direct sum definition between to linear codes, let us
construct a code $D$ based on $D_1$ and $D_2$ as

$$D=\{(u,u+v) \mid u \in D_1, v \in D_2\} \le \F_{q^2}^{2n}.$$
The code $D$ has size of $|D|=q^{2n -(r_1+r_2+k_1+k_2)=|D_1||D_2|}.$
Also, we can define the code $C$ based on the codes $C_1$ and $C_2$
as
$$C=\{(a,a+b) \mid a \in C_1, b \in C_2\} \le \F_{q^2}^{2n}.$$
The code $C$ is of size $|C|=|C_1||C_2|=q^{2n+r_1+r_2-k_1-k_2}.$ But
the trace-alternating dual of the code $D$ is

$$D^{\adual}= \{(u'+v'|,v') \mid u' \in D_1^\adual, v' \in D_2^\adual \}.$$
 We notice that $(u'+v',v')$ is orthogonal to $(u,u+v)$ because,
 from properties of the product,
\begin{eqnarray*}\acal{(u,u+v)}{(u'+v',v')}&=&\acal{u}{u'+v'}+
\acal{u+v}{v'}\\ &=& 0\end{eqnarray*} holds for $u\in D_1,v\in
D_2,u'\in D_1^\adual,$ and $v' \in D_2^\adual$.

Therefore, $D \subseteq D^\adual$ is a self-orthogonal code with
respect to the trace alternating product. Furthermore, $C^\adual=
\{(a'+b',b') \mid a' \in C_1^\adual, b' \in C_2^\adual\}.$ Hence,
$C\cap C^\adual = \{(a,a+b) \cap (aa+b',b')\}=D$. Therefore, there
exists an $\F_q$-linear subsystem code $Q=A \otimes B$ with the
following parameters.

\begin{compactenum}[i)]
\item \begin{eqnarray*}K&=&|A|=q^{2n}/(|C||D|)^{1/2}\\&=&
\frac{q^{2n}}{\sqrt{(q^{2n}R_1R_2/K_1K_2)(q^{2n}/K_1K_2R_1R_2)}}\\&=&\frac{q^{2n}}{\sqrt{q^{2n+r_1+r_2-k_1-k_2}q^{2n-r_1-r_2-k_1-k_2}}}\\&=&q^{k_1k_2}=K_1K_2.
\end{eqnarray*}
\item $R=(\frac{|C|}{|D|})^{1/2}=R_1R_2.$
\item the minimum distance is a direct consequence.
\end{compactenum}

\end{proof}

\medskip

%
\medskip

\begin{theorem}
If there exist two pure subsystem quantum codes $Q_1$ and $Q_2$ with
parameters $[[n_1,k_1,r_1,d_1]]_q$ and $[[n_2,k_2,r_2,d_2]]_q$,
respectively. Then there exists a pure subsystem code $Q'$ with
parameters $[[n_1+n_2,k_1+k_2+r_1+r_2-r,r,\geq \min (d_1,d_2)]]_q$.
\end{theorem}
\begin{proof}
This Lemma can be proved easily from~\cite[Theorem 5.]{aly08a}
and~\cite[Lemma 73.]{ketkar06}. The idea is to map a pure subsystem
code to a pure stabilizer code, and once again map the pure
stabilizer code to a pure subsystem code.
\end{proof}

\begin{theorem}
If there exist two pure subsystem quantum codes $Q_1$ and $Q_2$ with
parameters $[[n_1,k_1,r_1,d_1]]_q$ and $[[n_2,k_2,r_2,d_2]]_q$,
respectively. Then there exists a pure subsystem code $Q'$ with
parameters $[[n_1+n_2,k_1+k_2,r_1+r_2,\geq \min (d_1,d_2)]]_q$.
\end{theorem}
\begin{proof}
Existence of the code $Q_i$ with parameters $[[n,K_i,R_i,d_i]]_q$
implies existence of two additive  codes $C_i$ and $D_i$ for $i \in
\{1,2\}$ such that $ |C_i|=q^nR_i/K_i$ and $|D_i|=|C\cup
C^{\perp_s}|=q^n/R_iK_i$.

 Let us choose the codes $C$ and $D$ as follows.
 $$C=C_1 \oplus C_2=\{uv \mid v \in C_1, v\in C_2\},$$ and $$D=D_1 \oplus D_2=\{ab \mid a\in D_1, b\in
 C_2\},$$
 respectively. From this construction, and since $D_1$ and $D_2$ are
 self-orthogonal codes, it follows that $D$ is also a
 self-orthogonal code. Furthermore, $D_1 \subseteq C_1$ and $D_2
 \subseteq
 C_2$, then

  $$D_1 \oplus D_2 \subseteq C_1 \oplus C_2,$$
hence $D\subseteq C$.  The code $C$ is of size
\begin{eqnarray*}|C|&=&|C_1||C_2|=q^{(n_1+n_2)-(k_1+k_2)+(r_1+r_2)}\\&=&q^{n_1}q^{n_2}R_1R_2/K_1K_2\end{eqnarray*}
and  $D$ is of size
\begin{eqnarray*}|D|&=&|D_1||D_2|=q^{(n_1+n_2)-(k_1+k_2)-(r_1+r_2)}\\&=&q^{n_1}q^{n_2}/R_1R_2K_1K_2.\end{eqnarray*}
On the other hand,
\begin{eqnarray*}C^\sdual=(C_1 \oplus C_2)^\sdual = C_2^\sdual \oplus C_1^\sdual
\supseteq D_2\oplus D_1.\end{eqnarray*} Furthermore, $C\cap
C^\sdual=(C_1 \oplus C_2)\cap (C_2^\sdual \cap C_1^\sdual)=D$.

Therefore, there exists a subsystem code $Q=A\otimes B$ with the
following parameters.
\begin{compactenum}[i)]
\item \begin{eqnarray*}K&=&|A|=q^{n_1+n_2}/(|C||D|)^{1/2}\\&=&
\frac{q^{n_1+n_2}}{\sqrt{(q^{n_1+n_2}R_1R_2/K_1K_2)(q^{n_1+n_2}/K_1K_2R_1R_2)}}\\&=&\frac{q^{n_1+n_2}}{\sqrt{q^{n_1+n_2+r_1+r_2-k_1-k_2}q^{n_1+n_2-r_1-r_2-k_1-k_2}}}\\&=&q^{k_1k_2}=K_1K_2=|A_1||A_2|.
\end{eqnarray*}
\item \begin{eqnarray*}
R&=&(\frac{|C|}{|D|})^{1/2}=\sqrt{\frac{q^{n_1}q^{n_2}R_1R_2/K_1K_2}{q^{n_1}q^{n_2}/R_1R_2K_1K_2}}\\&=&R_1R_2=|B_1||B_2|.\end{eqnarray*}
\item the minimum weight of $D^{\perp_s} \backslash C$ is at least the
minimum weight of $D_1^{\perp_s}\backslash C_1$ or
$D_2^{\perp_s}\backslash C_2$.
\begin{eqnarray*}
d&=&\min \{\swt (D_1^{\perp_s} \backslash C_1),(D_2^{\perp_s}
\backslash C_2)\} \nonumber \\ &\geq& \min \{d_1,d_2\}.
\end{eqnarray*}
\end{compactenum}

\end{proof}

\bigskip

\begin{table}[h]

\begin{center}
\caption{Existence of subsystem propagation
rules}\label{table:propagationrulesall}
\begin{supertabular}{| c@{\hspace{0.2cm}}|| c@{\hspace{0.2cm}}| c@{\hspace{0.2cm}}|c@{\hspace{0.3cm}}|}
\hline \hline n $\backslash$ k &k-1 &k &k+1\\
 \hline \hline
n-1 & $[ r+2, d-1]_q$& $[\leq r+2,d]_q$, $[r+1,d-1]_q$& $[r,d-1]_q$\\
\hline
  n & $[ r+1, d]_q$, $[ r+1,\geq d]_q$&  $[r,d]_q$ $\rightarrow [ \leq r, \geq d]_q$&$[ r-1,  d]_q$\\
  &&$ \rightarrow [\geq r,\leq d]_q $&\\
\hline
   n+1 & $[\geq r,\geq d]_q$&$[\geq r,d]_q$ , $[r,\geq d]_q$&\\
 \hline
\end{supertabular}
\end{center}
\end{table}

\begin{theorem}
Given two pure subsystem codes $Q_1$ and $Q_2$ with parameters
$[[n_1,k_1,r_1,d_1]]_q$ and $[[n_2,k_2,r_2,d_2]]_q$, respectively,
with $k_2 \leq n_1$. An $[[n_1+n_2-k_2,k_1+r_1+r_2-r,r,d]]_q$
subsystem code exists such that $d\geq \min \{d_1,d_1+d_2-k_2\}$ and
$0\leq r <k_1+r_1+r_2$.
\end{theorem}
\begin{proof}
The proof is a direct consequence as shown in the previous theorems.
\end{proof}

\medskip

\begin{theorem} If an $((n,K,R,d))_{q^m}$ pure subsystem code exists, then there exists a
pure subsystem code with parameters $((nm,K,R,\geq d))_q$.
Consequently, if a pure subsystem code with parameters
$((nm,K,R,\geq d))_q$ exists, then there exist a subsystem code with
parameters $((n,K,R,\geq \lfloor d/m \rfloor))_{q^m}$..
\end{theorem}
\begin{proof}
Existence of a pure subsystem code with parameters
$((n,K,R,d))_{q^m}$ implies existence of a pure stabilizer code with
parameters $((n,KR,d))_{q^m}$ using~\cite[Theorem 5.]{aly08a}.
By~\cite[Lemma 76.]{ketkar06}, there exists a stabilizer code with
parameters  $((nm,KR,\geq d))_q$. From~\cite[Theorem 2,5.]{aly08a},
there exists a pure subsystem code with  parameters $((nm,K,R,\geq
d))_q$ that proves the first claim.  By~\cite[Lemma 76.]{ketkar06}
and ~\cite[Theorem 2,5.]{aly08a}, and repeating the same proof, the
second claim is a consequence.
\end{proof}

\newpage


\begin{small}
\tablefirsthead{\hline \hline n/k&k=1&k=2&k=3&k=4&k=5&k=6&k=7&k=8&k=9&k=10&k=11&k=12\\
 \hline  \hline  }
 \topcaption{Upper bounds on subsystem code parameters
using linear programming, $q=2$} \label{LPtable-q2}
\par
\begin{supertabular}{p{0.8cm}p{0.8cm}p{0.8cm}p{0.75cm}p{0.75cm}p{0.75cm}p{0.7cm}p{0.7cm}p{0.7cm}p{0.7cm}p{0.7cm}p{0.7cm}p{0.7cm}p{0.7cm}p{0.7cm}}
\hline n=6&(5,1), (3,2),
 (1,3),&(4,1), (2,2),&(3,1),
 (1,2),&(2,1),&(1,1),&&&&&\\
n=7&(6,1), (4,2), (2,3),&(5,1), (3,2),&(4,1), (2,2),&(3,1),
(1,2),&(2,1),&(1,1),&&&&
\\
n=8&(7,1), (5,2), (3,3),&(6,1), (4,2), (2,3),&(5,1), (3,2),&(4,1),
(2,2),&(3,1), (1,2),&(2,1),&(1,1),&&&
\\
n=9&(8,1), (6,2), (4,3), (2,4),&(7,1), (5,2), (3,3),&(6,1), (4,2),
(2,3),&(5,1), (3,2),&(4,1), (2,2),&(3,1),
(1,2),&(2,1),&(1,1),&&\\
n=10&(9,1), (7,2), (5,3), (3,4),&(8,1), (6,2), (4,3), (2,4),&(7,1),
(5,2), (3,3),&(6,1), (4,2), (1,3),&(5,1), (3,2),&(4,1),
(2,2),&(3,1), (1,2),&(2,1),&(1,1),&

\\n=11&(10,1),
(8,2), (6,3), (4,4), (2,5),&(9,1), (7,2), (5,3), (3,4),&(8,1),
(6,2), (4,3), (2,4),&(7,1), (5,2), (3,3),&(6,1), (4,2),
(1,3),&(5,1), (3,2),&(4,1), (2,2),&(3,1),
(1,2),&(2,1),&(1,1),\\
n=12&(11,1), (9,2), (7,3), (5,4), (3,5),&(10,1), (8,2), (6,3),
(4,4), (1,5),&(9,1), (7,2), (5,3), (3,4),&(8,1), (6,2), (4,3),
(1,4),&(7,1), (5,2), (3,3),&(6,1), (4,2), (1,3),&(5,1),
(3,2),&(4,1), (2,2),&(3,1), (1,2),&(2,1),&(1,1),\\
n=13&(12,1), (9,2), (8,3), (6,4), (4,5), (1,6),&(11,1), (9,2),
(7,3), (5,4), (3,5),&(10,1), (8,2), (6,3), (4,4),&(9,1), (7,2),
(5,3), (3,4),&(8,1), (6,2), (4,3), (1,4),&(7,1), (5,2),
(3,3),&(6,1), (4,2),& (5,1), (3,2),&(4,1), (2,2),& (3,1),
(1,2),&(2,1),&(1,1),\\

\vspace{5cm}\\

\bigskip

Table~\ref{LPtable-q2}.&& Continued\\
 \hline  \hline
 n/k&k=1&k=2&k=3&k=4&k=5&k=6&k=7&k=8&k=9&k=10&k=11&k=12\\
  \hline  \hline
\\
n=14&(13,1), (10,2), (9,3), (7,4), (5,5), (3,6),&(12,1), (10,2),
(8,3), (6,4), (4,5),&(11,1), (9,2), (7,3), (5,4), (2,5),&(10,1),
(8,2), (6,3), (4,4),&(9,1), (7,2), (5,3), (3,4),&(8,1), (6,2),
(4,3),&(7,1), (5,2), (2,3),&(6,1), (4,2),&(5,1), (3,2),&(4,1),
(2,2),&(3,1), (1,2),&(2,1),
\\n=15&(14,1),
(12,2), (10,3), (8,4), (6,5), (4,6),&(13,1), (11,2), (9,3), (7,4),
(5,5), (3,6),&(12,1), (10,2), (8,3), (6,4), (4,5),&(11,1), (9,2),
(7,3), (5,4), (2,5),&(10,1), (8,2), (6,3), (4,4),&(9,1), (7,2),
(5,3), (2,4),&(8,1), (6,2), (4,3),&(7,1), (5,2), (2,3),&(6,1),
(4,2),&(5,1), (3,2),&(4,1), (2,2),&(3,1), (1,2),
\\n=16&
(15,1), (13,2), (11,3), (9,4), (7,5), (5,6), (1,7),&(14,1), (12,2),
(10,3), (8,4), (6,5), (4,6),&(13,1), (11,2), (9,3), (7,4), (5,5),
(2,6),&(11,1), (10,2), (8,3), (6,4), (4,5),&(11,1), (9,2), (7,3),
(5,4), (1,5),&(10,1), (8,2), (6,3), (4,4),&(9,1), (7,2), (5,3),
(2,4),&(8,1), (6,2), (4,3),&(6,1), (5,2), (2,3),&(6,1),
(4,2),&(5,1), (3,2),&(4,1), (2,2),
\\n=17&(14,1),
(14,2), (12,3), (9,4), (8,5), (6,6), (4,7),&(15,1), (13,2), (11,3),
(9,4), (7,5), (5,6), (1,7),&(14,1), (12,2), (10,3), (8,4), (6,5),
(4,6),&(13,1), (11,2), (9,3), (7,4), (5,5), (1,6),&(11,1), (9,2),
(8,3), (6,4), (3,5),&(10,1), (9,2), (7,3), (5,4),&(10,1), (8,2),
(6,3), (4,4),&(9,1), (7,2), (5,3), (2,4),&(8,1), (6,2),
(4,3),&(7,1), (5,2), (1,3),&(5,1), (4,2),&(4,1), (3,2),
\\n=18&(17,1),  (13,2), (13,3), (11,4), (9,5), (7,6), (5,7),&(15,1),
(14,2), (12,3), (10,4), (8,5), (6,6), (4,7),&(15,1), (12,2), (11,3),
(9,4), (7,5), (4,6),&(13,1), (11,2), (10,3), (8,4), (6,5),
(3,6),&(13,1), (11,2), (9,3), (7,4), (5,5),&(12,1), (10,2), (8,3),
(6,4), (2,5),&(11,1), (9,2), (7,3), (5,4),&(9,1), (8,2), (6,3),
(4,4),&(8,1), (7,2), (5,3), (1,4),&(8,1), (6,2), (3,3),&(6,1),
(5,2), (1,3),&(5,1), (4,2),\\
\end{supertabular}

\end{small}

\newpage

\begin{small}
\tablefirsthead{\hline \hline n/k&k=1&k=2&k=3&k=4&k=5&k=6&k=7&k=8&k=9&k=10&k=11&k=12\\
 \hline  \hline  }
\topcaption{Upper bounds on subsystem code parameters using linear programming,
$q=3$} \label{LPtable-q3}
\par
\begin{supertabular}{p{0.8cm}p{0.8cm}p{0.8cm}p{0.75cm}p{0.75cm}p{0.75cm}p{0.7cm}p{0.7cm}p{0.7cm}p{0.7cm}p{0.7cm}p{0.7cm}p{0.7cm}p{0.7cm}}
\hline n=4&(3,1), (1,2),& (2,1),&(1,1),&&&&&&&
\\n=5&(4,1),
(2,2),&(3,1), (1,2),&(2,1),&(1,1),&&&&&&
\\n=6&(5,1),
(3,2), (1,3),&(4,1), (2,2),&(3,1), (1,2),&(2,1),&(1,1),&&&&&
\\n=7&(4,1),
(4,2), (2,3),&(4,1), (3,2), (1,3),&(4,1), (2,2),&(3,1),
(1,2),&(2,1),&(1,1),&&&&
\\n=8&(5,1), (5,2), (3,3), (1,4),&(5,1), (4,2), (2,3),&(5,1), (3,2),
(1,3),&(4,1), (2,2),&(3,1),
(1,2),&(2,1),&(1,1),&&&\\
n=9&(6,1), (6,2), (3,3), (2,4),&(5,1), (5,2), (3,3), (1,4),&(6,1),
(4,2), (2,3),&(4,1), (3,2), (1,3),&(4,1), (2,2),&(3,1),
(1,2),&(1,1),&(1,1),&&&\\
n=10&(9,1), (7,2), (5,3), (3,4), (1,5),&(8,1), (6,2), (4,3),
(2,4),&(7,1), (5,2), (3,3), (1,4),&(6,1), (4,2), (2,3),&(5,1),
(3,2), (1,3),&(4,1), (2,2),&(3,1),
(1,2),&(2,1),&(1,1),&\\
n=11&(10,1), (7,2), (6,3), (4,4), (2,5),&(9,1), (7,2), (5,3), (3,4),
(1,5),&(7,1), (5,2), (4,3), (2,4),&(7,1), (5,2), (3,3),
(1,4),&(6,1), (4,2), (2,3),&(5,1), (3,2), (1,3),&(4,1),
(1,2),&(2,1),
(1,2),&(2,1),&\\
n=12&(10,1), (8,2), (6,3), (5,4), (3,5), (1,6),&(9,1), (6,2), (6,3),
(4,4), (2,5),&(9,1), (4,2), (5,3), (3,4), (1,5),&(8,1), (4,2),
(4,3), (2,4),&(7,1), (3,2), (3,3), (1,4),&(6,1), (2,2),
(2,3),&(5,1), (2,2),& (4,1), (2,2),&(3,1),
(1,2),&\\
\end{supertabular}
\end{small}

\section{Conclusion and Discussion}
We have established a number of subsystem code constructions. In
particular, we have shown how one can derive subsystem codes from
stabilizer codes. In combination with the propagation rules that we
have derived, one can easily create tables with the best known
subsystem codes. Table~\ref{table:propagationrulesall}. shows the
propagation rules of subsystem code parameters and what the rules
are to derive new subsystem codes from existing ones. We
have constructed tables of subsystem code parameters over binary and finite
fields.

Tables~\ref{LPtable-q2} and \ref{LPtable-q3} present upper bounds on
subsystem code parameters using the linear programming bound
implemented using MAGMA\cite{magma} and Matlab 0.7 programs, for
small code lengths. As a future research, designing the encoding and
decoding circuits of subsystem codes will be conducted as well as
deriving tables of upper bounds for large code lengths. Finally, it
will be interesting to derive sharp upper and lower bounds on
subsystem code parameters.


\part{Quantum Convolutional Codes}

\chapter{Quantum Convolutional Codes}\label{ch_QCC_bounds}

\section{Introduction}
Quantum information is sensitive to noise and needs error correction
and recovery strategies.  Quantum block error-correcting code (QBC)
and quantum convolutional codes $(QCC)$ are means to protect quantum
information against noise. The theory of stabilizer block
error-correcting codes is widely studied over binary and finite
fields, see for
example~\cite{ashikhmin01,calderbank98,ketkar06,rains99} and
references therein.
Quantum convolutional codes (QCC) have not been studied well over binary and
finite fields. There remain many interesting and open questions regarding the
properties and the usefulness of quantum convolutional codes. At this point in
time, it is not known if quantum convolutional codes offer a decisive advantage
over quantum block codes. However, it appears that quantum convolutional codes
are more suitable for quantum communications.

\medskip

In this chapter, we extend the theory of quantum convolutional codes
over finite fields generalizing some of the previously known
results. After a brief review of previous work in quantum
convolutional codes, we give the necessary background in classical
and quantum convolutional codes in Sections~\ref{sec:CC}
and~\ref{sec:QCCparameters}. We reformulate the necessary
terminology of the theory of quantum convolutional codes. Then in
the next two chapters, we construct families of quantum
convolutional codes based on classical codes~\cite{aly07b}.  Sections~\ref{sec:QCCparameters},\ref{sec:CS_CSS},
\ref{sec:QCC_singletonbound},  and the next chapter are based on a joint work with P.K.
Sarvepalli and A. Klappenecker, for further details, see our
companion paper~\cite{aly07b}.
\section{Previous Work on QCC} We  review the previous work on
quantum convolutional codes. There have been examples of quantum
convolutional codes in literature; the most notable being the
$((5,1,3))$ code of Ollivier and Tillich, the $((4,1,3))$ code of
Almeida and Palazzo  and the rate $1/3$ codes of Forney and Guha.
\begin{itemize}
\item
Chau initiated the early work in quantum convolutional
codes~~\cite{chau98,chau99}. However, there are negative
arguments about  his work~\cite{almeida05} and many authors are
divided whether his codes are truly quantum convolutional codes
or not.

\item Ollivier and Tillich developed the stabilizer framework for quantum convolutional codes.
They also addressed the encoding and decoding aspects of quantum
convolutional
codes~\cite{ollivier04,ollivier03,ollivier05,ollivier06}.
Furthermore, they provided a maximum likelihood error estimation
algorithm. They showed, as an example, a code of rate $k/n=1/5$
that can correct only one error.

\item
Almedia and Palazzo constructed a concatenated convolutional
code of rate $1/4$ with memory $m=3$; i.e. a ((4,1,3)) code as
shown in \cite{almeida04}. Their construction is valid  only a
specific code parameter. It would be interesting if their work
can be generalized, if possible, to any two arbitrary
concatenated codes.

\item
Kong and Parhi constructed quantum convolutional codes with rates $1/(n+1)$ and
$1/n$ from a classical convolutional codes with rates $1/n$ and $1/(n-1)$, see~
\cite{kong05thesis,kong04}. Their work was not a general approach  for any
quantum convolutional codes, with arbitrary rate $k/n$ and $k>1$.

\item
Forney and Guha constructed quantum convolutional codes with rate $1/3$
\cite{forney05a}. Also, together with Grassl, they derived rate $(n-2)/n$ {
\qccs } \cite{forney05b}. They gave tables of optimal rate $1/3$ { \qccs } and
they also constructed good quantum block codes obtained by tail-biting
convolutional codes.

\item
Grassl and R{\"{o}}tteler constructed quantum convolutional codes from product
codes. They showed that starting with an arbitrary convolutional code and a
self-orthogonal block code, a quantum convolutional code can be
constructed~\cite{grassl05}.

\item Recently, Grassl and R{\"{o}}tteler~\cite{grassl06b} gave a
general algorithm to construct quantum circuits for non-catastrophic encoders
and encoder inverses for channels with memories. Unfortunately, the encoder
they derived is for a subcode of the original code.
\end{itemize}

It is apparent from the discussion above that several issues need to
be addressed regarding the efficiency of the decoding algorithms and
encoding circuits for quantum convolutional codes.   Somewhat
surprisingly there has been no work done on the bounds of quantum
convolutional codes. In this chapter we address this problem
partially by giving a bound for a class of QCC. This bound is
somewhat similar to the generalized Singleton bound for classical
convolutional codes.

\medskip

 \noindent \textbf{Motivation}
 In this chapter we  give a straightforward  extension of the theory of quantum
convolutional codes to nonbinary alphabets.  We give analytical constructions
for quantum convolutional codes unlike the previous work where most of the
codes were constructed by either heuristics or  computer search. In many cases,
we give the exact free
 distance of the quantum convolutional codes. The main contributions of our work are that we:
\begin{itemize}

\item
establish bounds on a class of quantum convolutional codes  similar to
generalized Singleton bound for classical convolutional codes.

\item
provide the necessary definitions and terminology of stabilizer formalization
of convolutional codes, free distance, error bases.

\item
construct families of quantum convolutional codes based on
classical block codes --  such as Reed-Solomon (RS), BCH, and
Reed-Muller codes.

\end{itemize}

\section{Background on Convolutional Codes}\label{sec:CC}

\subsection{Overview} Classical convolutional codes  appeared in a series of
seminal papers in the seventies of the last century. The algebraic structure of
these codes was initiated by Forney~\cite{forney70,forney70b} and
Justesen~\cite{massey73}. Cyclic convolutional codes were first introduced by
Piret~\cite{piret88,piret76,piret75} and generalized by Roos~\cite{roos79}.
Using this construction, one family of cyclic convolutional codes based on
Reed-Solomon codes was derived~\cite{piret88}. It was shown that any
convolutional code has a canonical direct decomposition into subcodes; and
hence it has a minimal encoder.

The subject became active, once again, by a series of recent papers by
Gluesing-Luerssen al et. in~\cite{gluesing03,gluesing04,gluesing06} and by
Rosenthal~\cite{rosenthal99}. Cyclic convolutional codes are defined as left
principle ideals in a skew-polynomial ring. Also, a subclass of cyclic
convolutional codes is described where the units of the skew polynomial ring is
used.

Unit memory convolutional codes are an important class of codes that is
appeared in a paper by Lee~\cite{lee76}. He also showed that these codes have
large free distance $d_f$ among other codes (multi-memory) with the same rate.
Upper and lower bounds on the free distance of unit memory codes were derived
by Thommesen and Justesen~\cite{thommesen83}, confirming superiority of these
codes in comparison to other convolutional codes. Since then, there were some
attempts to construct unit memory codes by using computer search and by
puncturing existing convolutional codes. For an algebraic method to construct
unit memory convolutional codes, classes of these codes were derived by Piret
based on RS codes~\cite{piret88} and by Hole based on BCH codes~\cite{hole00}.
 Also, a class of unit memory codes defined using circulant sub-matrices was
derived by Justesen et. al~\cite{justesen90}.

Bounds on convolutional codes have been studies as well. Rosenthal al et.
showed a generalized Singleton bound and MDS convolutional
codes~\cite{rosental01, rosenthal99b,rosenthal99}.

\subsection{Algebraic Structure of Convolutional Codes}
We give some background concerning classical convolutional codes,
following \cite[Chapter 14]{huffman03} and \cite{lally06}.

Let $\F_q$ denote a finite field with $q$ elements.  An $(n,k,\delta)_q$
\textit{convolutional code} $C$ is a submodule of $\F_q[D]^n$ generated by a
right-invertible matrix $G(D)=(g_{ij})\in \F_q[D]^{k\times n}$,
\begin{eqnarray}\label{eq:cc-def}
C =\{  \textbf{u}(D) G(D) \mid \mathbf{u}(D) \in \F_q[D]^k\},
\end{eqnarray}
such that $\sum_{i=1}^k \nu_i = \max\{ \deg \gamma\,|\, \gamma
\text{ is a $k$-minor of $G(D)$}\}$ $=:\delta$,\\ \noindent where
$\nu_i = \max_{1\leq j\leq n} \{\deg g_{ij} \}.$ We say $\delta$ is
the \textit{degree} of $C$. The \textit{memory} $\mu$ of $G(D)$ is
defined as $\mu=\max_{1\le i\le k} \nu_i$.  The \textit{weight}
$\wt(v(D))$ of a polynomial $v(D)$ in $\F_q[D]$ is defined as the
number of nonzero coefficients of $v(D)$, and the \textit{weight} of
an element $\mathbf{u}(D)\in \F_q[D]^n$ is defined as
$\wt(\mathbf{u}(D))=\sum_{i=1}^n \wt(u_i(D))$.  The \textit{free
distance} $d_f$ of $C$ is defined as $d_f =\wt(C)=\min \{ \wt
(u)\mid u \in C, u\neq 0 \}.$ We say that an $(n,k,\delta)_q$
convolutional code with memory $\mu$ and free distance $d_f$ is an
$(n,k,\delta;\mu,d_f)_q$ convolutional code.

Let $\NN$ denote the set of nonnegative integers. Let
\begin{eqnarray}\Gamma_q= \{ v\colon \NN\rightarrow \F_q\,|\, \text{ all but
finitely many coefficients of $v$ are 0}\}.\end{eqnarray} We can view  $v\in
\Gamma_q$ as a sequence $\{v_i=v(i)\}_{i\geq 0}$ of finite support.
We define a vector space isomorphism $\sigma\colon
\F_q[D]^n\rightarrow \Gamma_q$ that maps an element
$\mathbf{u}(D)=(u_1(D),\ldots,u_n(D))$ in $\F_q[D]^n$ to the
coefficient sequence of the polynomial $\sum_{i=0}^{n-1} D^i
u_i(D^n)$, that is, an element in $\F_q[D]^n$ is mapped to its
interleaved coefficient sequence. Frequently, we will refer to the
image $\sigma(C) =\{\sigma(c)\mid c\in C \}$ of a convolutional code
(\ref{eq:cc-def}) again as $C$, as it will be clear from the context
whether we discuss the sequence or polynomial form of the code. Let
$G(D)= G_0 + G_1 D +\cdots + G_\mu D^\mu$, where $G_i\in
\F_q^{k\times n}$ for $0\le i\le \mu$. We can associate to the
generator matrix $G(D)$ its semi-infinite coefficient matrix

\begin{eqnarray}\label{eq:Gmatrix}
G= \begin{pmatrix}
       G_0 & G_1 & \cdots & G_\mu & & \\
        & G_0 & G_1 & \cdots & G_\mu &  \\
        &  & \ddots  & \ddots &  & \ddots \\
     \end{pmatrix}.
\end{eqnarray}

If $G(D)$ is the generator matrix of a convolutional code $C$, then one easily
checks that $\sigma(C)=\Gamma_q G$.

In the literature, convolutional codes are often defined in the form $\{
p(D)G'(D) \mid p(D)\in \F_{q}(D)^k\}$, where $G'(D)$ is a matrix of full rank
in $\F_q^{k\times n}[D]$. In this case, one can obtain a generator matrix
$G(D)$ in our sense by multiplying $G'(D)$ from the left with a suitable
invertible  matrix $U(D)$ in $\F_q^{k\times k}(D)$, see~\cite{huffman03}.

\medskip

\noindent \textbf{Euclidean and Hermitian Inner Products.} We define
the \textit{Euclidean inner product} of two sequences $u$ and $v$ in
$\Gamma_q$ by $ \scal{u}{v} = \sum_{i\in \NN} u_iv_i$, and the
Euclidean dual of a convolutional code $C\subseteq \Gamma_q$ by
$C^\perp=\{ u\in \Gamma_q\,|\, \scal{u}{v}=0 \text{ for all } v\in
C\}$. A convolutional code $C$ is called self-orthogonal if and only
if $C\subseteq C^\perp$. It is easy to see that a convolutional code
$C$ is self-orthogonal if and only if $GG^T=0$.

Consider the finite field $\F_{q^2}$. The
 \textit{Hermitian inner product} of two
sequences $u$ and $v$ in $\Gamma_{q^2}$ is defined as $ \scal{u}{v}_h =
\sum_{i\in \NN} u_i\, v_i^q.$ We have $C^\hdual = \{ u\in
\Gamma_{q^2}\,|\, \scal{u}{v}_h=0 \text{ for all } v\in C\}$. 
Then, $C\subseteq C^\hdual$ if and only if $GG^\dagger=0$, where the Hermitian
transpose $\dagger$ is defined as $(a_{ij})^\dagger = (a_{ji}^q)$.

\medskip

\noindent \textbf{Delay Operator.} We can define the delay operator
as a shift operator in the codeword to the left or  right. Let
$g_i(D)$ be a row in the infinite generator polynomial $G(D)$, the
right $j-th$ shift is given by \begin{eqnarray}D^j
g_i(D)=g_{i+j}(D).\end{eqnarray}

\subsection{Duals of Convolutional Codes}
The dual  of a convolutional code plays an important role in constructing
quantum convolutional codes. Therefore, we first introduce the dual of a
convolutional code. We can define the inner product between two sequences
$\mathbf{v}$ and $\mathbf{w}$ as
\begin{eqnarray} \langle \mathbf{v}|\mathbf{w}\rangle = \sum_{i
\in \Z}\langle \mathbf{v}_i |\mathbf{w}_i \rangle. \end{eqnarray} Recall that
every codeword in $C$ is equivalent to a sequence. The dual convolutional code
$C^\perp$ is the set of all sequences that are orthogonal to every sequence
$\mathbf{v}$ in $C$.
\begin{lemma}[Dual of Convolutional Code]\label{lem:dualG}
Let $k/n$ be the rate of a convolutional code $C$ generated by a semi-infinite
generator matrix $G$. Also, let  $(n-k)/n$ be the rate of dual of a
convolutional code $C^\perp$ generated by the semi-infinite generator matrix
$G^\perp$, such that
\begin{eqnarray}
G= \begin{pmatrix}
       G_0 & G_1 & \cdots & G_m & & \\
        & G_0 & G_1 & \cdots & G_m &  \\
        &  & \ddots  & \ddots &  & \ddots \\
     \end{pmatrix}
     \nonumber
     \end{eqnarray}
     and
  \begin{eqnarray}
     G^\perp= \begin{pmatrix}
       G_0^\perp & G_1^\perp & \cdots & G_{m^\perp}^\perp & & \\
        & G_0^\perp & G_1^\perp & \cdots & G_{m^\perp}^\perp &  \\
        &  & \ddots  & \ddots &  & \ddots \\
     \end{pmatrix}
\end{eqnarray}
where $G_i$ are $k \times n$ matrices, for all $0 \leq i \leq m$. Then $G
(G^{\perp})^ T=0$.
\end{lemma}
\begin{proof}
see \cite[Theorem 2.63]{johannesson99}.
\end{proof}
A  convolutional code $C$ is said to be self-orthogonal if $C \subseteq
C^\perp$. Clearly, a convolutional code is self-orthogonal if and only if $G
G^T=0$. We can also define a relation between the polynomial generators
matrices $G(D)$ and $G^\perp(D)$. If $G_{r}^\perp (D)= G_{m^\perp}^\perp +
G_{m^\perp-1}^\perp D+\cdots + G_1^\perp D^{m^\perp-1} +G_0^\perp D^{m^\perp}$,
then $ G(D) (G_r^{\perp}(D) )^T=0$ (see \cite[Theorem 2.64]{johannesson99}).
The following Lemma gives the relation between the total constraint lengths of
a code and its dual code.

\begin{lemma}
The convolutional code $C$ is self-orthogonal if and only if
\begin{eqnarray}
G(D) G(D^{-1})^T=0
\end{eqnarray}
\end{lemma}
\begin{proof}

Let the polynomial $G(D)=G_0+G_1D+\ldots+G_mD^m$ and its dual
polynomial
$G^\perp(D)=G^\perp_0+G^\perp_1D+\ldots+G^\perp_{m^\perp}D^{m^\perp}$
be the polynomial generator matrices of $C$ and its dual,
respectively. We know that $G(D)G_r^\perp(D)^T=0$. But,
\begin{eqnarray}
G_r^\perp (D)&=& G_{m^\perp}^\perp + G_{m^\perp-1}^\perp D+\cdots + G_1^\perp
D^{m^\perp-1} +G_0^\perp D^{m^\perp} \nonumber \\ &=& \big(G_{m^\perp}^\perp
D^{-m^\perp} + G_{m^\perp-1}^\perp D^{1-m^\perp}+\cdots + G_1^\perp D^{-1}
+G_0^\perp \big) D^{m^\perp} \nonumber \\ &=& G^\perp(D^{-1}) D^{m^\perp}.
\end{eqnarray}
Therefore, $G(D)G_r^\perp(D)^T = G(D)  G^\perp(D^{-1})^T D^{m^\perp}=0$. So,
 $G(D)  G^\perp(D^{-1})^T=0$.  Let $C \leq C^\perp$ be a self-orthogonal convolutional code, we know that the
elements of $G(D)$ can be generated from the elements of $G^\perp (D)$.
 Since, $G(D)  G^\perp(D^{-1})^T=0$, it follows that $G(D)  G(D^{-1})^T=0$.

Conversely, if $G(D)G(D^{-1})^T=0$, then it implies that the convolutional code
generated by $G(D)$ must be a subcode of $G^\perp(D)$. Therefore, $C$ must be a
self-orthogonal convolutional code.
\end{proof}
We can also formulate the above condition in a slightly different manner as
follows. Let $G(D)= [g_{ij}(D)   ]$. Then $G(D) G(D^{-1})^T=\sum_{l=1}^n
g_{il}(D) g_{jl}(D^{-1})$. So, for a self-orthogonal code $\sum_{l=1}^n
g_{il}(D) g_{jl}(D^{-1})=0$, for all $1 \leq i, j \leq k$.  Alternatively,  if
\begin{eqnarray}G(D)= [\textbf{g}_1(D), \textbf{g}_2(D), \ldots, \textbf{g}_k(D) ]^T ,\end{eqnarray} where $\mathbf{g}_i(D)=[g_{i1}(D), g_{i2}(D),\ldots, g_{in}(D)] $,  then
\begin{eqnarray}G(D) G(D^{-1})^T=[g_{i}(D) g_j(D^{-1})^T]=0,\end{eqnarray}
i.e. $g_i(D) g_j(D^{-1})^T=0$
\noindent \textbf{Cross-Correlation.} It is also possible to derive
these conditions in terms of the cross-correlations between
codewords of a convolutional code as in \cite{forney05b}.
Let us define the Euclidean inner product between two (Laurent) series
$g(D)=\sum_{i \in \Z}g_iD^i$ and $h(D)=\sum_{i \in \Z}h_iD^i$ for $g_i,h_i \in
\F_q$ as
\begin{eqnarray}\langle g(D)|h(D)\rangle =\sum_{i \in \Z} g_ih_i.\end{eqnarray} If
the series are over $F_{q^2}$, we can define their Hermitian inner product as
\begin{eqnarray}\langle g(D)|h(D)\rangle_h=\sum_{i \in \Z}
g_i^qh_i.\end{eqnarray}

If $\textbf{v}(D)$ is equal to $[v_1(D),v_1(D),\dots,v_{n}(D)  \mid  v_i(D) \in
\F_q((D))]$ then we can define the Euclidean inner product with
$\mathbf{w}(D)=[w_1(D),w_1(D),\dots,w_{n}(D)]$ as
 \begin{eqnarray}\langle
\mathbf{v}(D)|\mathbf{w}(D) \rangle = \sum_{i=1}^n \langle v_i(D)
|w_i(D)\rangle.\end{eqnarray} Let us define the conjugate of $g(D) \in
\F_{q^2}((D)) $ as $g^\dagger(D)=\sum_{i \in \Z}g_i^q D^i$. Then, we can also
define the Hermitian inner product of $\textbf{v}(D)$ and $\textbf{w}(D)$ as

 \begin{eqnarray}\langle
\mathbf{v}(D)|\mathbf{w}(D) \rangle_h = \sum_{i=1}^n \langle v_i(D)
|w_i(D)\rangle_h=\sum_{i=1}^n \langle v_i(D)
|w_i^\dagger(D)\rangle.\end{eqnarray}

 Now, we define the cross-correlation between the sequences $\textbf{v}(D)$
and $\mathbf{w}(D)$ as
\begin{eqnarray}\label{Eq:cross_correlation}
R_{\mathbf{vw}}(D) &=& \sum_{i \in \Z} \langle
\mathbf{v}(D)|D^i\mathbf{w}(D)\rangle D^i=\sum_{i \in \Z} R_{\mathbf{vw},i}D^i.
\end{eqnarray}

If   $C$ is self-orthogonal, then $R_{\mathbf{vw}}(D)=0$ for any
$\mathbf{v}(D),\mathbf{w}(D)  \in C$.

\begin{lemma}
$R_{\mathbf{vw}}(D)=  \mathbf{v}(D) \mathbf{w}(D^{-1})^T$
\end{lemma}
\begin{proof}
The proof is a direct consequence from definition of  $R_{\mathbf{vw}}(D)$,
Equation (\ref{Eq:cross_correlation}).

\begin{eqnarray}
R_{\mathbf{vw}}(D)&=& \sum_{i \in \Z} \langle
\mathbf{v}(D)|D^i\mathbf{w}(D)\rangle D^i \nonumber \\ &=& \sum_{i \in \Z}
\sum_{j=1}^{n}\langle \mathbf{v}_j(D)|D^i\mathbf{w}_j(D)\rangle D^i \nonumber
\\&=&  \sum_{i \in \Z}\sum_{j=1}^{n}\mathbf{v}_j\mathbf{w}_{j-i}D^i
 = \sum_{i \in \Z}\sum_{j=1}^{n}\mathbf{v}_j D^j D^{-j}\mathbf{w}_{j-i}D^i
  \nonumber \\&=& \sum_{j=1}^{n}\mathbf{v}_j D^j \sum_{i \in \Z}
D^{-j}\mathbf{w}_{j-i}D^i  =  \sum_{j=1}^{n}\mathbf{v}_j D^j
\sum_{i \in \Z} \mathbf{w}_{j-i}D^{-(j-i)} \nonumber \\&=&
  \mathbf{v}(D) \mathbf{w}(D^{-1})^T
\end{eqnarray}
\end{proof}

If $\mathbf{v}(D)$ is orthogonal to $\mathbf{w}(D)$, then
$R_{\mathbf{vw}}(D)=0$. We can also define the cross-correlation with respect
to the Hermitian inner product as

\begin{eqnarray}
R_{\mathbf{vw}}^h(D) &=& \sum_{i \in \Z} \langle
\mathbf{v}(D)|D^i\mathbf{w}(D)\rangle_h D^i=\sum_{i \in \Z}
R_{\mathbf{vw},i}^hD^i, \nonumber \\ &= & \mathbf{v}(D)
\mathbf{w}^\dagger(D^{-1}).
\end{eqnarray}
 If   a code $C$ is Hermitian self-orthogonal, then $R_{\mathbf{vw}}^h(D)=0$
for any $\mathbf{v}(D),\mathbf{w}(D)  \in C$.

\begin{lemma}
Let $G(D)$ be  a minimal  encoder of a convolutional code $C$ with total
constraint length $\delta$. Then the dual encoder $G^\perp(D)$ of $C^\perp$ has
also a total constraint equals to $\delta$
\end{lemma}
\begin{proof}
See for example \cite[Theorem 7]{forney70}
\end{proof}


\medskip


\section{Quantum Convolutional Codes}\label{sec:QCCparameters}
The state space of a $q$-ary quantum digit is given by the complex
vector space $\C^q$. Let $\mbox{$\{ \ket{x}\,|\, x\in \F_q\}$}$
denote a fixed orthonormal basis of $\C^q$, called the computational
basis. For $a,b\in \F_q$, we define the unitary operators
\begin{eqnarray} X(a)\ket{x}=\ket{x+a}\;\; \text{and}\;\;
Z(b)\ket{x}=\exp(2\pi i \tr(bx)/p)\ket{x},\end{eqnarray} where the addition is
in $\F_q$, $p$ is the characteristic of $\F_q$, and
$\tr(x)=x^p+x^{p^2}+\cdots+x^q$ is the absolute trace from $\F_q$ to
$\F_p$. The set $\mathcal{E}=\{X(a),Z(b)\,|\, a,b\in\F_q\}$ is a
basis of the algebra of $q\times q$ matrices, called the
\textit{error basis}.

A quantum convolutional code  encodes a stream of quantum digits. One does not
know in advance how many qudits {\em i.e.}, quantum digits will be sent, so the
idea is to impose structure on the code that simplifies online encoding and
decoding. Let $n$, $m$ be positive integers. We will process $n+m$ qudits at a
time, $m$ qudits will overlap from one step to the next, and $n$ qudits will be
output.

For each $t$ in $\NN$, we define the Pauli group $P_t=\langle M |
M\in \mathcal{E}^{\otimes (t+1)n+m}\rangle$ as the group generated
by the \mbox{$(t+1)n+m$}-fold tensor product of the error
basis~$\mathcal{E}$. Let $I=X(0)$ be the $q\times q$ identity
matrix. For $i,j\in \NN$ and $i\le j$, we define the inclusion
homomorphism $\iota_{ij}\colon P_i\rightarrow P_j$ by
$\iota_{ij}(M)=M\otimes I^{\otimes n(j-i)}$. We have
$\iota_{ii}(M)=M$ and $\iota_{ik}=\iota_{jk}\circ \iota_{ij}$ for
$i\le j\le k$. Therefore, there exists a group \begin{eqnarray}
P_\infty = \dirlimit (P_i,\iota_{ij}),\end{eqnarray} called the
direct limit of the groups $P_i$ over the totally ordered set
$(\NN,\le)$. For each nonnegative integer $i$, there exists a
homomorphism $\iota_i\colon P_i\rightarrow P_\infty$ given by
$\iota_i(M_i)=M_i\otimes I^{\otimes \infty}$ for $M_i\in P_i$, and
$\iota_i =\iota_j\circ \iota_{ij}$ holds for all $i\le j$. We have
$P_\infty = \bigcup_{i=0}^\infty \iota_i(P_i)$; put differently,
$P_\infty$ consists of all infinite tensor products of matrices in
$\langle M\,|\,M\in \mathcal{E}\rangle$ such that all but finitely
many tensor components are equal to $I$. The direct limit structure
that we introduce here provides the proper conceptual framework for
the definition of convolutional stabilizer codes; see~\cite{ribes00}
for background on direct limits.

\begin{small}
\begin{eqnarray*}
S&=& \left(\begin{array}{lccc}
 \overbrace{\hskip 0.5in}^{n}  \overbrace{ \hspace{0.35in}}^{m}  &&\\
 \mbox{\begin{tabular}{|ccccc|} \hline
&&&&\\
&  &M& &\\
&&&&\\
\hline
\end{tabular}}&&\\
\left. \mbox{\hskip 0.5in\begin{tabular}{|ccccc|} \hline
&&&&\\
& & M&  &\\
&&&&\\
\hline
\end{tabular}}\hskip 0.1in \right\}  n-k& &\\
&&\\
 &\ddots& \\ &\mbox{t times} & \\\end{array}\right)
\end{eqnarray*}
\end{small}

We will define the stabilizer of the quantum convolutional code also through a
direct limit. Let $S_0$ be an abelian subgroup of $P_0$. For positive integers
$t$, we recursively define a subgroup $S_t$ of $P_t$ by $S_t=\langle N\otimes
I^{\otimes n}, I^{\otimes tn}\otimes M\,|\, N\in S_{t-1}, M\in S_0\rangle.$ Let
$Z_t$ denote the center of the group $P_t$. We will assume that
\begin{compactenum}[\bf S1)]
\item $I^{\otimes
tn} \otimes M$ and $N\otimes I^{\otimes tn}$ commute for all $N,M\in S_0$ and
all positive integers $t$.
\item $S_tZ_t/Z_t$ is an
$(t+1)(n-k)$-dimensional vector space over $\F_q$.
\item $S_t\cap Z_t$ contains only the identity
matrix.
\end{compactenum} Assumption \textbf{S1} ensures that $S_t$ is
an \textit{abelian} subgroup of~$P_t$, \textbf{S2} implies that $S_t$ is
generated by $t+1$ shifted versions of $n-k$ generators of $S_0$ and all these
$(t+1)(n-k)$ generators are independent, and \textbf{S3} ensures that the
stabilizer (or $+1$ eigenspace) of $S_t$ is nontrivial as long as $k<n$.

The abelian subgroups $S_t$ of $P_t$ define an abelian group
\begin{eqnarray} S=\dirlimit (S_i,\iota_{ij})= \langle \iota_t(I^{\otimes
tn}\otimes M)\,|\, t\ge 0, M\in S_0\rangle\end{eqnarray} generated by shifted
versions of elements in $S_0$.
\begin{definition}
Suppose that an abelian subgroup $S_0$ of $P_0$ is chosen such that
\textbf{S1}, \textbf{S2}, and \textbf{S3} are satisfied. Then the
$+1$-eigenspace of $S=\displaystyle\lim_{\longrightarrow} (S_i,\iota_{ij})$ in
$\bigotimes_{i=0}^\infty \C^q$ defines a convolutional stabilizer code with
parameters $[(n,k,m)]_q$.
\end{definition}
\smallskip

In practice, one works with a stabilizer $S_t$ for some large (but previously
unknown) $t$, rather than with $S$ itself. We notice that the rate $k/n$ of the
quantum convolutional stabilizer code defined by $S$ is approached by the rate
of the stabilizer block code $S_t$ for large $t$. Indeed, $S_t$ defines a
stabilizer code with parameters $[[(t+1)n+m,(t+1)k+m]]_q$; therefore, the rates
of these stabilizer block codes approach
\begin{eqnarray}
\lim_{t\rightarrow \infty} \frac{(t+1)k+m}{(t+1)n+m} = \lim_{t\rightarrow
\infty} \frac{k+m/(t+1)}{n+m/(t+1)} = \frac{k}{n}.\end{eqnarray}

We say that an error $E$ in $P_\infty$ is \textit{detectable} by a
convolutional stabilizer code with stabilizer $S$ if and only if a scalar
multiple of $E$ is contained in $S$ or if $E$ does not commute with some
element in $S$. The \textit{weight} $\wt$ of an element in $P_\infty$ is
defined as its number of non-identity tensor components. A quantum
convolutional stabilizer code is said to have \textit{free distance} $d_f$ if
and only if it can detect all errors of weight less than $d_f$, but cannot
detect some error of weight $d_f$.  Denote by $Z(P_\infty)$ the center of
$P_\infty$ and by $C_{P_\infty}(S)$ the centralizer of $S$ in $P_\infty$.  Then
the free distance is given by $d_f = \min \{\wt(e)\mid e\in
C_{P_\infty}(S)\setminus Z(P_\infty)S\}$.

Let $(\beta,\beta^q)$ denote a normal basis of $\F_{q^2}/\F_q$. Define a map
$\tau\colon P_\infty\rightarrow \Gamma_{q^2}$ by
$\tau(\omega^cX(a_0)Z(b_0)\otimes X(a_1)Z(b_1)\otimes \cdots)= (\beta
a_0+\beta^q b_0,\beta a_1+\beta^q b_1,\dots)$. For sequences $v$ and $w$ in
$\Gamma_{q^2}$, we define a trace-alternating form
\begin{eqnarray}
\scal{v}{w}_a = \tr_{q/p}\left(\frac{v\cdot w^q - v^q\cdot
w}{\beta^{2q}-\beta^2}\right).
\end{eqnarray}

\begin{lemma}\label{lm:commute}
Let $A$ and $B$ be elements of $P_\infty$. Then $A$ and $B$ commute if and only
if $\scal{\tau(A)}{\tau(B)}_a=0$.
\end{lemma}
\begin{proof}
This follows from \cite{ketkar06} and the direct limit structure.
\end{proof}

\begin{lemma}\label{lm:q2c}
Let $Q$ be an $\F_{q^2}$-linear $[(n,k,m)]_q$ quantum convolutional code with
stabilizer $S$, where $S=\dirlimit (S_i,\iota_{ij})$ and $S_0$ an abelian
subgroup of $P_0$ such that \textbf{S1}, \textbf{S2}, and \textbf{S3} hold.
Then $C=\sigma^{-1}\tau(S)$ is an $\F_{q^2}$-linear $(n,(n-k)/2;\mu\leq
\ceil{m/n})_{q^2}$ convolutional code generated by $\sigma^{-1}\tau(S_0)$.
Further, $C\subseteq C^\hdual$.
\end{lemma}
\begin{proof}
Recall that $\sigma : \F_{q^2}[D]^n \rightarrow \Gamma_{q^2}$, maps $u(D)$ in
$\F_{q^2}[D]^n$ to $\sum_{i=0}^{n-1}D^iu_i(D^n)$. It is invertible, thus
 $\sigma^{-1}\tau(e) =\sigma^{-1}\circ \tau(e) $ is
well defined for any $e$ in $P_\infty$. Since $S$ is generated by shifted
versions of $S_0$, it follows that $C=\sigma^{-1}\tau(S)$ is generated as the
$\F_{q^2}$ span of $\sigma^{-1}\tau(S_0)$ and its shifts, {\em i.e.},
$D^l\sigma^{-1}\tau(S_0)$, where $l\in \N$. Since $Q$ is an $\F_{q^2}$-linear
$[(n,k,m)]_q$ quantum convolutional code, $S_0$ defines an $[[n+m,k+m]]_q$
stabilizer code with $(n-k)/2$ $\F_{q^2}$-linear generators. Since the maps
$\sigma$ and $\tau$ are linear $\sigma^{-1}\tau(S_0)$ is also
$\F_{q^2}$-linear. As $\sigma^{-1}\tau(e)$ is in $\F_{q^2}[D]^n$ we can define
an $(n-k)/2\times n $ polynomial generator matrix that generates $C$. This
generator matrix need not be right invertible, but we know that there exists a
right invertible polynomial generator matrix that generates this code. Thus $C$
is an $(n,(n-k)/2;\mu)_{q^2}$ code. Since $S$ is abelian,
Lemma~\ref{lm:commute} and the $\F_{q^2}$-linearity of $S$ imply that
$C\subseteq C^\hdual$. Finally, observe that maximum degree of an element in
$\sigma^{-1}\tau(S_0)$ is $\ceil{m/n}$ owing to  $\sigma$. Together with
\cite[Lemma~14.3.8]{huffman03} this  implies that  the memory of
$\sigma^{-1}\tau(S)$ must be $\mu \leq \ceil{m/n}$.
\end{proof}

\section{CSS Code Constructions}\label{sec:CS_CSS}
We define the degree of an $\F_{q^2}$-linear $[(n,k,m)]_q$ quantum
convolutional code $Q$ with stabilizer $S$ as the degree of the classical
convolutional code $\sigma^{-1}\tau(S)$. It is possible to define the degree of
the quantum convolutional code purely in terms of the stabilizer too, but such
a definition is somewhat convoluted.  We denote an $[(n,k,m)]_q$ quantum
convolutional code with free distance $d_f$ and total constraint length
$\delta$ as $[(n,k,m;\delta,d_f)]_q$. It must be pointed out this notation is
at variance with the classical codes in not just the order but the meaning of
the parameters.

\begin{corollary}\label{co:q2c}
An $\F_{q^2}$-linear $[(n,k,m;\delta,d_f)]_q$ convolutional stabilizer code
implies the existence of an $(n,(n-k)/2;\delta)_{q^2}$ convolutional code $C$
such that $d_f=\wt(C^\hdual \setminus C)$.
\end{corollary}
\begin{proof}
As before let $C=\sigma^{-1}\tau(S)$, by Lemma~\ref{lm:commute} we can conclude
that $\sigma^{-1}\tau(C_{P_\infty}(S)) \subseteq C^\hdual$. Thus an
undetectable error is mapped to an element in $C^\hdual \setminus C$. While
$\tau$ is injective on $S$ it is not the case with $C_{P_\infty}(S)$.  However
we can see that if $c$ is in $C^\hdual \setminus C$, then surjectivity of
$\tau$ (on $C_{P_\infty}(S)$) implies that there exists an error $e$ in
$C_{P_\infty}(S)\setminus Z(P_\infty)S$ such that $\tau(e)=\sigma(c)$. As
$\tau$ and $\sigma$ are isometric $e$ is an undetectable error with $\wt(c)$.
Hence, we can conclude that $d_f=\wt(C^\hdual\setminus C)$. Combining with
Lemma~\ref{lm:q2c} we have the claim stated.
\end{proof}

An $[(n,k,m;\delta,d_f)]_q$ code is said to be a \textit{pure code} if there
are no errors of weight less than $d_f$ in the stabilizer of the code.
Corollary~\ref{co:q2c} implies that $d_f=\wt(C^\hdual\setminus
C)=\wt(C^\hdual)$.

\begin{theorem}\label{th:c2qHerm}
Let $C$ be $(n,(n-k)/2,\delta;\mu)_{q^2}$ convolutional code such that
$C\subseteq C^\hdual$. Then there exists an $[(n,k,n\mu;\delta,d_f)]_q$
convolutional stabilizer code, where $d_f=\wt(C^\hdual\setminus C)$. The code
is pure if $d_f=\wt(C^\hdual)$.
\end{theorem}
\begin{proof}[Sketch]
Let $G(D)$ be the polynomial generator matrix of $C$, with the semi-infinite
generator matrix $G$
 defined as in equation~(\ref{eq:Gmatrix}).
Let $C_t=\langle \sigma(G(D)),\ldots,\sigma(D^tG(D)) \rangle = \langle C_{t-1},
\sigma(D^t G(D)) \rangle$, where $\sigma$ is applied to every row in $G(D)$.
The self-orthogonality of $C$ implies that $C_t$ is also self-orthogonal. In
particular $C_0$ defines an $[n+n\mu,(n-k)/2]_{q^2}$ self-orthogonal code. From
the theory of stabilizer codes we know that there exists an abelian subgroup
$S_0\le P_0$ such that $\tau(S_0)=C_0$, where $P_t$ is the Pauli group over
$(t+1)n+m$ qudits; in this case $m=n\mu$. This implies that $\tau(I^{\otimes nt
}\otimes S_0) =\sigma(D^tG(D))$. Define $S_t =\langle S_{t-1}, I^{\otimes nt
}\otimes S_0\rangle$, then $\tau(S_t)=\langle\tau(S_{t-1},\sigma(D^t G(D))
\rangle$. Proceeding recursively, we see that $\tau(S_t)= \langle
\sigma(G(D)),\ldots, \sigma(D^tG(D)) \rangle=C_t$. By Lemma~\ref{lm:commute},
the self-orthogonality of $C_t$ implies that $S_t$ is abelian, thus \textbf{S1}
holds. Note that $\tau(S_tZ_t/Z_t)=C_t$, where $Z_t$ is the center of $P_{t}$.
Combining this with $\F_{q^2}$-linearity of $C_t$ implies that $S_tZ_t/Z_t$ is
a $(t+1)(n-k)$ dimensional vector space over $F_q$; hence \textbf{S2} holds.
For \textbf{S3}, assume that $z\neq \{1\}$ is in $S_t\cap Z_t$. Then $z$ can be
expressed as a linear combination of the generators of $S_t$. But $\tau(z)=0$
implying that the generators of $S_t$ are dependent. Thus $S_t\cap Z_t=\{ 1\}$
and \textbf{S3} also holds. Thus $S=\dirlimit (S_t,\iota_{tj} )$ defines an
$[(n,k,n\mu;\delta)]_q$ convolutional stabilizer code. By definition the degree
of the quantum code is the degree of the underlying classical code. As
$\sigma^{-1}\tau(S)=C$, arguing as in Corollary~\ref{co:q2c} we can show that
$\sigma^{-1}\tau(C_{P_\infty}(S))=C^\hdual$ and $d_f=\wt(C^\hdual\setminus C)$.
\end{proof}

\begin{corollary}\label{co:c2qEuclid}
Let $C$ be an $(n,(n-k)/2,\delta;\mu)_q$ code such that $C\subseteq C^\perp$.
Then there exists an $[(n,k,n\mu;\delta,d_f)]_q$ code with
$d_f=\wt(C^\perp\setminus C)$. It is pure if $\wt(C^\perp\setminus
C)=\wt(C^\perp)$.
\end{corollary}
\begin{proof}
Since $C\subseteq C^\perp$, its generator matrix $G$ as in
equation~(\ref{eq:Gmatrix}) satisfies $GG^T=0$. We can obtain an
$\F_{q^2}$-linear $(n,(n-k)/2,\delta;\mu)_{q^2}$ code, $C'$ from $G$ as
$C'=\Gamma_{q^2}G$. Since $G_i\in \F_q^{(n-k)/2\times n}$ we have
$GG^{\dagger}=GG^T=0$. Thus $C'\subseteq C'^\hdual$. Further, it can checked
that $\wt(C'^\hdual\setminus C')= \wt(C^\perp\setminus C)$. The claim follows
from Theorem~\ref{th:c2qHerm}.
\end{proof}

\section{QCC  Singleton Bound}\label{sec:QCC_singletonbound}

Three  main properties to measure  performance of a quantum convolutional
stabilizer code are code rate, minimum free distance, and complexity of its
encoders (decoders). We study bounds on the minimum free distance of QCC's. All
quantum block codes whether they are pure or impure saturate the quantum
Singleton bound. Also, classical convolutional codes obey modified Singleton
bound. We recall generalized Singleton bound for convolutional codes as shown
in the following Lemma.
\begin{lemma}[Generalized Singleton Bound]\label{th:genSBound}
The  free distance of a $(n,k,m;\delta,d_f)_q$ convolutional code is
upper-bounded by
\begin{eqnarray}
d_f \leq (n-k)\left( \left\lfloor \frac{\delta}{k} \right\rfloor+1
\right)+\delta+1=\mathfrak{B}(n,k,m;\delta).
\end{eqnarray}
\end{lemma}

\begin{proof}
See~\cite[Theorem~2.4]{rosental01}.
\end{proof}
If the free distance of the QCC is same as the free distance of the dual code,
i.e. $C^\perp \backslash C$, then QCC is called pure code. The following Lemma
shows the generalized Singleton bound for pure QCC's.

\begin{theorem}[Singleton bound]\label{th:qsb}
The free distance of an $[(n,k,m;\delta,d_f)]_q$ $\F_{q^2}$-linear pure
convolutional stabilizer code is bounded by
\begin{eqnarray}
d_f&\leq& \frac{n-k}{2}\left ( \left\lfloor \frac{2\delta}{n+k} \right\rfloor+1
\right) + \delta+1
\end{eqnarray}
\end{theorem}
\begin{proof}
By Corollary~\ref{co:q2c}, there exists an $(n,(n-k)/2,\delta)_{q^2}$ code $C$
such that $\wt(C^\hdual \setminus C) =d_f$, and the purity of the code implies
that $\wt(C^\hdual)=d_f$. The dual code $C^\perp$ or $C^\hdual$ has the same
degree as code \cite[Theorem~2.66]{johannesson99}.
 Thus, $C^\hdual$ is an $(n,(n+k)/2,\delta)_{q^2}$ convolutional code with
 free distance $d_f$.  By the generalized Singleton bound
\cite[Theorem~2.4]{rosental01} for classical convolutional codes, we have
\begin{eqnarray*}
d_f &\leq & \left( n - (n+k)/2 \right )\left( \left\lfloor
\frac{\delta}{(n+k)/2} \right\rfloor+1\right) + \delta + 1,
\end{eqnarray*}
which implies the claim.
\end{proof}


\section{QCC Example}
\begin{exampleX}[QCC with rate $1/3$ and single error correction]\label{QCC-example}
Consider the code $C$ generated by $$g_1=\begin{pmatrix} D & 1+D+D^2
& 1+D^2
\end{pmatrix}.$$
and the set of all generators can be given as $\{ D^i g_1(D), i \in \Z \}$. So,
the generator matrix of the code in the infinite form is

\begin{eqnarray}
G=\begin{pmatrix} g_1(x) \\ Dg_1(x) \\.\\.\\.\end{pmatrix}=
\begin{pmatrix}011 & 110 & 011 \\ & 011 &110&011& \\ & & \ddots& \ddots &\ddots \end{pmatrix}
\end{eqnarray}
Now, we can map the generator $G$ to a stabilizer subgroup $S$ with
two generators. The two generators of $S$ have infinite length of
Pauli matrices as
$$(\dots, III, IXX, XXI, IXX, III, \dots)$$ and $$(\cdots, III, IZZ, ZZI, IZZ,
III, \cdots).$$

It is  straight forward to check that $g_1$ is orthogonal to itself
using the cross correlated function. Also, row shifts of the matrix
$G$ are orthogonal to each other. Therefore, the code $C$ is
self-orthogonal, and the dual code $C^\perp$ has rate $2/3$ and
generated by.
\begin{eqnarray*}
H= \begin{pmatrix} D &1+D &1+D \\ 1 &1&1 \end{pmatrix}
\end{eqnarray*}
Also, $C^\perp$ can be mapped to a centralizer subgroup $C(S) \in \G$. One can
check that $C^\perp$ has minimum free distance $d_f=3$. Clearly, the
convolutional code has memory $v=2$, i.e. the max degree of $g_1$.

\end{exampleX}

\chapter{Quantum Convolutional Codes Derived from Reed-Solomon
Codes}\label{ch_QCC_RS} In this chapter I construct quantum
convolutional codes based on generalized Reed-Solomon and
Reed-Muller codes.  The quantum convolutional codes derived from the
generalized Reed-Solomon codes are shown to be optimal in the sense
that they attain the Singleton bound with equality, as shown in
Chapter~\ref{ch_QCC_bounds}.


\section{Convolutional GRS Stabilizer Codes}\label{section:RS}

In this section we will use Piret's construction of Reed-Solomon convolutional
codes \cite{piret88} to derive quantum convolutional codes. Let $\alpha\in
\F_{q^2}$ be a primitive $n$th root of unity, where $n|q^2-1$. Let
$w=(w_0,\ldots,w_{n-1}),\mathbf{\gamma}=(\gamma_0,\ldots,\gamma_{n-1})$ be in
$\F_{q^2}^n$ where  $w_i\neq 0$ and  all $\gamma_i\neq 0$ are distinct. Then
the generalized Reed-Solomon (GRS) code over $\F_{q^2}^n$ is the code with the
parity check matrix, (cf. \cite[pages~175--178]{huffman03})

\begin{eqnarray}
 H_{\gamma,w} =\left[ \begin{array}{llll}
w_0 &w_1&\cdots &w_{n-1}\\
w_0\gamma_0 &w_1\gamma_1  &\cdots &w_{n-1}\gamma_{n-1}\\
\vdots& \vdots  &\ddots &\vdots\\
w_0\gamma_0^{t-1} &w_1\gamma_1^{2(t-1)} &\cdots &w_{n-1}\gamma_{n-1}^{(t-1)(n-1)}\\
\end{array}\right].
\end{eqnarray}

The code is denoted by $\text{GRS}_{n-t}(\gamma,v)$, as its generator matrix is
of the form $H_{\gamma,v}$ for some $v\in \F_{q^2}^n$. It is an
$[n,n-t,t+1]_{q^2}$  MDS code \cite[Theorem~5.3.1]{huffman03}. If we choose
$w_i=\alpha^i$, then $w_i\neq 0$. If $\gcd(n,2)=1$, then $\alpha^2$ is also a
primitive $n$th root of unity; thus $\gamma_i=\alpha^{2i}$ are all distinct and
we have an $[n,n-t,t+1]_{q^2}$ GRS code with parity check matrix $H_0$, where

\begin{eqnarray}
 H_0 =\left[ \begin{array}{ccccc}
1 &\alpha &\alpha^2 &\cdots &\alpha^{n-1}\\
1 &\alpha^3 &\alpha^6 &\cdots &\alpha^{3(n-1)}\\
\vdots& \vdots &\vdots &\ddots &\vdots\\
1 &\alpha^{2t-1} &\alpha^{2(2t-1)} &\cdots &\alpha^{(2t-1)(n-1)}
\end{array}\right].
\end{eqnarray}

Similarly if $w_i=\alpha^{-i}$ and $\gamma_i=\alpha^{-2i}$, then we have
another $[n,n-t,t+1]_{q^2}$ GRS code with parity check matrix

\begin{eqnarray}
 H_1 =\left[ \begin{array}{ccccc}
1 &\alpha^{-1} &\alpha^{-2} &\cdots &\alpha^{-(n-1)}\\
1 &\alpha^{-3} &\alpha^{-6} &\cdots &\alpha^{-3(n-1)}\\
\vdots& \vdots &\vdots &\ddots &\vdots\\
1 &\alpha^{-(2t-1)} &\alpha^{-2(2t-1)} &\cdots &\alpha^{-(2t-1)(n-1)}
\end{array}\right].
\end{eqnarray}

The $[n,n-2t,2t+1]_{q^2}$  GRS code with $w_i=\alpha^{-i(2t-1)}$ and
$\gamma_i=\alpha^{2i}$ has a parity check matrix $H^*$ that is equivalent to
$\left[\begin{smallmatrix} H_0\\H_1\end{smallmatrix}\right]$ up to a
permutation of rows. Let us consider the convolutional code generated by the
generator polynomial matrix $H(D)=H_0+DH_1$, see Equation \ref{eq:h(d)}. The
polynomial generator matrix $H(D)$ can also be converted to a semi-infinite
matrix $H$ that defines the same code.

\begin{tiny}
\begin{figure*}[t]
 H(D) =\begin{eqnarray}\label{eq:h(d)}
\left[ \begin{array}{ccccc}
1+D &\alpha +\alpha^{-1} D&\alpha^2 +\alpha^{-2} D&\cdots &\alpha^{n-1}+\alpha^{(-n-1)}D\\
1 +D&\alpha^3 +\alpha^{-3} D&\alpha^6 +\alpha^{-6} D&\cdots &\alpha^{3(n-1)}+\alpha^{-3(n-1)}D\\
\vdots& \vdots &\vdots &\ddots &\vdots\\
1 +D&\alpha^{\mu-1} +\alpha^{-(\mu-1)}D&\alpha^{2(\mu-1)}+
\alpha^{-2(\mu-1)}D&\cdots &\alpha^{(\mu-1)(n-1)}+\alpha^{-(\mu-1)(n-1)}D
\end{array}\right]
\end{eqnarray}
\end{figure*}
\end{tiny}

Our goal is to show that under certain restrictions on $n$ the following
semi-infinite coefficient matrix $H$ determines an $\F_{q^2}$-linear Hermitian
self-orthogonal convolutional code

\begin{eqnarray}\label{eq:H}
H = \left[ \begin{array}{ccccc}
H_0 &H_1& \zero& \cdots & \cdots\\
\zero&H_0 &H_1& \zero& \cdots \\
\vdots&\vdots&\vdots &\cdots& \ddots
\end{array}\right].
\end{eqnarray}

To show that $H$ is Hermitian self-orthogonal, it is sufficient to show that
$H_0,H_1$ are both self-orthogonal and $H_0$ and $H_1$ are orthogonal to each
other. A portion of this result is contained in \cite[Lemma~8]{grassl04}, {\em
viz.}, $n=q^2-1$. We will prove a slightly stronger result.
 We will show that the matrices
$\overline{H}_0, \overline{H}_1$ are self-orthogonal and mutually orthogonal,
where
\begin{eqnarray}
 \overline{H}_0 =\left[ \begin{array}{ccccc}
1 &\alpha &\alpha^2 &\cdots &\alpha^{n-1}\\
1 &\alpha^2 &\alpha^4 &\cdots &\alpha^{2(n-1)}\\
\vdots& \vdots &\vdots &\ddots &\vdots\\
1 &\alpha^{\mu-1} &\alpha^{2(\mu-1)} &\cdots &\alpha^{(\mu-1)(n-1)}
\end{array}\right] \mbox{ and } \end{eqnarray} \begin{eqnarray}
\overline{H}_1 =\left[ \begin{array}{ccccc}
1 &\alpha^{-1} &\alpha^{-2} &\cdots &\alpha^{(-n-1)}\\
1 &\alpha^{-2} &\alpha^{-4} &\cdots &\alpha^{-2(n-1)}\\
\vdots& \vdots &\vdots &\ddots &\vdots\\
1 &\alpha^{-(\mu-1)} &\alpha^{-2(\mu-1)} &\cdots &\alpha^{-(\mu-1)(n-1)}
\end{array}\right].
\end{eqnarray}

\begin{lemma}\label{lem:RS_CC_selforthogonal_np}
Let $n|q^2-1$ such that $q+1< n\leq q^2-1$ and $2 \leq \mu=2t\leq
\lfloor n/(q+1)\rfloor $, then \begin{eqnarray}\overline{H}_0 =
(\alpha^{ij})_{1\le i<\mu, 0\le j< n}\quad \text{and} \quad
\overline{H}_1 =  (\alpha^{-ij})_{1\le i<\mu, 0\le j< n}
\end{eqnarray} are self-orthogonal with respect to the Hermitian inner
product. Further, $\overline{H}_0$ is orthogonal to
$\overline{H}_1$.
\end{lemma}
\begin{proof}
Denote by $\overline{H}_{0,j}= (1, \alpha^j, \alpha^{2j},\cdots,
\alpha^{j(n-1)})$ and $\overline{H}_{1,j}= (1, \alpha^{-j},
\alpha^{-2j},\cdots, \alpha^{-j(n-1)})$, where $1\leq j \leq \mu-1$. The
Hermitian inner product of $\overline{H}_{0,i}$ and $\overline{H}_{0,j}$ is
given by
\begin{eqnarray}
\langle \overline{H}_{0,i}|\overline{H}_{0,j}\rangle_h &=&\sum_{l=0}^{n-1}
\alpha^{il}\alpha^{jql} = \frac{\alpha^{(i+jq)n} - 1}{\alpha^{i+jq}-1},
\end{eqnarray}
which vanishes if $i+jq\not \equiv 0 \mod n$. If $1\leq i,j \leq \mu-1 =
\floor{n/(q+1)}-1$, then $q+1 \leq i+jq\leq (q+1)\floor{n/(q+1)}-(q+1) <n $;
hence, $\langle \overline{H}_{0,i}|\overline{H}_{0,j} \rangle_h=0$. Thus,
$\overline{H}_0$ is self-orthogonal. Similarly, $\overline{H}_1$ is also
self-orthogonal. Furthermore,
\begin{eqnarray}
\langle \overline{H}_{0,i}|\overline{H}_{1,j}\rangle_h &=&\sum_{l=0}^{n-1}
\alpha^{il}\alpha^{-jql} = 
\frac{\alpha^{(i-jq)n} - 1}{\alpha^{i-jq}-1}.
\end{eqnarray}
This inner product vanishes if $\alpha^{i-jq}\neq 1$ or, equivalently, if
$i-jq\not\equiv 0\mod n$. Since $1\leq i,j\leq \floor{n/(q+1)}-1\leq q-2$, we
have $1\leq i\leq \floor{n/(q+1)}-1 \leq q-2$ while $q\leq jq\leq
q\floor{n/(q+1)}-q <n$. Thus $i\not\equiv jq \mod n$ and this inner product
also vanishes, which proves the claim.
\end{proof}
Since $H_i$ is contained in $\overline{H}_i$, we obtain the following:
\begin{corollary}\label{cor:RS_CC_selforthogonal_np}
Let $2 \leq \mu=2t\leq \floor{n/(q+1)} $, where $n|q^2-1$ and $q+1<n\leq
q^2-1$. Then $H_0$ and $H_1$ are Hermitian self-orthogonal. Further, $H_0$ is
orthogonal to $H_1$ with respect to the Hermitian inner product.
\end{corollary}

\medskip

The following example explains our construction.
\begin{exampleX}\label{ex:RSexample2}
Let $q=5$ and $t=2$, then $n=24$ and $2 \leq \mu=4 \leq q-1$.
\begin{eqnarray*}
H_0 =\left[ \begin{array}{cccccccc}
1 &\alpha &\alpha^2 &\alpha^3&\alpha^4 & \cdots &\alpha^{22}&\alpha^{23}\\
1 &\alpha^{3} &\alpha^{6} &\alpha^{9}&\alpha^{12} & \cdots &\alpha^{66}&\alpha^{69}\\
\end{array}\right]  \mbox{ and}
\end{eqnarray*}
\begin{eqnarray*}
H_1 =\left[ \begin{array}{cccccccc}
1 &\alpha^{-1} &\alpha^{-2} &\alpha^{-3}&\alpha^{-4}& \cdots &\alpha^{-22}&\alpha^{-23}\\
1 &\alpha^{-3} &\alpha^{-6} &\alpha^{-9}&\alpha^{-12}& \cdots &\alpha^{-66}&\alpha^{-69}\\
\end{array}\right]
\end{eqnarray*}
We notice that $H_0^q H_0=0$, $H_1^q H_1=0$, and $H_0^q H_1=0$. Also if we
extend $H_0$ by one row, we find that $H_0^q H_0 \neq  0$.
\end{exampleX}

\medskip
Before we can construct quantum convolutional codes, we need to compute the
free distances of $C$ and $C^\hdual$, where $C$ is the convolutional code
generated by $H$.
\medskip

\begin{lemma}\label{lem:CCdfree}
Let $2\leq 2t\leq \floor{n/(q+1)}$, where $\gcd(n,2)=1 $,  $n |q^2-1$ and
$q+1<n\leq q^2-1$.  Then the convolutional code $C=\Gamma_{q^2} H$ has free
distance $d_f\geq n-2t+1 >2t+1=d_f^\perp$, where $d_f^\perp=\wt(C^\hdual)$ is
the free distance of $C^\hdual$.
\end{lemma}
\begin{proof}
Since $d_f^\perp=\wt(C^\hdual)=\wt(C^\perp)$, we compute the weight
$\wt(C^\perp)$. Let $c=(\ldots, 0,c_0,\ldots, c_l,0,\ldots)$ be a
codeword in $C^\perp$ with $c_i\in \F_{q^2}^n$, $c_0\neq 0$, and
$c_l\neq 0$.  It follows from the parity check equations $cH^T=0$
that $c_0H_1^T=0=c_lH_0^T$ holds.  Thus, $\wt(c_0),\wt(c_l)\geq
t+1$. If $l>0$, then $\wt(c)\geq \wt(c_0)+\wt(c_l)\geq 2t+2$.  If
$l=0$, then $c_0$ is in the dual of $H^*$, which is an
$[n,n-2t,2t+1]_{q^2}$ code. Thus $\wt(c)=\wt(c_0)\geq 2t+1$ and
$d_f^\perp\geq 2t+1$.  But if $c_x$ is in the dual of $H^*$, then
$(\ldots,0,c_x,0,\ldots)$ is a codeword of $C$. Thus $d_f^\perp
=2t+1$.

Let $(\ldots,c_{i-1},c_i,c_{i+1},\ldots)$ be a nonzero codeword in $C$.
Observing the structure of $C$, we see that any nonzero $c_i$ must be in the
span of $H^*$. But $H^*$ generates an $[n,2t,n-2t+1]_{q^2}$ code. Hence
$d_f\geq n-2t+1$. If $2t\leq \floor{n/(q+1)}$, then $t\leq n/6$; thus $d_f\geq
n-2t+1 >2t+1 =d_f^\perp$ holds.
\end{proof}
\medskip

The preceding proof generalizes \cite[Corollary~4]{piret88} where the free
distance of $C^\perp$  was computed for $q=2^m$.

\section{Quantum Convolutional Codes from RS Codes}

We derive a family of quantum convolutional codes based on the
previous construction of generalized Reed-Solomon Codes.
Furthermore, we show the optimality of the derived quantum codes.

\medskip

\begin{theorem}\label{th:qcc-rs}
Let $q$ be a power of a prime, $n$ an odd divisor of $q^2-1$, such that
$q+1<n\leq q^2-1$ and $2\leq \mu=2t \leq \floor{n/(q+1)}$. Then there exists a
pure quantum convolutional code with parameters $[(n, n-\mu,
n;\mu/2,\mu+1)]_q$. This code is optimal, since it attains the 
Singleton bound with equality.
\end{theorem}
\begin{proof}
The convolutional code generated by the coefficient matrix $H$ in
equation (\ref{eq:H}) has parameters $(n,\mu/2,\delta\leq
\mu/2;1,d_f)_{q^2}$. Inspecting the corresponding polynomial
generator matrix shows that $\delta\le \mu/2$, since $\nu_i=1$ for
$1\le i\le \mu/2$. By Corollary~\ref{cor:RS_CC_selforthogonal_np},
this code is Hermitian self-orthogonal; moreover,
Lemma~\ref{lem:CCdfree} shows that the distance of its dual code is
given by $d_f^\perp=\mu+1 <d_f$. By Theorem~\ref{th:c2qHerm}, we can
conclude that there exists a pure convolutional stabilizer code with
parameters $[(n,n-\mu,n; \delta \le \mu/2, \mu+1)]_q$. It follows
from Theorem~\ref{th:qsb} that \begin{eqnarray}
\begin{array}{l@{\,}c@{\,}l}
\mu+1 &\le& (\mu/2)\left(\left\lfloor
2\delta/(2n-\mu)\right\rfloor+1\right)+\delta+1 \\
&\le& (\mu/2)\left(\left\lfloor \mu/(2n-\mu)\right\rfloor+1\right)+\delta+1.
\end{array}
\end{eqnarray} Since $\left\lfloor \mu/(2n-\mu)\right\rfloor=0$, the right hand
side equals $\mu/2+\delta+1$, which implies $\delta=\mu/2$ and the optimality
of the quantum code.
\end{proof}

The following two examples explain our construction.
\begin{exampleX}
Let $q=4$ and $t=1$, then $n=15$ and $2 \leq \mu=2 \leq q-1$.
\begin{eqnarray}
H_0 =\left[ \begin{array}{cccccccc}
1 &\alpha &\alpha^2 &\alpha^3&\alpha^4 & \cdots &\alpha^{13}&\alpha^{14}\\
\end{array}\right]
\end{eqnarray}

and
\begin{eqnarray}
H_1 =\left[ \begin{array}{cccccccc}
1 &\alpha^{-1} &\alpha^{-2} &\alpha^{-3}&\alpha^{-4}& \cdots &\alpha^{-13}&\alpha^{-14}\\
\end{array}\right]
\end{eqnarray}
We notice that $H_0^q H_0=0$, $H_1^q H_1=0$, and $H_0^q H_1=0$. Also if we
extend $H_0$ by one row, we find that $H_0^q H_0 \neq 0$.
\end{exampleX}

\begin{exampleX}
Let $q=5$ and $t=2$, then $n=24$ and $2 \leq \mu=4 \leq q-1$.
\begin{eqnarray*}
H_0 =\left[ \begin{array}{cccccccc}
1 &\alpha &\alpha^2 &\alpha^3&\alpha^4 & \cdots &\alpha^{22}&\alpha^{23}\\
1 &\alpha^{3} &\alpha^{6} &\alpha^{9}&\alpha^{12} & \cdots &\alpha^{3}&\alpha^{21}\\
\end{array}\right]
\end{eqnarray*}

and
\begin{eqnarray*}
H_1 =\left[ \begin{array}{cccccccc}
1 &\alpha^{-1} &\alpha^{-2} &\alpha^{-3}&\alpha^{-4}& \cdots &\alpha^{-22}&\alpha^{-23}\\
1 &\alpha^{-3} &\alpha^{-6} &\alpha^{-9}&\alpha^{-12}& \cdots &\alpha^{-66}&\alpha^{-69}\\
\end{array}\right]
\end{eqnarray*}
We notice that $H_0^q H_0=0$, $H_1^q H_1=0$, and $H_0^q H_1=0$. Also if we
extend $H_0$ by one row, we find that $H_0^q H_0 \neq  0$.
\end{exampleX}


\section{Convolutional Codes from Quasi-Cyclic Subcodes of Reed-Muller Codes}\label{section:RM}

 An alternative method to construct convolutional codes from block codes is to
 use quasi-cyclic codes. We consider the Reed-Muller codes to construct a
 series quantum convolutional codes with varying memory. But first we review
 the necessary background on binary Reed-Muller codes. Furthermore, we use the
 framework developed by Esmaeili and Gulliver to construct quasi-cyclic
 subcodes RM codes from block RM codes over the binary field, see
 \cite{esmaeili97},\cite{esmaeili98} for more details.

\smallskip

 Let $u,v\in \F_2^n$, where $u=(u_1,u_2,\ldots, u_n)$ and
$v=(v_1,u_2,\ldots,v_n)$. We define the boolean product
\begin{eqnarray}uv=(u_1v_1,u_2v_2,\ldots,u_nv_n).\end{eqnarray} The product of $i$ such
$n$-tuples is said to have a degree of $i$. Let $v_0=(1,1,\ldots,1)
\in \F_2^{2^m}$. For $m>0$ and $1\leq i\leq m$, define $b_i \in
\F_2^{2^m}$ as concatenation of $2^{m-i}$ blocks of the form
$\zero\one$. Each block is of length $2^{i}$ and equal to
$(\zero\one)$, where $\zero,\one \in \F_2^{2^{i-1}}$.

Let $0\leq r<m$ and $B=\{b_1,b_2,\ldots,b_m\} \subseteq \F_2^{2^m}$. Then the
$r$th order Reed-Muller code is the span of $v_0$ and all products of elements
in $B$ upto and including the degree $r$ and it is denoted by $\RM(r,m)$. Let
$G_m^r$ denote the generator matrix of $\RM(r,m)$. Let $B_m^i$ denote all the
products with exactly degree $i$. Then for $0\leq i\leq r <m$ (see
\cite{esmaeili98} for details)
\begin{eqnarray}
G^r_m = \left[ \begin{array}{c}B^r_m\\B^{r-1}_m\\ \vdots\\ B^{i+1}_m\\
G_m^{i}\end{array} \right].
\end{eqnarray}
The dimension of $\RM(r,m)$ is given by $k(r)=\sum_{i=0}^{r}\binom{m}{i}$ and
its distance is given by $2^{m-r}$. The dual of $\RM(r,m)$ is given by
$\RM(r,m)^\perp=\RM(m-1-r,m)$. The dual distance of $\RM(r,m)$ is $2^{r+1}$ as
can be easily verified. Further details on the properties of Reed-Muller codes
can be found in \cite{huffman03}.


%
Let $w_{\mu}=(110\cdots0)\in \F_2^{2^\mu}$. Let $lw_{\mu}$ denote the vector
obtained by concatenating $l$ copies of $w_{\mu}$. For $0\leq i\leq l-1$, let
$QM_{i,l}= (2^{l-i-1}w_{i+1})\otimes B_{m-l}^{r-i}$  which is a matrix of size
$\binom{m-l}{r-i} \times 2^m$ and for $i=l$ let
$QM_{l,l}=\left[\begin{array}{cccc} G_{m-l}^{r-l} & \zero & \cdots& \zero
\end{array}\right]$. The convolutional code derived from the quasi-cyclic
subcode of $\RM(r,m)$ has the following generator matrix.
\begin{eqnarray}
G&=&\left[ \begin{array} {c}
QM_{0,l}\\
QM_{1,l}\\
\vdots\\
QM_{l-1,l}\\
QM_{l,l}
\end{array}\right]
\nonumber \\&=&\left[ \begin{array}{c|c|c|c|c|c|c|c} B_{m-l}^{r} & B_{m-l}^{r} & B_{m-l}^{r} & B_{m-l}^{r}& B_{m-l}^{r}& \cdots & B_{m-l}^{r}\\
B_{m-l}^{r-1} & B_{m-l}^{r-1} & \zero&\zero&B_{m-l}^{r-1} & \cdots&\cdots\\
\vdots&\vdots & \vdots &\vdots&\vdots & \ddots&\ddots  \\
B_{m-l}^{r-l+1} & B_{m-l}^{r-l+1} & \zero&\zero&\cdots & \zero  &\zero\\
G_{m-l}^{r-l} & \zero & \zero &\cdots&\zero & \zero &\zero
\end{array} \right],\nonumber \\
&=& \left[\begin{array}{cccccc} G_0&G_1&\cdots & \cdots & G_{2^{l}-1}
\end{array} \right].
\end{eqnarray}

We note that $G_0=G_{m-l}^{r}$ and for $1\leq i\leq 2^l-1$, the elements of
$G_i$ are a subset of the elements in  $G_0$. The convolutional code generated
by $G$ has rate $\sum_{i=0}^r{m-l \choose i}/2^{m-l}$ and free distance
$2^{m-r}$~\cite{esmaeili98}.
\begin{lemma}
The free distance of the convolutional code orthogonal to $G$ is $2^{r+1}$.
\end{lemma}
\begin{proof}
Assume that $c$ is codeword in the space orthogonal to $G$. Without loss of
generality we can take it to be of the form $c=(c_0, c_1, \ldots, c_i,\ldots)$,
where all the $c_i=\zero$, for $i<0$. Since $cG^T=0$, we have the following set
of constraints for $t\geq 0$.
\begin{eqnarray}
\sum_{t-2^{l}-1}^t c_iG_{t-i}^T=0.
\end{eqnarray}
Alternatively, we can write the above as a set of equations as
\begin{eqnarray}
c_0G_0^T &=&0,\nonumber\\
c_1G_0^T+c_0G_1^T&=&0,\nonumber\\
\vdots &=& \vdots \nonumber\\
c_iG_{0}^T+c_{i-1}G_1^T+\cdots+ c_{i-2^l+1}G_{2^l-1}^T&=&0\nonumber \\
\vdots &=& \vdots,
\end{eqnarray}
\end{proof}
If follows that $c_0\in \RM(r,m-l)^\perp$. Since the rowspace of $G_i$ is a
subset of the rowspace of $G_0$, it then follows that $c_0G_1^T=0$ giving
$c_1G_0^T=0$. Thus $c_1$ is also in $\RM(r,m-l)^\perp$. Proceeding like this we
see that $c_i\in \RM(r,m-l)^\perp$ for all $i\geq 0$. Thus the free distance of
the code orthogonal to $G$ is equal to the dual distance of $\RM(r,m-l)$ which
is $2^{r+1}$.

\begin{lemma}
Let $1\leq l\leq m$ and $0\leq r\leq \lfloor (m-l-1)/2\rfloor$, then the
convolutional code generated by $G$ is self-orthogonal.
\end{lemma}
\begin{proof}
It is sufficient to show that $G_iG_j^T=0$ for $0\leq i,j\leq 2^l-1$. Since the
rows of $G_i$  are a subset of the rows of $G_0$ it suffices to show that $G_0$
is self-orthogonal. For $G_0$ to be self-orthogonal we require that $r\leq
(m-l)-r-1$ which holds. Hence, $G$ generates a self-orthogonal convolutional
code.
\end{proof}

\section{Quantum Convolutional Codes from QC RM Codes}
We can derive a family of QC RM codes as shown in the following
Lemma.
\begin{lemma}
Let $1\leq l\leq m$ and $0\leq r\leq \lfloor (m-l-1)/2\rfloor$, then there
exist pure linear quantum convolutional codes with the parameters
$((2^{m-l},2^{m-l}-2k,2^{l}-1))$ and free distance $2^{r+1}$, where
$k=\sum_{i=0}^r {m-l\choose i}$.
\end{lemma}
\begin{proof}
Since $G$ defines a linear self-orthogonal convolutional code with parameters
$(2^{m-l},k(r),2^{l}-1)$ and free distance $2^{m-r}$, there exists a linear
quantum convolutional code with the parameters
$((2^{m-l},2^{m-l}-2k(r),2^l-1))$. For $0\leq r\leq \lfloor (m-l-1)/2\rfloor$,
the dual distance $2^{r+1}< 2^{m-r}$, hence the code is pure.
\end{proof}

It turns out  that the convolutional codes in~\cite{esmaeili98} that
are used here have degree 0, hence, are a sequence of juxtaposed
block codes disguised as convolutional codes. Consequently, the
codes constructed in the previous theorem have parameters
$[(2^{m-l},2^{m-l}-2k(r), 0; 0, 2^{r+1})]_2$.


\section{Conclusion and Discussion}
We constructed two families of quantum convolutional codes based on
RS and Reed-Muller codes. We showed that quantum convolutional codes
derived from our constructions have better parameters in comparison
to quantum block codes counterparts. We proved that the codes
derived from RS codes are optimal in a sense that they it attains
generalized Singleton bound with equality. One possible extension of
this work is to construct other good families of quantum
convolutional codes.

\chapter{Quantum Convolutional Codes derived from  BCH
Codes}\label{ch_QCC_BCH}

Quantum convolutional codes can be used to protect a sequence of
  qubits of arbitrary length against decoherence.  We introduce two new
  families of quantum convolutional codes. Our construction is based on
  an algebraic method which allows to construct classical
  convolutional codes from block codes,  in particular
  BCH codes.  These codes have the property that they contain their
  Euclidean, respectively Hermitian, dual codes. Hence, they can be
  used to define quantum convolutional codes by the stabilizer code
  construction. We compute BCH-like bounds on the free distances which
  can be controlled as in the case of block codes, and establish that
  the codes have non-catastrophic encoders. Some materials presented in this
  chapter are also published in~\cite{aly07d,aly07a} as a joint work with M. Grassl, A. Klappenecker, M. R\"{o}tteler, and P.K. Sarvepalli.

\section{Introduction}

\nix{We investigate the theory of quantum convolutional codes. We  derive new
families of quantum convolutional codes based BCH codes. We show that these
codes have higher rates than their quantum block codes counterparts.}

Unit memory convolutional codes are an important class of codes that appeared
in a paper by Lee~\cite{lee76}. He also showed that these codes have large free
distance $d_f$ among other codes (multi-memory) with the same rate. \nix{Upper
and lower bounds on the free distance of unit memory codes were derived by
Thommesen and Justesen~\cite{thommesen83}.} Convolutional codes are often
designed heuristically. However, classes of unit memory codes were constructed
algebraically by Piret based on Reed-Solomon codes~\cite{piret88} and by Hole
based on BCH codes~\cite{hole00}.
In a recent paper, doubly-cyclic convolutional codes are investigated which
include codes derived from Reed-Solomon and BCH codes \cite{gluesing04b}. These
codes are related, but not identical to the codes defined in this chapter.

A quantum convolutional codes encodes a sequence of quantum digits at a time. A
stabilizer framework for quantum convolutional codes based on direct limits was
developed in~\cite{aly07b} including necessary and sufficient conditions for
the existence of convolutional stabilizer codes. An $[(n,k,\mathrm{m};\nu)]_q$
convolutional stabilizer code with free distance $d_f= \wt(C^\perp \backslash
C)$ can also correct up to $\lfloor \frac{(d_f-1)}{2} \rfloor $
 errors. It is important to mention that the parameters of a quantum
convolutional code $Q$ are defined differently. The \emph{memory $\mathrm{m}$}
is defined as the overlap length among any two infinite sequences of the code
$Q$. Also, the \emph{degree $\nu$} is given by the degree of the classical
convolutional code $C^\perp$.  The code $Q$ is \emph{pure} if there are no
errors less than $d_f$ in the stabilizer of the code; $d_f=\wt(C^\perp
\backslash C)= \wt(C^\perp)$.

 Recall that one can construct convolutional stabilizer codes from self-orthogonal (or dual-containing) classical
convolutional codes over $\F_q$ (cf. \cite[Corollary~6]{aly07b}) and
$\F_{q^2}$ (see \cite[Theorem~5]{aly07b}) as stated in the following
theorem.

\begin{theorem}\label{CSS:F_q}
An $[(n,k,nm;\nu,d_f)]_q$ convolutional stabilizer code exists if and only if
there exists an $(n,(n-k)/2,m;\nu)_q$ convolutional code such that $C \leq
C^\perp$ where the dimension of $C^\perp$ is given by $(n+k)/2$ and
$d_f=\wt(C^\perp \backslash C).$
\end{theorem}

The main results of this chapter are:
\begin{inparaenum}[(a)]
\item a method to construct convolutional codes from block codes
\item a new class of convolutional stabilizer codes based on BCH codes.
\end{inparaenum}
These codes have non-catastrophic dual encoders making it possible to derive
non-catastrophic encoders for the quantum convolutional codes.

\section{Construction of Convolutional Codes from Block Codes}
In this section, we give a method to construct convolutional codes
from block codes. This generalizes an earlier construction by
Piret~\cite{piret88b} to construct convolutional codes from block
codes. One benefit of this method is that we can easily bound the
free distance using the techniques for block codes. Another benefit
is that we can give easily a non-catastrophic encoder.

Given an $[n,k,d]_q$ block code with parity check matrix $H$, it is possible to
split the matrix $H$ into $m+1$ disjoint submatrices $H_i$, each of length $n$
such that
\begin{eqnarray}
H=\left[\begin{array}{c} H_0\\H_1\\ \vdots\\ H_m
\end{array}\right]. \label{eq:splitH}
\end{eqnarray}
Then we can form the polynomial matrix
\begin{eqnarray}
H(D)=\H_0+\H_1 D+\H_2 D^2+\ldots+\H_m D^m,\label{eq:ccH}
\end{eqnarray}
where the number of rows of $H(D)$ equals the maximal number $\kappa$ of rows
among the matrices $H_i$.  The matrices $\H_i$ are obtained from the matrices
$H_i$ by adding zero-rows such that the matrix $\H_i$ has $\kappa$ rows in
total. Then $H(D)$ generates a convolutional code. Of course, we already knew
that $H_i$ define block codes of length $n$, but taking the $H_i$ from a single
block code will allow us to characterize the parameters of the convolutional
code and its dual using the techniques of block codes. Our first result
concerns a non-catastrophic encoder for the code generated by $H(D)$.

\begin{theorem}\label{th:noncataDualEnc}
Let $C\subseteq \F_q^n$ be an $[n,k,d]_q$ linear code with parity check matrix
$H$ in $\F_q^{(n-k)\times n}$. Assume that $H$ is partitioned into submatrices
$H_0,H_1,\ldots,H_m$ as in equation~(\ref{eq:splitH}) such that $\kappa = \rk
H_0$ and $\rk H_i\le \kappa$ for $1\le i\le m$. Define the polynomial matrix
\begin{eqnarray}
H(D)=\H_0+\H_1 D+\H_2 D^2+\ldots+\H_m D^m,
\end{eqnarray}
where $\H_i$ are obtained from the matrices $H_i$ by adding zero-rows such that
the matrix $\H_i$ has a total of $\kappa$ rows. Then we have:
\begin{compactenum}[(a)]
\item \label{lm:CCbasic}
The matrix $H(D)$ is a reduced basic generator matrix.

\item \label{lm:CCdual}
If the code $C$ contains its Euclidean dual $C^\bot$ or its Hermitian dual
$C^\hdual$, then the convolutional code $U=\{\mbf{v}(D) H(D)\,|\, \mbf{v}(D)\in
\F_q^{n-k}[D]\}$ is respectively contained in its dual code $U^\perp$ or
$U^\hdual$.

\item \label{lm:CCdist}
Let $d_f$ and $d_f^\perp$ respectively denote the free distances of $U$ and
$U^\perp$. Let $d_i$ be the minimum distance of the code $C_i=\{ v\in
\F_q^n\,|\, v\H_i^t =0\}$, and let $d^\perp$ denote the minimum distance of
$C^\perp$. Then the free distances are  bounded by $\min \{d_0+d_m,d \}\leq
d_f^\perp\leq d$ and $d_f \geq d^\perp$.
\end{compactenum}
\end{theorem}

\begin{proof}
To prove the claim (a), it suffices to show that
\begin{compactenum}[i)]
\item $H(0)$ has full rank $\kappa$;
\item $(\coeff(H(D)_{ij},D^{\nu_i}))_{1\le i\le \kappa, 1\le j\le n}$
has full rank $\kappa$;
\item $H(D)$ is non-catastrophic;
\end{compactenum}
cf.~\cite[Theorem 2.16 and Theorem 2.24]{piret88}.

By definition, $H(0)=\H_0$ has rank $\kappa$, so i) is satisfied. Condition ii)
is satisfied, since the rows of $H$ are linearly independent; thus, the rows of
the highest degree coefficient matrix are independent as well.

It remains to prove iii). Seeking a contradiction, we assume that the generator
matrix $H(D)$ is catastrophic.  Then there exists an input sequence
$\mathbf{u}$ with infinite Hamming weight that is mapped to an output sequence
$\mathbf{v}$ with finite Hamming weight, i.\,e. $v_i=0$ for all $i\ge i_0$. We
have
\begin{equation}\label{eq:encoding}
v_{i+m} = u_{i+m} \H_0 + u_{i+m-1}\H_1+\ldots+u_i\H_m,
\end{equation}
where $v_{i+m}\in\F_q^n$ and $u_j\in\F_q^\kappa$.  By construction, the vector
spaces generated by the rows of the matrices $H_i$ intersect trivially. Hence
$v_i=0$ for $i\ge i_0$ implies that $u_{i-j}\H_j=0$ for $j=0,\ldots,m$.  The
matrix $\H_0$ has full rank. This implies that $u_i=0$ for $i\ge i_0$,
contradicting the fact that $\mathbf{u}$ has infinite Hamming weight; thus, the
claim (a) holds.

To prove the claim (b), let $\mbf{v}(D)$, $\mbf{w}(D)$ be any two codewords in
$U$. Then from equation~(\ref{eq:encoding}), we see that $v_i$ and $w_j$ are in
the rowspan of $H$ {\em i.e.} $C^\perp$, for any $i,j\in \Z$. Since $C^\perp
\subseteq C$, it follows that $v_i\cdot w_j=0 $, for any $i,j \in \Z$ which
implies that $\scal{\mbf{v}(D)}{\mbf{w}(D)} =\sum_{i\in\Z} v_i\cdot w_i =0$.
Hence $U\subseteq U^\perp$. Similarly, we can show that if $C^\hdual \subseteq
C$, that $U \subseteq U^\hdual$.

\nix{ that since $H_i$ are submatrices of $H$ the codes $C_i$ with parity check
matrices $H_i$ all contain the code $C$ with parity check matrix $H$ {\em
i.e.} $C \subseteq C_i$.  Hence if $C^\perp\subseteq C$ it follows that
$C_i^\perp\subseteq C^\perp\subseteq C \subseteq C_i$. Similarly we can show
that if $ C^\hdual \subseteq C$ then $C_i^\hdual \subseteq C_i$. }

For the claim (c), without loss of generality assume that the codeword
$\mbf{c}(D) =\sum_{i=0}^l c_iD^i$ is in $U^\perp$, with $c_0\ne 0 \ne c_l$.
\nix{ then, $c=(\ldots,0,c_0,\ldots,c_l,0,\ldots)\in (\F_q^n)^{\N}$,  with
$c_0\ne 0 \ne c_l$, belongs to the sequence version of the code $U^\perp$, then
} Then $\mbf{c}(D)D^m$ and $\mbf{c}(D)D^{-l}$ are orthogonal to every element
in $H(D)$, from which we can conclude that $c_0H_m^t=0=c_lH_0^t$. It follows
that $c_0\in C_0$ and $c_l\in C_l$. If $l>0$, then $\wt(c_0)\geq d_m$ and
$\wt(c_l)\geq d_0$ implying $\wt(\mbf{c}(D))\geq d_0+d_m$. If $l=0$, then
$c_0D^i$, where $0\leq i\leq m$  is orthogonal to every element in $H(D)$, thus
$c_0H_i^t=0$, whence $c_0H^t=0$ and $c_0\in C$, implying that $\wt(c_0)\geq d$.
It follows that $\wt(c)\geq \min \{ d_0+d_m, d\}$, giving the lower bound on
$d_f^\perp$.

For the upper bound note that if $c_0$ is a codeword $C$, then $c_0H_i^t=0$.
Therefore codeword $\mbf{c}(D)$ and its shifts $\mbf{c}(D)D^i$ for $0\leq i\leq
m$ are orthogonal to $H(D)$. Hence  $\mbf{c}(D)\in U^\perp$ and $d_f^\perp\leq
d$.

Finally, let $\mbf{c}(D)$ be a codeword in $U$. We saw earlier in the proof of
(b) that that every $c_i$ is in  $C^\perp$. Thus $d_f\geq \min\{\wt(c_i)\} \geq
d^\perp$.
\end{proof}

A special case of our claim (a) has been established by a different method
in~\cite[Proposition 1]{hole00}.

\section{Convolutional BCH Codes}
One of the attractive features of BCH codes is that they allow us to design a
code with desired distance. There have been prior approaches to construct
convolutional BCH codes most notably~\cite{rosenthal99} and~\cite{hole00},
where one can control the free distance of the convolutional code. Here we
focus on codes with unit memory. In the literature on convolutional codes there
is a subtle distinction between unit memory and partial unit memory codes,
however for our purposes, we will disregard such nuances. Our codes have better
distance parameters as compared to Hole's construction and are easier to
construct compared to~\cite{rosenthal99}.

\subsection{Unit Memory Convolutional BCH Codes}
Let $\F_q$ be a finite field with $q$ elements, $n$ be a positive integer such
that $\gcd(n,q)=1$. Let $\alpha$ be a primitive $n$th root of unity.  A BCH
code $C$ of designed distance $\delta$ and length $n$ is a cyclic code with
generator polynomial $g(x)$ in $\F_q[x]/\langle x^n-1\rangle$ whose defining
set is given by $Z=C_b\cup C_{b+1}\cup \cdots \cup C_{b+\delta-2}$, where
$C_x=\{xq^i\bmod n \mid i\in \Z, i\ge 0 \}$. Let
\begin{eqnarray*}
H_{\delta,b} =\left[ \begin{array}{ccccc}
1 &\alpha^b &\alpha^{2b} &\cdots &\alpha^{b(n-1)}\\
1 &\alpha^{b+1} &\alpha^{2(b+1)} &\cdots &\alpha^{(b+1)(n-1)}\\
\vdots& \vdots &\vdots &\ddots &\vdots\\
1 &\alpha^{(b+\delta-2)} &\alpha^{2(b+\delta-2)} &\cdots
&\alpha^{(b+\delta-2)(n-1)}
\end{array}\right].
\end{eqnarray*}
Then $C=\{ v\in \F_q^n\,|\,vH_{\delta,b}^t=0\}$. If $r=\ord_n(q)$, then a
parity check matrix, $H$ for $C$ is given by writing every entry in the matrix
$H_{\delta,b}$ as a column vector over some $\F_q$-basis of $\F_{q^r}$, and
removing any dependent rows. \nix{ Let us make a simple remark concerning the
expansion of a vector of length $n$ over $\F_{q^r}$ in terms of a basis over
$\F_q$. } Let $B=\{b_1,\dots,b_r\}$ denote a basis of $\F_{q^r}$ over $\F_q$.
Suppose that $w=(w_1,\dots,w_n)$ is a vector in $\F_{q^r}^n$, then we
can write $w_j = w_{j,1}b_1+ \cdots + w_{j,r}b_r$ for $1\le j\le n$. Let
$w^i=(w_{1,i},\dots, w_{n,i})$ be vectors in $\F_q^n$ with $1\le i\le r$, For a
vector $v$ in $\F_q^n$, we have $v\cdot w=0$ if and only if $v\cdot w^i=0$ for
all $1\le i\le r$.

For a matrix $M$ over $\F_{q^r}$, let $\ex_B(M)$ denote the matrix that is
obtained by expanding each row into $r$ rows over $\F_q$ with respect to the
basis $B$, and deleting all but the first rows that generate the rowspan of the
expanded matrix. Then $H=\ex_B(H_{\delta,b})$. \nix{ By a slight abuse, we will
refer to $H_{\delta,b}$ itself as the parity check matrix of $C$. }

It is well known that the minimum distance of a BCH code is greater than or
equal to its designed distance $\delta$, which is very useful in constructing
codes.  Before we can construct convolutional BCH codes we need the following
result on the distance of
cyclic codes.

\begin{lemma}\label{lm:hartman}
Let $\gcd(n,q)=1$ and $2\leq \alpha\leq \beta <n$. Let $C\subseteq \F_q^n$ be a
cyclic code with defining set
\begin{equation}\label{eq:defining_set}
Z=\{z\mid z\in C_x, \alpha\le x\le\beta, x\not\equiv 0\bmod q\}.
\end{equation}
\nix{Further,  let $r=\ord_n(q)$ and $q\leq \delta \leq \delta_{\max}$, where
$$\delta_{\max}=\left\lfloor\frac{n}{q^{r}-1} (q^{\lceil r/2\rceil}-1-(q-2)[r
\textup{ odd}])\right\rfloor.$$ } Then the minimum distance
$\Delta(\alpha,\beta)$ of $C$ is lower bounded as
\begin{eqnarray}
\Delta(\alpha,\beta) \geq
\begin{cases}
q+\floor{(\beta-\alpha+3)/q}-2, & \text{if $\beta-\alpha \geq 2q-3$;}\\
\floor{(\beta-\alpha+3)/2}, & \text{otherwise.}
\end{cases}
\end{eqnarray}
\end{lemma}
\begin{proof}
Our goal is to bound the distance of $C$ using the Hartmann-Tzeng bound (for
instance, see~\cite{huffman03}). Let  $A=\{z,z+1,\ldots, z+a-2 \} \subseteq Z$.
Let $\gcd(b,q)<a$ and $A+jb =\{ z+jb,z+1+jb,\ldots, z+a-2+jb\} \subseteq Z$ for
all $0\leq j\leq s$. Then  by~\cite[Theorem~4.5.6]{huffman03}, the minimum
distance of $C$ is $\Delta(\alpha,\beta) \geq a+s$.

We choose $b=q$, so that $\gcd(n,q)=1<a$ is satisfied for any $a>1$. Next we
choose $A\subseteq Z$ such that $|A|=q-1$ and $A+jb\subseteq Z$ for $0\leq
j\leq s$, with $s$ as large as possible. Now two cases can arise. If
$\beta-\alpha+1 < 2q-2$, then there {\em may not} always exist a set $A$ such
that $|A|=q-1$. In this case we relax the constraint that $|A|=q-1$ and choose
$A$ as the set of maximum number of consecutive elements. Then $|A|=a-1 \geq
\floor{(\beta-\alpha+1)/2}$ and $s\geq 0$ giving the distance
$\Delta(\alpha,\beta)\geq \floor{(\beta-\alpha+1)/2}+1=
\floor{\beta-\alpha+3)/2}$.

If $(\beta-\alpha+1) \geq 2q-2$, then we can always choose a set $A\subseteq
\{z \mid \alpha \leq z\leq \alpha+2q-3, z\not\equiv 0\bmod q\}$ such that
$|A|=q-1$. Since we want to make $s$ as large as possible, the worst case
arises when $A=\{\alpha+q-1,\ldots,\alpha+2q-3\}$.  Since $A+jb\subseteq Z$
holds for $0\leq j\leq s$, it follows $\alpha+2q-3+sq\leq \beta$. Thus $s \leq
\floor{(\beta-\alpha+3)/q}-2$. Thus the distance $\Delta(\alpha,\beta)\geq
q+\floor{(\alpha-\beta+3)/q}-2$.
\end{proof}

\begin{theorem}[Convolutional BCH codes]\label{th:bchCC}
Let $n$ be a positive integer such that $\gcd(n,q)=1$, $r=\ord_n(q)$
and $2\leq 2\delta <\delta_{\max}$, where
\begin{eqnarray}\delta_{\max}=\left\lfloor\frac{n}{q^{r}-1} (q^{\lceil
r/2\rceil}-1-(q-2)[r \textup{ odd}])\right\rfloor.\end{eqnarray} Then there
exists a unit memory rate $k/n$ convolutional BCH code with free
distance $d_f\geq \delta+1+\Delta(\delta+1,2\delta)$ and
$k=n-\kappa$, where $\kappa=r\ceil{\delta(1-1/q)}$.  The free
distance of the dual is $\geq \delta_{\max}+1$. \nix{ If
$r\ceil{\delta(1-1/q)}=
r\ceil{2\delta(1-1/q)}-r\ceil{\delta(1-1/q)}$, then there also
exists a $\sigma$-cyclic convolutional code with these parameters. }
\end{theorem}
\begin{proof}
Let $C\subseteq \F_q^n$ be a narrow-sense BCH code of designed distance
$2\delta+1$ and $B$ a basis of $\F_{q^r}$ over $\F_q$. Recall that a parity
check matrix for $C$ is given by $H=\ex_B(H_{2\delta+1,1})$. Further, let
$H_0=\ex_B(H_{\delta+1,1})$, then from
\begin{eqnarray}
H_{2\delta+1,1}=\left[ \begin{array}{c}H_{\delta+1,1}\\
H_{\delta+1,\delta+1}\end{array}\right],\label{eq:bchH}
\end{eqnarray}
\nix{ Then $C=\{ v\in \F_q^n\,|\, vH_{2\delta+1,1}=0\}$, where the parity check
matrix $H_{2\delta+1,1}$ of $C$ is of the form
\begin{eqnarray}
H_{2\delta+1,1}=\left[ \begin{array}{c}H_{\delta+1,1}\\
H_{\delta+1,\delta+1}\end{array}\right], \label{eq:bchH}
\end{eqnarray}
One must note that these matrices are over $\F_{q^r}$. } it follows that
$H=\left[
\begin{array}{c} H_0
\\ H_1
\end{array}\right]$, where
$H_1$ is the complement of $H_0$ in $H$. It is obtained from
$\ex_B(H_{\delta+1,\delta+1})$ by removing all rows common to
$\ex_B(H_{\delta+1,1 })$. The code $D_0$ with parity check matrix
$H_0=\ex_B(H_{\delta+1,1})$
coincides with  narrow-sense BCH code of length $n$ and design distance
$\delta+1$.

By \cite[Theorem~10]{aly07a}, we have $\dim C = n- r\ceil{2\delta(1-1/q)}$ and
$\dim D_0 = n-r\ceil{\delta(1-1/q)}$; hence $\rk H= r\ceil{2\delta(1-1/q)}$,
$\rk H_0 = r\ceil{\delta(1-1/q)}$, and $\rk H_1 = \rk H - \rk H_0 =
r\ceil{2\delta(1-1/q)}-r\ceil{\delta(1-1/q)}$. For $x>0$, we have $\ceil{x}
\geq \ceil{2x} -\ceil{x}$; therefore, $\kappa:= \rk H_0 \geq \rk H_1$.

By Theorem~\ref{th:noncataDualEnc}(\ref{lm:CCbasic}), the matrix $H$ defines a
reduced basic generator matrix
\begin{eqnarray}
 H(D) = \H_0+D\H_1 \label{eq:bchHD}
\end{eqnarray}
of a convolutional code of dimension $\kappa$, while its dual which we refer to
as a convolutional BCH code is of dimension $n-\kappa$.

\nix{ The code $C_1$ is a cyclic code of the form given in
Lemma~\ref{lm:hartman} with $\alpha =\delta+1$ and $\beta=2\delta$. Indeed,

By Theorem~\ref{noncataDualEnc}~\ref{lm:CCdist}, the free distance is bounded
by $\min\{d_0+d_1,d \}\leq d_f\leq d$.  By Lemma~\ref{lm:hartman}, $d_1
=\Delta(\delta+1,2\delta) $ and by the BCH bound $d_0\geq \delta+1$. Thus
$d_f\geq \delta+1+\Delta(\delta+1,2\delta)$.  The dual free distance also
follows from Theorem~\ref{noncataDualEnc}~\ref{lm:CCdist} as $d_f^\perp\geq
d^\perp$. But $d^\perp \geq \delta_{\max}+1$ by \cite[Lemma~12]{aly07a}.

Let $C$ be a narrow-sense BCH code of length $n$ and designed distance
$2\delta+1$. Recall that a parity check matrix $H$ for $C$ is given by
expanding each row of $H_{2\delta+1,1}$ over $\F_{q^r}$ into $r$ rows over
$\F_q$ (and removing any dependent rows). Note that we can write
$H_{2\delta+1,1}$ as
\begin{eqnarray}
H_{2\delta+1,1}=\left[ \begin{array}{c}H_{\delta+1,1}\\
H_{\delta+1,\delta+1}\end{array}\right].\label{eq:bchH}
\end{eqnarray}
Let the expanded matrices of $H_{2\delta+1,1}$ and $ H_{\delta+1,1}$ be denoted
as $H$ and $H_0$, respectively. Then we have $H=\left[
\begin{array}{c} H_0
\\ H_1
\end{array}\right]$, where
$H_1$ is the complement of $H_0$ in $H$. }

Now $H_1$ is the parity check matrix of a cyclic code, $D_1$ of the form given
in Lemma~\ref{lm:hartman}, {\em i.e.} the defining set of $D_1$ is
$Z_1$ as defined in (\ref{eq:defining_set}) with  $\alpha =\delta+1$ and
$\beta=2\delta$. Since $H_1$ is linearly independent of $H_0$ we have
$x\not\equiv 0\bmod q$ in the definition of $Z_1$.

\nix{ Consider the convolutional code obtained by forming the polynomial parity
check matrix $H(D)$ as in Theorem~\ref{th:noncataDualEnc} {\em i.e.},
\begin{eqnarray}
H(D)=H_{0} +H_{1}D.\label{eq:bchHD}
\end{eqnarray}
This code has memory $m=1$, since rank $H_1>0$ as the following discussion will
show.

To compute the dimension of the convolutional code we need to know the ranks of
$H,H_0$ and $H_1$. First observe that $H_{2\delta+1,1}$ and $H_{\delta+1,1}$
both define narrow-sense BCH codes, therefore using \cite[Theorem~10]{aly07a},
the ranks of $H_0$ and $ H_1$ can be computed as $r\ceil{\delta(1-1/q)}$ and
$r\ceil{2\delta(1-1/q)}-r\ceil{\delta(1-1/q)}$ respectively. Since the larger
of the matrices $H_0,H_1$ determines the dimension of $H(D)$, we compute the
dimension of the convolutional code as $n-\kappa$, where $\kappa=
\max\{r\ceil{\delta(1-1/q)}, r\ceil{2\delta(1-1/q)}-r\ceil{\delta(1-1/q)}\}$.
But observe that for $x>0$, $\ceil{x} \geq \ceil{2x} -\ceil{x}>0$. Therefore,
$\kappa=r\ceil{\delta(1-1/q)}$ and rank $H_1>0$. }

By Theorem~\ref{th:noncataDualEnc}(\ref{lm:CCdist}), the free distance of the
convolutional BCH code is bounded as $\min\{d_0+d_1,d \}\leq d_f\leq d$. By
Lemma~\ref{lm:hartman}, $d_1\geq\Delta(\delta+1,2\delta) $ and by the BCH bound
$d_0\geq \delta+1$. Thus $d_f\geq \delta+1+\Delta(\delta+1,2\delta)$. The dual
free distance also follows from
Theorem~\ref{th:noncataDualEnc}(\ref{lm:CCdist}) as $d_f^\perp\geq d^\perp$.
But $d^\perp \geq \delta_{\max}+1$ by \cite[Lemma~12]{aly07a}.
\end{proof}
\subsection{Hole's Convolutional BCH Codes}
In the previous construction of convolutional BCH codes we started with a BCH
code with parity check matrix $H=H_{2\delta+1,1}$, see
equation~(\ref{eq:bchH}), and obtained $H_0$ to be the expansion of
$H_{\delta+1,1}$. An alternate splitting of $H$ gives us the Hole's
convolutional BCH codes~\cite{hole00}. Because of space constraints we will not
explore the details or other choices of splitting the parity check matrix of
the parent BCH code.

We notice that if the matrix $H$  satisfies the conditions in
Theorem~\ref{th:noncataDualEnc}, then the convolutional code has
non-catastrophic encoder. Furthermore the minimum free distance of this code is
given by $d_f \geq d_{H_0}+d_{H_1}$ if $d_{H_0H_1} > d_{H_0}+d_{H_1}$, where
$d_{H_0}$, $d_{H_1}$, and $d_{H_0H_1}$ are the minimum distances of the block
codes $[n,n-\mu]$, $[n,n-\mu+\lambda]$, and $[n,n-2\mu+\lambda]$ respectively,
see~\cite[Proposition 2]{hole00} for more details. Also, $d_f=d_{H_0H_1}$ if
$d_{H_0H_1} \leq d_{H_0}+d_{H_1}$. We have showed in~\cite{aly06a} that there
exist a $[n,n-r\lceil (\delta-1)(1-1/q) \rceil]$ nonbinary dual-containing BCH
code with designed distance $\delta=2t+1$ and length $n=q^r-1$ for $2\leq
\delta  < \delta_{\max}=  (q^{\lceil r/2\rceil}-1-(q-2)[r \textup{ odd}])$ and
$r=\ord_n(q)$.

Let us construct the matrices $H_0$ and $H_1$ as follows. Let $\alpha$ be a
primitive element in $\F_{q^r}$. Let $2\leq t < q^{\lceil r/2
\rceil-1}+1$ and $r \geq 3$. Assume the matrix $\textbf{H}=\Big[ \begin {array}{cc} H_0\\ H_1\\
\end{array}\Big]$ has size $t(1-1/q) \times n$. We can extend every row of $H$ into $r$-tuples of powers of $\alpha$.
Now, the matrix $H_0$ has size $(\lceil t(1-1/q) \rceil-1)r \times n$ taking
the first $(\lceil t(1-1/q) \rceil-1)r$ rows of $H$.

\begin{eqnarray}
H_0 =\left[ \begin{array}{ccccc}
1 &\alpha &\alpha^2 &\cdots &\alpha^{n-1}\\
1 &\alpha^3 &\alpha^6 &\cdots &(\alpha^3)^{(n-1)}\\
\vdots& \vdots &\vdots &\ddots &\vdots\\
1 &\alpha^{\delta-4} &\alpha^{2(\delta-4)} &\cdots &\alpha^{(\delta-4)(n-1)}
\end{array}\right].
\end{eqnarray}
The matrix $H_1$ has size $(\lceil t(1-1/q) \rceil-1)r \times n$ where all
elements are zero except at the last row of $H$.
\begin{eqnarray}
H_1 =\left[ \begin{array}{ccccc}
0 &0 &0 &\cdots &0\\
0 &0 &0 &\cdots &0\\
\vdots& \vdots &\vdots &\ddots &\vdots\\
1 &\alpha^{\delta-2} &\alpha^{2(\delta-2)} &\cdots &\alpha^{(\delta-2)(n-1)}
\end{array}\right].
\end{eqnarray}

\begin{theorem}\label{th:HoleBCH}
Let $H$ be a parity check matrix defined by $H_0+D H_1$. If $H$ is canonical,
then there exists an $(n,k,m;d_f)$ convolutional code with $n=q^r-1$,
$k=n-r\lceil t(1-1/q) \rceil-r$,  $m=r$, and $d_f \geq \delta$ for $2\leq
\delta=2t+1  < \delta_{\max}= (q^{\lceil r/2\rceil}-1-(q-2)[r \textup{ odd}])$.
\end{theorem}
\begin{proof}
We first show that the parity check matrix $H=H_0+DH_1$ is canonical. We notice
that a) $H_0$ has full rank $(\lceil t(1-1/q) \rceil-1)r$ rows; since it
generates a BCH code with parameters $[n,n-(\lceil t(1-1/q) \rceil-1)r]$. b)
the last $r$ rows of $H_1$ are linearly independent. c) the rows of the matrix
$H_0$ are different and linearly independent of the last $r$ rows of $H_1$.
Therefore from~\cite[Proposition 1]{hole00}, The parity check matrix $H$ is
canonical and it generates a convolutional code $C$ with parameters
$(n,n-(\lceil t(1-1/q) \rceil-1)r,r)$.
 Second, we compute the free distance of $C$. Notice that the matrix $H_0$ defines a
BCH code with minimum distance $d_{H_0} \geq 2t-1=\delta-2$ from the BCH bound.
Also, the matrix $H_1$ defines a BCH code with minimum distance at least
2 if two columns are equal. Therefore, the BCH code generated by $\textbf{H}=\Big[ \begin {array}{cc} H_0 \\ H_1 \\
\end{array}\Big]$ with parameters $[n,n-\lceil t(1-1/q) \rceil r]$ has minimum distance $d_{\textbf{H}}\geq
\delta=2t+1$. From~\cite[Proposition 2]{hole00}, the convolutional code $C$ has
free distance $d_f \geq \delta$. \nix{ The minimum weight $\omega_0$ is at
least as the minimum distance of the block code consists of the first $(\lceil
t(1-1/q) \rceil-2)r$ rows of $H$ for $t \geq 2$. Therefore, $\omega \geq
2t-3=\delta-4$. }
\end{proof}

\section{Constructing Quantum Convolutional Codes from Convolutional BCH Codes}

In this section we derive one family of quantum convolutional codes
derived from BCH codes.  We briefly describe the stabilizer
framework for quantum convolutional codes, see also
\cite{aly07b,grassl07,ollivier04}. The stabilizer is
given by a matrix
\begin{equation}\label{stab-mat}
S(D)=(X(D)|Z(D)) \in\F_q[D]^{(n-k)\times 2n}.
\end{equation}
which satisfies the symplectic orthogonality condition $0 = X(D) Z(1/D)^t -
Z(D) X(1/D)^t$.  Let ${\cal C}$ be a quantum convolutional code defined by a
stabilizer matrix as in eq.~(\ref{stab-mat}). Then $n$ is called the frame
size, $k$ the number of logical qudits per frame, and $k/n$ the rate of ${\cal
C}$. It can be used to encode a sequence of
blocks with $k$ qudits in each block (that is, each element in the sequence
consists of $k$ quantum systems each of which is $q$-dimensional) into a
sequence of blocks with $n$ qudits.

 The memory of the quantum convolutional code is defined as
 \begin{eqnarray}
 m = \max_{1 \leq i
\leq n-k,1 \leq j \leq n}(\max(\deg X_{ij}(D),\deg Z_{ij}(D))).
\end{eqnarray}
 We use the
notation $[(n,k,m)]_q$ to denote a quantum convolutional code with the above
parameters. We can identify $S(D)$ with the generator matrix of a
self-orthogonal classical convolutional code over $\F_q$ or $\F_{q^2}$, which
gives us a means to construct convolutional stabilizer codes. Analogous to the
classical codes we can define the free distance, $d_f$ and the degree $\nu$,
prompting an extended notation $[(n,k,m;\nu,d_f)]_q$. All the parameters of the
quantum convolutional code can be related to the associated classical code as
the following propositions will show. For proof and further details see
\cite{aly07b}\footnote{A small difference exists between the notion of memory
defined here and the one used in \cite{aly07b}.}.

\begin{proposition}\label{pr:css}
  Let $(n,(n-k)/2,\nu;m)_q$ be a convolutional code such that $C \leq
  C^\perp$, where the dimension of $C^\perp$ is given by $(n+k)/2$.
  Then an $[(n,k,m;\nu,d_f)]_q$ convolutional stabilizer code exists
  whose free distance is given by $d_f=\wt(C^\perp \backslash C)$, which
  is said to be pure if $d_f = \wt(C^\perp)$.
\end{proposition}

\begin{proposition}\label{pr:c2qHerm}
Let $C$ be an $(n,(n-k)/2,\nu;m)_{q^2}$ convolutional code such that
$C\subseteq C^\hdual$.  Then there exists an $[(n,k,m;\nu,d_f)]_q$
convolutional stabilizer code, where $d_f=\wt(C^\hdual\setminus C)$.
\end{proposition}

\medskip

Under some restrictions on the designed free distance, we can use convolutional
codes derived in the previous section to construct quantum convolutional codes.
These codes are slightly better than the quantum block codes of equivalent
error correcting capability in the sense that their rates are slightly higher.

\begin{theorem}\label{th:bchQccEuclid}
Assume the same notation as in Theorem~\ref{th:bchCC}. Then there exists a
quantum convolutional code with parameters $ [(n,n-2\kappa,n)]_q$, where
$\kappa = r\ceil{\delta(1-1/q)}$. Its free distance $d_{f}\geq
\delta+1+\Delta(\delta+1,2\delta)$, and it is pure to $d'\geq \delta_{\max}+1$.
\end{theorem}
\begin{proof}
We construct a unit memory $(n,n-\kappa)_q$ classical convolutional BCH code as
per Theorem~\ref{th:bchCC}. Its polynomial parity check matrix $H(D)$ is as
given in equation~(\ref{eq:bchHD}). Using the same notation in the proof, we
see that the code contains its dual if $H$ is self-orthogonal. But given the
restrictions on the designed distance, we know from \cite[Theorem~3]{aly07a}
that the BCH block code defined by $H$ contains its dual. It follows from
Theorem~\ref{th:noncataDualEnc}(\ref{lm:CCdual}) that the convolutional BCH
code contains its dual. From \cite[Corollary~6]{aly07b}, we can conclude that
there exists a convolutional code with the parameters $[(n,n-2\kappa,n)]_q$. By
Theorem~\ref{th:bchCC} the free distance of the dual is $d' \geq
\delta_{\max}+1$, from whence follows the purity.
\end{proof}

Another popular method to construct quantum codes makes use of codes over
$\F_{q^2}$. \nix{We can use convolutional BCH codes to derive another family of
quantum convolutional codes.}

\begin{lemma}
Let $2\leq 2\delta  < \floor{n(q^r-1)/(q^{2r}-1)} $, where
and $r=\ord_n(q^2)$. Then there exist quantum convolutional codes with
parameters $ [(n,n-2\kappa,n)]_q$ and free distance $d_{f} \geq
\delta+1+\Delta(\delta+1,2\delta)$, where $\kappa=r\ceil{\delta(1-1/q^2)}$.
\end{lemma}
\begin{proof}
By Theorem~\ref{th:bchCC} there exists an  $(n,n-\kappa,1)_{q^2}$ convolutional
BCH code with the polynomial parity check matrix as in
equation~(\ref{eq:bchHD}). The parent BCH code has design distance $2\delta+1$
and given the range of $\delta$,  we know by \cite[Theorem~14]{aly07b} that it
contains its Hermitian dual. By
Theorem~\ref{th:noncataDualEnc}(\ref{lm:CCdual}), the convolutional code also
contains its Hermitian dual. By~\cite[Theorem~5]{aly07b}, we can conclude that
there exists a convolutional stabilizer code with parameters
$[(n,n-2\kappa,n)]_q$.
\end{proof}

 In~\cite{aly07b}, we have shown generalized Singleton bound for
convolutional stabilizer codes. The free distance of an $[(n,k,m;\nu,d_f)]_q$
$\F_{q^2}$-linear pure convolutional stabilizer code is bounded by
\begin{eqnarray}
d_f&\leq& \frac{n-k}{2}\left ( \left\lfloor \frac{2\nu}{n+k} \right\rfloor+1
\right) + \nu+1.
\end{eqnarray}
 The bound can be reformulated in terms of the memory $m$ instead
of the total constraint length $\nu$. Observe that if $m=0$, then it reduces to
the quantum Singleton bound viz. $d_f\leq (n-k)/2+1$.
\begin{corollary}
A pure $((n,k,m,d_f))_q$ linear quantum convolutional code obeys
$$ d_f \leq \frac{n-k}{2} \left\lfloor \frac{m(n-k)}{n+k} \right\rfloor +(n-k)(m+1)/2+1.$$
\end{corollary}
\begin{proof}
The proof is actually  straightforward. It follows from~\cite[Theorem
7]{aly07b} and the fact that $\delta \leq m(n-k)/2$
\end{proof}

\section{QCC from Product Codes} Let $(n,k,m)$ be a classical
convolutional code that encodes $k$ information into $n$ bits with
memory order $m$. We construct quantum convolutional codes based on
product codes as shown in~\cite{grassl05}. We explicitly determine
parameters of the constructed codes with the help of results
from~\cite{aly07a}. We follow the natation that has been used
in~\cite{grassl07}.

\begin{lemma}\label{lem:qcc-productcodes}
Let $C_1=(n_1,k_1,m_1)$ be a classical linear convolutional code
over $\F_q$ . Also, let $C_2=(n_2,k_2,m_2)$ be an Euclidean
self-orthogonal linear code over $\F_q$ . Then the product code $C_1
\otimes C_2=(n_1n_2-m,n_1n_2-k_1k_2,m)$ defines a quantum
convolutional code with memory $m_1*m_2$.
\end{lemma}
\begin{proof}
See~\cite[Theorem 10]{grassl05}.
\end{proof}
Now, we can restrict ourselves to one class of codes. Consider the
convolutional BCH codes derived in this chapter~\cite{aly07d}. We
know that the code is dual-containing if $\delta \leq \delta_{max}$.
In our construction, we do not require both $C_1$ and $C_2$ to be
convolutional codes or even self-orthogonal. We choose $C_1$ to be
an arbitrary convolutional code and $C_2$ can be self-orthogonal
block or convolutional code as shown in
Theorem~\ref{lem:qcc-productcodes}. Therefore, it is straightforward
to derive quantum convolutional BCH codes from BCH product  codes as
shown in Theorem~\ref{lem:qcc_bchproductcodes}. The reason we use
this construction rather than the convolutional unit memory code
construction is because the quantum codes derived from product codes
have efficient encoding circuits as shown in~\cite{grassl07}.

\begin{theorem}\label{lem:qcc_bchproductcodes}
Let $n$ be a positive integer such that $\gcd(n,q)=1$. Let $C_1$ be a
convolutional BCH code with length $n$, designed distance $\delta_1$ and memory
$m$. Let $C_2^\perp$ be a BCH code with designed distance $2\leq \delta_2  \leq
q^{\lceil r/2\rceil}-1-(q-2)[r \textup{ odd}]$. then there exists a quantum
convolutional BCH code constructed from the product code $C_1 \otimes C_2$ and
with the same parameters as $C_1$.
\end{theorem}

\begin{proof}
We know that the code $C_2$ is self-orthogonal since $2\leq \delta_2  \leq
q^{\lceil r/2\rceil}-1-(q-2)[r \textup{ odd}]$. From~\cite{grassl05}, the
convolutional product code $C_1 \otimes C_2$ is self-orthogonal and it has
memory $m$. From~\cite[Proposition 1.]{aly07d}, there exists a quantum
convolutional BCH code with the given parameters.
\end{proof}

\section{Efficient Encoding and Decoding Circuits of QCC-BCH}

Quantum convolutional codes promise to make quantum information more
reliable because they have online encoding and decoding circuits.
What we mean by online encoder and decoder is that the encoded and
decoded qudits can be sent or received with a constant delay. The
phase estimation algorithm can be used to measure the received
quantum information. In this section, we design efficient encoding
and decoding circuits for unit memory quantum convolutional codes
derived in this chapter~\cite{aly07d,aly07b}. We use the framework
established in~\cite{grassl06b,grassl07}.

 Grassl and R\"otteler showed that an
encoder circuit $\mathcal{E}$ for a quantum convolutional code $C$ exists if
the gates in $\mathcal{E}$ can be arranged into  a circuit of finite depth.
This can be applied to quantum convolutional codes derived from CSS-type
classical codes, as well as  product codes as shown in~\cite[Theorem
5]{grassl07}.

Let us assume we have two classical codes $C_1$ and $C_2$ with parameters
$(n,k_1)$ and $(n,k_2)$ and represented by a parity check matrices $H_1$ and
$H_2$, respectively. Let us construct the matrix
$$\left(\begin{array}{c|c} H_2(D)&0\\0&H_1(D)
\end{array}\right) \subseteq \F_q[D]^{(2n-k_1-k_2)\times 2n }$$
where $H_i(D)$ is the polynomial matrix of the matrix $H_i$.

We can assume that the matrix $H=H_1+H_2D$ defines a convolutional
BCH code. The matrices $H_1(D)$ and $H_2(D)$ correspond to
non-catastrophic and delay-free encoders. They also have full-rank
$k_1$ and $k_2$~\cite{aly07d}. \nix{\begin{theorem}[Convolutional
BCH codes]\label{th:bchCC} Let $n$ be a positive integer such that
$\gcd(n,q)=1$, $r=ord_n(q)$ and $2\leq 2\delta <\delta_{\max}$,
where
$$\delta_{\max}=\left\lfloor\frac{n}{q^{r}-1} (q^{\lceil
r/2\rceil}-1-(q-2)[r \textup{ odd}])\right\rfloor.$$ Then there exists a unit
memory rate $k/n$ convolutional BCH code with free distance $d_f\geq
\delta+1+\Delta(\delta+1,2\delta)$ and $k=n-\kappa$, where
$\kappa=r\ceil{\delta(1-1/q)}$.  The free distance of the dual is $\geq
\delta_{\max}+1$.
\end{theorem}}
The following theorem shows that there exists an encoding circuit
for quantum convolutional codes derived from  convolutional BCH
codes.
\begin{theorem}
Let $Q$  be a quantum convolutional code derived from convolutional BCH code as
shown in Theorem~\ref{th:bchCC}. Then $Q$ has an encoding circuit whose depth
is finite.
\end{theorem}
\begin{proof}
We know that there is a convolutional BCH code with a generator
matrix $H=H_1+H_2D$. Furthermore, the matrices $H_1$ and $H_2$
define two BCH codes with parameters $(n,k_1)$ and $(n,k_2)$.  Let
us construct the stabilizer matrix \begin{eqnarray}(X(D)|Z(D)=
\left(\begin{array}{c|c} H_2(D)&0\\0&H_1(D)
\end{array}\right) \subseteq \F_q[D]^{(2n-k_1-k_2)\times 2n }.\end{eqnarray}

The matrices $H_1(D)$ and $H_2(D)$ correspond to two encoders satisfying
\begin{inparaenum}[i)] \item they correspond to non-catastrophic encoders as shown
in~\cite[Theorem 3.]{aly07d}. \item they have full-ranks $n-k_1$ and $n-k_2$.
\item they have delay-free encoders.
\end{inparaenum}
Therefore, they have a Smith normal form given by
\begin{eqnarray}A_1(D)H_2(D)B_1(D)=\Big(I\hspace{0.3cm} 0\Big),\end{eqnarray}
for some chosen matrices  of $A_1(D) \in \F_q[D]^{(n-k_2)\times
(n-k_2)}$ and $B_1(D)\in \F_q[D]^{n\times n}$.

\end{proof}

\section{Conclusion and Discussion}

In this chapter, we presented a general method to derive unit memory
convolutional codes, and applied it to construct convolutional BCH codes. In
addition, we derived two families of quantum convolutional codes based on BCH
codes. By this construction, other families of convolutional cyclic codes can
be derived and convolutional stabilizer codes can be also constructed.

\part{Quantum and Classical  LDPC Codes}
\chapter{A Class of Quantum LDPC Codes Constructed From Finite Geometries}

Low-density parity check (LDPC) codes are a significant class of
classical codes with many applications. Several good LDPC codes have
been constructed using random, algebraic, and finite geometries
approaches, with containing cycles of length at least six in their
Tanner graphs. However, it is impossible to design a self-orthogonal
parity check matrix of an LDPC code without introducing cycles of
length four.

In this chapter, a new class of quantum LDPC codes based on lines and
points of finite geometries is constructed. The parity check
matrices of these codes are adapted to be self-orthogonal with
containing only one cycle of length four in each pair of two rows.
Also, the column and row weights, and bounds on the minimum distance
of these codes are given. As a consequence, these codes can be
encoded using shift-register encoding algorithms and can be decoded
using iterative decoding algorithms over various quantum
depolarizing channels.

\section{Introduction}\label{sec:intro}
Low density parity check (LDPC) codes are a
capacity-approaching~(\emph{Shannon limit}) class of codes that were
first described in a seminal work by Gallager~\cite{gallager62}. In
Tanner~\cite{tanner81}, LDPC codes were rediscovered and presented
in a graphical interpretation~(\emph{codes over graphs}).  Iterative
decoding of LDPC and turbo codes highlighted the importance of these
classes of codes for communication and storage channels.
Furthermore, they have been used extensively in many
applications~\cite{macKay98,lin04,liva06}.

 There have been several notable attempts to construct regular and irregular
good LDPC codes using algebraic combinatorics and random
constructions, see~\cite{song06,liva06}, and references
therein. Liva~\emph{et al.}~\cite{liva06} presented a survey of the
previous work done on algebraic constructions of LDPC codes based on
finite geometries, elements of finite fields, and RS codes.
Furthermore, a good construction of LDPC codes should have a girth
of the Tanner graph, of at least six~\cite{liva06,lin04}.

Quantum information is sensitive to noise and needs error
correction, control, and recovery strategies. Quantum block and
convolutional codes are means to protect quantum information against
noise and decoherence. A well-known class of quantum codes is called
stabilize codes, in which it can be easily constructed using
self-orthogonal (or dual-containing) classical codes,
see~\cite{calderbank98,aly07a,ketkar06} and references therein.
 Recently, subsystem codes combine the features of decoherence free subspaces,
noiseless subsystems, and quantum error-correcting codes,
see~\cite{aly06c,bacon06,kribs05b,lidar98} and references
therein.

Quantum block LDPC codes have been proposed
in~\cite{postol01,macKay04}. MacKay \emph{et al.} in~\cite{macKay04}
constructed sparse graph quantum LDPC codes based on cyclic matrices
and using a computer search. Recently, Camera \emph{el al.} derived
quantum LDPC codes in an analytical method~\cite{camara05}. Hagiwara
and Imai constructed quasi-cyclic (QC) LDPC codes and derived  a
family of quantum QC LDPC codes from a nested pair of classical
codes~\cite{hagiwara07}.

 In this chapter, we construct LDPC codes based on finite
geometry. We show that the constructed LDPC codes have quasi-cyclic
structure and their parity check matrices can be adapted to satisfy
the  self-orthogonal (or dual-containing) conditions. The
motivations for this work are that \begin{inparaenum}[(i)] \item
LDPC codes constructed from finite geometries can be encoded using
linear shift-registers. The column weights remain fixed with the
increase in  number of rows and length of the code.
\item The adapted parity check matrix has exactly one cycle with length four between any two rows and many cycles with length of at least six. \item A class of quantum LDPC codes is constructed that can be decoded using
known iterative decoding algorithms over quantum depolarizing
channels; some of these algorithms are stated in~\cite{poulin07}.
\end{inparaenum}

 \emph{Notation:} Let $q$ be a prime power
$p$ and $\F_q$ be a finite field with $q$ elements. Any two binary
vectors $\textbf{v}=(v_1,v_2,\ldots,v_n)$ and
$\textbf{u}=(u_1,u_2,\ldots,u_n)$ are orthogonal if their inner
product vanishes, i.e., $\sum_{i=1}^n v_iu_i \mod 2=0$. Let
$\textbf{H}$ be a parity check matrix defined over $\F_2$, then
\textbf{H} is self-orthogonal if the inner product between any two
arbitrary rows of \textbf{H} vanishes.

\section{LDPC Code Constructions and Finite Geometries}
\subsection{LDPC Codes}
\begin{defn} An $(\rho, \lambda)$ regular LDPC code is defined by a sparse
binary parity check matrix $\textbf{H}$ satisfying the following
properties.
\begin{compactenum}[i)]
\item $\rho$ is the number of one's in a column.
\item $\lambda$ is the number of one's in  a row.
\item Any two rows have at most one nonzero element in common. The code does not have cycles of length four in its Tanner graph.
\item $\rho$ and $\lambda$ are small in comparison to the number of rows and length of
the code. In addition, rows of the matrix \textbf{H} are not
necessarily linearly independent.
\end{compactenum}
\end{defn}

The third condition guarantees that  iterative decoding algorithms
such as sum-product or message passing perform well over
communication channels. In general it is hard to design regular LDPC
satisfying the above conditions, see
\cite{song06,liva06,lin04} and references therein.
\subsection{Finite Geometry}
Finite geometries can be classified into Euclidean and projective
geometry over finite fields.  Finite geometries codes are an
important class of cyclic and quasi-cyclic codes because their
encoder algorithms can be implemented using linear feedback shift
registers and their decoder algorithms can be implemented using
various decoding algorithms such as majority logic (MLG),
sum-product (SPA), and weighted BF, see~\cite{kou01,liva06,lin04}.

\begin{figure}[h]
  \includegraphics[scale=0.7]{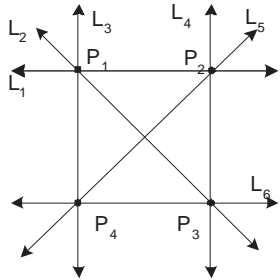}
 \centering
  \caption{Euclidean geometry with points $n=4$ and lines $l=6$}\label{ldpc1}
\end{figure}
\begin{defn}
A finite geometry with  a set of $n$ points $\{p_1,p_2,\dots,p_n\}$,
a set of $l$ lines $\{L_1,L_2,\ldots,L_l\}$ and an integer pair
$(\lambda,\rho)$ is defined as follows:
\begin{compactenum}[i)]
\item Every line $L_i$ passes through $\rho$ points.
\item Every point $p_i$ lies in $\lambda$ lines, i.e., every point
$p_i$ is intersected by $\lambda$ lines.
\item Any two points $p_1$ and $p_j$ can define one and only one line $L_k$ in between.
\item Any two lines $L_i$  and $L_j$ either intersect at only one point $p_i$ or they are parallel.
\end{compactenum}
\end{defn}

Therefore, we can form a binary matrix $\textbf{H}=[h_{i,j}]$ of
size $l \times n$ over $\F_2$. The rows and columns of \textbf{H}
correspond the $l$ lines and $n$ points in the Euclidean geometry,
respectively. If the i\emph{th} line $L_i$ passes through the point
$p_i$ then $h_{i,j}=1$, and otherwise $h_{i,j}=0$
Fig.~\ref{ldpc1} shows an example of Euclidean geometry with $n=4$,
$l=6$, $\lambda=3$, and $\rho=2$. We can construct the incidence
matrix $\textbf{H}$ based on this geometry where every point and
line correspond to a column and row, respectively.  For $\rho << l$
and $\lambda <<n$, The matrix $\textbf{H}$ is a sparse low density
parity check matrix. In this example, the matrix $\textbf{H}_{EG-I}$
is given by

\begin{eqnarray}
\textbf{H}_{EG-I}=\left( \begin{array}{cccc}1&1&0&0 \\1&0&1&0 \\ 1&0&0&1 \\ 0&1&1&0\\
0&1&0&1\\ 0&0&1&1 \end{array} \right)
\end{eqnarray}

We call the Euclidean geometry defined in this type as a
$\textbf{Type-I EG}$. The Tanner graph of \textbf{Type-I EG}  is a
regular bipartite graph with $n$ code variable vertices and $l$
check-sum vertices. Also, each variable bit vertex has degree
$\lambda$ and each check-sum has degree $\rho$.

If we can take the transpose of this matrix $\textbf{H}_{EG-I}$,
then we can also define a $(\rho,\lambda)$ LDPC code with length $l$
and minimum distance is at least $\rho+1$. The codes defined in this
type are called LDPC codes based on $\textbf{Type-II EG}$. In this
type, any two rows intersect at exactly one position.

\begin{figure}[t]
  \includegraphics[scale=0.7]{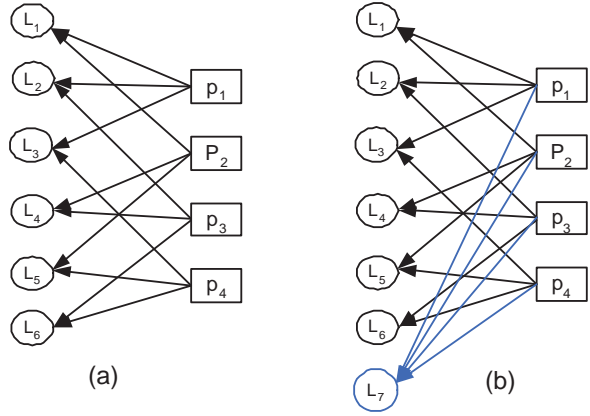}
 \centering
  \caption{(a) EG with $n=4$ points  and $l=6$ lines (b) The Tanner graph of a self-orthogonal H
matrix.}\label{ldpc2}
\end{figure}

\subsection{Adapting the Matrix $\textbf{H}_{EG-II}$ to be Self-orthogonal}
Let $\textbf{H}_{EG-II}$  be  a parity check matrix of a regular
LDPC code constructed based on \textbf{Type-II EG} Euclidean
geometry. We can construct a self-orthogonal matrix
$\textbf{H}_{EG-II}^{orth}$ from $\textbf{H}_{EG-II}$ in two cases.

\textbf{Case 1.} If the number of one's in a row is odd and any two
rows intersect at exactly one position, i.e., any line connects two
points.  As shown in Fig.~\ref{ldpc2}, the Tanner graph corresponds
to a self-orthogonal parity check matrix $\textbf{H}_{EG-II}^{orth}$
if and only if every check-sum has even degree and any any two
check-sum nodes meet at even code variable nodes. This condition is
the same as every row in the parity check matrix
$\textbf{H}_{EG-II}^{orth}$ has an even weight and any two rows
overlap in even nonzero positions.

\begin{small}
\begin{eqnarray}
\textbf{H}_{EG-II}^{orth}&=&\Big(\begin{array}{c|c}
\textbf{H}^T&\textbf{1}
\end{array} \Big)\end{eqnarray}
\end{small}
The vector $\textbf{1}$ of length $n$ is added as the last column in
$\textbf{H}_{EG-II}^{orth}$.

 \textbf{Case 2.} Assume the number of
one's in a line is even and any two rows intersect at exactly one
position. We can construct a self-orthogonal parity check matrix
$\textbf{H}_{EG-II}^{orth}$ as follows. We add the vector
$\textbf{1}$ along with the identity matrix \textbf{I} of size
$n\times n$. We guarantee that  any two rows of the matrix
$\textbf{H}_{EG-II}^{orth}$ intersect at two nonzero positions and
every row has an even weight.
\begin{eqnarray}
\textbf{H}_{EG-II}^{orth}&=&\left(\begin{array}{c|c|c}
\textbf{H}^T&\textbf{1}&\textbf{I}\end{array} \right).
\end{eqnarray}

\subsection{Characteristic Vectors and Matrices}

Let $n$ be a positive integer such that $n=q^m-1$, where
$m=\ord_n(q)$ is the multiplicative order of $q$ modulo $n$. Let
$\alpha$ denote a fixed primitive element of~$\F_{q^m}$. Define a
map $\textbf{z}$ from $\F_{q^m}^*$ to $\F_2^n$ such that all entries
of $\textbf{z}(\alpha^i)$ are equal to 0 except at position $i$,
where it is equal to 1. For example,
$\textbf{z}(\alpha^2)=(0,1,0,\ldots,0)$.  We call
$\textbf{z}(\alpha^k)$ the location (or characteristic) vector of
$\alpha^k$. We can define the location vector
$\textbf{z}(\alpha^{i+j+1})$ as the right cyclic shift of the
location vector $\textbf{z}(\alpha^{i+j})$, for $0 \leq j \leq n-1$,
and the power is taken module $n$. The location vector can be
extended to two or more nonzero positions. for example, the location
vector of $\alpha^2$, $\alpha^3$ and $\alpha^5$ is given by
$\textbf{z}(\alpha^2,\alpha^3,\alpha^5)=(0,1,1,0,1,0,\ldots,0)$.

\begin{defn}\label{def:Amatrix}We can define a map $A$ that associates to an element $\F_{q^m}^*$ a circulant
matrix in $\F_2^{n\times n}$ by
\begin{eqnarray}\label{label:mapA} A(\alpha^i)=\left ( \begin{array}{ccc}
\textbf{z}(\alpha^i) \\  \textbf{z}(\alpha^{i+1})
\\ \vdots \\  \textbf{z}(\alpha^{i+n-1})
\end{array} \right).
\end{eqnarray}
By construction, $A(\alpha^k)$ contains a 1 in every row and column.
\end{defn}

We will use the map $A$ to associate to a parity check matrix
$H=(h_{ij})$ in $(\F_{q^m}^*)$ the (larger and binary) parity check
matrix $\textbf{H}=(A(h_{ij}))$ in $\F_2^{n \times n}$. The matrices
$A(h_{ij})$$'s$ are $n \times n$ circulant permutation matrices
based on some primitive elements $h_{ij}$ as shown in
Definition~\ref{def:Amatrix}.

\section{Constructing Self-Orthogonal Cyclic LDPC Codes from Euclidean
Geometry}\label{sec:LDPCcodes}
In this section we construct self-orthogonal algebraic Low Density
Parity Check (LDPC) codes based on finite   geometries. Particulary,
there are two important classes of finite geometries: Euclidean and
projective geometry.

\subsection{Euclidean Geometry $EG(m,q)$}
We construct regular LDPC codes based on lines and points of
Euclidean geometry. The class we derive has a cyclic structure, so
it is called cyclic LDPC codes. Cyclic LDPC codes can be defined by
a sparse parity check matrix or by a generator polynomial and can be
encoded using shift-register. Furthermore, they can be decoded using
well-known iterative decoding algorithms~\cite{lin04,liva06}.

Let $q$ be power of a prime $p$, i.e. $q=p^s$ for some integer
$s\geq 2$. Let $EG(m,q)$ be the m-dimensional Euclidean geometry
over $\F_q$ for some integer $m \geq 2$. It consists of $p^{ms}=q^m$
points and every point is represented by an m-tuple,
see~\cite{kou01}. A line in $EG(m,q)$ can be described by a
$1$-dimensional subspace of the vector space of all $m$-tuples over
$\F_q$ or a coset of it. The number of lines in $EG(m,q)$ is given
by \begin{eqnarray}(q^{m-1})(q^{m}-1)/(q-1),\end{eqnarray} and each
line passes through $q$ points. Every line has $q^{(m-1)}-1$ lines
parallel to it. Also, for any point in $EG(m,q)$, there are
\begin{eqnarray}(q^{m}-1)/(q-1),\end{eqnarray} lines intersect at this point. Two lines can
intersect at only one point or they are parallel.

Let $\F_{q^m}$ be the extension field of $\F_q$. We can represent
each element in $\F_{q^m}$ as an $m$-tuple over $\F_q$. Every
element in the finite field $\F_{q^m}$ can be looked as  a point in
the Euclidean geometry $EG(m,q)$, henceforth $\F_{q^m}$ can be
regarded as the Euclidean geometry $EG(m,q)$.

Let $\alpha$ be a primitive element of $\F_{q^m}$. $q^m$ points of
$EG(m,q)$ can be represented by elements of the set $\{
0,1,\alpha,\alpha^2,\ldots,\alpha^{q^m-2} \}$. We can also define a
line $L$ as the set of points of the form $\{ \textbf{a}+\gamma
\textbf{ b} \mid \gamma \in \F_q \}$, where $\textbf{a}$ and
$\textbf{b}$ are linearly independent over $\F_{q}$. For a given
point $\textbf{a}$, there are $(q^m-1)/(q-1)$ lines in $EG(m,q)$
that intersect at $\textbf{a}$.

\textbf{Type-I EG.} Let $n=q^m-1$ be the number of points excluding
the original point $\textbf{0}$ in $EG(m,q)$. Assume $L$ be a line
not passing through $\textbf{0}$. We can define the binary vector
\begin{eqnarray}\textbf{v}_L=(v_1,v_1,\ldots,v_{n}),\end{eqnarray}
 where $v_i=1$ if the
point $\alpha^i$ lies in a line $L$.  The vector $\textbf{v}_L$ is
called the incidence vector of $L$. Elements of the vector
$\textbf{v}_L$ correspond to the elements
$1,\alpha,\alpha^2,\ldots,\alpha^{n-1}$. $\alpha L$ is also a line
in $EG(m,q)$, therefore $\alpha \textbf{v}_L$ is a right
cyclic-shift of the vector $\textbf{v}_L$. Clearly, the lines
$L,\alpha L,\ldots,\alpha^{n-1}L$ are all different. But, they  may
not be linearly independent.

Consider the vectors $L_i,\alpha L_i,\ldots,\alpha^{n-1} L_i$. We
can construct an $n \times n$ matrix $H_i$ in the form

\begin{eqnarray} H_i=\left ( \begin{array}{cccccc} \textbf{v}_{L_i} \\ \alpha
\textbf{v}_{L_i}\\ \vdots \\ \alpha^{n-1}\textbf{v}_{L_i}
\end{array} \right)
\end{eqnarray}
Clearly, $H_i$ is a circulant matrix with column and row weights
equals to $q$, the number of points that lie in a  line $\alpha^j
L_i$, for $0 \leq j \leq n-1$. $H_i$ has size of $n \times n$. The
total number of lines   in $EG(m,q)$ that do not pass through the
origin $\textbf{0}$ are given by
\begin{eqnarray}(q^{m-1}-1)(q^m-1)/(q-1)\end{eqnarray}
They can be partitioned into $(q^{m-1}-1)/(q-1)$ cyclic classes,
see~\cite{liva06}. Every class $\mathcal{H}_i$ can be defined by an
incidence vector $L_i$ as $\{L_i,\alpha L_i,\alpha^2
L_i,\ldots,\alpha^{n-1} L_i\}$ for $1 \leq i \leq
(q^{m-1}-1)/(q-1)$. Let $1 \leq \ell \leq (q^{m-1}-1)/(q-1)$, then
$\mathcal{H}_{EG,\ell}$ is defined as

\begin{eqnarray} \mathcal{H}_{EG,\ell}= \Big[ \begin{array}{cccccc} \mathcal{H}_1 & \mathcal{H}_2 &
\ldots & \mathcal{H}_\ell
\end{array} \Big]^T.
\end{eqnarray}

For each cyclic class $\mathcal{H}_i$, we can form the matrix
$\mathbf{H}_i$ over $\F_2$ of size $n \times n$. Therefore,
$\mathbf{H}_i$ is a circulant binary matrix of row and column
weights of q.

If we assume that there are $1\leq \ell \leq (q^{m-1}-1)/(q-1)$
incidence lines in $EG(m,q)$ not passing through the origin, then we
can form the binary matrix

\begin{eqnarray} \textbf{H}_{EG,\ell}=\Big [ \begin{array}{cccccc} \textbf{H}_1 &
\textbf{H}_2&
\ldots& \textbf{H}_\ell
\end{array} \Big]^T.
\end{eqnarray}

The matrix $\textbf{H}_{EG,\ell}$ consists of a $\ell$ sub-matrices
$\textbf{H}_i$ of size $n \times n$ and it has column and row
weights $\ell q$ and $q$, respectively. The null space of the matrix
$\textbf{H}_{EG,\ell}$ gives a cyclic EG-LDPC code of length
$n=q^m-1$ and minimum distance $\ell q+1$, whose Tanner graph has a
girth of at least six, see~\cite{song06,liva06}.

The Tanner graph of \textbf{Type-I EG}  is a regular bipartite graph
with $q^m-1$ code variable vertices and $l$ check-sum vertices.
Also, Each variable bit vertex has degree $\rho=q$ and each
check-sum has degree $\lambda=\ell q$.


\textbf{Type-II EG.} We can take the transpose of the parity check
matrix $\mathcal{H}_{(EG,\ell)}$ over $\F_{q^m}$ as defined  in
$\textbf{Type-I}$ to define a new parity check matrix with the
following properties, see~\cite{kou01}.
\begin{eqnarray} \mathcal{H}_{EG,\ell}^T= \Big[ \begin{array}{cccccc} \mathcal{H}_1^T & \mathcal{H}_2^T&
\ldots& \mathcal{H}_\ell^T
\end{array} \Big]
\end{eqnarray}
So, the matrix $\mathcal{H}_i^T$ is the transpose matrix of
$\mathcal{H}_i$. Consequently, we can define the binary matrix
$\textbf{H}_{EG,\ell}$

\begin{eqnarray} \textbf{H}_{EG,\ell}^T=\Big [ \begin{array}{cccccc} \textbf{H}_1^T & \textbf{H}_2^T&
\ldots& \textbf{H}_\ell^T
\end{array} \Big].
\end{eqnarray}

Let $\ell=(q^{m-1}-1)/(q-1)$, then the matrix
$\textbf{H}_{EG,\ell}^T$ has the following properties
\begin{compactenum}[i)]
\item
The total number of columns is given by $\ell n=
(q^{m-1}-1)(q^m-1)/(q-1)$.

\item Number of rows is given by $n=q^m-1$.
\item The rows of this matrix correspond to the nonorigin points of $EG(m,q)$
and the columns correspond to the lines in $EG(m,q)$ that do not
pass through the origin. \item $\lambda=\ell q=
q(q^{m-1}-1)/(q-1)=(q^m-1)/(q-1)-1$ is the row weight for $\ell =
(q^{m-1}-1)/(q-1)$. Also $\rho=q$ is the column weight.
\item Any two rows  in $\textbf{H}_{EG,\ell}^T$ have exactly one nonzero element in common. Also, any two
columns have at most one nonzero element in common.
\item The binary sub-matrix $\textbf{H}_i^T$ has size $(q^m-1)\times (q^m-1)$. Also, it can be
constructed using only one vector $\textbf{v}_L$ that will be
cyclically shifted $q^m-1$ times.
\end{compactenum}

\subsection{QC LDPC Codes} The matrix $\textbf{H}^{T}_{EG,\ell}$
defines a quasi-cyclic (QC) LDPC code of length $N=\ell
n=(q^{(m-1)}-1)(q^m-1)/(q-1)$ for $\ell=(q^{m-1}-1)/(q-1)$. The
matrix $\textbf{H}^{T}_{EG,\ell}$ has $ n=q^m-1$ rows that are not
necessarily independent. We can define a QC LDPC code over $\F_2$ as
the null-space of the matrix $\textbf{H}^{T}_{EG,\ell}$ of sparse
circulant sub-matrices of equal size. The matrix
$\textbf{H}^{T}_{EG,\ell}$ with parameters $(\rho, \lambda)$ has the
following properties.
\begin{compactenum}[i)]
\item $\rho=q$ is the weight of a column $c_i$. $\rho$ does not
depend on $m$, hence length of the code can be increased without
increasing the column weight.
\item $\lambda=\ell q$ is the weight of a row $r_i$. $\lambda$
depends on $m$, but the length of the code increases much faster
than $\lambda$.
\item Every  two columns intersect at most at one nonzero position. Every two rows  have exactly one and only one nonzero
position in common.
\end{compactenum}

 From this definition, the minimum distance of the LDPC code defined by the null-space of
$\textbf{H}^{T}_{EG,\ell}$ is at least $\rho+1$. This is because we
can add at least $\rho+1$ columns in the parity check matrix
$\textbf{H}^{T}_{EG,\ell}$ to obtain the zero column (rank of
$\textbf{H}^{T}_{EG,\ell}$ is at least $(\rho+1$)). Furthermore, the
girth of the Tanner graph for this matrix $\textbf{H}_i$ is at least
six, see~\cite{macKay98,song06}. This is a $(\rho,\lambda)$ QC LDPC
code based on $\textbf{Type-II EG}$.

\subsection{Self-orthogonal QC LDPC
Codes}\label{sec:LDPCcodesorthogonal}

We can define a self-orthogonal parity check matrix
$\textbf{H}^{orth}_{EG,\ell}$ from $\textbf{Type-II EG}$
construction as follows. The binary matrix
$\textbf{H}^{T}_{EG,\ell}$ of size $n \times \ell n$ for $1 \leq
\ell \leq (q^{m-1}-1)/(q-1)$  has row and column weights of
$\lambda=\ell q$ and $\rho=q$, respectively. Let \textbf{1} be the
column vector of size $(q^m-1) \times 1$ defined as
$\textbf{1}=(1,1,\ldots,1)^T$. If the weight of a row in
$\textbf{H}^{T}_{EG,\ell}$ is odd, then we can add the vector
\textbf{1} to form the matrix $\textbf{H}^{orth}_{EG,\ell}=\Big[
\textbf{H}_{EG,\ell}^T  \mid \textbf{1} \Big]$. Also, if the weight
of a row in $\textbf{H}^{T}_{EG,\ell}$ is even, then we can add the
vector \textbf{1} along with the identity matrix of size
$(q^m-1)\times (q^m-1)$ to form $\textbf{H}^{orth}_{EG,\ell}=\Big[
\textbf{H}_{EG,\ell}^T  \mid \textbf{1} \mid \textbf{I} \Big]$.
Therefore, we can prove that $\textbf{H}^{orth}_{EG,\ell}$ is
self-orthogonal as shown in the following Lemma.

\begin{lemma}\label{lem:Hself-orthogonal}
The parity check matrix $\textbf{H}^{orth}_{EG,\ell}$ defined as
\begin{eqnarray}\textbf{H}^{orth}_{EG,\ell} =\left\{
  \begin{array}{ll}
    \Big[\begin{array}{cccc|c}\textbf{H}_1^T&\textbf{H}_2^T&\dots&\textbf{H}_\ell^T
&\textbf{1}
\end{array}\Big], \mbox{for odd $\ell q$;}  \\ \hspace{0.2cm} \\
    \Big[\begin{array}{cccc|c|c}\textbf{H}_1^T&\textbf{H}_2^T&\dots&\textbf{H}_\ell^T
&\textbf{1}&\textbf{I}
\end{array}\Big],  \mbox{for even $\ell q$} \nonumber
  \end{array}
\right.
\end{eqnarray}
is self-orthogonal.
\end{lemma}

\begin{proof}
From the construction $\textbf{Type-II EG}$, any two different rows
intersect (overlap)  in exactly one nonzero position. If $\ell q$ is
odd, then adding the column vector \textbf{1} will result an even
overlap as well as rows of even weights. Therefore, the inner
product $\mod 2$ of any arbitrary rows vanishes. Also, if $\ell q$
is even, adding the columns $\Big[ \textbf{1} \mid \textbf{I} \Big]$
will produce row of even weights and the inner product $\mod 2$ of
any arbitrary rows vanishes.
\end{proof}

$\textbf{H}^{orth}_{EG,\ell}$ has size $n \times N$ for odd $\ell q$
where $n=q^m-1$, $N= n\ell +1$, and $1 \leq \ell \leq
(q^{(m-1)}-1)/(q-1)$. Also, it has length  $N=n(\ell+1)+1$ for even
$\ell q$.

%
The minimum distance of the LDPC codes constructed in this
type can be shown using the BCH bound as stated in the following
result.

\begin{lemma}
The minimum distance of an LDPC defined by the parity check matrix
$\textbf{H}^{orth}_{EG,\ell}$  is at least  $q+1$.
\end{lemma}


\section{Quantum LDPC Block Codes}\label{sec:QLDPCcodes}

In this section we derive a family of LDPC stabilizer codes derived
from LDPC codes based on finite geometries.  Let $P=\{I,X,Z, Y=iXZ\}$ be a
set of Pauli matrices defined as
\begin{eqnarray} I=\left( \begin{array}{cc} 1 &0 \\0&1 \end{array}\right),
X=\left(
\begin{array}{cc} 0 &1 \\1&0 \end{array}\right), Z=\left( \begin{array}{cc} 1
&0 \\0&-1 \end{array}\right) \end{eqnarray}
and the matrix $Y$ is the combination of the matrices $X$ bit-flip
and $Z$ phase-flip defined as $Y=iXZ=\left(
\begin{array}{cc} 0 &-i
\\i&0
\end{array}\right)$. Clearly, $$X^2=Z^2=Y^2=I.$$

A well-known method to construct quantum codes is by using the
stabilizer formalism, see for
example~\cite{aly08thesis,calderbank98,gottesman97,macKay04} and references
therein. Assume we have a stabilizer group $S$ generated by a set
$\{S_1,S_2,\ldots,S_{n-k}\}$ such that every two row operators
commute with each other. The error operator $S_j$ is a tensor
product of $n$ Pauli matrices. $$S_j=E_1\otimes E_2\otimes \ldots
\otimes E_n, \hspace{0.3cm} E_i \in P.$$ $S_j$ can be seen as a
binary vector of length $2n$~\cite{macKay04,calderbank98}.  A
quantum code $Q$ is defined as +1 joint eigenstates of the
stabilizer $S$. Therefore, a codeword state $\ket{\psi}$ belongs to
the code $Q$ if and only if
\begin{eqnarray}S_j\ket{\psi}=\ket{\psi} \mbox{ for all } S_j \in
S.\end{eqnarray}

\textbf{CSS Construction:} Let $\textbf{G}$ and $\textbf{H}$ be two
binary matrices define the classical code $C$ and dual code
$C^\perp$, respectively. The CSS construction assumes that the
stabilizer subgroup (matrix) can be written as
\begin{eqnarray}
\textbf{S} = \left( \begin{array}{c|c} \textbf{H} & \textbf{0}
\\\textbf{0} &\textbf{G}
\end{array}\right)
\end{eqnarray}
where $\textbf{H}$ and $\textbf{G}$ are $k \times n$ matrixes
satisfying $\textbf{HG}^T=\textbf{0}$. The quantum code with
stabilizer $\textbf{S}$ is able to encode $n-2k$ logical qubits into
$n$ physical qubits. If $\textbf{G}=\textbf{H}$, then  the
self-orthgonality or dual-containing condition becomes
$\textbf{HH}^T=\textbf{0}$. If $C$ is a code that has a parity check
matrix $\textbf{H}$, then $C^\perp \subseteq C$.

 \textbf{Constructing Dual-containing LDPC Codes:} Let us
construct the stabilizer matrix
\begin{eqnarray}
S_{stab} = \Big( \begin{array}{c|c} H_{X} & 0 \\0 &H_{Z}
\end{array}\Big).
\end{eqnarray}

The matrix $\textbf{H}_{EG,\ell}^{orth}$ is a binary self-orthogonal
matrix as shown in Section~\ref{sec:LDPCcodesorthogonal}. We replace
every nonzero element in $\textbf{H}_{EG,\ell}^{orth}$ by the Pauli
matrix $X$ to form the matrix $H_X$. Similarly, we replace every
nonzero element in $\textbf{H}_{EG,\ell}^{orth}$ by the Pauli matrix
$Z$ to form the matrix $H_Z$. Therefore the matrix $S_{stab}$ is
also self-orthogonal. We can assume that the matrix $H_X$ corrects
the bit-flip errors, while the matrix $H_Z$ corrects the phase-flip
errors, see~\cite{macKay04,aly08thesis}.
\begin{lemma}\label{def:qldpc}
A  quantum LDPC code $Q$ with rate $(n-2k)/n$ is a code whose
stabilizer matrix $S_{stab}$ of size $2k \times 2n$ has a pair
$(\rho,\lambda)$ where $\rho$ is the number of non-zero error
operators in a column and $\lambda$ is the number of non-zero error
operators in a row. Furthermore, $S_{stab}$ is constructed from a
binary self-orthogonal parity check matrix
$\textbf{H}_{EG,\ell}^{orth}$ of size $k \times n$.
\end{lemma}

Using Lemma~\ref{def:qldpc} and LDPC codes given by the parity check
matrix $\textbf{H}_{EG,\ell}^{orth}$ as shown in
Section~\ref{sec:LDPCcodesorthogonal}, we can derive a class of
quantum LDPC codes as stated in the following Lemma. 

\begin{theorem}\label{lem:qldpcparameters}
 Let $\textbf{H}_{EG,\ell}^{orth}$ be a parity check matrix of an LDPC code based on $EG(m,q)$, where $n=q^m-1$ and $1\leq \ell \leq (q^{m-1}-1)/(q-1)$. Then,
there exists a quantum LDPC code $Q$ with parameters $[[N,N-2n,\geq
q+1]]_2$ where $N=\ell n+1$ for odd $\ell q$ and $N=(\ell +1)n+1$
for even $\ell q$.
\end{theorem}

\begin{proof}
By Lemma~\ref{lem:Hself-orthogonal}, $\textbf{H}_{EG,\ell}^{orth}$
is self-orthogonal. Using Lemma~\ref{def:qldpc}, there exists a
quantum LDPC code with the given parameters.
\end{proof}

\section{Conclusion}
We constructed a class of quantum LDPC codes derived from finite
geometries. The constructed codes have high rates and their minimum
distances are bounded. They only have one cycle of length four between any two rows and many cycles of length of at least six.  A new class of
quantum LDPC codes based on projective geometries can be driven in a
similar way.

\chapter{Quantum LDPC Codes Derived from \emph{Latin} Squares}
 In this chapter I construct a class of regular Low Density
Parity Check (LDPC) codes derived from  \emph{Latin} squares. The parity check
matrices of these codes are constructed by permuting orthogonal \emph{Latin}
squares of order $n$ in block-rows and block-columns. I show that the
constructed LDPC codes are self-orthogonal and their minimum and stopping
distances are bounded. This helps us to construct a family of quantum LDPC
block codes. Consequently, I demonstrate that these constructed codes have good
error correction capabilities  and can be decoded using iterative decoding
algorithms similar to  their classical counterpart. Therefore, this work shows
that cycles of length $4$ in the Tanner graphs of the parity check matrices do
not greatly affect performance of LPDC codes if they can be distributed
regularly.

\section{Introduction}

Low Density Parity Check  (LDPC) codes are a capacity approaching
(\emph{Shannon limit}) class of codes that first appeared in a seminal work by
Gallager~\cite{gallager63}. LDPC codes were rediscovered by
Tanner~\cite{tanner81}, in which he showed the  interpretation graphical view
of these codes (\emph{codes over graphs}). Iterative decoding of LDPC and turbo
codes highlighted these codes as  important classes of codes (modern coding
theory) for communication and storage channels. Furthermore, they have been
used intensively in many applications~\cite{macKay98,lin04}. Rather than, BCH
and Reed-Solomon cyclic codes, LDPC  codes are often historically constructed
by a computer search. Also, their encoding complexity is high in comparison to
other codes. However, LDPC codes have high performance and better error
correction capabilities because they have iterative decoding
algorithms~\cite{tanner04,song06,liva06,lin04}.

Quantum information is sensitive to noise and needs error correction
strategies. Quantum block and convolutional codes are means to protect quantum
information.  Quantum block LDPC codes have been introduced using a computer
search by MacKay in~\cite{macKay04}. He constructed sparse graph quantum codes
from classical LDPC codes. Recently, Camera \emph{el al.} derived quantum LDPC
codes in an analytical method~\cite{camara05}. Quantum convolutional codes
(\emph{quantum memory codes}) are an alternate to quantum blocks codes
(\emph{quantum memoryless codes}). Quantum convolutional codes promise to make
quantum communication more reliable because of their online encoding and
decoding algorithms, see~\cite{grassl07,forney05b,aly07b}.

We investigate the problem of constructing good quantum error correcting codes.
Recently, Hagiwara and Imai constructed quasi-cyclic (QC) LDPC codes and
derived  a family of quantum QC LDPC codes from a nested pair of classical
codes~\cite{hagiwara07}. In our work we establish sufficient conditions for the
parity check matrix $\textbf{H}$ of a LDPC code to be self-orthogonal.

In this chapter, a new class of quantum LDPC codes based on our construction of
LDPC codes is proposed. We derive regular LDPC codes from elements of finite
fields (\emph{Latin} squares) and algebraic combinatorics~\cite{aly07b}.
Quantum LPDC block codes constructed in this chapter have some advantages;
\begin{inparaenum}[(a)]
\item  quantum block codes constructed from LDPC are good codes  as shown by MacKay
et al.~\cite{macKay04}, \item LDPC codes are capacity achieving codes and have
high rates, \item the constructed codes can be decoded using standard iterative
decoding algorithms.
\end{inparaenum}

The constructed codes  have cycles with length $4$ to guarantee
self-orthogonality  as we will show in
section~\ref{ch_QBC_LDPC:sec:background}. Moreover, we show that the
performance of these codes is reasonable and can be improved by reducing the
number of 4-cycles in the parity check matrix. We also note that the these
codes have high rates. This is due to the fact that we try to  have less
4-cycle, dimension of the parity check matrix is reduced, i.e. $R \geq 1-k/n$.
Finally, performance of our constructed codes can be improved by shortening and
puncturing the parity check matrices of these codes to reduce the number of
cycles with length $4$.

\emph{Notation:} We will refer to a row of matrices (block) as a block-row and
a regular row of elements through out some matrices as  a row. This is also
applied to a block-column.
\section{Classical and Quantum LDPC Codes}\label{ch_QBC_LDPC:sec:background}
In this section we introduce quantum and classical LDPC codes. Our goal is to
make this chapter as self-contained as possible.

\subsection{Quantum LDPC Codes}

Quantum LDPC first appeared in a paper by Mackay~\emph{el. al.}
in~\cite{macKay04}. He showed that good quantum block codes can be constructed
from classical codes with low-weight codewords. So, it is not necessary to
start with a good classical code that has high minimum distance.
He showed analytically that:\\
\begin{proposition}\label{pros1:QLDPC}
 A $(\rho,\lambda,n)$-LDPC code is a dual-containing code if it has a
parity check matrix $H$ over $\F_2$ such that
\begin{compactenum}[i)]
\item Every row has fixed weight $\lambda$ and every column has fixed weight $\rho$.
\item Every pair of rows in $H$ has an even overlap, and every row has even
weight, meaning every pair of rows is \emph{multiplicity even}.
\end{compactenum}
\end{proposition}

MacKay used the random construction of LDPC codes to derive quantum codes.
Recently, Camara \emph{el al.} showed  quantum convolutional LDPC codes using
analysis method~\cite{camara05}. They presented a class of quantum codes that
can be decoded using iterative algorithms. We now can define  quantum LDPC
codes using the row and column weights.
\begin{definition}
A quantum LDPC code is a code whose stabilizer matrix $S_{stab}$ has a pair
$(\rho,\lambda)$ where $\rho$ is the number of non-zero error operators per
column and $\lambda$ is the number of non-zero error operators per  row.
\end{definition}

For the binary case, the error operator can be an element in the Pouli group
generated by the matrices $\{I,X,Z,Y=iXZ \}$.

\subsection{Classical LDPC Codes}
LDPC codes, whether they are block or convolutional, have better encoding and
decoding algorithms in comparison to other codes. In fact this class of codes
can be encoded using shift register circuits, see for
example~\cite{song06,song06-b,macKay04,tanner04} and the recent survey
paper~\cite{liva06}. LDPC codes that have an algebraic structure are superior
because
\begin{inparaenum}[i)] \item they perform well in terms of bit and block error
probabilities, and \item they are easy to encode and decode.
\end{inparaenum}

We pursue our construction by defining some terms. Let $\F_q$ be a finite field
with $q$ elements. We can define a QC-LDPC code over $\F_q$ as the null-space
of a matrix $\textbf{H}$ of sparse circulants of equal size. The matrix
$\textbf{H}$ with parameters $(\rho, \lambda)$ has the following properties:
\begin{compactenum}
\item $\rho$ is the weight of a column $c_i$,
\item $\lambda$ is the weight of a row $r_i$.

\end{compactenum}
From this definition, the minimum distance of the QC-LDPC defined by the
null-space of $\textbf{H}$ is at least $\rho+1$. This is because we can add at
least $\rho+1$ columns in the parity check matrix $H$ to get the zero column
(rank of \textbf{H} is at least $\rho+1$). Furthermore, the girth of the Tanner
graph for this matrix $\textbf{H}$ is at least 6, see~\cite{macKay98}.

Consider $q=p^m$ for some prime $p$ and positive integer $m\geq 2$. Let
$\alpha$ be a primitive element in $\F_q$. The finite field $\F_{p^m}$ can be
generated by some primitive elements $\alpha^i$ for $1 \leq i \leq p$. So, the
set $S=\{\alpha^0=1, \alpha, \alpha^2, \ldots, \alpha^{q-2}, \alpha^{q-1}=1,
\alpha^\infty=0\}$ form all elements in $\F_{p^m}$. Clearly if $m=1$, then
there are $q-1$ primitive elements in this field. We also note that the set $S
\backslash \{0\}$, equivalently $\F_q^*$, form a multiplicative group of order
$n$. This is a curial part of our construction.

Every nonzero element $\alpha^i$ in $\F_q$ can be written as a zero vector of
length $n=q-1$ except at position $i$. So,
$\textbf{z}(\alpha^i)=(z_0,z_1,\ldots,z_n)$ for $z_i=\alpha^i$ and $z_j=0$
where $i \ne j$. Also,  $\textbf{z}(0)=(0,0,\ldots,0)$. Clearly, the weight of
the vector $\textbf{z}(\alpha^i)$ is equal to one. We will assume the vector
$\textbf{z}$ is defined over $F_2$ instead of $F_q$. For example,
$\textbf{z}(\alpha^2)=(0,1,0,\ldots,0)$.

Let $\gamma$ be a nonzero element in $\F_q$. We can define the location vector
$ \textbf{ z}( \gamma \alpha^i)$ as the cyclic shift of the location vector
$\textbf{z}(\alpha^i)$. Let $A$ be a $n \times n$ matrix over $\F_2$.

\begin{eqnarray} A=\left ( \begin{array}{ccc} \textbf{z}(\alpha^i) \\  \textbf{ z}(\gamma \alpha^i)
\\ \vdots \\  \textbf{z}(\gamma^{n-1} \alpha^i) \end{array} \right) \end{eqnarray}

From this construction every row or column of the matrix $A$ contains only one
nonzero entry.  Now, we give two definitions to measure the performance of the decoding
algorithms of LDPC codes: girth of a Tanner graph and stopping sets. The
minimum stopping set is analogous to the minimum Hamming distance of linear
block codes.

\medskip

\begin{defn}[Girth of a Tanner graph]
The girth $g$ of the Tanner graph is a length of its minimum cycle.
\end{defn}

The stopping set of a Tanner graph is a subset of the variable nodes $V$ such
that its neighboring check nodes in $L$ are connected to at least two nodes in
this subset as shown in the following definition. The stopping distance is the
size of the smallest stopping set and it determines the number of correctable
 erasures by an iterative decoding algorithm, see for
example~\cite{orlitsky05,schwartz06,di00}.
\medskip

\begin{defn}[Stopping sets]\label{def:stoppingset1}
The set $S\subseteq C$ is called the stopping set of a graph $G=(V,C,E)$ if the
degree of each vertex in $\Gamma(S)$ in the induced graph $G_S$ on $S \cup
\Gamma(S)$ is at least two, where $\Gamma(S)$ is the set of neighbors of $S$ in
$V$.
\end{defn}
Let $s$ be the size of the smallest stopping set, i.e., $s$ is the stopping
distance (number).  We can also define the stopping distance from $\textbf{H}$
directly as follows~\cite{schwartz06}.
\begin{defn}[Stopping distance]\label{def:stoppingset2}
 The stopping distance of the parity check matrix $\textbf{H}$ is defined as the largest integer
$s(\textbf{H})$ such that every set of $(s(\textbf{H})-1)$ or less columns of
\textbf{H} contains at least one row of weight one.
\end{defn}
The stopping ratio $\sigma$ of the Tanner graph is defined by $s/n$. The
minimum Hamming distance is a property of the code to measure its performance
for maximum-likelihood (ML) decoding, while the stopping distance is a property
of the parity check matrix $\textbf{H}$ or the Tanner graph $G$ of a specific
code. Hence it varies for different choices of $\textbf{H}$ for the same code
$\mathcal{C}$.
 The stopping distance $s(\textbf{H})$ gives a lower bound of the minimum distance of the code
 $\mathcal{C}$ defined by a the low density parity check matrix
 \textbf{H}. Hence,
\begin{eqnarray}
 s(\textbf{H}) \leq d_{min}.\end{eqnarray}
 It has been shown that finding the stopping sets with minimum
cardinality is an NP-hard problem since the minimum-set vertex covering problem
can be reduced to it~\cite{krishnan07}. One can also define the trapping sets
for AWGN and BSC communication channels.

\section{Constructing LDPC Codes From \emph{Latin} Squares}\label{sec:LDPC-Latin}
In this section we construct self-orthogonal algebraic Low Density Parity Check
(LDPC) codes derived from \emph{Latin} squares. The class that we show has a
quasi-cyclic (QC) structure and hense is is called QC-LDPC codes. There have
been some constructions of LDPC and QC LDPC based on \emph{Latin} squares such
as the construction in~\cite{vasic02} based on mutually orthogonal and cyclic
\emph{Latin} squares. Also, in~\cite{milekovic04,laendner07} the authors
designed LDPC codes based on idempotent and symmetric \emph{Latin} squares.
These constructions are beneficial because they have girth of at least $6$ and
the codes are regular and irregular with arbitrary rates.  In addition, the
authors computed the stopping sets to measure  performance of LDPC codes over
the binary erasure channel.

\subsection{\emph{Latin} Square}
A \emph{Latin} square of order $n$ is a square matrix of size $n \times n$
defined over $\F_q^*$ or (i.e., $\textbf{Z}_q$) such that  each element
$\alpha^i \in F_q$ appears only once in every row and column. Clearly many
\emph{Latin} squares can be defined over the same alphabet, but the exact
number is not know for large $n$. \emph{Latin} squares have been used in many
applications and there are various methods to construct them. In addition,
there is a connection between \emph{Latin} squares and permutation groups. In
other words, one can look at a permutation group of order $n$ as a \emph{Latin}
squares of order $n$. We can define the \emph{main} and \emph{isotopy} classes
of \emph{Latin} squares as follows, see~\cite{matsumoto02,laendner07,kelley06}.

\begin{defn}Let $L$ and $L'$ be two \emph{Latin} squares of order $n$.
\begin{compactenum}[i)]
\item
If the square $L'$ can be obtain from $L$ under row, column and symbol
permutations, then $L$ is isotopy to $L'$. The set of all \emph{Latin} squares
isomorphic to $L$ is called \emph{isotropy} class.

\item The \emph{main} class of $L$ is given by the set of all squares which are
isomorphic to some conjugate of $L$. \emph{Paratopic} squares are a set of
squares which belong to the same \emph{main} class.
\item  We call a \emph{Latin} square $L$ of order $n$  reduced if $(1,2,3,\dots,n)$
appears in the first row and column.
\item For $1 \leq  k  \leq n$, a \emph{Latin} rectangle is an array of size $k\times n$
such that every element appears once in a row and may or may not appear in a
column. Clearly, \emph{Latin} squares are special cases of \emph{Latin}
rectangles where $k=n$, see~\cite{mcKay06}.
\end{compactenum}
\end{defn}

Let $R_n$ be the total number of reduced \emph{Latin} squares, the total number
of \emph{Latin} squares of order $n$ is given by $$L_n=n!(n-1)!R_n.$$

 We can also study
properties of some classes of \emph{Latin} squares.
\begin{defn} Let
$L$ and $L'$ be two \emph{Latin} squares of order $n$ \begin{compactenum}[i)]
\item $L$ is orthogonal to $L'$ if the cell $(i,j)$ in $L$ is different from
the cell $(i,j)$ in $L'$ for all $2 \leq i \leq n$ and $1 \leq j \leq n$.

\item There are at most $n-1$ \emph{mutually orthogonal} \emph{Latin} squares of order
$n$. Therefore, the set $L_1,L_2,\ldots,L_{n-1}$ is \emph{mutually orthogonal}
if $L_i$ and $L_j$ are orthogonal for $1 \leq i <j \leq n-1$.
\end{compactenum}
\end{defn}

As an example, two orthogonal  \emph{Latin} squares of order $n=4$ are given by

\begin{eqnarray} L_1=\left ( \begin{array}{cccccc}
1&2&3&4\\
2&1&4&3\\
3&4&1&2\\
4&3&2&1\\
\end{array} \right),
L_2=\left ( \begin{array}{cccccc}
1&2&3&4\\
3&4&1&2\\
4&3&2&1\\
2&1&4&3
\end{array} \right).
\end{eqnarray}

One way to obtain all orthogonal \emph{Latin} squares is by fixing the first
row and permute all other rows by one to obtain a new square matrix. Therefore,
we have $n-1$ permuted orthogonal \emph{Latin} squares.

\emph{Latin} squares have been used to construct efficient LDPC codes,
see~\cite{milekovic04,laendner07}. A \emph{Latin} square $L$ of order n is
idempotent if the cell $(i,j)$ contains the symbol $i$ for $1 \leq i \leq n$.
$L$ is symmetric if the cells $(i,j)$ and $(j,i)$ for $1 \leq i < j \leq n$
contain the same symbol. We define a special class of \emph{Latin} squares
called Cayley \emph{Latin} squares where the elements $\{1,\ldots,n \}$ form a
cyclic group of order $n$.
\begin{theorem}
The \emph{Latin} square $L$ derived from the Cayley table of a group $G$ is
atomic if and only if $G$ is a cyclic group of prime order.
\end{theorem}
\begin{proof}
See~\cite{wanless05}.
\end{proof}
Clearly, the transpose of a (orthogonal) \emph{Latin} square is also a
(orthogonal) \emph{Latin} square. We can also define the minimum distance
between two rows in a \emph{Latin} square as the number of nonzero elements in
the difference among these two rows. We can see that the Hamming distance
between any two rows of an $n \times n$ \emph{Latin} square is $n$.

\subsection{A Class of LDPC}
We construct a class of LDPC based on primitive elements of a finite field
$\F_q$. For simplicity, let us assume $q$ is a prime. This is equivalent to
constructing a \emph{Latin} square of order $n=q-1$.

Let $\alpha^i$ be an element in $\F_q$ for $1 \leq i \leq n$ such that
$\gcd(\alpha^i,q)=1$. Let $S$  be the set of primitive elements excluding $1$,
 $S= \{\alpha^1,\alpha^2,\ldots,\alpha^{n}\}.$ We can form the matrix $G$ of
size $n \times n$ as a result of the multiplicative group $\Z/q\Z$

\begin{eqnarray} G&=&\left( \begin{array}{ccccc} g_1 \\g_2\\ \vdots \\ g_{n} \end{array}  \right)=
\left( \begin{array}{ccccc} h_1 &h_2&\ldots& h_{n} \end{array}  \right)\nonumber \\
&=& \left ( \begin{array}{cccccc} \alpha^1 & \alpha^2&\alpha^3&\ldots & \alpha^{n} \\
\alpha^{2}& \alpha^4 &\alpha^6&\ldots &  \alpha^{n-1} \\ \vdots
&\vdots&\vdots&\vdots
\\  \alpha^{n}& \alpha^{n-1} &\alpha^{n-2}&\ldots&\alpha^{1}
\end{array} \right),
\end{eqnarray}
where $g_i$ is the $i$th row in $G$ and $h_j$ is the $j$th column in $G$. The
matrix $G$ has the following structure:
\begin{compactenum}[i)]
\item any two distinct rows differ in all positions. \item any two distinct columns differ in any
positions.
\item all elements of the field are presented in a row (column). 
\end{compactenum}
This matrix $G$ is equivalent to the \emph{Latin} square of order $n$. We know
that there are $n-1$ orthogonal \emph{Latin} squares of order $n$, we call them
$B_1,B_2,\ldots,B_{n-1}$ where $G=B_1$.

We form the matrix $B$ by permuting rows  of the matrix $G$ in a certain order.
So, the matrix $B_j$ is  a permutation of the matrix $B_i$ under row
permutation.

\begin{eqnarray} B= \left( \begin{array}{ccccc} B_1 &B_2& \ldots & B_{n-1} \end{array}
\right).
\end{eqnarray}
We have formed an $n \times (n-1)n$ matrix $B$ where every row in $G$ is
extended horizontally $(n-1)$ times.
\begin{corollary}
Any two rows in the matrix $B$ differ in all positions. I.e., $B$ is a
self-orthogonal matrix.
\end{corollary}
\begin{proof}
This is a direct consequence of our construction. Any two rows of the matrix
$B_j$ satisfies this condition. Therefore, any two rows in all matrices $B_j$'s
are orthogonal. Also, for any length $n$, the multiplication $(n-1)n$ is even.
Therefore, the inner product of a row by itself always vanishes.
\end{proof}
We can also see that the Hamming distance between any two rows of the matrix
$B$ is $n(n-1)$. This is because any two rows in the sub-matrix $B_i$ have
Hamming distance equal to zero.

We can also extend every matrix $B_j$ in $B$ vertically to form the matrix

\begin{eqnarray} H_j&=&\!\! \left( \begin{array}{ccccc} B_{j} \\B_{j+1}\\ \ldots \\ B_{j+\rho-1} \end{array}  \right)
 \\ &=& \!\! \left(\!\! \!\begin{array}{ccccc} h_{1,j}
&h_{2,j}&\ldots&\ldots&h_{n,j}\\h_{1,j+1}&h_{2,j+1}&\ldots&\ldots&h_{n,j+1}\\
h_{1,j+2}&h_{2,j+2}&\ldots&\ldots&h_{n,j+2}\\ \vdots&\vdots&\vdots&\vdots&\vdots \\
h_{1,(j+\rho-1)}&h_{2,(j+\rho-1)}&\ldots&\ldots&h_{n,(j+\rho-1)}
\end{array} \!\! \!\right)\nonumber,
\end{eqnarray}
where the element $h_{i,j+\ell}$ is a column of $n$ elements.
 Now the matrix $H_j$ has
size $(\rho)n \times n$. Therefore we formed  a $(\rho)n \times (n-1)n$ matrix
$H$.
\begin{eqnarray} H&=& \left( \begin{array}{ccccc} H_{1} &H_{2}& H_3&\ldots & H_{(n-1)} \end{array}
\right).
\end{eqnarray}
The matrix $H_j$ has the following properties:
\begin{compactenum}[i)]
\item Every $n$
components of every column are distinct and they form all the $n$ nonzero
elements of $\F_q^*$. \item any two columns differ in every position.
\item Any two rows have even number of elements in common.
\end{compactenum}

\begin{lemma}
For $1 \leq i,j \leq \rho n$, $i \ne j$,  any two rows $g_i$ and $g_j$ in $H$
have no common symbol from $\F_q$ or they have an even number of symbols in
common.
\end{lemma}
\begin{proof}
The proof is straightforward from the construction of the  matrix $H$ and
permutations of its rows and columns. the block $B_{j+\ell}$ is an orthogonal
\emph{Latin} square and a row permutation of the block $B_{j+\ell'}$.
\end{proof}
We now can replace every entry in $H$ by its location vector  to obtain a
$(\rho)n \times (n-1)n^2$ matrix
\begin{eqnarray} \mathcal{G}_{j}= \left[ \begin{array}{ccccc}A_{j,1} &A_{j,2}&\ldots&A_{j,n-1}
\end{array} \right],\end{eqnarray}
We construct the $\rho \times (n-1)n$ matrix $\textbf{H}$ of $n \times n$
submatrices over $\F_2$.
\begin{eqnarray}\textbf{H}&=& \left ( \begin{array}{ccccc} \mathcal{G}_1 \\ \mathcal{G}_2\\ \ldots \\ \mathcal{G}_{\rho} \end{array} \right)
\nonumber \\&=& \left ( \begin{array}{ccccc} A_{1,1} & A_{1,2}&\ldots & A_{1,n-1} \\
A_{2,1}& A_{2,2} &\ldots &  A_{21n-1} \\ \vdots &\vdots&\vdots&\vdots
\\  A_{\rho,1} & A_{\rho,2}&\ldots&A_{\rho,n-1}
\end{array} \right)
\end{eqnarray}
and the matrices $A_{i,j}'s$ are $n \times n^2$ circulant permutation matrices
of \emph{Latin} squares.

By this construction we built an $\rho n \times (n-1)n^2$ matrix $\textbf{H}$
over $\F_2$, where we replace $\alpha^i$ by $1$ at position $i$ in the vector
$\textbf{z}(\alpha^i)$.
The previous steps are summarized  in algorithm~\ref{fig:algorithm1}. We notice
that the row weight of $\textbf{H}$ is $(n-1)n$ and the column weight is
$\rho$. 
\begin{figure}
\begin{algorithmic}[1]
\STATE Input: A finite field $GF(q)$, where $q$ is a prime, \STATE Output: A
parity check matrix $\textbf{H}$ of size $\rho n \times (n-1)n^2$.

\* \\
    \STATE Construct the matrix $G$ as the multiplication group of $\F_q^*$, \emph{Latin} square of order $n=q-1$.

    \FOR{j = 1 to (n-1)}
        \STATE construct the sub-matrices $B_1, B_2,...,B_{n-1}$ as orthogonal
        \emph{Latin} squares.

    \ENDFOR
    \FOR{j = 1 to n-1}

        \STATE for each sub-matrix $B_j$ construct the column submatrices
        $H_{ij}$.
     \ENDFOR
         \STATE Form the matrix $H$.
         \STATE Convert every element in $H$ to a locator vector to
         form the matrix $\textbf{H}$.

\end{algorithmic}
\caption{Constructing LDPC codes based on elements of a finite field
(\emph{Latin} Square)}\label{fig:algorithm1}
\end{figure}

%
\subsection{Parameters  of LDPC Codes} Let $\rho$ and $\lambda$ be two integers
such that $1 \leq \rho < \lambda < n$.  Let $H(\rho, \lambda)$ be a sub-matrix
of the matrix $\textbf{H}$ satisfying the row (column) constraints as above.
The parameter $\rho$ represents the number of nonzero positions in a column;
$\rho$ is a weight of a column. Also, the parameter $\lambda$ represents the
number of nonzero positions in a row; $\lambda$ is a weight of a row. We can
always assume that $\lambda=n-1$ for the \emph{Latin} square construction. The
null-space of the matrix $\textbf{H}(\rho, \lambda)$ gives a $(\rho, \lambda)$
regular dual-containing LDPC code of length $\lambda n^2$ and rate $(\lambda n
-\rho)/\lambda n$. The minimum distance of the code is $\geq \rho$. This
construction gives a class of regular LDPC codes.

\begin{theorem}\label{thm:ldpc-Hmatrix}
For a prime integer $q$, the regular LDPC code generated by the parity check
matrix $\textbf{H}$ is dual-containing and it has rate $\frac{\lambda
n-\rho}{\lambda n}$. .
\end{theorem}

\begin{proof}
We need to show that the matrix $\mathcal{G}_j$ is also self-orthogonal as well
as $\mathcal{G}_j \times \mathcal{G}_i^T=0$ for $1 \leq i \leq \rho$.
\begin{compactenum}[i)]
\item Since $q$ is a prime, then $n$ is an even integer. Let $g_l$ and $g_k$
be two rows in $\mathcal{G}_j$ over $\F_2$. Then $g_k$ must be a permutation of
the row $g_l$ for $k\ne l$, hence they do not intersection at any position or
they have even weight of their inner product. So, $g_l * g_k^T=0$. Now, for $l
=k$, from the assumption $n$ is even and $g_l$ has exactly one nonzero element,
therefore, $g_l$ has even weight (multiplicity even), hence it is
self-orthogonal.

\item Now, let us choose any two arbitrary  rows $g_{jl}$ in $\mathcal{G}_j$ and $g_{ik}$ in
$\mathcal{G}_i$. Using a similar argument as in i) one can show that $g_{jl}
* g_{ik}^T=0$.
\item The claim about the rate comes from our algorithm in
Fig.~\ref{fig:algorithm1}. The result follows.
\end{compactenum}
\end{proof}
\begin{lemma}
The stopping distance of LDPC codes derived from \emph{Latin} squares is
exactly $n$.
\end{lemma}
\begin{proof}
By applying Definition~\ref{def:stoppingset2}, one can see that the number of
columns that have rows with weight one is $n$.
\end{proof}
By a similar argument one can also compute the stopping set and number of
cycles with length $4$.

\begin{figure}[h]
  \includegraphics[scale=0.6]{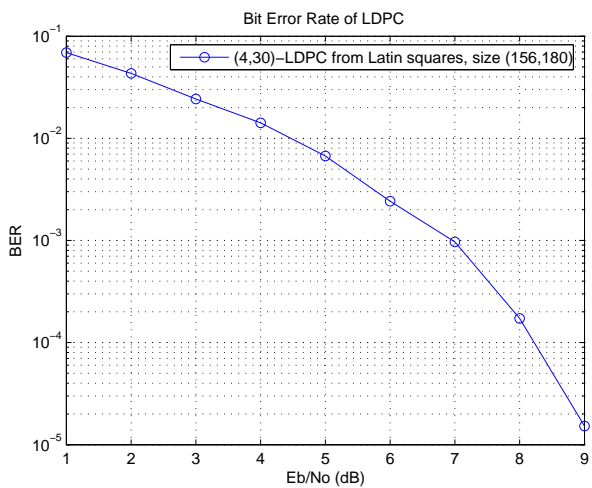}
 \centering
  \caption{Performance of a (4,30) LDPC code with parameters (156,180) based on \emph{Latin} squares}\label{ldpc1}
\end{figure}

We finish this construction by giving an example.
\begin{figure*}[t]
\begin{small}
\begin{eqnarray}\label{equ:hexample}  H= \left ( \begin{array}{cccc|cccc|cccc}
\alpha^1&\alpha^2&\alpha^3&\alpha^4  &\alpha^2&\alpha^4&\alpha^1&\alpha^3  &\alpha^3&\alpha^1&\alpha^4&\alpha^2\\
\alpha^2&\alpha^4&\alpha^1&\alpha^3  &\alpha^3&\alpha^1&\alpha^4&\alpha^2  &\alpha^4&\alpha^3&\alpha^2&\alpha^1\\
\alpha^3&\alpha^1&\alpha^4&\alpha^2  &\alpha^4&\alpha^3&\alpha^2&\alpha^1  &\alpha^1&\alpha^2&\alpha^3&\alpha^4\\
\alpha^4&\alpha^3&\alpha^2&\alpha^1  &\alpha^1&\alpha^2&\alpha^3&\alpha^4  &\alpha^2&\alpha^4&\alpha^1&\alpha^3\\
\\
\alpha^2&\alpha^4&\alpha^1&\alpha^3  &\alpha^3&\alpha^1&\alpha^4&\alpha^2 &\alpha^1&\alpha^2&\alpha^3&\alpha^4  \\
\alpha^3&\alpha^1&\alpha^4&\alpha^2  &\alpha^4&\alpha^3&\alpha^2&\alpha^1 &\alpha^2&\alpha^4&\alpha^1&\alpha^3   \\
\alpha^4&\alpha^3&\alpha^2&\alpha^1  &\alpha^1&\alpha^2&\alpha^3&\alpha^4 &\alpha^3&\alpha^1&\alpha^4&\alpha^2  \\
\alpha^1&\alpha^2&\alpha^3&\alpha^4  &\alpha^2&\alpha^4&\alpha^1&\alpha^3 &\alpha^4&\alpha^3&\alpha^2&\alpha^1  \\
\end{array} \right)
\end{eqnarray}
\end{small}
\end{figure*}

\begin{figure*}[t]
\begin{small}
\begin{eqnarray}\label{equ:Hexample}
 \textbf{H}= \left ( \begin{array}{cccc|cccc|cccc}
1000&0100&0010&0001 &0100&0001&1000&0010  &0010&1000&0001&0100\\
0100&0001&1000&0010 &0010&1000&0001&0100 &0001&0010&0100&1000\\
0010&1000&0001&0100 &0001&0010&0100&1000  &1000&0100&0010&0001\\
0001&0010&0100&1000 &1000&0100&0010&0001 &0100&0001&1000&0010 \\
\\

0100&0001&1000&0010  &0010&1000&0001&0100 &1000&0100&0010&0001 \\
0010&1000&0001&0100 &0001&0010&0100&1000& 0100&0001&1000&0010  \\
0001&0010&0100&1000  &1000&0100&0010&0001 &0010&1000&0001&0100 \\
1000&0100&0010&0001 &0100&0001&1000&0010  &0001&0010&0100&1000 \\
\end{array} \right)
\end{eqnarray}
\end{small}
\end{figure*}

\begin{exampleX}
Let $q=5=n+1$ and $\alpha$ be a primitive element in $\F_q$. Let $\lambda=n-1$
and $\rho=2$, the generator matrix is give by
\begin{eqnarray} G&=& \left ( \begin{array}{ccccc} g_1 \\ g_2\\g_3\\g_4\end{array}
\right)=\left ( \begin{array}{cccccc} h_1 & h_2&h_3&h_4\end{array} \right)
\nonumber \\&=& \left ( \begin{array}{cccccc}
\alpha^1&\alpha^2&\alpha^3&\alpha^4\\
\alpha^2&\alpha^4&\alpha^1&\alpha^3\\
\alpha^3&\alpha^1&\alpha^4&\alpha^2\\
\alpha^4&\alpha^3& \alpha^2&\alpha^1
\end{array} \right)
\end{eqnarray}
One can construct the matrices $B$, $H$ and $\textbf{H}$, and can check that
the matrix $\textbf{H}(2,12)$ is self-orthogonal.

The matrix $B$ is given by
\begin{eqnarray} B&=& \left ( \begin{array}{ccccc} B_1 &B_2&B_3\end{array}
\right),
\end{eqnarray}
where $B_1=G$ and
\begin{eqnarray} B_2 \!\!=\!\!
\left (\!
\begin{array}{cccc}
\alpha^2&\alpha^4&\alpha^1&\alpha^3\\
\alpha^3&\alpha^1&\alpha^4&\alpha^2\\
\alpha^4&\alpha^3& \alpha^2&\alpha^1\\
\alpha^1&\alpha^2&\alpha^3&\alpha^4
\end{array} \! \right) \nonumber, B_3 \!\!=\!\!
\left (\!
\begin{array}{cccc}
\alpha^3&\alpha^1&\alpha^4&\alpha^2\\
\alpha^4&\alpha^3& \alpha^2&\alpha^1\\
\alpha^1&\alpha^2&\alpha^3&\alpha^4\\
\alpha^2&\alpha^4&\alpha^1&\alpha^3
\end{array} \! \right) \nonumber.
\end{eqnarray}

 Also, the matrices $H$  and $\textbf{H}$ are shown in
Equations~\ref{equ:hexample} and \ref{equ:Hexample}.
\end{exampleX}

\section{Quantum LDPC Block Codes}
In this section we derive a family of stabilizer codes based on self-orthogonal
LDPC codes constructed from elements of orthogonal \emph{Latin} squares as
shown in section~\ref{sec:LDPC-Latin}.  Let us construct
the stabilizer matrix
\begin{eqnarray}
S_{stab} = \Big( \begin{array}{c|c} H_{X} & 0 \\0 &H_{Z}
\end{array}\Big).
\end{eqnarray}

The matrix $\textbf{H}$ is a binary self-orthogonal matrix, where we replace
every nonzero element in $\textbf{H}$ by the Pauli matrix $X$ to form the
matrix $H_X$. Similarly, we replace every nonzero element in $\textbf{H}$ by
the Pauli matrix $Z$ to form the matrix $H_Z$. Therefore the matrix $S_{stab}$
is also self-orthogonal. We can assume that the matrix $H_X$ corrects the
bit-flip errors, while the matrix $H_Z$ corrects the phase-flip errors,
see~\cite{macKay04}.

\begin{proposition}\label{prop2:QLDPC}
A quantum LDPC code $Q$ with rate $(n-2k)/n$ is a code whose stabilizer matrix
$S_{stab}$ of size $2k \times 2n$ has a parity check matrix $\textbf{H}$ with
pair $(\rho,\lambda)$ where $\rho$ is the number of non-zero error operators in
a column and $\lambda$ is the number of non-zero error operators in a row.
\end{proposition}

 We
now give a  family of quantum LDPC codes constructed from self-orthogonal LDPC
codes that is based on elements of \emph{Latin} squares.
\begin{lemma}
Let $n$ be the order of a \emph{Latin} square where $q=n+1$ for some prime $q$.
Let $\textbf{H}(\rho,\lambda)$ be a parity check matrix of a LDPC code over
$\F_2$ with  column weight $\rho$ and row weight $\lambda$. Then, there exists
a quantum LDPC code with parameters $[[\lambda n,\lambda n-2 n\rho,\geq
\rho]]_2$.
\end{lemma}
\begin{proof}
We know that there exists a regular LDPC code with a parity check matrix
$\textbf{H}$ constructed from \emph{Latin} squares of order $n=q-1$, see steps
in Fig.~\ref{fig:algorithm1}. The matrix $\textbf{H}$ of size $\rho n \times
\lambda n$ has row weight $\rho$ and column weight $\lambda=n(n-1)$. From
Theorem~\ref{thm:ldpc-Hmatrix}, the parity check matrix $\textbf{H}$ is
self-orthogonal and by Proposition~\ref{prop2:QLDPC} it defines a stabilizer
matrix in the form $S_{stab} = \Big(
\begin{array}{c|c} \textbf{H} & 0 \\0 &\textbf{H}
\end{array}\Big).$

The quantum code is also defined over $\F_2$ and has parameters
$[[N,M,d_{min}]]$ where $N=\lambda n$ and $M=\lambda n-2 \rho n,$ and $d_{min}
\geq \rho$.

\end{proof}

The stabilizer matrix of the quantum code $Q$ is derived from a QC-LDPC code.
Consequently, we can use any classical iterative decoding algorithm to estimate
error operators. A step in this regard has been taken by Camara \emph{el al.}
in~\cite{camara05}. They also constructed regular LDPC code from group theory.
We can conclude that our method of constructing QC-LDPC codes is simple and
benefits from  iterative decoding algorithms as well as easy encoders.

\section{Discussion}
We note that the constructed codes have reasonable performance in comparison to
MacKay's work in random constructions of LDPC codes.

LDPC codes shown in~\cite{milekovic04} and \cite{vasic02} have good performance
because these constructions of LDPC based on Latin squares do not need the
parity check matrices to be self-orthogonal. So, they have fewer (orthogonal)
Latin squares spread in the parity check matrices. In comparison to our work,
we have reasonable  performance, and our parity check matrices are
self-orthogonal, consequently they have some cycles of length $4$. Based on our
work, we can highlight the following issues:

\begin{compactenum}[i)]

\item It will be interesting to bound the maximum number of 4-cycle in the parity
check matrix. In our construction, it can be checked that the upper bound is
the length of the Latin squares, but this is not a tight bound since many rows
in the parity check matrix have at most 2 or 4 positions in common.

\item Other constructions of LDPC codes based on finite geometry might give
better performance of self-orthogonal LDPC codes. In addition, the minimum
distance and the stopping set of these codes can be computed easily.

\item Cyclic LDPC and QC LDPC are beneficial codes because, in addition to
their iterative decoding algorithms, they have efficient encoding algorithms
using shift registers.
\end{compactenum}

\section{Conclusion}
We introduced  a family of quantum LDPC codes based on \emph{Latin} squares.
Our construction is simple in comparison to other constructions that use random
approaches. Furthermore, one can use iterative decoding algorithms to decode
these codes. We plan to derive more families of quantum LDPC and convolutional
codes.

\chapter{Families of LDPC Codes Derived from Nonprimitive BCH Codes and Cyclotomic Cosets}

Low-density parity check (LDPC) codes are an important class of
codes with many applications. Two algebraic methods for constructing
regular LDPC codes are derived -- one based on nonprimitive
narrow-sense BCH codes and the other directly based on cyclotomic
cosets. The constructed codes have high rates and are free of cycles
of length four; consequently, they can be decoded using standard
iterative decoding algorithms. The exact dimension and bounds for
the minimum distance and stopping distance are derived. These
constructed codes can be used to derive quantum error-correcting
codes.

\section{Introduction}\label{sec:intro}
Bose-Chaudhuri-Hochquenghem (BCH) codes are  an interesting class of
linear codes that has been investigated for nearly  half of century.
This type of codes has a rich algebraic structure.  BCH codes with
parameters $[n,k,d\geq \delta]_q$ are interesting because one can
choose their dimension and minimum distance once given their design
distance $\delta$ and length $n$. A linear code defined by a
generator polynomial $g(x)$ has dimension $k=n-deg(g(x))$ and rate
$k/n$.  It was not an easy task to show the dimension of
nonprimitive BCH codes over finite fields. In~\cite{aly06a,aly07a},
we have given an explicit formula for the dimension of these codes
if their deigned distance $\delta$ is less than a constant
$\delta_{\max}$.

Low-density parity check (LDPC) codes are a
capacity-approaching~(\emph{Shannon limit}) class of codes that were
first described in a seminal work by Gallager~\cite{gallager62}.
Tanner in~\cite{tanner81} rediscovered LDPC codes using a graphical
interpretation.  A regular $(\rho, \lambda)$ LDPC code is measured
by the weights of its columns $\rho$ and rows $\lambda$. Iterative
decoding of LDPC and turbo codes highlighted the importance of these
classes of codes for communication and storage channels.
Furthermore, these codes are practical and have been used in many
beneficial applications~\cite{macKay98,lin04}. In contrast to BCH
and Reed-Solomon (RS) cyclic codes, LDPC cyclic codes with sparse
parity check matrices are customarily constructed by a computer
search. In practice, LDPC codes can achieve higher performance and
better error correction capabilities than many other codes, because
they have efficient iterative decoding algorithms, such as the
product-sum algorithm~\cite{tanner04,luby01,liva06,lin04}. Some BCH
codes turned out to be LDPC cyclic codes as well; for example, a
$(15,7)$ BCH code is also an LDPC code with a minimum distance five.

Regular and irregular LDPC codes have been constructed based on
algebraic and random approaches~\cite{song06,djurdjevic03,song06-b},
and references therein. Liva~\emph{et al.}~\cite{liva06} presented a
survey of the previous work done on algebraic constructions of LDPC
codes based on finite geometry, elements of finite fields, and RS
codes. Yi~\emph{et al.}~\cite{yi05} gave a construction for LDPC
codes, based on binary narrow-sense primitive BCH codes, and their
method is free of cycles of length $4$. Furthermore, a good
construction of LDPC codes should have a girth of the Tanner graph,
of at least $6$~\cite{liva06,lin04}. One might wonder how do the
rates and minimum distance of BCH codes compare to LDPC codes? Do
self-orthogonal BCH codes give raise to self-orthogonal LDPC codes
as well under the condition $\delta \leq \delta_{max}$. We  show
that how to derive LDPC codes from nonprimitive BCH codes.

One way to measure the decoding performance of linear codes is by
computing their~\emph{minimum distance} $d_{min}$. The performance
of low-density parity check codes under iterative decoding can also
be gauged by measuring their \emph{stopping sets} $S$
and~\emph{stopping distance} $s$, which is the size of the smallest
stopping set~\cite{schwartz06,orlitsky05}. For any given parity
check matrix $\textbf{H}$ of an LDPC code $\mathcal{C}$, one can
obtain the Tanner graph $G$ of this code and computes the stopping
sets. Hence, $s$ is a property of $\textbf{H}$, while $d_{min}$ is a
property of $\mathcal{C}$. The minimum distance is also bounded by
$d_{min} \geq s$. BCH codes are decoded invertible matrices such as
Berkcampe messay method, LDPC codes ar decoded using iterative
decoding  and Belief propagation (BP) algorithms.

In this Chapter,  we give a series of regular LDPC and Quasi-cyclic (QC)-LDPC
code constructions based on non-primitive narrow-sense BCH codes and elements
of cyclotomic cosets. The constructions are called \textbf{Type-I} and
\textbf{Type-II} regular LDPC codes. The algebraic structures of these codes
help us to predict additional properties of these codes. Hence, The constructed
codes have the following characteristics:
\begin{compactenum}[i)]
\item Two classes of regular  LDPC codes are constructed that have high rates and free of cycles of length
$4$. Their properties can be analyzed easily.
\item The exact dimension is computed and the minimum distance is bounded for the constructed codes. Also, the
stopping sets and stopping distance can be determined from the structure of
their parity check matrices. They can be decoded with known standard iterative
decoders.
\end{compactenum}

 The motivation for our work is to construct Algebraic regular LDPC
codes that can be used to derive quantum error-correcting codes.
Alternatively, they can also be used for wireless communication
channels. Someone will argue about the performance and usefulness of
the constructed regular LDPC codes in comparison to irregular LDPC
codes. Our first motivation is to derive quantum LDPC codes based on
nonprimitive BCH codes. Hence, the constructed codes can be used to
derive classes of symmetric quantum
codes~\cite{calderbank98,macKay04} and asymmetric quantum
codes~\cite{evans07,steane96}. The literature lacks many
constructions of algebraic quantum LDPC codes, see for
example~\cite{macKay04,aly08c} and references therein.

\section{Constructing LDPC Codes}\label{sec:LDPC}
Let $\F_q$ denote a finite field of characteristic $p$ with $q$
elements. Recall that the set $\F_q^*=\F_q\setminus \{0\}$ of nonzero
field elements is a multiplicative cyclic group of order $q-1$.  A
generator of this cyclic group is called a primitive element of the
finite field $\F_q$.

\subsection{Definitions}
Let $n$ be a positive integer such that $\gcd(n,q)=1$ and $q^{\lfloor
m/2\rfloor} <n \leq \mu=q^m-1$, where $m=\ord_n(q)$ is the
multipicative order of $q$ modulo $n$.

Let $\alpha$ denote a fixed primitive element of~$\F_{q^m}$.  Define a map
$\textbf{z}$ from $\F_{q^m}^*$ to $\F_2^\mu$ such that all entries of
$\textbf{z}(\alpha^i)$ are equal to 0 except at position $i$, where it is equal
to 1. For example, $\textbf{z}(\alpha^2)=(0,1,0,\ldots,0)$.  We call
$\textbf{z}(\alpha^k)$ the location (or characteristic) vector of $\alpha^k$.
We can define the location vector $\textbf{z}(\alpha^{i+j+1})$ as the right
cyclic shift of the location vector $\textbf{z}(\alpha^{i+j})$, for $0 \leq j
\leq \mu-1$, and the power is taken module $\mu$.

\medskip

\begin{defn}\label{def:Amatrix}We can define a map $A$ that associates to an element $\F_{q^m}^*$ a circulant
matrix in $\F_2^{\mu\times \mu}$ by
\begin{eqnarray}\label{label:mapA} A(\alpha^i)=\left ( \begin{array}{ccc}
\textbf{z}(\alpha^i) \\  \textbf{z}(\alpha^{i+1})
\\ \vdots \\  \textbf{z}(\alpha^{i+\mu-1})
\end{array} \right).
\end{eqnarray}
By construction, $A(\alpha^k)$ contains a 1 in every row and column.
\end{defn}

For instance,  $A(\alpha^1)$ is the identity matrix of size
$\mu \times \mu$, and $A(\alpha^2)$ is the shift matrix
\begin{eqnarray} A(\alpha^2)=\left ( \begin{array}{cccccc} 0&1&0&\ldots&0 \\
0&0&1&\ldots&0\\
\vdots&\vdots&\vdots&\vdots&\vdots\\1&0&0&\ldots&0
\end{array} \right).
\end{eqnarray}

We will use the map $A$ to associate to a parity check matrix $H=(h_{ij})$ in
$(\F_{q^m}^*)^{a\times b}$ the (larger and binary) parity check matrix
$\textbf{H}=(A(h_{ij}))$ in $\F_2^{\mu a\times \mu b}$. The matrices
$A(h_{ij})$$'s$ are $\mu \times \mu$ circulant permutation matrices based on
some primitive elements $h_{ij}$ as shown in Definition~\ref{def:Amatrix}.

\subsection{Regular LDPC Codes}
A low-density parity check code (or LPDC short) is a binary block code
that has a parity check matrix $\textbf{H}$ in which each row (and
each column) is sparse. An LDPC code is called \textit{regular} with
parameters $(\rho,\lambda)$ if it has a sparse parity check matrix $H$
in which each row has $\rho$ nonzero entries and each column has
$\lambda$ nonzero entries.

A regular LDPC code defined by a parity check matrix $\textbf{H}$ is said to
satisfy the \textit{row-column condition} if and only if any two rows (or,
equivalently, any two columns) of $\textbf{H}$ have at most one position of a
nonzero entry in common.  The row-column condition ensures that the Tanner
graph does not have cycles of length $4$.

A Tanner graph of a binary code with a parity check matrix
$\textbf{H}=(h_{ij})$ is a graph with vertex set $V\stackrel{.}{\cup} C$ that
has one vertex in $V$ for each column of $\textbf{H}$ and one vertex in $C$ for
each row in $\textbf{H}$, and there is an edge between two vertices $i$ and $j$
if and only if $h_{ij}\neq 0$. Thus, the Tanner graph is a bipartite graph. The
vertices in $V$ are called the variable nodes, and the vertices in $C$ are
called the check nodes.  We refer to $d(v_i)$ and $d(c_j)$ as the degrees of
variable node $v_i$ and check node $c_j$ respectively.

Two values used to measure the performance of the decoding algorithms of LDPC
codes are: girth of a Tanner graph and stopping sets. The minimum stopping set
is analogous to the minimum Hamming distance of linear block codes.


\medskip

\begin{defn}[Grith of a Tanner graph]\label{def:girth}
The girth $g$ of the Tanner graph is the length of its shortest cycle (minimum
cycle).
\end{defn}
A Tanner graph with large girth is desirable, as iterative decoding converges
faster for graphs with large girth.

\medskip

\begin{defn}[Stopping set]\label{defn:stopset}
A \textit{stopping set} $S$ of a Tanner graph is a subset of the variable nodes
$V$ such that each vertex in the neighbors of $S$ is connected to at least two
nodes in $S$.
\end{defn}

The \textit{stopping distance} is the size of the smallest stopping set. The
stopping distance determines the number of correctable erasures by an iterative
decoding algorithm, see~\cite{orlitsky05,schwartz06,di00}.
\medskip

\begin{defn}[Stopping distance]
The stopping distance of the parity check matrix $\textbf{H}$ can be defined as
the largest integer $s(\textbf{H})$ such that every set of at most
$(s(\textbf{H})-1)$ columns of \textbf{H} contains at least one row of weight
one, see~\cite{schwartz06}.
\end{defn}
 The stopping ratio $\sigma$ of the Tanner graph of
a code of length $n$ is defined by $s$ over the code length.

The minimum Hamming distance is a property of the
code used to measure its performance for maximum-likelihood
decoding, while the stopping distance is a property of the parity
check matrix $\textbf{H}$ or the Tanner graph $G$ of a specific
code. Hence, it varies for different choices of $\textbf{H}$ for the
same code $\mathcal{C}$.  The stopping distance $s(\textbf{H})$ gives
a lower bound of the minimum distance of the code $\mathcal{C}$
defined by \textbf{H}, namely
\begin{eqnarray}
 s(\textbf{H}) \leq d_{min}
\end{eqnarray}
It has been shown that finding the stopping sets of minimum
cardinality is an NP-hard problem, since the minimum-set vertex
covering problem can be reduced to it~\cite{krishnan07}.


\section{LDPC Codes based on  BCH Codes}\label{sec:LDPC-BCH}
In this section we give two constructions of LDPC codes derived from
nonprimitive BCH codes, and from elements of cyclotomic cosets.
In~\cite{yi05}, the authors derived a class of regular LDPC codes
from primitive BCH codes but they did not prove that the
construction has free of cycles of length four in the Tanner graph.
In fact, we will show that not all primitive BCH codes can be used
to construct LDPC with cycles greater than or equal to six in their
Tanner graphs. Our construction is free of cycles of length four if
the BCH codes are chosen with prime lengthes as proved in
Lemma~\ref{lemma:freecycles4}; in addition the stopping distance is
computed. Furthermore, We are able to derive a formula for the
dimension of the constructed LDPC codes as given in
Theorem~\ref{thm:Hrank}. We also infer the dimension and cyclotomic
coset structure of the BCH codes based on our previous results
in~\cite{aly06a,aly07a}.

We keep the definitions of the previous section. Let $q$ be a power of
a prime and $n$ a positive integer such that $\gcd(q,n)=1$. Recall
that the cyclotomic coset $C_x$ modulo $n$ is defined as
\begin{eqnarray}C_x=\{xq^i\bmod n \mid i\in \Z, i\ge 0\}.
\end{eqnarray}

Let $m$ be the multiplicative order of $q$ modulo $n$. Let $\alpha$ be
a primitive element in $\F_{q^m}$. A nonprimitive narrow-sense BCH
code $\mathcal{C}$ of designed distance $\delta$ and length $n$ over
$\F_{q}$ is a cyclic code with a generator monic polynomial $g(x)$
that has $\alpha, \alpha^2, \ldots, \alpha^{\delta-1}$ as zeros,
\begin{eqnarray}
g(x)=\prod_{i=1}^{\delta -1} (x-\alpha^i).
\end{eqnarray}
Thus,  $c$ is a codeword in $\mathcal{C}$ if and only if
$c(\alpha)=c(\alpha^2)=\ldots=c(\alpha^{\delta-1})=0$. The parity check matrix
of this code can be defined as
\begin{eqnarray}\label{bch:parity}
 H_{bch} =\left[ \begin{array}{ccccc}
1 &\alpha &\alpha^2 &\cdots &\alpha^{n-1}\\
1 &\alpha^2 &\alpha^4 &\cdots &\alpha^{2(n-1)}\\
\vdots& \vdots &\vdots &\ddots &\vdots\\
1 &\alpha^{\delta-1} &\alpha^{2(\delta-1)} &\cdots &\alpha^{(\delta-1)(n-1)}
\end{array}\right].
\end{eqnarray}

We note the following fact about the cardinality of cyclotomic cosets.
\medskip

\begin{lemma}\label{th:bchnpcosetsize}
Let $n$ be a positive integer and $q$ be a power of a prime, such that
$\gcd(n,q)=1$ and $q^{\lfloor m/2\rfloor} <n \leq q^m-1$, where $m=ord_n(q)$.
The cyclotomic coset $C_x=\{ xq^j\bmod n \mid 0\le j<m\}$ has a cardinality of
$m$ for all $x$ in the range $1\leq x\leq nq^{\lceil m/2\rceil}/(q^m-1).$
\end{lemma}
\begin{proof}
See~\cite[Lemma 8]{aly07a}.
\end{proof}

Therefore, all cyclotomic cosets have the same size $m$ if their
range is bounded by a certain value. This lemma enables one to
determine the dimension in closed form for BCH code of small
designed distance~\cite{aly06a,aly07a}. In fact, we show the
dimension of nonprimitve BCH codes over $\F_q$.
\medskip

\begin{theorem}\label{th:bchnpdimension}
Let $q$ be a prime power and $\gcd(n,q)=1$, with $ord_n(q)=m$. Then a
narrow-sense BCH code of length $q^{\lfloor m/2\rfloor} <n \leq q^m-1$ over
$\F_q$ with designed distance $\delta$ in the range $2 \leq \delta \le
\delta_{\max}= \min\{ \lfloor nq^{\lceil m/2 \rceil}/(q^m-1)\rfloor,n\}$, has
dimension of
\begin{equation}\label{eq:npdimension}
k=n-m\lceil (\delta-1)(1-1/q)\rceil.
\end{equation}
\end{theorem}
\begin{proof}
See~\cite[Theorem 10]{aly07a}.
\end{proof}

Based on these two observations, we can construct regular LDPC codes from BCH
codes with a known dimension and cyclotomic coset size.

\subsection{ {Type-I Construction}}
In this construction, we use the parity check matrix of a nonprimitive
narrow-sense BCH code over $\F_q$ to define the parity check
matrix of a regular LDPC over $\F_2$.

Consider the narrow-sense BCH code of prime length $q^{\lfloor
m/2\rfloor} <n \leq q^m-1$ over $\F_q$ with designed distance
$\delta$ and  $ord_n(q)=m$. We use the fact that there must be some
primes in the integer range $(q^{\lfloor m/2\rfloor}, q^m-1)$. In
fact, there must exist a prime between $x$ and $2x$ for some integer
x, in which it ensures existence primes in the given interval. A
parity check matrix $\textbf{H}$ of an LDPC code can be obtained by
applying the map $A$ in Equation~(\ref{label:mapA}) to each entry of
the parity check matrix~(\ref{bch:parity}) of this BCH code,
\begin{eqnarray}\label{eq:LDPCtype-I} &\textbf{H}&=\\ &&\!
\left[ \begin{array}{ccccc}
A(1) &A(\alpha)&A(\alpha^2)\!\!\!&\cdots\!\!\!&\!\!\! A(\alpha^{n-1})\!\!\!\\
A(1) &A(\alpha^2)&A(\alpha^4 )\!\!\!&\cdots\!\!\!&\!\!\! A(\alpha^{2(n-1)})\!\!\!\\
\vdots& \vdots &\vdots &\ddots\!\!\!  &\!\!\! \vdots\!\!\! \\
A(1) &A(\alpha^{\delta-1})\!\!\!&A(\alpha^{2(\delta-1)}
)&\cdots&\!\!\!A(\alpha^{(\delta-1)(n-1)}) \!\!\!
\end{array}\right].\nonumber
\end{eqnarray}
The matrix $\textbf{H}$ is of size $(\delta-1)\mu \times n\mu$ and by
construction it has the following properties:
\begin{compactitem}
\item Every column has a weight of  $\delta-1$.
\item Every row has a weight of $n$.
\end{compactitem}

The matrix $\textbf{H}$ of size $(\delta-1)\mu \times n\mu$ has a
weight of $\rho=\delta-1$ in every column, and a weight of
$\lambda=n$ in every row. The null space of the matrix $\textbf{H}$
defines a $(\rho,\lambda)$ LDPC code with a high rate for a small
designed distance $\delta$ as we will show. The minimum distance of
the BCH code is bounded by
\begin{eqnarray}
d_{min} \geq \left\{
               \begin{array}{ll}
                 \delta+1, & \hbox{odd $\delta$;} \\
                 \delta+2, & \hbox{even $\delta$.}
               \end{array}
             \right.
\end{eqnarray}
Also, the minimum distance of the LDPC codes is bounded by
$d_{min}$.
Now, we will show that in general  regular $(\rho,\lambda)$ LDPC
codes derived from  primitive BCH codes of length $n$  are not free
of cycles of length four as claimed in~\cite{yi05}.

\medskip

\begin{lemma}\label{lemma:freecycles4}
The Tanner graph of LDPC codes constructed in~\textbf{Type-I} are
free of cycles of length four for a prime length $n$.
\end{lemma}
\begin{proof}
Consider the block-column indexed by $n-j$ for $1 \leq j \leq n-1$
and let $r_i$ and $r_i'$ be two different block-rows  for $1 \leq
r_i,r_i' \leq (\delta-1)$. Assume by contradiction that we have
$A(\alpha^{r_i(n-j)}) = A(\alpha^{r_j'(n-j)})$. Thus $r_i(n-j) \mod
n=r_i'(n-j) \mod n$  or $n(r_i-r_i') \mod n=(r_i-r_i')j \mod n=0$.
This contradicts the assumption that $n> j \geq 1$ and $r_i\neq
r_i'$.
\end{proof}
Hence primitive BCH codes of composite length $n$ can not be used to
derive LDPC codes that are cycles-free of length four using our
construction.

%
The proof of the following lemma is straight forward by  exchanging,
adding, and permuting a  block-row.
\medskip

\begin{lemma}\label{lem:twoblocks}
Let $(\ldots,1_\ell,\ldots)$ be a vector of length $\mu$ that has 1 at position
$\ell$. Under the cyclic shift, the following two blocks $h_a$ and $h_b$ of
size $\mu \times \mu$  are equivalent, where $h_a$ and $h_b$ are generated by
the rows $\left(\begin{array}{ccccc} 1&\ldots&1_i&\ldots\end{array}\right)$ and
$\left(\begin{array}{ccccc} 1&\ldots&1_j&\ldots\end{array}\right)$ and their
cyclic shifts, respectively.
\end{lemma}

One might imagine that the rank of the parity check matrix
$\textbf{H}$ in~(\ref{eq:LDPCtype-I}) is given by $(\delta-1)\mu$
since rows of every block-row $h_a$ is linearly independent. A
computer program has been written to check the exact formula and
then we drove a formula to give the rank of the matrix $\textbf{H}$.
\medskip

\begin{theorem}\label{thm:Hrank}
Let $n$ be a prime in the range $q^{\lfloor m/2\rfloor} <n \leq
\mu=q^m-1$ and $\delta$ be an integer in the range $2 \leq \delta <
n$ for some prime power $q$ and $m=\ord_q(n)$. The rank of the
parity check matrix $\textbf{H}$ given by

\begin{eqnarray}\label{eq:LDPCtype-I} \textbf{H}=\!
\left[ \begin{array}{ccccc}
\mathcal{A}^o &\mathcal{A}^1&\mathcal{A}^2\!\!\!&\cdots\!\!\!&\!\!\! \mathcal{A}^{n-1}\!\!\!\\
\mathcal{A}^0 &\mathcal{A}^2&A^4\!\!\!&\cdots\!\!\!&\!\!\! \mathcal{A}^{2(n-2)}\!\!\!\\
\vdots& \vdots &\vdots &\ddots\!\!\!  &\!\!\! \vdots\!\!\! \\
\mathcal{A}^0 &\mathcal{A}^{\delta-1}\!\!\!&\mathcal{A}^{\delta-1}
)&\cdots&\!\!\!\mathcal{A}^{(\delta-1)(n-1)} \!\!\!
\end{array}\right]
\end{eqnarray}
 is
$(\delta-1)\mu-(\delta-2)$, where $\mathcal{A}^i=A(\alpha^i )$.
\end{theorem}
\begin{proof}
The proof of this theorem can be shown by mathematical induction for
$1,2,\ldots,\delta \leq n$. We know that every block-row is linearly
independent.
\begin{compactenum}[i)]
\item Case i. Let $\delta=2$, the statement is true since ever block-row has
only 1 in every column, the first n columns represent the identity matrix.

\item Case ii-1. Assume the statement is true for $\delta-2$. In this case,  the matrix
\textbf{G} has a full rank given by $(\delta-2)\mu-(\delta-3)$.  So, we have
$$\textbf{G}=\left(\begin{array}{cccccc} h_{11}&h_{12}&h_{13}&\ldots&\ldots&h_{1n}\\0&h_{22}&h_{23}&\ldots&\ldots&h_{2n} \\
0&0&h_{33}&\ldots&\ldots&h_{3n}\\ 0&0&0&\vdots&\vdots&h_{in}\\
0&0&\ldots&h_{(\delta-2) (\delta-2)}&\ldots&h_{(\delta-2)n}
\end{array}\right).$$
The elements $h_{ii}'s$ have 1's in the diagonal and zeros
everywhere using simple Gauss elimination method and
Lemma~\ref{lem:twoblocks}.

\item Case iii-1.
We can form the sub-matrix $\textbf{H}_2$ of size $(\delta-1)\mu
\times (\delta-1)\mu$ by adding one block-row to the matrix
$\textbf{G}$. The last block-row is generated by
$$(A(\alpha^0),A(\alpha^{\delta-1}),A(\alpha^{2(\delta-1)}),\ldots,A(\alpha^{n-1(\delta-1)})).$$
 All $\mu-1$ rows of the last block-row are linearly independent
and can not be generated from the previous $\delta-2$ blocks-row.
Now, in order to obtain the last row-block to be zero at positions
$h_{(\delta-1)1},h_{(\delta-1)2},\ldots,h_{(\delta-1)(\delta-2)}$,
we can add the element $h_{jj}$ to the element $h_{(\delta-1)j}$. In
addition, the last row (row indexed by $(\delta-1)\mu$) of block-row
$\delta-1$ can be generated by adding all elements of the first
block-row  to the first $\mu-1$ rows of the last block-row.
$$\textbf{G}=\left(\begin{array}{cccccc} h_{11}&h_{12}&h_{13}&\ldots&\ldots&h_{1n}\\0&h_{22}&h_{23}&\ldots&\ldots&h_{2n} \\
0&0&h_{33}&\ldots&\ldots&h_{3n}\\ 0&0&0&\vdots&\vdots&h_{in}\\
0&0&\ldots&h_{(\delta-1) (\delta-1)}&\ldots&h_{(\delta-1)n}
\end{array}\right).$$

 Therefore, the matrix $\textbf{G}$ has rank of
$(\delta-2)\mu-(\delta-3)+\mu-1=(\delta-1)\mu-(\delta-2)$. We notice that the
matrix $\textbf{H}$ has the same rank as the matrix $\textbf{G}$, hence the
proof is completed.

\end{compactenum}
\end{proof}
The proof can also be shown by dropping the last row of every
block-row except at the last row in the first block-row. Hence, the
remaining matrix has a full rank.
Obtaining a formula for  rank of the parity check matrix \textbf{H}
allows us to compute rate of the constructed LDPC codes.  Now, we
can deduce the relationship between nonprimitive narrow-sense BCH
codes and LDPC codes constructed in~\textbf{Type-I}.

\medskip

\begin{theorem}[LDPC-BCH Theorem]\label{lem:BCHtype-I}
Let $n$ be a prime and $q$ be a power of a prime, such that
$\gcd(n,q)=1$ and $q^{\lfloor m/2\rfloor} <n \leq q^m-1$, where
$m=ord_n(q)$. A nonprimitive narrow-sense BCH code with parameters
$[n,k,d_{min}]_q$ gives  a $(\delta-1,n)$ LDPC code with rate
$(n\mu-[(\delta-1)\mu-(\delta-2)])/n\mu$, where $k=n-m\lceil
(\delta-1)(1-1/q)\rceil$ and $2 \leq \delta \leq \delta_{max}$. The
constructed codes are free of cycles with length four.
\end{theorem}
\begin{proof}
By \textbf{Type-I} construction of LDPC codes derived from
nonprimitive BCH codes using Equation~(\ref{eq:LDPCtype-I}), we know
that every element $\alpha^i$ in $H_{bch}$ is a circulant matrix
$A(\alpha^i)$ in $\textbf{H}$. Therefore, there is a parity check
matrix $\textbf{H}$ with size $(\delta-1)\mu \times n\mu$.
$\textbf{H}$ has a row weight of $n$ and a column weight of
$\delta-1$. Hence, the null space of the matrix $\textbf{H}$ defines
an LDPC code with the given rate using Lemma~\ref{thm:Hrank}.

The constructed code is free of cycles of length four, because the
matrix $H_{bch}$ has no two rows with the same value in the same
column, except in the first column. Hence, the matrix $\textbf{H}$
has, at most, one position in common between two rows due to
circulant property and Lemma~\ref{lemma:freecycles4}. Consequently,
they have a Tanner graph with girth greater than or equal to six.
\end{proof}


Based on \textbf{Type-I} construction of regular LDPC codes, we
notice that every variable node has a degree $\delta-1$ and every
check nodes has a degree $n$. Also, the maximum number of columns
that do not have one in common is $n$. Therefore, the following
Lemma counts the stopping distance of the Tanner graph defined by
$\textbf{H}$.

\smallskip

\begin{lemma}\label{lem:stoppingset}
The cardinality of the smallest stopping set of the Tanner graph of
\textbf{Type-I} construction of regular LDPC codes is $\mu+1$.
\end{lemma}
\begin{proof}
Let $\textbf{H}$ be the parity check matrix of an $(\delta-1,n)$
LDPC code given in  \textbf{Type-I} construction. We know that every
row has a weight of $n$ and every column has a weight of $\delta-1$.
Let $c_j$ be a node in $C$ and $v_i$ be a node in $V$, therefore,
$d(c_j)=n$ and $d(v_i)=\delta-1$. If we choose a set of the first
$\mu$ columns in \textbf{H}, then every row has a weight of exactly
one. Therefore, the result follows.
\end{proof}

\smallskip

\begin{example}
Let $n=\mu=q^m-1$, with $m=7$ and $q=2$.  Consider a BCH code with $\delta=5$
and length $n$. Assume $\alpha$ to be a primitive element in $\F_{q^m}$.  The
matrix $H$ can be  written as
\begin{eqnarray} H=\left ( \begin{array}{ccccccc} 1&\alpha&\alpha^2& \ldots&\alpha^{126} \\
1&\alpha^2&\alpha^4&\ldots&\alpha^{125}\\
1&\alpha^3&\alpha^6&\ldots&\alpha^{124}\\
1&\alpha^4&\alpha^8&\ldots&\alpha^{123}\\
\end{array} \right),
\end{eqnarray}
and the matrix $\textbf{H}$ has size $ 508 \times 16129 $. Therefore, we
constructed a $(4,127)$ regular LDPC with a rate of $123/127$, see
Fig.~\ref{fig:ldpc1}.
\end{example}

\begin{table}[ht]
\caption{Parameters of LDPC codes derived from NP BCH codes}
\label{table:bchtable}
\begin{center}
\begin{tabular}{|l|l|c|c|c|}
\hline   $q$ &$\mu$&  BCH Codes & LDPC code&rank of \textbf{H}\\&&&size of \textbf{H}    &\\
 \hline
 2&31&$[23,12,4]$&(93,713)&91\\
 3&26&$[23,12,5]$&(104,598)& 101 \\
 2&31&$[31,26,3]$&(62,961)&61\\
 2&31&$[31,21,5]$&(124,961,)&121\\
 2&31&$[31,26,6]$&(155, 961)&151\\
  2&31&$[31,16,7]$&(186,961)&181\\
    2&63&$[47,24,4]$&(189 ,1961)&187\\
  2&63&$[61,21,6]$&(315, 3843)& 311\\
  2 &63&$[61,11,10]$&(567,3843)&559\\
    2 &127&$[127,113,15]$&(1778,16129)&1765\\
2 &127&$[127,103,25]$&(3048,16129)&3025\\
&&&&\\
 \hline
\end{tabular}
\end{center}
\end{table}
%
\section{LDPC Codes Based on Cyclotomic Cosets} In this section we will construct regular
LDPC codes based on the structure of cyclotomic cosets. Assume that
we use the same notation as shown in Section~\ref{sec:LDPC}. Let
$C_x$ be a cyclotomic coset modulo prime integer $n$, defined as
$C_x=\{xq^i \bmod n \mid i\in \Z, 1 \leq x < n \}.$ We can also
define the location vector $\textbf{y}$ of a cyclotomic coset $C_x$,
instead of the location vector $\textbf{z}$ of an element
$\alpha^i$.

\medskip

\begin{defn}
The location vector $\textbf{y}(C_x)$ defined over a cyclotomic coset $C_x$ is
the vector $\textbf{y}(C_x)=(z_0,z_1,\dots,z_n)$, where all positions  are
zeros except at positions corresponding to elements of $C_x$.
\end{defn}
Let $\ell$ be the number of different cyclotomic cosets $C_{x}^i$'s that are
used to construct the matrices $H_{C_j}^i$'s. We can index the $\ell$ location
vectors   corresponding to $C_{x_1},C_{x_2},\ldots,C_{x_\ell}$, as
$\textbf{y}^1,\textbf{y}^2,\ldots,\textbf{y}^\ell$.  Let $\textbf{y}^1(\gamma
C_x)$ be the cyclic shift of  $\textbf{y}^1(C_x) $ where every element in $C_x$
is incremented by 1.

\subsection{Type-II Construction}
 We construct the
matrix $H_{C_x}^1$ from the cyclotomic $C_x$ as
\begin{eqnarray}\label{eq:HcycloBCH}
 H_{C_x}^1=\left ( \begin{array}{ccc} \textbf{y}^1(C_x) \\  \textbf{y}^1(\gamma C_x)
\\ \vdots \\   \textbf{y}^1(\gamma^{n-1} C_x) \end{array} \right), \end{eqnarray}
where $  \textbf{y}^1(\gamma^{j+1} C_x)$ is the cyclic shift of  $
\textbf{y}^1(\gamma^{j} C_x) $ for $0\leq j \leq n-1$.

 From Lemma~\ref{th:bchnpcosetsize}, we know that all cyclotomic cosets $C_x$'s
 have a size of $m$ if $1\leq x\leq nq^{\lceil m/2\rceil}/(q^m-1).$

We can generate all rows of $H_{C_x}$, by shifting the first row one position
to the right. Our construction of the matrix $H_{c_x}^i$ has the following
restrictions.
\begin{compactitem}
\item Let $x \leq \Theta(\sqrt{n})$, this will guarantee that all
cyclotomic cosets have the same size $m$.
\item Any two rows of  $H_{c_x}^i$ have only one nonzero position in common.
\item Every row (column) in $H_{c_x}^i$ has a weight of $m$.
\end{compactitem}

We can construct the matrix $\textbf{H}$ from different cyclotomic cosets as
follows.
\begin{small}
\begin{eqnarray}
\textbf{H} &=& \Big[ \begin{array}{cccc}H_{C_1}^1&
H_{C_3}^2&\ldots&H_{C_j}^\ell
\end{array}\Big]\\&=&\left ( \begin{array}{ccccc}
\textbf{y}^1(C_1)&  \textbf{y}^2(C_2)&\ldots& \textbf{z}^\ell(C_j)\\
  \textbf{y}^1( \gamma C_1)&\textbf{y}^2(\gamma  C_2)&\ldots&
\textbf{y}^\ell(\gamma C_j)
\\ \vdots &\vdots&\vdots&\vdots \\\textbf{y}^1(\gamma^{n-1}  C_1)& \textbf{y}^2(\gamma^{n-1} C_2)&\ldots& \textbf{y}^\ell(\gamma^{n-1} C_j) \end{array} \right), \nonumber\end{eqnarray}
\end{small}
where we choose  the number $\ell$ of different sub-matrices $H_{C_j}$. The
$n\times (\ell*n)$ matrix $\textbf{H}$ constructed in \textbf{Type-II} has the
following properties.
\begin{compactenum}[i)]

\item Every column has a weight of $m$ and every row has a weight of $m*\ell$, where $\ell$ is
the  number of matrices $H_{C_j}'s$.
\item For a large n, the matrix $\textbf{H}$ is a sparse low-density parity check matrix.
\end{compactenum}
 We can also show that the null space  of the matrix $\textbf{H}$
defines an $(m,m\ell)$ LDPC code with rate $(\ell-1)/\ell$. Clearly, an
increase in $\ell$, increases the rate of the code. 
%

Since all cyclotomic cosets $C_{x_1},C_{x_2},\ldots,C_{x_\ell}$ used to
construct \textbf{H} are different, then  the first column in each sub-matrix
$H_{C_x}^j$ is different from the first column in all sub-matrices  $H_{C_x}^i$
for $j\neq i$ and $1 \leq i \leq \ell$. Now, we can give a lower bound in the
stopping distance of \textbf{Type-II} LDPC codes.

\medskip

\begin{lemma}
The stopping distance of LDPC codes, that are in \textbf{Type-II} construction,
is at least $\ell+1$.
\end{lemma}

 One can improve this bound, by counting the number of columns in each
sub-matrix $H_{C_x}^i$ that do not have one in common in addition to all
columns in the other sub-matrices.

\smallskip

\begin{example}
Consider $n=q^m-1$ with $m=5$, $q=2$, and $\delta=5$.  We can compute the
cyclotomic cosets $C_1$, $C_3$ and $C_5$ as $C_1=\{1,2,4,8,16\},$
$C_3=\{3,6,12,24,17\}$ and $C_5=\{5,10,20,9,18 \}$. The matrices $H_{C_1}^1$,
$H_{C_3}^2$ and $H_{C_5}^3$ can be defined based on $C_1$, $C_3$ and $C_5$,
respectively.
\begin{small}\begin{eqnarray} H_{C_1}^1=\left (
\begin{array}{ccccccccccccccccccccccccccccccccccccccccccc}
\!\!\!  1101 & \!\!\!0001 &\!\!\! 0000 &\!\!\! 0001 &\!\!\! 0000 &\!\!\! 0000 &\!\!\! 0000 &\!\!\! 000 \!\!\! \\
\!\!\!  0110 & \!\!\!1000 &\!\!\! 1000 &\!\!\! 0000 &\!\!\! 1000 &\!\!\! 0000 &\!\!\! 0000 &\!\!\! 000 \!\!\! \\
\!\!\!  0011 & \!\!\!0100 &\!\!\! 0100 &\!\!\! 0000 &\!\!\! 0100 &\!\!\! 0000 &\!\!\! 0000 &\!\!\! 000 \!\!\! \\
\!\!\!  0001 & \!\!\!1010 &\!\!\! 0010 &\!\!\! 0000 &\!\!\! 0010 &\!\!\! 0000 &\!\!\! 0000 &\!\!\! 000 \!\!\! \\
\!\!\!  0000 & \!\!\!1101 &\!\!\! 0001 &\!\!\! 0000 &\!\!\! 0001 &\!\!\! 0000 &\!\!\! 0000 &\!\!\! 000 \!\!\! \\
\!\!\!  \vdots & \!\!\! \vdots &\!\!\! \vdots &\!\!\! \vdots &\!\!\! \vdots &\!\!\! \vdots &\!\!\! \vdots &\!\!\! \vdots \!\!\! \\
\!\!\!  0100 & \!\!\!0100 &\!\!\! 0000 &\!\!\! 0100 &\!\!\! 0000 &\!\!\! 0000 &\!\!\! 0000 &\!\!\! 011 \!\!\! \\
\!\!\!  1010 & \!\!\!0010 &\!\!\! 0000 &\!\!\! 0010 &\!\!\! 0000 &\!\!\! 0000 &\!\!\! 0000 &\!\!\! 001 \!\!\! \\
\end{array} \right)
\end{eqnarray}
\end{small}
The matrix $\textbf{H}$ of size (31,93) is given by \begin{eqnarray} \textbf{H}
= \Big[
\begin{array}{cccc}H_{C_1}^1& H_{C_3}^2&H_{C_5}^3
\end{array}\Big],
\end{eqnarray} therefore, the null space of $\textbf{H}$ defines an (5,15) LDPC
code with parameters $(62,93)$, see Fig.~\ref{fig:ldpc2}.
\end{example}

\smallskip

We note that \textbf{Type-I} and \textbf{Type-II} constructions can
be used to derive quantum codes, if the parity check matrix
\textbf{H} is modified to be self-orthogonal. Recall that quantum
error-correcting codes over $\F_q$ can be constructed from
self-orthogonal classical codes over $\F_q$ and $\F_{q^2}$, see for
example~\cite{aly07a,calderbank98,hagiwara07,macKay04} and
references therein. In our future research, we plan to derive
quantum LDPC codes from \textbf{Type-I} and \textbf{Type-II}
constructions that are based on nonprimitve BCH codes.

\section{Simulation Results}\label{sec:simulation}

We simulated the performance of the constructed codes using standard iterative
decoding algorithms. Fig.~\ref{fig:ldpc1} shows the BER curve for an (4,31)
LDPC code~\textbf{Type I} with a length of 961, dimension of 837, and  number
of iterations  of 50. This performance can also be improved for various lengths
and the designed distance of BCH codes. Fig.~\ref{fig:ldpc2} shows the BER
curve for a (5,15) LDPC  \textbf{Type II} code with a size of (62,93) and
number of iterations 30. The performance of these constructed codes can be
improved for large code length in comparison to other LDPC codes constructed
in~\cite{lin04,liva06}. As shown in Fig.~\ref{fig:ldpc1} at the $10^{-4}$ BER,
the code performs at $5.5$ $Eb/No(dB)$, which is $1.7$ units from the Shannon
limit. Also,  in Fig.\ref{fig:ldpc2}  at the BER of $10^{-4}$, the code
performs at $5.3$ $Eb/No(dB)$.


\begin{figure}[h]
  \includegraphics[scale=0.45]{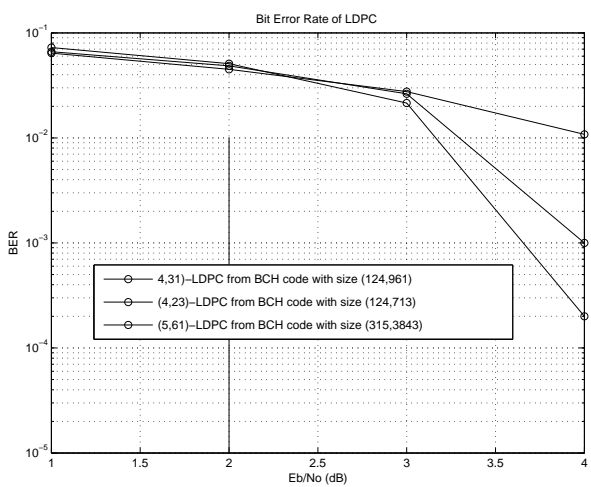}
 \centering
  \caption{\textbf{Type I:} Performance of an (4,31) LDPC code with rate $27/31$ and code size $(837,961)$.}\label{fig:ldpc1}
\end{figure}


\section{Conclusion}\label{sec:conclusion}
We introduced  two families of regular LDPC codes based on
nonprimitive narrow-sense BCH codes and structures of cyclotomic
cosets. We gave a systematic method to writ every element in parity
check matrix of BCH codes as vector of length $\mu$.  We
demonstrated that these constructed codes have high rates and a
uniform structure that made it easy to compute their dimensions,
stopping distance, and bound their minimum distance. Furthermore,
one can use standard iterative decoding algorithms to decode these
codes. we plan to investigate more properties of these codes and
evaluate their performance over different communication channels.
One can easily derive irregular LDPC codes based on these codes and
possibly increase performance of the iterative coding. Also, in a
future research, these constructed codes can be used to derive
quantum LDPC error-correcting codes.

\part{Applications}

\chapter{Asymmetric  Quantum  BCH   Codes}

\noindent {\bf Summary:}
Recently, the theory of quantum error control codes has been extended to
include quantum codes over asymmetric quantum channels --- qubit-flip and
phase-shift errors may have equal or different probabilities.  Previous work in constructing quantum error control codes has focused on code
constructions for symmetric quantum channels. In this chapter we establish a method to construct  asymmetric quantum codes based on classical codes. We
derive families of  asymmetric quantum codes derived, once again, from
classical BCH and RS codes over finite fields. Particularly, we present
interesting asymmetric quantum codes based on BCH codes with parameters
$[[n,k,d_z/d_x]]_q$ for certain values of code lengths, dimensions, and
various minimum distance. Finally, our constructions are well explained by
an illustrative example.


\section{Introduction}\label{sec:intro}

In 1996, Andrew Steane stated in his seminal work~\cite[page 2, col.
2]{steane96}\cite{steane96b,steane99} \emph{''The notation
$\{n,K,d_1,d_2\}$ is here introduced to identify a 'quantum code,' meaning
a code by which n quantum bits can store K bits of quantum information and
allow correction of up to $\lfloor (d_1-1)/2 \rfloor$ amplitude errors,
and simultaneously up to $\lfloor (d_2-1)/2 \rfloor$ phase errors.''} This
work is motivated by this statement, in which we  construct efficient
quantum codes that correct amplitude (qubit-flip) errors and phase-shift
errors separately. In~\cite{macwilliams77}, it was said that \emph{"BCH
codes are among the powerful codes"}. We address constructions of quantum
 codes based on Bose-Chaudhuri-Hocquenghem (BCH) codes over finite fields
for quantum symmetric and asymmetric channels.

Many quantum error control codes (QEC)  have been constructed over
the last decade to protect quantum information against noise and
decoherence. In coding theory, researchers have focused on bounds
and the construction aspects of quantum codes for large and
asymptomatic code lengths. On the other hand, physicists intend to
study the physical realization and mechanical quantum operations of
these codes for short code lengths. As a result, various approaches
to protect quantum information against noise and decoherence are
proposed including stabilizer block codes, quantum convolutional
codes, entangled-assisted quantum error control codes, decoherence
free subspaces, nonadditive codes, and subsystem
codes~\cite{ashikhmin99,calderbank98,forney05b,gottesman97,rains99,lidar98,poulin07,smolin07,zanardi97}
and references therein.

Asymmetric quantum control codes (AQEC), in which quantum errors
have different probabilities --- $\Pr{Z} > \Pr{X}$, are more
efficient than the symmetric quantum error control codes (QEC), in
which quantum errors have  equal probabilities --- $\Pr{Z} =
\Pr{X}$. It is argued in~\cite{ioffe07} that dephasing (loss of
phase coherence, phase-shifting) will happen more  frequently than
relaxation (exchange of energy with the environment,
qubit-flipping).  The noise level in a qubit is specified by the
relaxation $T_1$ and dephasing time $T_2$; furthermore the relation
between these two values is given by $1/T_1=1/(2T_1)+\Gamma_p$; this
has been well explained by physicists
in~\cite{evans07,ioffe07,stephens07}. The ratio between the
probabilities of qubit-flip X and phase-shift Z is typically $\rho
\approx 2T_1/T_2$. The interpretation is that $T_1$ is much larger
than $T_2$, meaning the photons take much more time to flip from the
ground state to the excited state. However, they change rapidly from
one  excited state to another. Motivated by this, \textbf{one needs
to design quantum codes that are suitable for this physical
phenomena.} The fault tolerant operations of a quantum computer
carrying controlled and measured quantum information  over
asymmetric channel have been investigated
in~\cite{aliferis07,bacon06,bacon06b,steane04,stephens07,aliferis07thesis}
and references therein. Fault-tolerant operations of QEC are
investigated for example
in~\cite{aliferis06,aliferis07thesis,gottesman97,preskill98,shor96,steane04,knill04}
and references therein.

Subsystem codes (SSC) as we prefer to call them were mentioned in
the unpublished work by Knill~\cite{knill06,knill96b}, in which he
attempted to generalize the theory of quantum error-correcting codes
into subsystem codes. Such codes with their stabilizer formalism
were reintroduced
recently~\cite{aly06c,bacon06,bacon06b,klappenecker0608,kribs05,poulin05}.
The construction aspects of these codes are given
in~\cite{aly08a,aly08f,aly06c}. Here we expand our understanding and
introduce asymmetric subsystem codes (ASSC).

Our following theorem establishes the connection between two
classical codes  and QEC, AQEC, SCC, ASSC.

\begin{theorem}[CSS AQEC and ASSC]\label{lem:AQEC}
Let $C_1$ and $C_2$ be two classical codes with parameters
$[n,k_1,d_1]_q$ and $[n,k_2,d_2]_q$ respectively, and $d_x=
\min\big\{\wt(C_{1} \backslash C_2^\perp), \wt(C_{2} \backslash C_{1
}^\perp)\big\}$, and $d_z= \max\big\{\wt(C_{1} \backslash
C_2^\perp), \wt(C_{2} \backslash C_{1 }^\perp)\big\}$.
\begin{compactenum}[i)]
\item if
  $C_2^\perp \subseteq C_1$, then there exists an AQEC with parameters $[[n,\dim C_1 -\dim
C_2^\perp,\wt(C_2\backslash C_1^\perp)/\wt(C_1\backslash
C_2^\perp)]]_q$ that is $[[n,k_1+k_2-n,d_z/d_x]]_q$. Also, there
exists a QEC with parameters $[[n,k_1+k_2-n,d_x]]_q$.

\item From [i], there exists an SSC with parameters
$[[n,k_1+k_2-n-r,r,d_x]]_q$ for $0 \leq r <k_1+k_2-n$.

\item If  $C_2^\perp=C_1 \cap C_1^\perp \subseteq C_{2}$, then there exists an ASSC with parameters  $[[n,k_2-k_1,k_1+k_2-n,d_z/d_x]]_q$
and $[[n,k_1+k_2-n,k_2-k_1,d_z/d_x]]_q$.
\end{compactenum} Furthermore, all constructed codes are pure to their minimum
distances.
\end{theorem}

The codes derived in~\cite{aly07a,aly06a} for primitive and
nonprimitive quantum BCH codes assume that qubit-flip errors,
phase-shift errors, and their combination occur with equal
probability, where $\Pr{Z}=\Pr{X}=\Pr{Y}=p/3$, $\Pr{I}=1-p$, and
$\{X,Z,Y,I\}$ are the binary Pauli operators $P$ shown in
Section~\ref{sec:AQEC}, see~\cite{calderbank98,shor95}. We aim to
generalize these codes over asymmetric quantum channels. In this
work we give  families of asymmetric quantum error control codes
(AQEC's) motivated by the work
from~\cite{evans07,ioffe07,stephens07}. Assume we have a classical
good error control code $C_i$ with parameters $[[n,k_i,d_i]]_q$ for
$i\in \{1,2\}$ --- codes with high minimum distances $d_i$ and high
rates $k_i/n$. We can construct a quantum code based on these two
classical codes, in which $C_1$ controls the qubit-flip errors while
$C_2$ takes care of the phase-shift errors, see
Lemma~\ref{lem:AQEC}.

A well-known construction on the  theory of quantum error control
codes is called CSS constructions. The codes $[[5,1,3]]_2$,
$[[7,1,3]]_2$, $[[9,1,3]]_2$, and $[[9,1,4,3]]_2$ have been
investigated in several research papers that analyzed their
stabilizer structure, circuits, and fault tolerant quantum computing
operations. On this work, we present several AQEC codes, including
a $[[15,3,5/3]]_2$ code, which encodes three logical qubits into
$15$ physical qubits, detects $2$  qubit-flip and  $4$ phase-shift
errors, respectively. As a result, many of the quantum constructed
codes and families of QEC for large lengths need further
investigations. We believe that their generalization is a direct
consequence.

\section{Asymmetric Quantum Codes}\label{sec:AQEC}
In this section we shall give some primary definitions and introduce
AQEC constructions.   Consider a quantum system with two-dimensional
state space $\mathcal{C}^2$. The basis vectors
\begin{eqnarray} v_0=\left(
\begin{array}{c} 1
\\ 0 \end{array} \right), \texttt{ } v_1=\left( \begin{array}{c} 0\\
1 \end{array}\right)\end{eqnarray} can be used to represent the
classical bits $0$ and $1$.  It is customary in quantum information
processing to use Dirac's ket notation for the basis vectors;
namely, the vector $v_0$ is denoted by the ket $\ket{0}$ and the
vector $v_1$ is denoted by ket $\ket{1}$. Any possible state of a
two-dimensional quantum system is given by a linear combination of
the form \begin{eqnarray}a
\ket{0}+b\ket{1}\!=\!\left(\begin{array}{c} \!\! a \!\\ \!\! b\!
\end{array}\right)\!,
 \mbox{ where } a, b \in \! \mathcal{C} \mbox{ and }
 |a|^2+|b|^2=\!1,\end{eqnarray}

In quantum information processing, the operations manipulating
quantum bits follow the rules of quantum mechanics, that is, an
operation that is not a measurement must be realized by a unitary
operator.  For example, a quantum bit can be flipped by a quantum
NOT gate $X$ that transfers the qubits $\ket{0}$ and $\ket{1}$ to
$\ket{1}$ and $\ket{0}$, respectively. Thus, this operation acts on
a general quantum state as follows.
$$X(a\ket{0}+b\ket{1})=a \ket{1}+ b \ket{0}.$$ With respect to the
computational basis, the quantum NOT gate  $X$ represents  the
qubit-flip errors.

\begin{eqnarray}
X=\ket{0}\bra{1}+\ket{1}\bra{0}=\left( \begin{array}{cc} 0 &1\\ 1&0\\
\end{array}\right).
\end{eqnarray}

 Also, let $Z=\left(\!\! \begin{array}{cc} 1 &0\\ 0&-1\\
\end{array}\right)$ be a matrix represents
the quantum phase-shift errors that changes  the phase of a quantum
system (states). \begin{eqnarray}Z(a\ket{0}+b\ket{1})=a \ket{0}- b
\ket{1}.\end{eqnarray}
 Other popular operations include  the combined bit and phase-flip $Y=iZX$, and the Hadamard gate
$H$, which are represented with respect to the computational basis
by the matrices

\begin{eqnarray}
Y=\left(
\begin{array}{cc} 0&-i\\ i&0\\ \end{array}\right),
H=\frac{1}{\sqrt{2}}\left(\!\!
\begin{array}{cc} 1&1\\ 1&-1\\ \end{array}\right).
\end{eqnarray}

\bigbreak

\noindent \textbf{Connection to Classical Binary Codes.} Let $H_i$
and $G_i$ be the parity check and generator  matrices of a classical
code $C_i$
 with parameters $[n,k_i,d_i]_2$ for $i \in \{1,2\}$. The
 commutativity condition of $H_1$ and $H_2$ is stated as

\begin{eqnarray}
H_1.H_2^T+H_2.H_1^T=\textbf{0}.
 \end{eqnarray}
The stabilizer of a quantum code based on the parity check matrices
$H_1$ and $H_2$ is given by \begin{eqnarray}H_{stab}=\Big( H_1 \mid
H_2\Big).\end{eqnarray}

One of these two classical codes controls the phase-shift errors,
while the other codes controls the bit-flip errors. Hence the CSS
construction of a binary AQEC can be stated as follows. Hence the
codes $C_1$ and $C_2$ are mapped to $H_x$ and $H_z$, respectively.

\begin{definition}Given two classical binary codes $C_1$ and $C_2$ such that $C_2^\perp
\subseteq C_1$. If we form $ G=\begin{pmatrix}
G_1&0\\0&G_2\end{pmatrix}, \mbox{  and   } H =\begin{pmatrix}
H_1&0\\0&H_2\end{pmatrix}, $ then
\begin{eqnarray}
H_1.H_2^T-H_2.H_1^T=0
\end{eqnarray}
 Let $d_1=\wt(C_1\backslash C_2)$ and $d_2=wt(C_2\backslash
C_1^\perp)$, such that $d_2 >d_1$ and $k_1+k_2>n$. If we assume that
 $C_1$ corrects the qubit-flip errors and $C_2$ corrects the phase-shift errors, then there exists
AQEC with parameters
\begin{eqnarray}
[[n,k_1+k_2-n,d_2/d_1]]_2.
\end{eqnarray}\end{definition}
We can always change the rules of $C_1$ and $C_2$ to adjust the
parameters.
\subsection{Higher Fields and Total Error Groups}
We can briefly discuss the theory in terms of higher
finite fields $\F_q$. Let $\mathcal{H}$ be the Hilbert space
$\mathcal{H}=\C^{q^n}=\C^q \otimes \C^q \otimes ... \otimes \C^q$.
Let $\ket{x}$ be the vectors of orthonormal basis of $\C^q$, where
the labels $x$ are  elements in the finite field $\F_q$. Let $a,b
\in \F_q$,  the unitary operators $X(a)$ and $Z(b)$ in $\C^q$ are
stated as:
\begin{eqnarray}X(a)\ket{x}=\ket{x+a},\qquad
Z(b)\ket{x}=\omega^{\tr(bx)}\ket{x},\end{eqnarray} where
$\omega=\exp(2\pi i/p)$ is a primitive $p$th root of unity and $\tr$
is the trace operation from $\F_q$ to $\F_p$

Let $\mathbf{a}=(a_1,\dots, a_n)\in \F_q^n$ and
$\mathbf{b}=(b_1,\dots, b_n)\in \F_q^n$. Let us denote by
\begin{eqnarray} X(\mathbf{a}) &=& X(a_1)\otimes\, \cdots \,\otimes X(a_n)
\mbox {  and},\nonumber \\ Z(\mathbf{b}) &=& Z(b_1)\otimes\, \cdots \,\otimes
Z(b_n)\end{eqnarray} the tensor products of $n$ error operators.  The sets
\begin{eqnarray}\textbf{E}_x&=&\{X(\mathbf{a})=\bigotimes_{i=1}^n X(a_i)\mid
\mathbf{a} \in \F_q^n, a_i \in \F_q\}, \nonumber \\ \textbf{E}_z&=&\{Z(\mathbf{b})=\bigotimes_{i=1}^n Z(b_i)\mid \mathbf{b} \in
\F_q^n,b_i \in \F_q\}\end{eqnarray} form an error basis on
$\C^{q^n}$. We can define the error group $\mathbf{G}_x$ and
$\mathbf{G}_z$ as follows
\begin{eqnarray} \mathbf{G}_x = \{
\omega^{c}\textbf{E}_x=\omega^{c}X(\mathbf{a})\,|\, \mathbf{a} \in
\F_q^n, c\in \F_p\},\nonumber \\
\mathbf{G}_z = \{\omega^{c}\textbf{E}_z=\omega^{c}Z(\mathbf{b})\,|\,
\mathbf{b} \in \F_q^n, c\in \F_p\}.\end{eqnarray}
 Hence the total error group
 \begin{eqnarray}
 \textbf{G}&=&\big\{\mathbf{G}_x,\mathbf{G}_z\big\}\nonumber \\ &=&\Big\{ \omega^{c}\bigotimes_{i=1}^n X(a_i), \omega^{c}\bigotimes_{i=1}^n Z(b_i) \mid a_i,b_i \in \F_q \Big\}
 \end{eqnarray}

Let us assume that the sets $\mathbf{G}_x$ and  $\mathbf{G}_z$
represent the qubit-flip and  phase-shift errors, respectively.
\medskip

Many constructed quantum codes assume that the quantum errors
resulted from decoherence and noise have equal probabilities,
$\Pr{X}=\Pr{Z}$. This statement as shown by experimental physics is
not true~\cite{stephens07,ioffe07}. This means the qubit-flip and
phase-shift errors happen with different probabilities. Therefore,
it is needed to construct quantum codes that deal with the realistic
quantum noise. We derive families  of asymmetric quantum error
control codes that differentiate between these two kinds of errors,
$\Pr{Z}>\Pr{X}$.

\begin{definition}[AQEC]
A $q$-ary asymmetric quantum code $Q$, denoted by
$[[n,k,d_z/d_x]]_q$, is a $q^k$ dimensional subspace of the Hilbert
space $\mathbb{C}^{q^n}$ and can control all bit-flip errors up to
$\lfloor \frac{d_x-1}{2}\rfloor$ and all phase-flip errors up to
$\lfloor \frac{d_z-1}{2}\rfloor$. The code $Q$ detects $(d_1-1)$
qubit-flip errors as well as detects $(d_1-1)$ phase-shift errors.
\end{definition}

We use different notation from the one given in~\cite{evans07}. The
reason is that we would like to compare $d_z$ and $d_x$ as a factor
$\rho =d_z/d_x$ not as a ratio. Therefore, if $d_z>d_x$, then the
AQEC has a factor great than one. Hence, the  phase-shift errors
affect the quantum system more than qubit-flip errors do.  In our
work, we would like to increase both the  factor $\rho$ and
dimension $k$ of the quantum code.

\bigbreak

 \noindent \textbf{Connection to Classical nonbinary Codes.} Let $C_1$ and $C_2$ be two linear codes over the finite field $\F_q$, and
let $[n,k_1,d_1]_q$ and $[n,k_2,d_2]_q$ be their parameters. For
$i\in \{1,2\}$, if $H_i$  is the parity check matrix of the code
$C_i$, then $\dim{C_i^{\perp}}=n-k_i$ and rank of $H_i^\perp$ is
$k_i$. If $C_{i}^\perp \subseteq C_{1+(i\mod 2)}$, then $C_{1+(i
\mod 2)}^\perp \subseteq C_i$. So, the rows of $H_i$ which form a
basis for $C_i^\perp$ can be extended to form a basis for
$C_{1+(i\mod 2)}$ by adding some vectors. Also, if $g_i(x)$ is the
generator polynomial of a cyclic code $C_i$ then
$k_i=n-deg(g_i(x))$, see~\cite{macwilliams77,huffman03}.

The error groups $\G_x$ and $\G_z$ can be mapped, respectively,  to
two classical codes $C_1$ and $C_2$ in  a similar manner as in QEC.
This connection is well-know, see for
example~\cite{calderbank98,rains99,sarvepalli07a}. Let $C_i$ be a
classical code such that $C_{1+(i\mod 2)}^\perp \subseteq C_i$ for
$i \in \{1,2\}$, then we have a symmetric quantum control code
(AQEC) with parameters $[[n,k_1+ k_2-n,d_z /d_x]]_q$. This can be
illustrated in the following result.

\begin{figure}[t]
  \begin{center}
  \includegraphics[scale=0.65]{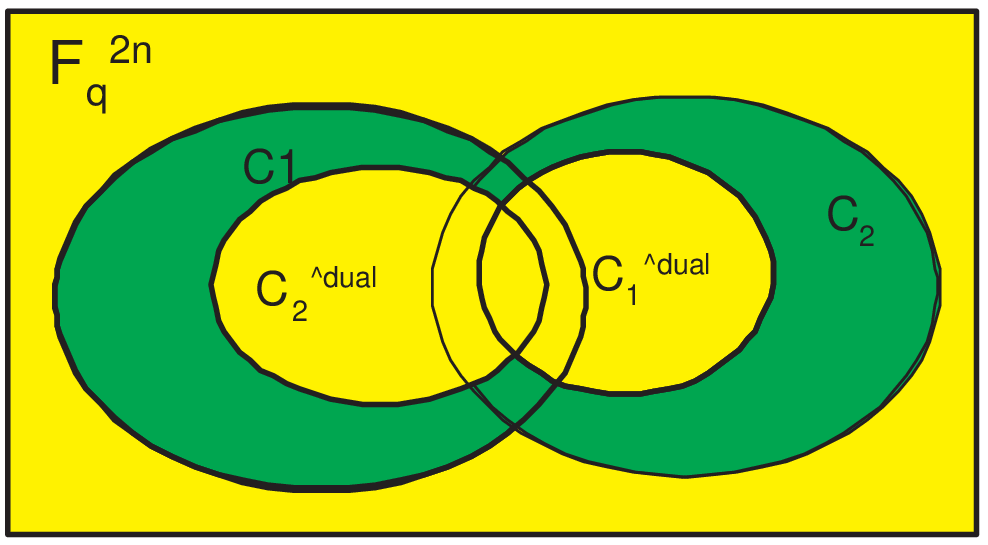}
  \caption{Constructions of asymmetric quantum codes based on two classical codes $C_1$ and $C_2$}
  {with parameters $[n,k_1]$ and $[n,d_2]$ such that $C_i \subseteq C_{1+(i \mod 2)}$ for $i=\{1,2\}$. AQEC has parameters $[[n,k_1+k_2-n,d_z/d_x]]_q$ where
  $d_x=\wt(C_1 \backslash C_2^\perp)$ and $d_z=\wt(C_2 \backslash C_1^\perp)$}\label{fig:subsys1}
  \end{center}
\end{figure}

\begin{lemma}[CSS AQEC]\label{lem:AQEC}
Let $C_i$ be a classical code with parameters $[n,k_i,d_i]_q$ such
that $C_i^\perp \subseteq C_{1+(i\mod 2)}$ for $i \in \{1,2\}$ , and
$d_x= \min\big\{\wt(C_{1} \backslash C_2^\perp), \wt(C_{2}
\backslash C_{1 }^\perp)\big\}$, and $d_z= \max\big\{\wt(C_{1}
\backslash C_2^\perp), \wt(C_{2} \backslash C_{1 }^\perp)\big\}$.
Then there is asymmetric quantum code with parameters
$[[n,k_1+k_2-n,d_z/d_x]]_q$. The quantum code is pure to its minimum
distance meaning that if $\wt(C_1)=\wt(C_1\backslash C_2^\perp)$
then the code is pure to $d_x$, also if $\wt(C_2)=\wt(C_2\backslash
C_1^\perp)$ then the code is pure to $d_z$.
\end{lemma}

Therefore, it is straightforward to derive asymmetric quantum
control codes from two classical codes as shown in
Lemma~\ref{lem:AQEC}. Of course, one wishes to increase the values
of $d_z$ vers. $d_x$ for the same code length and dimension.
\begin{remark}
The notations of purity and impurity of AQEC remain the same as
shown for QEC, the interested reader might consider any primary
papers on QEC.
\end{remark}

\section{Asymmetric Quantum BCH and RS Codes}\label{sec:AQEC-BCH}
In this section we derive classes of AQEC based on classical BCH and
RS codes. We will restrict ourself to the Euclidean construction for
codes defined over $\F_q$. However, the generalization to the
Hermitian construction for codes defined over $\F_{q^2}$ is straight
forward. We keep the definitions of BCH codes to a minimal since
they have been well-known, see  example~\cite{aly07a} or any
textbook on classical coding
theory~\cite{macwilliams77,huffman03}. Let $q$ be a
power of a prime and $n$ a positive integer such that $\gcd(q,n)=1$.
Recall that the cyclotomic coset $S_x$ modulo $n$ is defined as
\begin{eqnarray}S_x=\{xq^i\bmod n \mid i\in \Z, i\ge 0\}.
\end{eqnarray}

Let $m$ be the multiplicative order of $q$ modulo $n$. Let $\alpha$
be a primitive element in $\F_{q^m}$. A nonprimitive narrow-sense
BCH code $C$ of designed distance $\delta$ and length $n$ over
$\F_{q}$ is a cyclic code with a generator monic polynomial $g(x)$
that has $\alpha, \alpha^2, \ldots, \alpha^{\delta-1}$ as zeros,
\begin{eqnarray}
g(x)=\prod_{i=1}^{\delta -1} (x-\alpha^i).
\end{eqnarray}
Thus,  $c$ is a codeword in $\mathcal{C}$ if and only if
$c(\alpha)=c(\alpha^2)=\ldots=c(\alpha^{\delta-1})=0$. The parity
check matrix of this code can be defined as
\begin{eqnarray}\label{bch:parity}
 H_{bch} =\left[ \begin{array}{ccccc}
1 &\alpha &\alpha^2 &\cdots &\alpha^{n-1}\\
1 &\alpha^2 &\alpha^4 &\cdots &\alpha^{2(n-1)}\\
\vdots& \vdots &\vdots &\ddots &\vdots\\
1 &\alpha^{\delta-1} &\alpha^{2(\delta-1)} &\cdots
&\alpha^{(\delta-1)(n-1)}
\end{array}\right].
\end{eqnarray}

In general the dimensions and minimum distances of BCH codes are not
known. However,  lower bounds on these two parameters for such codes
are  given by $d \geq \delta$ and $k \geq n-m(\delta-1)$.
Fortunately, in~\cite{aly07a,aly06a} exact formulas for the
dimensions and minimum distances are given under certain conditions.
The following result shows the dimension of BCH codes.

\medskip

\begin{theorem}[Dimension BCH Codes]\label{th:bchnpdimension}
Let $q$ be a prime power and $\gcd(n,q)=1$, with $ord_n(q)=m$. Then
a narrow-sense BCH code of length $q^{\lfloor m/2\rfloor} <n \leq
q^m-1$ over $\F_q$ with designed distance $\delta$ in the range $2
\leq \delta \le \delta_{\max}= \min\{ \lfloor nq^{\lceil m/2
\rceil}/(q^m-1)\rfloor,n\}$, has dimension of
\begin{equation}\label{eq:npdimension}
k=n-m\lceil (\delta-1)(1-1/q)\rceil.
\end{equation}
\end{theorem}
\begin{proof}
See~\cite[Theorem 10]{aly07a}.
\end{proof}
Steane first derived binary quantum BCH codes
in~\cite{steane96,steane99}. In addition Grassl \emph{el. at.} gave
a family of quantum BCH codes along with tables of best
codes~\cite{grassl99b}.

In~\cite{aly06a,aly07a}, while it was a challenging task to derive
self-orthogonal or dual-containing conditions for BCH codes, we can
relax and omit these conditions by looking for  BCH codes that are
nested. The following result shows a family of QEC derived from
nonprimitive narrow-sense BCH codes.

We can also switch between the code and its dual to construct a
quantum code.  When the BCH codes contain their duals, then we can
derive the following codes.

\begin{theorem}\label{sh:euclid}
Let $m=\ord_n(q)$  and $q^{\lfloor m/2\rfloor} <n \leq q^m-1$ where
$q$ is a power of a prime and $2\le \delta\le \delta_{\max},$ with
$$\delta_{\max}^*=\frac{n}{q^m-1}(q^{\lceil m/2\rceil}-1-(q-2)[m  \textup{ odd}]),$$
then there exists a quantum code with parameters
$$[[n,n-2m\lceil(\delta-1)(1-1/q)\rceil,\ge \delta]]_q$$ pure to $\delta_{\max}+1$
\end{theorem}
\begin{proof}
See~\cite[Theorem 19]{aly07a}.
\end{proof}

\medskip

\subsection{AQEC-BCH}
Fortunately, the mathematical structure of BCH codes always us
easily to show the nested required structure as needed in
Lemma~\ref{lem:AQEC}. We know that $g(x)$ is a generator polynomial
of a narrow sense BCH code that has roots
$\alpha^2,\alpha^3,\ldots,\alpha^{\delta-1}$ over $\F_{q}$. We know
that the generator polynomial has degree $m \lfloor
(\delta-1)(1-1/\delta)\rfloor$ if $\delta \leq \delta_{max}$.
Therefore the dimension is given by $k=n-deg(g(x))$. Hence, the
nested structure of BCH codes is obvious and can be described as
follows. Let
\begin{eqnarray}\delta_{i+1} > \delta_i > \delta_{i-1} \geq \ldots \geq 2,\end{eqnarray}
and let $C_i$ be a BCH code that has generator polynomial $g_i(x)$,
in which it has roots $\{2,3,\ldots,\delta-1\}$. So, $C_i$ has
parameters $[n,n-deg(g_i(x)),d_i\geq \delta_i]_q$, then
\begin{eqnarray}
C_{i+1} \subseteq C_{i} \subseteq C_{i-1} \subseteq \ldots
\end{eqnarray}

We need to ensure that $\delta_i$ and $\delta_{i+1}$ away of each
other, so the elements (roots) $\{2,\ldots,\delta_i-1\}$ and
$\{2,\ldots,\delta_{i+1}-1\}$ are different. This means that the
cyclotomic cosets generated by $\delta_i$ and $\delta_{i+1}$ are not
the same, $S_1\cup \ldots \cup S_{\delta_i-1} \neq S_1\cup \ldots
\cup S_{\delta_{i+1}-1} $. Let $\delta_i^\perp$ be the designed
distance of the code $C_i^\perp$. Then the following result gives a
family of AQEC BCH codes over $\F_q$.

\begin{table}[t]
\caption{Families of asymmetric quantum BCH codes~\cite{magma}}
\label{table:bchtable}
\begin{center}
\begin{tabular}{|c|c|c|c|c|}
\hline   q & $C_1$ BCH Code & $C_2$ BCH Code &AQEC \\
 \hline
 &&&\\
 2&$[15,11,3]$&$[15,7,5]$&$[[15,3,5/3]]_2$\\
 2&$[15,8,4]$&$[15,7,5]$&$[[15,0,5/4]]_2$\\
 2&$[31, 21, 5]$ & $[31, 16, 7]$& $[[31,6, 7/5]]_2$\\
 2&$[31,26,3]$&$[31,16,7]$&$[[31,11,7/3]]$\\
 2&$[31,26,3]$&$[31,16,7]$&$[[31,10,8/3]]$\\
 2&$[31,26,3]$&$[31,11,11]$&$[[31,6,11/3]]$\\
  2&$[31,26,3]$&$[31,6,15]$&$[[31,1,15/3]]$\\
2&$[127,113,5]$&$[127,78,15]$&$[[127,64,15/5]]$\\
2&$[127,106,7]$&$[127,77,27]$&$[[127,56,25/7]]$\\

  \hline
\end{tabular}
\end{center}
\end{table}
%


\begin{theorem}[AQEC-BCH]\label{thm:AQEC-bch}
Let $q$ be a prime power and $\gcd(n,q)=1$, with $ord_n(q)=m$. Let
$C_1$ and $C_2$ be two  narrow-sense BCH codes of length $q^{\lfloor
m/2\rfloor} <n \leq q^m-1$ over $\F_q$ with designed distances
$\delta_1$ and $\delta_2$ in the range $2 \leq \delta_1, \delta_2
 \le \delta_{\max}= \min\{ \lfloor nq^{\lceil m/2
\rceil}/(q^m-1)\rfloor,n\}$ and $\delta_1 <\delta_2^\perp \leq
\delta_2 <\delta_1^\perp$.

Assume $S_1\cup \ldots \cup S_{\delta_1-1} \neq S_1\cup \ldots \cup
S_{\delta_{2}-1}$, then there exists an asymmetric quantum error
control code with parameters $[[n,n-m\lceil
(\delta_1-1)(1-1/q)\rceil-m\lceil (\delta_2-1)(1-1/q)\rceil,\geq
d_z/d_x]]_q$, where $d_z=\wt(C_2 \backslash C_1^\perp) \geq \delta_2
> d_x=\wt(C_1 \backslash C_2^\perp) \geq \delta_1$.
\end{theorem}

\begin{proof}
From the nested structure of BCH codes, we know that if $\delta_1 <
\delta_2^\perp$, then $C_2^\perp \subseteq C_1$, similarly if
$\delta_2 < \delta_1^\perp$, then $C_1^\perp \subseteq C_2$. By
Lemma~\ref{th:bchnpdimension}, using the fact that $\delta \leq
\delta_{\max}$, the dimension of the code $C_i$ is given by
$k_i=n-m\lceil (\delta_i-1)(1-1/q)\rceil$ for $i=\{1,2\}$. Since
$S_1\cup \ldots \cup S_{\delta_1-1} \neq S_1\cup \ldots \cup
S_{\delta_{2}-1}$, this means that $deg(g_1(x)) < deg(g_2(x))$,
hence $k_2 <k_1$. Furthermore $k_1^\perp < k_2^\perp$.

By Lemma~\ref{lem:AQEC} and  we  assume $ d_x=wt(C_1 \backslash
C_2^\perp) \geq \delta_1$ and $ d_z=wt(C_2 \backslash C_1^\perp)
\geq \delta_2$ such that $d_z>d_x$ otherwise we exchange the rules
of $d_z$ and $d_x$; or the code $C_i$ with $C_{1+(i\mod 2)}$.
Therefore, there exists AQEC with parameters $[[n,k_1+k_2-n,\geq
d_z/d_z]]_q$.
\end{proof}

The problem with BCH codes is that we have lower bounds on their
minimum distance given their arbitrary designed distance. We argue
that their minimum distance meets with their designed distance for
small values that are particularly interesting to us. One can also
use the condition shown in~\cite[Corollary 11.]{aly07a} to ensure
that the minimum distance meets the designed distance.

The condition regarding the designed distances $\delta_1$ and
$\delta_2$ allows us to give  formulas for the dimensions of BCH
codes $C_1$ and $C_2$, however, we can derive AQEC-BCH without this
condition as shown in the following result. This is  explained by an
example in the next section.

\begin{lemma}\label{thm:AQEC-bch2}
Let $q$ be a prime power, $\gcd(m,q)=1$, and $q^{\lfloor m/2\rfloor}
<n \leq q^m-1$  for some integers $m=\ord_n(q)$. Let $C_1$ and $C_2$
be two BCH codes with parameters $[n,k_1,d_x \geq \delta_1]_q$ and
$[n,k_2,d_z \geq \delta_2]_q$, respectively, such that $\delta_{1}
<\delta_2^\perp \leq \delta_2 <\delta_1^\perp$, and $k_1+k_2 >n$.
Assume $S_1\cup \ldots \cup S_{\delta_1-1} \neq S_1\cup \ldots \cup
S_{\delta_{2}-1}$, then there exists an asymmetric quantum error
control code with parameters $[[n,k_1+k_2-n, \geq d_z/d_x]]_q$,
where $d_z=\wt(C_1 \backslash C_2^\perp)=\delta_2 >  d_x=\wt(C_2
\backslash C_1^\perp)=\delta_1$.
\end{lemma}
In fact the previous theorem can be used to derive any asymmetric
cyclic quantum control codes. Also, one can construct AQEC based on
codes that are defined over $\F_{q^2}$.

\subsection{RS Codes} We can also derive a family of asymmetric
quantum control codes based on Redd-Solomon codes. Recall that a RS
code with length $n=q-1$ and designed distance $\delta$ over a
finite field $\F_q$ is a code with parameters $[[n,n-d+1
,d=\delta]]_q$ and generator polynomial
\begin{eqnarray}
g(x)=\prod_{i=1}^{d-1}(x-\alpha^i).
\end{eqnarray}
It is much easier to derive conditions for AQEC derived from RS as
shown in the following theorem.

\begin{theorem}
Let $q$ be a prime power and $n=q-1$. Let $C_1$ and $C_2$ be two RS
codes with parameters $[n,n-d_1+1,d_1]]_q$ and $[n,n-d_2+1,d_2]_q$
for $d_1<d_2<d_1^\perp=n-d_1$. Then there exists AQEC code with
parameters $[[n, n-d_1-d_1+2,d_z/d_x]]_q$, where $d_x=d_1 <d_z=d_2$.
\end{theorem}
\begin{proof}
since $d_1<d_2<d_1^\perp$, then $n-d_1^\perp+1<n-d_2+1<n-d_1+1$ and
$k_1^\perp<k_2 < k_1$. Hence $C_2^\perp \subset C_1$ and $C_1^\perp
\subset C_2$. Let $d_z = \wt(C_2\backslash C_1^\perp)= d_2$ and $d_x
= \wt(C_1\backslash C_2^\perp)= d_1$. Therefore there must exist
AQEC with parameters $[[n,n-d_1-d_1+2,d_z/d_x]]_q$.
\end{proof}
It is obvious from this theorem that the constructed code is a pure
code to its minimum distances.  One can also derive asymmetric
quantum RS codes based on RS codes over $\F_{q^2}$. Also,
generalized RS codes can be used to derive similar results. In fact,
one can derive AQEC from any two classical cyclic codes obeying the
pair-nested structure over $\F_q$.

\section{Illustrative Example}\label{sec:example}
We have demonstrated a family  of asymmetric quantum codes with
arbitrary length, dimension, and minimum distance parameters. We
will present a simple example to explain our construction.

Consider a BCH code $C_1$ with parameters $[15,11,3]_2$  that has
designed distance $3$ and generator matrix given by

\begin{eqnarray} \left[\begin{array}{p{0.1cm}p{0.1cm}p{0.1cm}cc ccccc cccccc}
1& 0& 0& 0& 0& 0& 0& 0& 0& 0& 0& 1& 1& 0& 0\\
0& 1& 0& 0& 0& 0& 0& 0& 0& 0& 0& 0& 1& 1& 0\\
0& 0& 1& 0& 0& 0& 0 &0 &0& 0& 0& 0& 0& 1& 1\\0& 0& 0& 1& 0 &0 &0& 0&
0& 0& 0& 1& 1& 0& 1\\0& 0& 0& 0& 1& 0& 0& 0& 0& 0& 0& 1& 0& 1& 0\\0
&0& 0& 0& 0& 1& 0& 0& 0& 0& 0& 0& 1& 0& 1\\0& 0& 0& 0& 0& 0& 1& 0&
0& 0& 0& 1& 1& 1& 0\\0& 0& 0& 0& 0& 0& 0& 1& 0& 0& 0& 0& 1& 1& 1\\0&
0& 0 &0& 0& 0& 0& 0& 1& 0& 0& 1& 1& 1& 1\\0& 0& 0& 0& 0& 0& 0& 0& 0&
1& 0& 1& 0& 1& 1\\0& 0& 0& 0& 0& 0& 0& 0& 0& 0&1 &1 &0 &0& 1
\end{array}\right]
\end{eqnarray}

and the code $C_1^\perp$ has parameters  $[15, 4, 8]_2$ and
generator matrix
\begin{eqnarray} \left[\begin{array}{p{0.1cm}p{0.1cm}p{0.1cm}cc ccccc cccccc}
1 &0& 0 &0 &1& 0 &0& 1& 1& 0& 1& 0& 1 &1 &1\\0& 1& 0& 0& 1 &1 &0& 1&
0& 1& 1& 1& 1& 0 &0\\0& 0& 1&
0& 0 &1& 1& 0& 1 &0 &1 &1& 1 &1 &0\\0 &0& 0 &1 &0 &0 &1& 1 &0 &1& 0& 1& 1& 1& 1\\
\end{array}\right]
\end{eqnarray}

Consider a BCH code $C_2$ with parameters $[15,7,5]_2$ that has
designed distance $5$ and generator matrix given by

\medskip

\begin{eqnarray}\left[\begin{array}{p{0.1cm}p{0.1cm}p{0.1cm}cc ccccc cccccc}
   1 &0& 0& 0& 0& 0& 0& 1& 0& 0& 0& 1& 0 &1 &1\\
  0 &1 &0& 0 &0 &0& 0& 1& 1 &0& 0& 1 &1& 1 &0\\ 0& 0& 1& 0& 0& 0& 0& 0& 1& 1& 0& 0& 1& 1& 1\\
0 &0& 0 &1 &0 &0 &0 &1 &0& 1& 1& 1& 0& 0& 0\\ 0& 0& 0& 0& 1& 0& 0& 0& 1& 0& 1& 1& 1& 0& 0\\
0& 0& 0& 0& 0& 1& 0& 0& 0& 1& 0& 1& 1 &1 &0\\ 0& 0& 0 &0 &0& 0& 1& 0
&0 &0 &1 &0& 1& 1& 1
\end{array}\right]
\end{eqnarray}

and the code $C_2^\perp$ has parameters $[15, 8, 4]_2$ and generator
matrix
\begin{eqnarray} \left[\begin{array}{p{0.1cm}p{0.1cm}p{0.1cm}cc ccccc cccccc}
1 &0 &0& 0& 0& 0& 0& 0& 1& 1& 0& 1& 0& 0& 0\\0& 1 &0 &0 &0 &0 &0 &0&
0& 1& 1& 0& 1& 0& 0\\0& 0& 1& 0& 0& 0& 0& 0& 0& 0& 1& 1& 0& 1& 0\\0&
0& 0& 1& 0& 0& 0& 0& 0& 0& 0& 1& 1& 0& 1\\0& 0& 0& 0& 1& 0& 0& 0& 1&
1& 0& 1& 1& 1& 0\\0& 0& 0& 0& 0& 1& 0& 0& 0& 1& 1& 0& 1& 1& 1\\0& 0&
0& 0& 0& 0& 1& 0& 1& 1&1 &0 &0& 1& 1\\0& 0& 0& 0& 0& 0& 0& 1& 1& 0&
1& 0& 0& 0& 1
\end{array}\right]
\end{eqnarray}

\bigskip

\noindent \textbf{AQEC.} We can consider the code $C_1$ corrects the
bit-flip errors such that $C_2^\perp \subset C_1$. Furthermore,
$C_1^\perp \subset C_2$. Furthermore and $d_x=\wt (C_1 \backslash
C_2^\perp)=3$ and $d_z=\wt (C_2 \backslash C_1^\perp)=5$. Hence, the
quantum code  can detect four phase-shift errors and two bit-flip
errors, in other words, the code can correct two phase-shift errors
and one bit-flip errors. There must exist asymmetric quantum error
control codes (AQEC) with parameters $[[n,k_1+k_2-n,d_z/d_x ]]_2
=[[15,3,5/3]]_2$. We ensure that this quantum code encodes three
qubits into $15$ qubits, and  it might also be easy to design a
fault tolerant circuit for this code similar to $[[9,1,3]]_2$ or
$[[7,1,3]]_2$, but one can use the cyclotomic structure of this
code. We ensure that many other quantum BCH can be constructed using
the approach given in this work that may or may not have better
fault tolerant operations and better threshold values.

\begin{remark}
An $[7,3,4]_2$  BCH code is used to derive Steane's code
$[[7,1,4/3]]_2$.  AQEC might not be interesting for Steane's code
because it can only detect $3$ shift-errors and  $2$ bit-flip
errors, furthermore, the code corrects one bit-flip and one
phase-shift at most. Therefore, one needs to design AQEC with $d_z$
much larger than $d_x$.

One might argue on how to choose the distances  $d_z$ and $d_x$, we
think the answer comes from the physical system point of view. The
time needed to phase-shift errors is much less that the time needed
for qubit-flip errors, hence depending on the factor between them,
one can design AQEC with factor a $d_z/d_x$.
\end{remark}

\section{Conclusion and Discussion}\label{sec:conclusion}
This chapter introduces a new theory of asymmetric quantum codes.  It
establishes a link between asymmetric and symmetric quantum control codes.
Families of AQEC are derived based on RS and BCH codes over finite fields.
Tables of AQEC-BCH and CSS-BCH are shown over $\F_q$.

We pose it as open quantum to study the fault tolerance operations of the
constructed quantum BCH codes in this work.  Some BCH codes are turned
out to be also LDPC codes. Therefore, one can use the same method shown
in~\cite{aly08c} to construct asymmetric quantum LDPC codes.

\scriptsize

\chapter{Asymmetric Quantum Cyclic Codes}

Recently in quantum information processing,  it has been shown  that phase-shift errors occur with high probability than qubit-flip errors, hence phase-shift errors are more
disturbing to quantum information than  qubit-flip
errors. This leads to constructing  asymmetric quantum codes to protect quantum information over asymmetric channels, $\Pr Z \geq \Pr X$. In this chapter we present two generic
methods to derive asymmetric quantum cyclic codes  using the generator
polynomials and defining sets of classical cyclic codes. Consequently, the
methods allow us to construct several families of asymmetric quantum
BCH, RS, and RM codes. Finally, the methods are used to construct families of subsystem codes.

\section{Introduction}
Recently, the theory of quantum error-correcting codes is extended to include construction of such codes over asymmetric quantum channels --- qubit-flip and phase-shift errors may have equal or
different probabilities, $\Pr Z \geq \Pr X$. Asymmetric quantum error control codes (AQEC) are quantum codes
defined over biased quantum channels. Construction of such codes first
appeared in~\cite{evans07,ioffe07,stephens07}. In~\cite{aly08aqec}
two families of AQEC are derived based on classical BCH and RS
codes. The code construction of AQEC is the CSS construction of QEC
based on two classical cyclic codes. For more details on the CSS
constructions of QEC see for
example~\cite{shor95,ashikhmin01,steane96,steane96b,steane97,calderbank98}

There have been several attempts to characterize the noise error
model in quantum information~\cite{Chang00}. In~\cite{steane96} the
CSS construction of a quantum code that corrects the errors
separated was stated. However, the percentage between the qubit-flip
and phase-shift error probabilities was not known for certain
physical realization. Recently, quantum error correction has been
extended over amplitude-damping channels~\cite{fletcher07}.

We expand the construction of  quantum error correction by
designing stabilizer codes that can correct phase-flip and
qubit-flip errors separately.  Assume that the quantum noise
operators occur independently and with different probabilities in
quantum states. Our goal is to adapt the constructed quantum codes
to more realistic noise models based on physical phenomena.

Motivated by their classical counterparts, the asymmetric quantum cyclic codes
that we derive have online simple encoding and decoding circuits
that can be implemented using shift-registers with feedback
connections. Also, their algebraic structure makes it easy to derive
their code parameters. Furthermore, their stabilizer can be defined
easily using generator polynomials of classical cyclic codes, in addition, it
is simple to derive self-orthogonal nested-code conditions for these
cyclic classes of codes.

In this work we construct quantum error-correcting codes that correct
quantum errors that may destroy quantum information with different
probabilities. We derive two generic framework methods that can be
applied to any classical cyclic codes in order to derive asymmetric quantum
cyclic codes. Special cases of our construction are shown
in~\cite{aly08aqec,ioffe07}.

\bigskip

\emph{Notation:} Let $q$ be a power of a prime integer $p$. We
denote by $\F_q$ the finite field with $q$ elements. We define the
Euclidean inner product $\langle x|y\rangle =\sum_{i=1}^nx_iy_i$ and
the Euclidean dual of a code $C\subseteq \F_{q}^n$ as $$C^\perp = \{x\in
\F_{q}^n\mid \langle x|y \rangle=0 \mbox{ for all } y\in C \}.$$ We
also define the Hermitian inner product for vectors $x,y$ in
$\F_{q^2}^n$ as $\langle x|y\rangle_h =\sum_{i=1}^nx_i^qy_i$ and the
Hermitian dual of $C\subseteq \F_{q^2}^n$ as
$$C^\hdual= \{x\in \F_{q^2}^n\mid \langle x|y \rangle_h=0 \mbox{ for all } y\in
C \}.$$
An $[n,k,d]_q$ denotes a classical code $C$ with with length $n$, dimension $k$, and minimum distance $d$ over $\F_q$.  A quantum code $Q$ is denoted by $[[n,k,d]]_q$.

\bigskip

\section{Classical Cyclic Codes}
Cyclic
codes are of greater interest because they have efficient encoding
and decoding algorithms. In addition, they have well-studied
algebraic structure. Let $n$ be a positive integer and $\F_q$ be a
finite field with $q$ elements. A cyclic code $C$ is a principle
ideal of

$$R_n=\F_q[x] / (x^n-1),$$

where $\F_q[x]$ is the ring of polynomials in invariant $x$. Every
cyclic code $C$ is generated by either a generator polynomial $g(x)$
or generator matrix $G$. Furthermore, every cyclic code is a linear
code that has dimension $k=n-deg(g(x))$. Let $c(x)$ be a codeword in
$\F_q^n[x]$ then $c(x)=m(x)g(x)$, where $m(x)$ is the message to be
encoded. Consequently, every codeword can be written uniquely using
a polynomial in $\F_q^n[x]$. Also, a codeword $c$ in $C$ can be
written as $(c_0,c_1,...,c_{n-1}) \in \F_q^n$. A codeword $c(x) \in
\F_q^n[x]$ is in $C$ with defining set $T$ if and only if
$c(\alpha^i)=0$ for all $i \in T$. Every cyclic code generated by a
generator polynomial $g(x)$ has a parity check polynomial
$x^{k}h(1/x)/h(0)$ where $h(x)=(x^n-1)/g(x)$. Clearly, the parity
check polynomial $h(x)$ can be used to define the dual code
$C^\perp$ such that $g(x)h(x) \mod (x^n-1)=0$. Recall that the dual
cyclic code $C^\perp$ is defined by the generator polynomial
$g^\perp(x)=x^{k}h(x^{-1})/h(0)$. Let $\alpha$ be an element in
$\F_q$. Then sometimes, the code is defined by the roots of the
generator polynomial $g(x)$. Let $T$ be the set of roots of $g(x)$,
$T$ is the defining set of $C$, then
$$
g(x) = \prod_{i \in T}(x-\alpha^i).$$ The set $T$ is the union of
cyclotomic cosets modulo $n$ that has $\alpha^i$ as a root.  More
details in cyclic codes can be found
in~\cite{huffman03,macwilliams77}. The following Lemma is needed to
derive cyclic AQEC.

\medskip

\begin{lemma}\label{lem:twocycliccodes}
Let $C_i$ be cyclic codes of length $n$ over $\F_q$ with defining
set $T_i$ for $i=1,2$. Then \begin{compactenum}[i)] \item $C_1 \cap
C_2$ has defining set $T_1 \cup T_2$. \item $C_1+C_2$ has defining
set $T_1 \cap T_2$. \item $C_1 \subseteq C_2$ if and only if $T_2
\subseteq T_1$.
\item $C_i^\perp \subseteq C_{1 + i(\mod 2)}$ if and only if  $C_{1 + i(\mod 2)}^\perp  \subseteq C_i
$.
\end{compactenum}
\end{lemma}

\bigskip

We will provide an analytical method not a computer search method to derive such codes. The
benefit of this method is that it is much easier to derive families
of AQEC. We define the classical cyclic code using the defining
set and generator polynomial~\cite{aly07a}, \cite{huffman03}. The following lemma establishes
conditions when $C_2^\perp \subseteq C_1$.

\medskip

\begin{lemma}
Let $T_{C_i}$ and $g_i(x)$ be the defining set and generator
polynomial of a cyclic code $C_i$ for $i=\{1,2\}$. If one of the
following conditions

\begin{compactenum}[i)]
\item $T_{C_1} \subseteq T_{C_2}$,
\item $g_1(x)$ divides $g_2(x)$,
\item $h_2(x)$ divides $h_1(x)$,
\end{compactenum}
then  $C_2 \subseteq C_1$.
\end{lemma}
\begin{proof}
The proof is straight forward from the definition of the codes $C_1$
and $C_2$ and by using Lemma~\ref{lem:twocycliccodes}.
\end{proof}

The following theorem shows the CSS construction of asymmetric
quantum error control codes over $\F_q$.

\medskip

\begin{theorem}[CSS AQEC]\label{lem:AQEC}
Let $C_1$ and $C_2$ be two classical codes with parameters
$[n,k_1,d_1]_q$ and $[n,k_2,d_2]_q$ respectively, and $d_x=
\min\big\{\wt(C_{1} \backslash C_2^\perp), \wt(C_{2} \backslash C_{1
}^\perp)\big\}$, and $d_z= \max\big\{\wt(C_{1} \backslash
C_2^\perp), \wt(C_{2} \backslash C_{1 }^\perp)\big\}$.
\begin{compactenum}[i)]
\item if
  $C_2^\perp \subseteq C_1$, then there exists an AQEC with parameters $[[n,\dim C_1 -\dim
C_2^\perp,d_z/d_x]]_q$ that is $[[n,k_1+k_2-n,d_z/d_x]]_q$.
\item Also, there
exists a QEC with parameters $[[n,k_1+k_2-n,d_x]]_q$.
\end{compactenum} Furthermore, all constructed codes are pure to their minimum
distances.
\end{theorem}

Therefore, it is straightforward to derive asymmetric quantum
control codes from two classical codes as shown in
Lemma~\ref{lem:AQEC}. Of course, one wishes to increase the values
of $d_z$ vers. $d_x$ for the same code length and dimension.
\bigskip

If the AQEC has minimum distances $d_z$ and $d_x$ with $d_z \geq d_x$, then it can
correct all qubit-flip errors $ \leq \lfloor (d_x-1)/2\rfloor$ and
all phase-shift errors $ \leq \lfloor (d_z-1)/2\rfloor$,
respectively, as shown in the following result.

\medskip

\begin{lemma}
An $[[n,k,d_z/d_x]]_q$ asymmetric quantum code corrects all
qubit-flip errors up to $\lfloor (d_x-1)/2\rfloor$ and all
phase-shift errors up to $\lfloor (d_z-1)/2 \rfloor$.
\end{lemma}

The codes derived in~\cite{aly07a,aly06a} for primitive and
nonprimitive quantum BCH codes assume that qubit-flip errors,
phase-shift errors, and their combination occur with equal
probability, where $\Pr{Z}=\Pr{X}=\Pr{Y}=p/3$, $\Pr{I}=1-p$, and
$\{X,Z,Y,I\}$ are the binary Pauli operators $P$, see~\cite{calderbank98,shor95}. We aim to
generalize these quantum BCH codes over asymmetric quantum channels. Furthermore, we will derive a much larger class of AQEC based on any two cyclic codes. Such codes include RS, RM, and Hamming codes.

\bigskip

\section{Asymmetric Quantum Cyclic Codes}
In this section we will give two
methods to derive asymmetric quantum cyclic codes. One method is
based on the generator polynomial of a cyclic code, while the other
is directly from the defining set of cyclic code.

\subsection{AQEC Based on Generator Polynomials of Cyclic
Codes}  Let $C_1$ be a cyclic code with parameters
$[[n,k,d]]_q$ defined by a generator polynomial $g_1(x)$. Let
$S=\{1,2,\ldots,\delta_1-1\}$, for some integer $\delta_1 <n$, be the set of roots of the polynomial
$g_1(x)$ such that
\begin{eqnarray}g_1(x)=\prod_{i \in S}(x-\alpha^i)\end{eqnarray}

 It is a well-known fact that the dimension of
the code $C_1$ is given by $
k_1=n-\deg(g_1(x))
$
We also know that the dimension of the dual code $C_1^\perp$ is
given by $k_1^\perp=n-k_1=\deg(g_1(x))$.

The idea that we propose is  simple. Let $f(x)=(x^b-1)$ be a
polynomial such that $1 \leq \deg (f(x)) \leq n-k$. We extend the
polynomial $g_1(x)$ to the polynomial $g_2^\perp(x)$ such that

\begin{eqnarray}
g_2^\perp(x)=f(x) g_1(x)
\end{eqnarray}

Now, let $g_2^\perp(x)$ be the generator polynomial of the code
$C_2^\perp$ that has dimension $k_2^\perp=n-deg(f(x)g_1(x)) <k_1$.
From the cyclic structure of the codes $C_1$ and $C_2^\perp$, we can
see that $C_2^\perp<C_1$, therefore $C_1^\perp < C_2$. Let
$d_1=\wt(C_1 \backslash C_2^\perp)$ and $d_2=\wt(C_2 \backslash
C_1^\perp)$ then we have the following theorem. We can also change
the rules of the code $C_1$ and $C_2$ to make sure that $d_2 >d_1$.

\medskip

\begin{theorem}
Let $C_1$ be a cyclic code with parameters $[n,k_1,d_1]_q$ and a
generator polynomial $g_1(x)$. Let $C_2^\perp$ be a cyclic code
defined by the polynomial $f(x)g_1(x)$ such that $b= \deg(f(x)) \geq
1$, then there exists AQEC with parameters
$[[n,2k_1-b-n,d_z/d_x]]_q$, where $d_x=\min \{\wt(C_1 \backslash
C_2^\perp), \wt(C_2 \backslash C_1^\perp)\}$ and $d_z=\max \{\wt(C_1
\backslash C_2^\perp), \wt(C_2 \backslash C_1^\perp)\}$. Furthermore
the code can correct $\lfloor (d_x-1)/2 \rfloor$ qubit-flip errors
and  $\lfloor (d_z-1)/2 \rfloor$ phase-shift errors.
\end{theorem}
\begin{proof}We proceed the proof as follows.
\begin{compactenum}[i)]
\item
We know that the dual code $C_1^\perp$ has dimension
$k_1^\perp=\deg(g_1(x))$. Also, $C_1^\perp$ has a generator
polynomial $h_1(x)=x^{n-k} h_1'(1/x)$ where
$h_1'(x)=(x^n-1)/g_1(x)$. Let $f(x)$ be a nonzero polynomial such
that $f(x)g_1(x)$ defines  a code $C_2^\perp$. Now the code
$C_2^\perp$ has dimension
$k_2^\perp=n-\deg(f(x)g_1(x))=n-(k_1+b)<k_1$.
\item
We notice that the polynomial $g_1(x)$ is a factor of the polynomial
$f(x)g_1(x)$, therefore the code generated by later is a subcode of
the code generated by the former. Then we have $C_2^\perp \subset
C_1$. Hence, the code $C_2^\perp$ has dimension
$k_2^\perp=n-(k_1+b)$.
\item
Also, the code $C_2$ has dimension $k_1+b$ and generator polynomial
given by $g_2(x)=(x^n-1)/(f(x)g_1(x))=h_1(x)/f(x)$. Hence the
$g_2(x)$ is a factor of $h_1(x)$, therefore $C_1^\perp$ is  a
subcode in $C_2$, $C_1^\perp \subseteq C_2$. There exists asymmetric
quantum cyclic code with parameters

\begin{compactenum}
\item
$\dim C_1 - \dim C_2^\perp=k_1-(n-k_1-b).$
\item $d_x=\min \{ \wt(C_2\backslash C_1^\perp),\wt(C_1\backslash
C_2^\perp)\}$ and  $d_z=\max \{ \wt(C_2\backslash
C_1^\perp),\wt(C_1\backslash C_2^\perp)\}$.
\end{compactenum}
\end{compactenum}
\end{proof}

\bigskip

\begin{figure}[t]
  \begin{center}
  \includegraphics[scale=0.65]{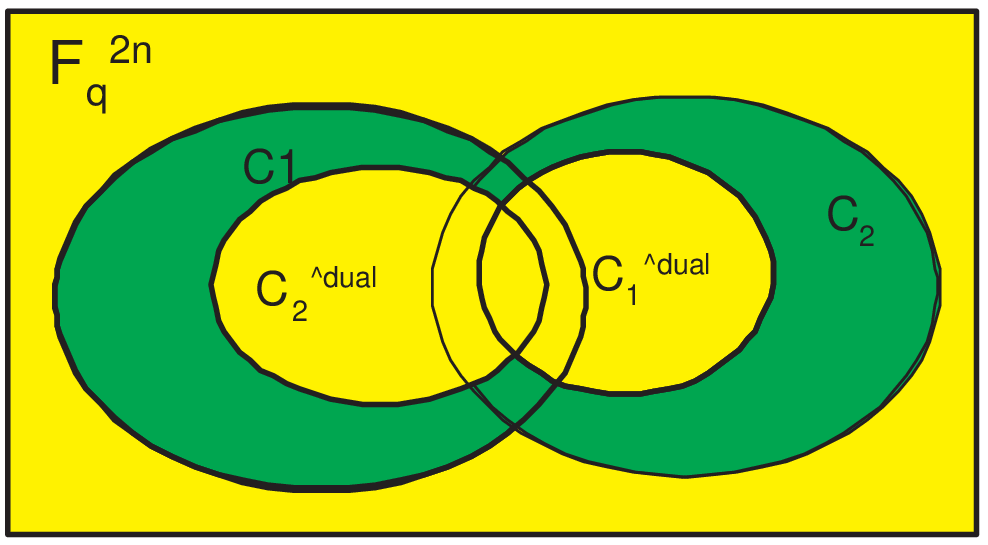}
  \caption{Constructions of asymmetric quantum codes based on two classical cyclic codes}
  { $C_1$ and $C_2$ with parameters $[n,k_1]$ and $[n,d_2]$ such that $C_i \subseteq C_{1+(i \mod 2)}$ for $i=\{1,2\}$. AQEC has parameters $[[n,k_1+k_2-n,d_z/d_x]]_q$ where
  $d_x=\wt(C_1 \backslash C_2^\perp)$ and $d_z=\wt(C_2 \backslash C_1^\perp)$}\label{fig:subsys1}
  \end{center}
\end{figure}

\subsection{Cyclic AQEC Using the Defining Sets Extension}
We can  give a general construction for a cyclic AQEC over $\F_q$ if
the defining sets of the classical cyclic codes are known.

\medskip

\begin{theorem}\label{lem:cyclic-subsysI}
Let $C_1$ be a $k$-dimensional cyclic code of length $n$ over
$\F_q$. Let $T_{C_1}$ and $T_{C_1^\perp}$ respectively denote the
defining sets of $C_1$ and $C_1^\perp$. If $T$ is a subset of
$T_{C_1^\perp} \setminus T_{C_1}$ that is the union of cyclotomic
cosets, then one can define a cyclic code $C_2$ of length $n$ over
$\F_q$ by the defining set $T_{C_2}= T_{C_1^\perp} \setminus (T \cup
T^{-1})$. If $b=|T\cup T^{-1}|$ is in the range $0\le b< 2k-n$ then
there exists asymmetric quantum code with parameters
$$[[n,2k-b-n,d_z/d_x]]_q,$$ where $d_x= \min \{ \wt(C_2 \setminus C_1^\perp),\wt(C_1 \setminus C_2^\perp) \}$
and $d_z= \max \{ \wt(C_2 \setminus C_1^\perp),\wt(C_1 \setminus
C_2^\perp) \}$.
\end{theorem}
\begin{proof}
Observe that if $s$ is an element of the set  $S=
T_{C_1^\perp}\setminus T_{C_1} = T_{C_1^\perp} \setminus (N\setminus
T_{C_1^\perp}^{-1})$, then $-s$ is an element of $S$ as well. In
particular, $T^{-1}$ is a subset of $T_{C_1^\perp}\setminus
T_{C_1}$.

By definition, the cyclic code $C_2$ has the defining set $T_{C_2}=
T_{C_1^\perp} \setminus (T \cup T^{-1})$; thus, the dual code
$C_2^\perp$ has the defining set
$$T_{C_2^\perp}=N\setminus T_{C_2}^{-1} =
T_{C_1}\cup (T\cup T^{-1}).$$

Since $n-k=|T_{C_1}|$ and $b=|T\cup T^{-1}|$, we have $\dim_{\F_q}
C_1=n-|T_{C_1}| = k$ and $\dim_{\F_q} C_2 = n-|T_{C_2}|=k+b$. Thus,
 there exists an
$\F_q$-linear asymmetric quantum code $_Q$ with parameters
$[[n,k_Q,d_z/d_x]]_q$, where
\begin{compactenum}[i)]
\item $k_Q = \dim C_1 -\dim C_2^\perp=k-(n-(k+b))=2k+b-n$,
\item $d_x= \min \{ \wt(C_2 \setminus C_1^\perp),\wt(C_1 \setminus C_2^\perp) \}$
and $d_z= \max \{ \wt(C_2 \setminus C_1^\perp),\wt(C_1 \setminus
C_2^\perp) \}$.
\end{compactenum}
as claimed.
\end{proof}

The usefulness of the previous theorem is that one can directly
derive asymmetric quantum codes from the set of roots (defining set)
of a cyclic code. We also notice that the integer $b$ represents a
size of a cyclotomic coset (set of roots), in other words, it does
not represent one root in $T_{C_1^\perp}$. Table~\ref{table:bchtable2} presents some   AQEC derived from BCH codes

\begin{table}[t]
\caption{Families of asymmetric quantum Cyclic codes}
\label{table:bchtable2}
\begin{center}
\begin{tabular}{|c|c|c|c|c|}
\hline   q & $C_1$ BCH Code & $C_2$ BCH Code &AQEC \\
 \hline
 &&&\\
 2&$[15,11,3]$&$[15,7,5]$&$[[15,3,5/3]]_2$\\
 2&$[15,8,4]$&$[15,7,5]$&$[[15,0,5/4]]_2$\\
 2&$[31, 21, 5]$ & $[31, 16, 7]$& $[[31,6, 7/5]]_2$\\
 2&$[31,26,3]$&$[31,16,7]$&$[[31,11,7/3]]$\\
 2&$[31,26,3]$&$[31,16,7]$&$[[31,10,8/3]]$\\
 2&$[31,26,3]$&$[31,11,11]$&$[[31,6,11/3]]$\\
  2&$[31,26,3]$&$[31,6,15]$&$[[31,1,15/3]]$\\
2&$[127,113,5]$&$[127,78,15]$&$[[127,64,15/5]]$\\
2&$[127,106,7]$&$[127,77,27]$&$[[127,56,25/7]]$\\

  \hline
\end{tabular}
\end{center}
\end{table}
In this section we establish the connection between AQEC and
subsystem codes. Furthermore we derive a larger class of quantum
codes called asymmetric subsystem codes (ASSC). We derive families
of subsystem BCH codes and cyclic subsystem codes over $\F_q$.
In~\cite{aly08a} we construct several families of subsystem
cyclic, BCH, RS and MDS codes over $\F_{q^2}$ with much more details

We expand our understanding of the theory of quantum error control
codes by correcting the quantum errors $X$ and $Z$ separately using
two different classical codes, in addition to correcting only errors
in a small subspace. Subsystem codes are a generalization of the
theory of quantum error control codes, in which errors can be
corrected as well as avoided (isolated).

Let $Q$ be a quantum code  such that $\mathcal{H}=Q\oplus Q^\perp$,
where $Q^\perp$ is the orthogonal complement of $Q$. We can define
the subsystem code $Q=A\otimes B$, see Fig.\ref{fig:subsys1}, as
follows

\medskip

\begin{definition}[Subsystem Codes]
An $[[n,k,r,d]]_q$ subsystem code is a decomposition of the subspace
$Q$ into a tensor product of two vector spaces A and B such that
$Q=A\otimes B$, where $\dim A=q^k$ and $\dim B= q^r$. The code $Q$
is able to detect all errors  of weight less than $d$ on subsystem
$A$.
\end{definition}
Subsystem codes can be constructed  from the classical codes  over
$\F_q$ and $\F_{q^2}$. Such codes do not need the classical codes to
be self-orthogonal (or dual-containing) as shown in the Euclidean
construction. We have given  general constructions of subsystem
codes in~\cite{aly06c} known as the subsystem CSS and Hermitian
Constructions. We provide a proof for the following special case of
the CSS construction.

\medskip

\begin{lemma}[SSC Euclidean Construction]\label{lem:css-Euclidean-subsys}
If $C_1$ is a $k'$-dimensional $\F_q$-linear code of length $n$ that
has a $k''$-dimensional subcode $C_2=C_1\cap C_1^\perp$ and
$k'+k''<n$, then there exist
\begin{eqnarray}
[[n,n-(k'+k''),k'-k'',\wt(C_2^\perp\setminus C_1)]]_q,  \nonumber \end{eqnarray}
\begin{eqnarray}[[n,k'-k'',n-(k'+k''),\wt(C_2^\perp\setminus C_1)]]_q \nonumber\end{eqnarray}
subsystem codes.
\end{lemma}
\begin{proof}
Let us define the code $X=C_1\times C_1 \subseteq \F_q^{2n}$,
therefore $X^\sdual=(C_1\times C_1)^\sdual=C_1^\sdual\times
C_1^\sdual$. Hence $Y=X \cap X^\sdual=(C_1\times C_1)\cap
(C_1^\sdual\times C_1^\sdual)= C_2 \times C_2$. Thus, $\dim_{\F_q}
Y=2k''$. Hence $|X||Y|=q^{2(k'+k'')}$ and $|X|/|Y|=q^{2(k'-k'')}$.
By Theorem~\cite[Theorem 1]{aly06c}, there exists a subsystem code
$Q=A\otimes B$ with parameters $[[n,\log_q\dim A,\log_q\dim B,
d]]_q$ such that
\begin{compactenum}[i)]
\item $\dim A=q^n/(|X||Y|)^{1/2}=q^{n-k'-k''}$.
\item $\dim B=(|X|/|Y|)^{1/2}=q^{k'-k''}$.
\item $d= \swt(Y^\sdual \backslash X)=\wt(C_2^\perp\setminus C_1)$.
\end{compactenum}
Exchanging the rules of the codes $C_1$ and $C_1^\perp$ gives us the
other subsystem code with the given parameters.
\end{proof}

Subsystem codes (SCC) require the code $C_2$ to be self-orthogonal,
$C_2 \subseteq C_2^\perp$. AQEC and SSC are both can be constructed
from the pair-nested classical codes, as we call them. From this
result, we can see that any two classical codes $C_1$ and $C_2$ such
that $C_2=C_1 \cap C_1^\perp \subseteq C_2^\perp$, in which they can
be used to construct a subsystem code (SSC), can be also used to
construct asymmetric quantum code (AQEC). Asymmetric subsystem codes
(ASSC) are much larger class than  the class of symmetric subsystem
codes, in which the quantum errors occur with different
probabilities in the former one and have equal probabilities in the
later one. In short, AQEC does  not require the intersection
code to be self-orthogonal.

The construction in Lemma~\ref{lem:css-Euclidean-subsys} can be
generalized to ASSC CSS construction in a similar way. This means
that we can look at an AQEC with parameters $[[n,k,d_z/d_x]]_q$. as
subsystem code with parameters $[[n,k,0,d_z/d_x]]_q$. Therefore all
results shown in~\cite{aly08a,aly06c} are a direct
consequence by just fixing the minimum distance condition.

We have shown in~\cite{aly08a} that All stabilizer codes
(pure and impure) can be reduced to subsystem codes as shown in the
following result.

\medskip

\begin{theorem}[Trading Dimensions of SSC and Co-SCC]\label{th:FqshrinkK}
Let $q$ be a power of a prime~$p$. If there exists an $\F_q$-linear
$[[n,k,r,d]]_q$ subsystem code (stabilizer code if $r=0$) with $k>1$
that is pure to $d'$, then there exists an $\F_q$-linear
$[[n,k-1,r+1,\geq d]]_q$ subsystem code that is pure to
$\min\{d,d'\}$.  If a pure ($\F_q$-linear) $[[n,k,r,d]]_q$ subsystem
code exists, then a pure ($\F_q$-linear) $[[n,k+r,d]]_q$ stabilizer
code exists.
\end{theorem}

\section{AQEC Based on Two Cyclic Codes}

In this section we can also derive asymmetric quantum codes based on
two cyclic codes and their intersections. We do not necessarily
assume that the code $C_1$ is an extension of the code $C_2^\perp$.
However, we assume that $C_2^\perp \subset C_1$. The benefit of
designing AQEC based on two different classical codes is that we
guarantee the minimum distance $d_z$ to be large in comparison to
$d_x$. In this case we can assume that $C_1$ is a binary BCH code
with small minimum distance, while $C_2$ is an LDPC code with large
minimum distance.

The only requirement one needs to satisfy is that $C_{i} \subseteq
C_{1+i(\mod 2)}$. There have been many families that satisfy this
condition. For example $(15,7)$ BCH code turns out to be an LDPC
code. We will show an example to illustrate our theory.

The following two examples  illustrate the previous constructions.
\begin{exampleX}
Let $C_1$ be the Hamming code with parameters $[n,k,3]_2$ where
$n-2^m-1$ and $k=2^m-m-1$. Consider $C_2$ be a BCH code with
parameters $n$ and designed distance $\delta \geq 5$. Clearly the
$d_z=\wt(C_2) > d_x=\wt(C_1)=3$. Let $k_2$ be the dimension of
$C_2$, then one can derive asymmetric quantum code with parameters
$[[n,k_1+k_2-n,d_z/3]]_q$. In fact, one can short the columns of the
parity check matrix of the Hamming code $C_1$ to obtain a cyclic
code with less dimension and large minimum distance, in which it can
be used as $C_2$.
\end{exampleX}
\bigskip

\begin{exampleX}
Let $F_{13}$ be the finite field with $q=13$ elements. Let $C_1$ be
the narrow-sense Reed-Solomon code of length $n=12$ and designed
distance $\delta=5$ over $F_{13}$. So, $C_1$ has defining set
$T_{C_1}=\{1,2,3,4 \}$. Therefore, $C_1$  is  an MDS code with
parameters $[12,8,5]$. The dual of $C_1$ is a RS code $C_1$ with
defining set $T_{C_1^\perp}=\{ 0,1,2,3,4,5,6,7 \}$. Also,
$C_1^\perp$ is an MDS code with parameters $[12,4,9]$.

Now, let us define the code $C_2$ by choosing a defining set
$T_{C_2}=\{1,2,3,4,7 \}$. So, $C_1^\perp \subseteq C_2
\Longleftrightarrow T_{C_2} \subset T_{C_1^\perp}$. Also compute the
defining set of $C_2$ as $T_{C_2}=\{ 0,1,2,3,4,6,7 \}$. So,
$C_1^\perp \subset C_2 \Longleftrightarrow T_{C_2} \subset
T_{C_1^\perp}$.  Hence, we can compute the parameters of the
asymmetric quantum error-correcting codes as follows. The minimum
distance is given by $d_{min}=C_1^\perp \backslash C =5$, dimension
$k=dim(C_1)-dim(C)=8-7=1$, and gauge qubits
$r=dim(C)-dim(C_1^\perp)=7-4=3$. Therefore, we have a subsystem code
with parameters $[[12,1,3,5]]$, which is also an MDS code obeying
Singleton bound $k+r+2d=n+2$.
\end{exampleX}

\bigskip

\section{Conclusion and Discussion}
We presented two  generic methods to derive asymmetric quantum error
control codes based on two classical cyclic codes over finite fields. We showed that one
can always start by a cyclic code with arbitrary dimension and minimum distance,  and will be able to derive AQEC
using the CSS construction. The method is also used to derive a family of subsystem codes.

\medskip

Based on the generic methods that we develop, all classical cyclic codes can be used to construct asymmetric quantum cyclic codes and subsystem codes. In a quantum computer that utilizes asymmetric quantum cyclic codes to protection quantum information, such codes are superior in  a sense that online encoding and decoding circuits will be used. In addition quantum shirt registers can be implemented.  Our future will include bounds on the minimum distance and dimension of such codes. Furthermore such work will include the best optimal and perfect asymmetric quantum codes.

\bigskip

Such asymmetric quantum error
control codes aim to correct the  phase-shift
errors that occur more frequently than qubit-flip errors. An attempt to address the fault tolerant operations and quantum circuits of such codes are given in~\cite{stephens07}, where an analysis for Becan-shor asymmetric subsystem code is analyzed and a fault-tolerant circuit is given. 

\newcommand{\XXstud}{{}}
\newcommand{\XXar}[1]{}

\end{document}